\documentclass[prx,
twocolumn,
superscriptaddress,nofootinbib,longbibliography]{revtex4-2}

\usepackage[pdftex]{graphicx}
\usepackage{amsmath,amsfonts,amsbsy,amssymb,amsthm,mathtools}
\usepackage{mathrsfs}
\usepackage{epstopdf}
\usepackage{color}
\usepackage{enumitem}
\usepackage{siunitx}
\usepackage[normalem]{ulem}
\usepackage{MnSymbol}
\usepackage{environ}
\usepackage[dvipsnames]{xcolor}
\usepackage{stmaryrd,bbm,dsfont}
\usepackage{subcaption}
\usepackage{ragged2e}
\usepackage{tikz}
\usetikzlibrary{quantikz2}
\usetikzlibrary{automata,arrows,positioning,calc}
\usepackage{tikz-cd}

\usepackage{pgfplots}
\usepackage{pstool}
\usepackage{tikzsymbols}
\usetikzlibrary{plotmarks,spy}
\usepackage{titletoc}
\titlecontents{section}
  [0em]
  {\vspace{0.3em}}
  {\makebox[1em][l]{\thecontentslabel.}}
  {}
  {\titlerule*[0.5pc]{.}\contentspage}
  
\titlecontents{subsection}
  [1em]
  {}
  {\makebox[1em][l]{\thecontentslabel.}}
  {}
  {\titlerule*[0.5pc]{.}\contentspage}

\titlecontents{subsubsection}
  [2em]
  {}
  {\makebox[1em][l]{\thecontentslabel.}}
  {}
  {\titlerule*[0.5pc]{.}\contentspage}

\newtheorem{theorem}{Theorem}
\newtheorem{proposition}[theorem]{Proposition}
\newtheorem{lemma}[theorem]{Lemma}

\newtheorem*{theoreminf*}{Theorem 1 (informal)}

\theoremstyle{definition}
\newtheorem{definition}[theorem]{Definition}

\usepackage[colorlinks = true, citecolor=blue, linkcolor=blue, urlcolor=blue]{hyperref}
\usepackage[T1]{fontenc}
\usepackage{mwe}

\renewcommand{\ll}{\llangle}
\newcommand{\rr}{\rrangle}
\newcommand{\ketbra}[1]{|{#1}\>\mkern-4mu\<{#1}|}
\newcommand{\kketbra}[1]{|{#1}\rr\mkern-4mu\ll{#1}|}
\newcommand{\tr}{\textup{Tr}}
\renewcommand{\>}{\rangle}
\newcommand{\<}{\langle}
\newcommand{\C}{{\mathbb{C}}} %
\newcommand{\R}{{\mathbb{R}}} %

\newcommand{\Z}{{\mathbb{Z}}}

\newcommand{\cov}{{\mathtt{cov}}}

\newcommand{\mse}{{\mathtt{MSE}}}

\newcommand{\seq}{{\mathsf{seq}}}

\newcommand{\D}{{\mathsf{D\,}}}
\renewcommand{\L}{{\mathsf{L\,}}}

\newcommand{\sdp}{{\mathtt{SDP}}}
\newcommand{\supp}{{\mathtt{supp}}}

\newcommand{\outs}{{\textrm{out}}}
\newcommand{\opt}{{\textrm{opt}}}
\renewcommand{\det}{{\textrm{det}}}

\newcommand{\invgate}[1]{ \gate[1,style={blue!0,fill=blue!0,fill opacity=0,opacity=0}]{#1} }

\newcommand{\bshade}[3]{{ \gategroup[#1,steps=#2,style={blue!10,rounded corners,fill=blue!10, inner xsep=1pt},background,label style={label position=above,anchor=north,yshift=0.3cm}]{{#3}} }}

\newcommand{\gshade}[3]{{ \gategroup[#1,steps=#2,style={green!10,fill=green!10, inner xsep=1pt},background,label style={label position=below,anchor=north,yshift=-0.2cm}]{{#3}} }}
\newcommand{\wshade}[3]{{ \gategroup[#1,steps=#2,style={green!0,fill=green!0, inner xsep=1pt},background,label style={label position=below,anchor=north,yshift=-0.2cm}]{{#3}} }}

\newcommand{\CQT}{Centre for Quantum Technologies, National University of Singapore, 3 Science Drive 2, Singapore 117543.\looseness=-1}
\newcommand{\NTU}{Nanyang Quantum Hub, School of Physical and Mathematical Sciences, Nanyang Technological University, Singapore 639673.\looseness=-1}
\def\QUICS{Joint Center for Quantum Information and Computer Science, NIST/University of Maryland, College Park, Maryland 20742, USA}
\def\JQI{Joint Quantum Institute, NIST/University of Maryland, College Park, Maryland 20742, USA}
\def\UMIACS{Department of Computer Science and Institute for Advanced Computer Studies, University of Maryland, College Park, Maryland 20742, USA}

\definecolor{KB}{rgb}{0.4,0.3,0.9}

\definecolor{darkred}{RGB}{0,0,0}
\newenvironment{chgenv}{ \color{darkred} }{%
    \par%
}
\newcommand{\chg}[1]{{\textcolor{darkred}{#1}}}
\newenvironment{atenv}{\begin{chgenv}}{\end{chgenv}}
\newcommand{\atchg}[1]{\chg{#1}}

\begin{document}

\normalem
\newlength\figHeight 
\newlength\figWidth

\title{Optimal noisy sequential multiparameter quantum sensing}

\author{Andrew Tanggara}
\email{andrew.tanggara@gmail.com}
\affiliation{\CQT}
\affiliation{\NTU}
\author{Lorc{\'a}n O Conlon}
\affiliation{\JQI}
\affiliation{\QUICS}
\author{Alexey~V.~Gorshkov}
\affiliation{\JQI}
\affiliation{\QUICS}
\author{Kishor Bharti}
\altaffiliation{Currently at IonQ Inc.}
\affiliation{\QUICS}
\affiliation{\UMIACS}

\date{\today}

\begin{abstract}
The most general strategy for estimating parameters of a quantum channel is to probe the channel many times sequentially with the help of coherent ancillary systems.
However, finding the strategy that minimizes the mean-squared error (MSE) is a notoriously hard problem, requiring optimization over the initial probe state, intermediate evolutions, final measurement, and classical postprocessing.
In this work, we study how the minimum attainable MSE can be bounded from below in this setting. 
We propose a lower bound for the MSE of any physically implementable sequential sensing strategy that applies to estimating parameters of any physical quantum evolution, including quantum channels probed over multiple time steps under temporally correlated noise and general temporally correlated quantum processes.
We show that this lower bound can be formulated as a semidefinite program (SDP).
Our bound generalizes the Nagaoka-Hayashi lower bound (that applies to parameter estimation tasks with separable measurements) to the regime where the quantum channel can be accessed sequentially multiple times.
For regular finite-dimensional single-parameter sequential estimation tasks admitting an optimal strategy, we show that our SDP bound is tight
and provide a procedure to obtain a locally optimal strategy from the SDP optimizer.
Using our optimization bound, we compute the MSE lower bound and derive new protocols for several other physically motivated problems including error-corrected quantum sensing under an unknown correlated noise model, multiparameter estimation, and estimation with qudit systems.
\end{abstract}

\maketitle

\section{Introduction}\label{sec:background}

Quantum-enhanced sensing involves the application of uniquely quantum mechanical resources to 
estimate
properties of a system better than is possible using only classical resources~\cite{giovannetti2011advances}. Quantum-enhanced sensing is one of the most advanced quantum technologies, with practical advantages already demonstrated in gravitational wave detection~\cite{aasi2013enhanced,tse2019quantum}, biosensing~\cite{casacio2021quantum}, atomic clocks~\cite{ludlow2015optical}, gravimetry~\cite{wu2019gravity}, and magnetometry~\cite{troullinou2023quantum}. Several other proof-of-principle experiments show the potential for further enhancements using entangled quantum states and non-separable measurements~\cite{marciniak2022optimal,conlon2023approaching,conlon2023discriminating,nielsen2023deterministic,daryanoosh2018experimental,hou2018deterministic,zhou2025experimental,yi2025optimal,mansouri2026experimental,li2023optimal,yung2025saturating,yung2026beating}. With this potential improvement, there has been great interest in determining optimal protocols for quantum-enhanced sensing. 

Any quantum sensing protocol can be divided into three stages: state preparation, quantum control, and measurement, as shown in Fig.~\ref{fig:quantum_sensing}.\footnote{
    In the figure, the classical post-processing of the measurement outcomes to construct the estimator $\hat{\theta}$ is implied.
} Each of these stages is important in its own right and has been extensively studied. 
Given access to $N$ copies of a unitary $U(\theta)$ that encodes information about an unknown parameter $\theta$ onto a quantum state of interest, one of the main aims of quantum sensing is to achieve the Heisenberg limit, $\mse_\theta\propto1/N^2$, where $\mse_\theta$ is the mean squared error (MSE) with which $\theta$ can be estimated. 
For parallel sensing, where $N$ copies of $U(\theta)$ are accessed simultaneously, it is known that this limit can be achieved by optimal state preparation, using an entangled state~\cite{bollinger1996optimal}. 
However, for sequential sensing where $N$ copies of $U(\theta)$ are accessed sequentially, optimal state preparation alone is not always sufficient to achieve the ultimate limits in quantum sensing.
For example, it is well established that the presence of noise can rapidly degrade sensitivity away from the Heisenberg limit~\cite{Huelga1997,dorner2009optimal,Escher2011,demkowicz2012elusive}.

\begin{figure}
    \centering
    \begin{subfigure}[c]{0.17\columnwidth}
        \centering
        \caption{}
        \includegraphics[width=\linewidth]{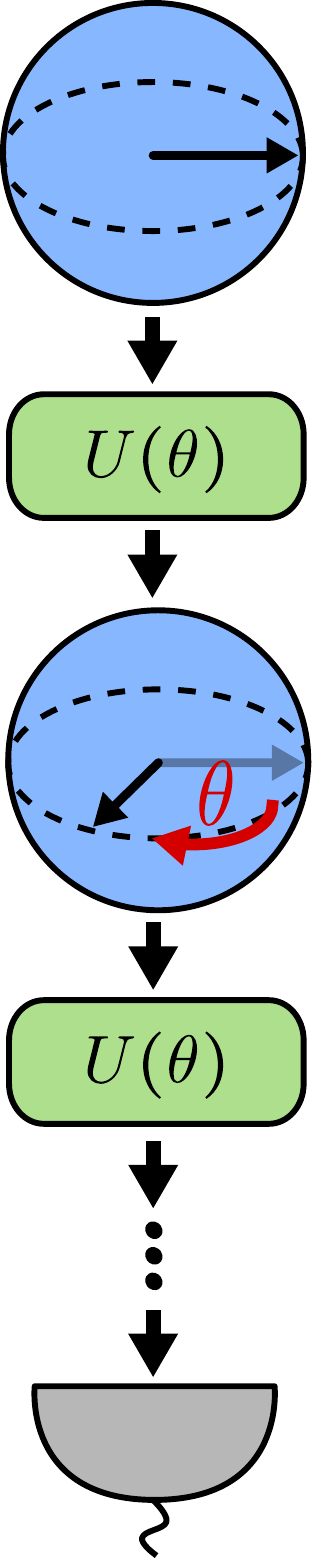}
    \end{subfigure}
    \hspace{3em}
    \begin{subfigure}[c]{0.17\columnwidth}
        \centering
        \caption{}
        \includegraphics[width=\linewidth]{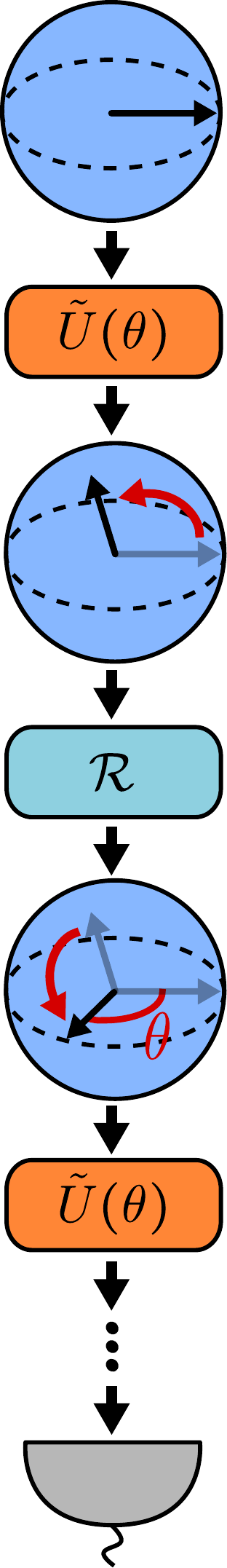}
    \end{subfigure}
    \caption{(a) Sequential quantum sensing for estimating the parameter $\theta$ of parameterized unitary $U(\theta)$ which applies a $\theta$-angle rotation on the probe state.
    (b) Sequential sensing for noisy parameterized unitary $\Tilde{U}(\theta)$ where the noise causes a rotation that erroneously deviates from the noiseless case of $U(\theta)$.
    A control operation $\mathcal{R}$ detects and corrects the error to recover the correct $\theta$ rotation.}
    \label{fig:quantum_sensing}
\end{figure}

In general, in the presence of noise, optimal quantum control is required to restore the Heisenberg limit. 
The restoration of the Heisenberg limit in the presence of noise via quantum control was first observed for specific noise models~\cite{ozeri2013heisenberg,kessler2014quantum,Arrad2014,dur2014improved}, before being developed into a more general theory~\cite{sekatski2017quantum,demkowicz2017adaptive,zhou2018achieving,zhou2021asymptotic}. 
This more general theory, known as error-corrected quantum sensing, presents general geometrical conditions under which the Heisenberg limit can be achieved. 
Note that optimal quantum control can offer enhanced precision even when the Heisenberg limit is not attainable~\cite{chiribella2012optimal,zhou2021asymptotic}, and when noise is not present~\cite{yuan2015optimal,yuan2016sequential}, further emphasizing its importance.

Finally, an optimal measurement is required to extract information about the unknown parameters. For estimating only a single parameter $\theta$, it is known that the quantum Fisher information (QFI) $\mathcal{J}_\text{Q}$ sets the ultimate attainable precision, $\mse_\theta \geq 1/\mathcal{J}_\text{Q}$~\cite{helstrom1967minimum}.\footnote{
    We use the notation $\mathcal{J}_\text{Q}$ for the QFI as there are multiple ways of defining a quantum version of the Fisher information, and here we focus on the QFI based on the symmetric logarithmic derivative.
} 
Here, $\mse_{\theta}$ is the mean-squared error (MSE) of the parameter $\theta$, which is the expectation value of the squared difference between the estimator $\hat{\theta}$ and the true parameter $\theta$, i.e. $\mathbb{E}[(\hat{\theta}-\theta)^2]$.
Furthermore, when estimating only a single parameter, the corresponding measurement saturating the QFI is known~\cite{barndorff2000fisher,hayashi2005statistical}. However, for estimating multiple parameters $(\theta_1,\theta_2,\dots,\theta_m)$, it is considerably more difficult to determine the optimal measurement, and the Cramér-Rao bound obtained from the inverse of the QFI matrix is not always attainable~\cite{conlon2025role,ragy2016compatibility}. The ultimate limit on precision is only known in some special cases~\cite{helstrom1968minimum,young1975asymptotically,belavkin2004generalized,holevo2011probabilistic,matsumoto2002new,bradshaw2017tight,bradshaw2018ultimate,yung2025most}, and in general it is very complex to determine the optimal measurement~\cite{kimizu2024adaptive,zhang2024qestoptpovm}. This is true even when the input state and control operations are fixed such that a fixed output state $\rho_\theta=\rho_\text{out}(\theta_1,...,\theta_m)$ is obtained prior to the final measurement. 
As such it is common to turn to Cramér-Rao bounds, represented here by $\mathcal{C}$, that are efficiently computable, but not necessarily tight, bounds that satisfy $\sum_i\mse_{\theta_i}\geq\mathcal{C}$, where $\sum_i\mse_{\theta_i}$ is the MSE attainable by any measurement on $\rho_\theta$. 
Here, $\mse_{\theta_i}=\mathbb{E}[(\hat{\theta}_i-\theta_i)^2]$ is the MSE of the $i$-th parameter.
A commonly used bound is the Holevo Cramér-Rao bound~\cite{holevo2011probabilistic} (also referred to as the Holevo--Nagaoka Cramér-Rao bound~\cite{nagaoka2005new}). 
However, for estimating many parameters with mixed states, there are several reasons not to use this bound~\cite{das2025holevo,conlon2022gap}.
As such, in this work, we focus on the Nagaoka--Hayashi Cramér-Rao bound (NHB), which is a tighter bound~\cite{conlon2021efficient}. 
The NHB is efficiently computable via a semidefinite program (SDP) and offers a good approximation to the ultimate attainable precision.

For single-parameter estimation, there has been significant progress in simultaneously optimizing the above three steps, determining optimal quantum sensing strategies~\cite{chiribella2012optimal,liu2023optimal,kurdzialek2025quantum,liu2022optimal,liu2024fully}. 
However, for multiparameter estimation, the situation is considerably more complex. 
For example, in the state preparation step, the same probe state may not be optimal for every parameter, a phenomenon known as probe incompatibility~\cite{albarelli2022probe}. 
Similarly, as highlighted above, at the measurement stage, the optimal measurements for each parameter may be incompatible~\cite{conlon2025role}, introducing fundamental trade-offs~\cite{lu2021incorporating,yung2024comparison}. 
Recent work by Hayashi and Ouyang has considered optimizing both the initial probe state and the measurement, but does not handle time-dependent control~\cite{Hayashi_2024}.

In this work, we present a systematic and efficiently computable bound in the system size and number of parameters for optimal quantum multiparameter estimation over $T$ time steps, although the computation grows exponentially with $T$. 
Our bound optimizes over all three stages of a quantum sensing protocol. 
This is achieved by deriving an extension of the NHB via the quantum comb formalism, which represents a sequential sensing protocol over multiple steps as a matrix, computable via an SDP. 
Our SDP bound is applicable to any physical sensing task, including those which allow probing a noisy quantum channel multiple times with temporally correlated noise, or signals of temporally correlated quantum processes more generally.
For single-parameter sensing tasks which optimal protocol exists, we show that our bound is tight.
We also derive a sufficient condition for a certain family of sequential sensing tasks under correlated noise to achieve the Heisenberg scaling, using calibration schemes to learn the noise parameter.
As an application of our analytical results, we compute the SDP bound for several physically-motivated sequential sensing tasks, including sequential quantum channel parameter estimation under various correlated-noise model and sequential multiparameter sensing tasks.

This paper is organized as follows.
We discuss quantum multiparameter estimation in Section~\ref{sec:optimal_multiparameter_estimation}, review existing Cramér-Rao bounds in Section~\ref{sec:existing_bounds}, discuss time-dependent sensing in Section~\ref{subsec:timedependent}, and explain the formalism for sequential sensing protocols in Section~\ref{sec:sequential_strategies_combs}.
In Section~\ref{sec:sequential_bound}, we present the main result on the efficiently computable SDP bound for general sequential sensing and present a proof sketch.
\atchg{In Section~\ref{sec:MSE_optimality_conditions}, we give conditions for determining when the sequential MSE bound is attainable and when Heisenberg scaling can be achieved.
In particular, for every regular finite-dimensional single-parameter sensing problem admitting an optimal tester, we show that the SDP bound is locally attainable and that its optimizer can be used to construct an optimal strategy.
We give a sufficient criterion for Heisenberg scaling in temporally correlated models governed by a finite
latent noise label.}
In Section~\ref{sec:examples}, we numerically apply our SDP bound to multiple examples, mainly on sensing problems under temporally correlated noise and compare SDP lower bounds for one-copy-per-step and two-copy-per-step control in noisy multiparameter sensing.
Additional examples that recover and extend known results are also presented in Appendix~\ref{app:additional_examples}.
\atchg{For the single-parameter examples, we obtain explicit optimal protocols from the SDP optimizer by decomposing it into an initial probe state and a sequence of intermediate isometries, as described in Appendices~\ref{app:single_par_achievability} and~\ref{app:tester_decomposition}, and then using the symmetric logarithmic derivative (SLD) of the corresponding output state.
}

\section{Preliminary Material}
\label{sec:prelims}
In this section, we present a general framework for sequential quantum multiparameter estimation.

\subsection{Optimal quantum multiparameter estimation}\label{sec:optimal_multiparameter_estimation}
A quantum channel parameter estimation problem is defined by real-valued parameters $\theta=(\theta_1,\dots,\theta_m)\in\R^m$ of a quantum channel $\Lambda_\theta:\L(A)\rightarrow\L(A)$ over the space $\L(A)$ of bounded linear operators from $A$ to itself for a $D$-dimensional probe register $A$.
Given access to $T$ copies of channel $\Lambda_\theta$, the most general strategy is to probe each copy of $\Lambda_\theta$ sequentially, while keeping a quantum memory register which assists in probing the channels in sequence and measuring at the end of the protocol. A conditional control operation can be applied between each channel use. At the output of such a procedure one obtains $\rho_\theta=\rho_\text{out}(\theta_1,\dots,\theta_m)$.

The experimenter then performs a measurement described by a positive operator-valued measure (POVM) which is a set of positive operators $\{\Pi_k\}_k$ that sum to the identity $\sum_{k}\Pi_{k}=I_D$. 
The $k$-th measurement outcome occurs with probability \mbox{$p(k|\theta)=\text{Tr}[\rho_\theta\Pi_{k}]$}. 
Based on the observed measurement outcomes, the experimenter produces an estimated value $\hat{\theta}_i$ of $\theta_i$. 
In this work, we focus on locally unbiased estimation where the classical Fisher information (CFI) $\mathcal{J}$ places a bound on the MSE. 
For locally unbiased estimation around a known point, $\theta_0=(\theta_{0,1},\dots,\theta_{0,m})$, we require~\cite{fujiwara2006strong,holevo2011probabilistic}
\begin{equation}\label{eqn:locally_unbiased}
    \langle\hat{\theta_i}\rangle=\theta_{0,i}
    \quad\text{and}\quad
    \left.\frac{\partial\langle\hat{\theta}_i\rangle}{\partial\theta_j}\right|_{\theta=\theta_0}=\delta_{i,j}\;,
\end{equation}
where $\<\hat{\theta}_i\> = \sum_k \hat{\theta}_{k,i} p(k|\theta)$ is the expected value of estimator $\hat{\theta}_i$.
Thus, the MSE of parameter $\theta_i$ is given by $\mse_{\theta_i} = \sum_k (\hat{\theta}_{k,i}-\theta_i)^2 p(k|\theta)$.
In this setting, the covariance matrix, $\cov_\theta$, is given by
\begin{equation}\label{eq:MSEmatrix}
    [\cov_\theta]_{ij}=\sum_{k}(\hat{\theta}_{k,i}-\theta_{i})(\hat{\theta}_{k,j}-\theta_{j}) \,p(k|\theta)\;,
\end{equation}
so that our quantity of interest is $\sum_{i=1}^m\mse_{\theta_i}=\text{Tr}[\cov_\theta]$. The CFI matrix can be computed as
\begin{equation}
\label{eq:CFIdefin}
    \mathcal{J}_{ij}
    =\sum_k p(k|\theta) \left(\frac{\partial \log p(k|\theta)}{\partial \theta_i}\right)\left(\frac{\partial \log p(k|\theta)}{\partial \theta_j}\right) \;.
\end{equation}
The CFI provides a bound on the MSE as
\begin{equation}
\label{eq:CFI}
    \sum_{i=1}^m\mse_{\theta_i} = \text{Tr}(\cov_\theta) \geq \text{Tr}(\mathcal{J}^{-1})\;.
\end{equation}

Given access to $T$ applications of the quantum channel $\Lambda_\theta$, the experimentalist therefore has control over the initial density matrix $\rho$, $T-1$ control unitaries in between each channel use denoted as $\mathbf{U}$, and the final measurement $\{\Pi_k\}_k$ and estimator coefficients $\hat{\theta}_{i,k}$. 
We note that the setup described above also covers strategies that probe $T$ copies of $\Lambda_\theta$ in parallel where $\rho$ is a quantum state over $A^{\otimes T}$ and the $T-1$ control unitaries in $\mathbb{U}$ swap the $A^{\otimes T}$ systems such that we have a quantum state $\Lambda_\theta^{\otimes T}(\rho)$ before the final measurement $\{\Pi_k\}_k$.
The most ambitious goal in quantum sensing is to solve the following minimization problem:
\begin{equation}
\label{eq:genoptimisation}
    \mathcal{C}_\text{opt}=\min_{\rho,\mathbf{U}, \{\Pi\},\hat{\theta}_i} \sum_{i=1}^m \mse_{\theta_i} \;.
\end{equation}
However, as discussed in the introduction, this optimization problem is a challenging problem to say the least. 
As such, it is common to introduce efficiently computable bounds on this quantity.

\subsection{Existing bounds}\label{sec:existing_bounds}

We now discuss existing bounds that apply in certain restricted versions of the problem stated in eqn.~\eqref{eq:genoptimisation}. 
When the input state $\rho$ is fixed, and we only allow for a single time step, i.e.~there is no optimization over $\mathbf{U}$ or $\rho$, the optimization over the measurement can be handled using the NHB. 
The NHB applies for locally unbiased estimation only, hence in this setting we write $\sum_i\mse_{\theta_i}=\text{Tr}(\cov_\theta)$. 
The NHB bounds the MSE as
\begin{equation}
\begin{gathered}
    \text{Tr}(\cov_\theta) \geq \mathcal{C}_\text{NH} \\
    \text{where}\\
    \mathcal{C}_\text{NH} := \min_{\mathbf{L},\,X} \tr(\mathbf{S}_\theta \mathbf{L}) \\
    \text{subject to}\\
       \mathbf{L} \geq  XX^\intercal \;,
        \quad \mathbf{L}_{jk}=\mathbf{L}_{kj}
        \;\mathrm{Hermitian} \;,\\
        X_j \text{ Hermitian, satisfying eqn.~\eqref{eq:unb},}
\end{gathered}
\end{equation}
where $\mathbf{S}_\theta= I_m\otimes \rho_\theta$ and $\mathbf{L}$ is an $m$-by-$m$ matrix of Hermitian operators $\mathbf{L}_{jk}$ acting on probe register $A$, i.e. $\mathbf{L} = \sum_{i,j=1}^m |i\>\<j|\otimes \mathbf{L}_{jk}$. 
On the other hand, $X$ is a block matrix with a single column and $m$ rows, where in each of the $m$ blocks $\{X_j\}_{j=1}^m$ are the estimator matrices.
They also need to satisfy the following matrix form of the local-unbiasedness conditions in eqn.~\eqref{eqn:locally_unbiased}:
\begin{equation}\label{eq:unb}
\begin{gathered}
    \text{Tr}(\rho_\theta X_j)=0 \quad
    \text{and}\quad
    \text{Tr}\Big( \frac{\partial\rho_\theta}{\partial\theta_k} X_j \Big) = \delta_{j,k} \;,
\end{gathered}
\end{equation}
for all $j,k\in\{1,\dots,m\}$.
For a POVM $\{\Pi_k\}_k$ and estimator values $\{\hat{\theta}_{k,j}\}_{k,j}$, the estimator matrices are defined at the fixed known point $\theta_0$ as $X_j:=\sum_k(\hat{\theta}_{k,j}-\theta_{0,j})\Pi_k$.
The evaluation at $\theta=\theta_0$ is implicit in eqn.~\eqref{eq:unb}, and $X_j$ is held fixed when the derivative in that equation is taken.
Thus the derivative acts only on $\rho_\theta$, and the second condition in eqn.~\eqref{eq:unb} is exactly the derivative local-unbiasedness condition in eqn.~\eqref{eqn:locally_unbiased}.

As mentioned in the introduction, Hayashi and Ouyang~\cite{Hayashi_2024} have extended the NHB to include optimization over both the initial state $\rho$ and measurement $\{\Pi_k\}_k$, simultaneously. 
Such optimization can be done efficiently in the probe dimension $D$ and number of parameters $m$ by using conic programming.

\subsection{Optimal time-dependent quantum sensing}\label{subsec:timedependent}

Optimal single-parameter estimation protocols and their MSE bounds are well-understood in the literature~\cite{chiribella2012optimal,liu2023optimal,kurdzialek2025quantum,liu2022optimal,liu2024fully}.
While comparatively little is known about optimal time-dependent sensing protocols for multiparameter estimation, there is a lot of insight to be gleaned from error-corrected quantum sensing. 
Error-corrected quantum sensing will also serve as one of the main applications of our new result, hence we introduce it in detail here. 
In this context, it is instructive to consider sensing problems in a continuous time limit. 
We consider a parameter $\theta$ of a Hamiltonian $H_\theta$ acting on our probe system, while simultaneously the probe experiences Markovian noise generated by Lindblad operators $\{L_k\}_k$,
\begin{equation}\label{eq:lindblad_metrology}
  \frac{d\rho}{dt} = -\iota[ H_\theta,\rho] + \sum_k\Big(L_k\rho L_k^\dagger - \tfrac12\{L_k^\dagger L_k,\rho\}\Big) \;,
\end{equation}
where $\iota=\sqrt{-1}$.
It is standard to discretize time into short windows $d t$ during which the state evolves according to eqn.~\eqref{eq:lindblad_metrology}. 
Following each $dt$ window, a recovery operation is applied designed to undo the effect of the noise accumulated during that segment \cite{kessler2014quantum,Arrad2014,dur2014improved,zhou2018achieving}. 
In the limit of arbitrarily fast correction, this approach can suppress noise to leading order in $d t$ while allowing coherent signal accumulation for a total time $\tau=M\,d t$ for $M$ channel uses. 

Let $\mathcal{S}$ denote the complex linear span of all Hermitian jump operators and their products:
\begin{equation}
    \mathcal{S} = \text{span}_{\C}\{\, I,\; L_k, L_k^{\dagger},\; L_k^{\dagger}L_j : \forall j,k \,\}.
\end{equation}
Then the condition under which Heisenberg scaling can be achieved for a single-parameter estimation is known as Hamiltonian not in Lindblad span (HNLS)~\cite{zhou2018achieving,demkowicz2017adaptive}. 
That is, Heisenberg scaling in time can be achieved if and only if the Hamiltonian generator $H_\theta=\theta h_0$ has a nonzero component orthogonal to the Lindblad span, $h_0 \notin \mathcal{S}$. 
Under this condition, it is known that there exists an error correcting code capable of achieving the Heisenberg limit~\cite{zhou2018achieving}. 
For multiparameter estimation, the corresponding necessary and sufficient conditions for achieving the Heisenberg limit under Markovian noise with adaptive protocols is also known~\cite{gorecki2020optimal}.

\subsection{Sequential quantum processes}\label{sec:sequential_strategies_combs}

As we have discussed in Section~\ref{sec:optimal_multiparameter_estimation}, a quantum sensing protocol on $T$ copies of quantum channel $\Lambda_\theta$ with parameters $\theta=(\theta_1,\dots,\theta_m)$ consists of an initial probe state $\rho$, a sequence of unitaries $\mathbf{U}$, a final POVM measurement $\Pi=\{\Pi_k\}_k$, and an estimator $\hat{\theta}=(\hat{\theta}_1,\dots,\hat{\theta}_m)$.
In the most general case, the protocol may involve ancillary registers $C$ which may exhibit entanglement with the probe register $A$ initially in the initial joint state $\rho$ over $C\otimes A$, acted on jointly with $A$ by unitaries in $\mathbf{U}=(U^{(1)},\dots,U^{(T-1)})$, and finally measured jointly with $A$ by POVM $\{\Pi_k\}_k$.
The initial state $\rho$, control unitaries $\mathbf{U}$, and POVM $\Pi$ altogether describes a \textit{sequential probe} $\mathbf{M}$.
Such a sequential probe $\mathbf{M}$ acting on $T$ copies of $\Lambda_\theta$ can be illustrated as
\begin{equation}
    \resizebox{0.97\linewidth}{!}{
    \begin{quantikz}[
        classical gap=0.03cm,
        column sep=0.28cm,
        font=\normalsize,
        labels={font=\footnotesize}
    ]
        \bshade{2}{9}{$\mathbf{M}$}
        \setwiretype{n} & \lstick[2]{$\rho$} & \wire[l][1]["C"{above,pos=0}]{q} \setwiretype{b} & \gate[2,style={fill=blue!20}]{U^{(1)}} \wire[r][1]["C"{above,pos=0.6}]{q} & 
            \invgate{\dots} & \wire[l][1]["C"{above,pos=0}]{q} & \gate[2,style={fill=blue!20}]{U^{(T-1)}} &
            \wire[l][1]["C"{above,pos=0}]{q} & \gate[2,style={fill=blue!20}]{\Pi} \\
        \setwiretype{n} & & \gate[1,style={fill=green!20}]{\Lambda_\theta} \wire[l][1]["{A}"{above,pos=0.45}]{q} \setwiretype{q} \gshade{1}{1}{} & \wire[l][1]["{A}"{above,pos=0.45}]{q} \wire[r][1]["{A}"{above,pos=0.45}]{q} & 
            \invgate{\dots} & \gate[1,style={fill=green!20}]{\Lambda_\theta} \wire[l][1]["{A}"{above,pos=0.6}]{q} \wire[r][1]["{A}"{above,pos=0.45}]{q} \gshade{1}{1}{} & &
            \gate[1,style={fill=green!20}]{\Lambda_\theta} \wire[l][1]["{A}"{above,pos=0.6}]{q} \wire[r][1]["{A}"{above,pos=0.45}]{q} \gshade{1}{1}{} &
    \end{quantikz} ,} 
\end{equation}
where the object in blue represents the sequential probe $\mathbf{M}$ and the objects in green are the $T$ copies of the channel of interest.

Instead of $T$ copies of quantum channels, we may also consider a general multi-time-step quantum process $\mathbf{\Lambda}_\theta$ parameterized by $\theta=(\theta_1,\dots,\theta_m)$ which is being probed at $T$ time points, e.g.~where there is an environment $E$ that the quantum process acts on jointly with probe register $A$ over time but is inaccessible to the experimenter.
In such a case, we can describe this entire process $\mathbf{\Lambda}_\theta$ as a sequence of quantum channels $\Lambda_\theta^{(1)},\dots,\Lambda_\theta^{(T)}$ acting jointly on probe register $A$ and some inaccessible environment $E$ which accounts for any temporal correlation in the process $\mathbf{\Lambda}_\theta$:
\begin{equation}
\begin{aligned}
    \resizebox{0.8\linewidth}{!}{
    \begin{quantikz}[
        classical gap=0.03cm,
        column sep=0.28cm,
        font=\normalsize,
        labels={font=\footnotesize}
    ]
        \setwiretype{n} & \gate[2,style={fill=green!20}]{\Lambda_\theta^{(1)}} \wire[l][1]["{A}"{above,pos=0.45}]{q} \setwiretype{q} \gshade{1}{1}{} & \quad \wire[l][1]["{A}"{above,pos=0.45}]{q} \wire[r][1]["{A}"{above,pos=0.45}]{q} & 
            \gate[2,style={fill=green!20}]{\Lambda_\theta^{(2)}} \gshade{1}{1}{} & 
            \invgate{\dots} \wire[l][1]["{A}"{above,pos=0.45}]{q} & 
            \gate[2,style={fill=green!20}]{\Lambda_\theta^{(T)}} \wire[l][1]["{A}"{above,pos=0.65}]{q} \wire[r][1]["{A}"{above,pos=0.45}]{q} \gshade{1}{1}{} & \\
        \setwiretype{n} & \gshade{1}{5}{$\mathbf{\Lambda}_\theta$} & \wire[l][1]["{E}"{above,pos=0.45}]{q} \setwiretype{q} & 
            & 
            \invgate{\dots} \wire[l][1]["{E}"{above,pos=0.45}]{q} & 
            \wire[l][1]["{E}"{above,pos=0.65}]{q} & \setwiretype{n}
    \end{quantikz}
    } ,
\end{aligned}
\end{equation}
where each channel $\Lambda_\theta^{(t)}$ may be individually dependent only on a subset of parameters $\theta$, and channels $\Lambda_\theta^{(t)},\Lambda_\theta^{(t')}$ at different steps need not be the same. 
Thus, the action of a sequential probe $\mathbf{M}$ on parameterized process $\mathbf{\Lambda}_\theta$ over $T$ time steps can be illustrated as
\begin{equation}
    \resizebox{0.95\linewidth}{!}{
    \begin{quantikz}[
        classical gap=0.03cm,
        column sep=0.28cm,
        font=\normalsize,
        labels={font=\footnotesize}
    ]
        \bshade{2}{9}{$\mathbf{M}$}
        \setwiretype{n} & \lstick[2]{$\rho$} & \wire[l][1]["C"{above,pos=0}]{q} \setwiretype{b} & \gate[2,style={fill=blue!20}]{U^{(1)}} \wire[r][1]["C"{above,pos=0.6}]{q} & 
            \invgate{\dots} & \wire[l][1]["C"{above,pos=0}]{q} & \gate[2,style={fill=blue!20}]{U^{(T-1)}} &
            \wire[l][1]["C"{above,pos=0}]{q} & \gate[2,style={fill=blue!20}]{\Pi} \\
        \setwiretype{n} & & \gate[2,style={fill=green!20}]{\Lambda_\theta} \wire[l][1]["{A}"{above,pos=0.45}]{q} \setwiretype{q} \gshade{2}{1}{} & \wire[l][1]["{A}"{above,pos=0.45}]{q} \wire[r][1]["{A}"{above,pos=0.45}]{q} & 
            \invgate{\dots} & \gate[2,style={fill=green!20}]{\Lambda_\theta} \wire[l][1]["{A}"{above,pos=0.6}]{q} \wire[r][1]["{A}"{above,pos=0.45}]{q} \gshade{2}{1}{} & &
            \gate[2,style={fill=green!20}]{\Lambda_\theta} \wire[l][1]["{A}"{above,pos=0.6}]{q} \wire[r][1]["{A}"{above,pos=0.45}]{q} \gshade{2}{1}{} & \\
        \setwiretype{n} & & \gshade{1}{6}{$\mathbf{\Lambda}_\theta$} & \wire[l][1]["{E}"{above,pos=0}]{q} \setwiretype{q} & 
            \invgate{\dots} & & \wire[l][1]["{E}"{above,pos=0}]{q} &
             & \setwiretype{n}
    \end{quantikz} } .
\end{equation}

A sequential probe $\mathbf{M}$ is an object known as a \textit{quantum tester}~\cite{chiribella2009theoretical,taranto2025higherorderquantumoperations}, which can be thought of as an analogue of a POVM for multi-time-step processes.
It gives the measurement statistics from a sequence of interactions with the process $\mathbf{\Lambda}$ followed by a measurement, analogous to how a POVM gives the output statistics when directly measuring a quantum state.
Both sequential probes $\mathbf{M}$ and multi-time-step processes $\mathbf{\Lambda}$ are more generally known as \textit{higher-order quantum operations}~\cite{taranto2025higherorderquantumoperations} or \textit{quantum combs}~\cite{chiribella2009theoretical}, and they can be represented as positive semidefinite operators over the joint probe systems over $T$ steps $\Bar{A}_T = A^{\otimes 2T}$ satisfying a set of certain linear constraints.
Throughout this work, we denote the positive semidefinite matrix representation of a sequential probe and a parameterized process as $\mathbf{M}$ and $\mathbf{\Lambda}_\theta$ themselves, respectively, since their matrix representation contains their complete information operationally. 
We also denote the set of all sequential probes over $T$ time steps as $\seq(\Bar{A}_T)$, which is the set of all physical sequential sensing strategies.
The complete definition of $\seq(\Bar{A}_T)$ is given in eqn.~\eqref{eqn:testers} of Appendix~\ref{app:multistep_processes}.

In this matrix representation, we can conveniently express the output probabilities and therefore the covariance matrix and MSE of a given sensing protocol.
Here we give an outline of how this is done. For a more detailed and rigorous approach, see Appendix~\ref{app:sequential_sensing}.

Since a sequential probe $\mathbf{M}$ is a positive semidefinite matrix, it can be decomposed into positive semidefinite operators $\mathbf{M}_1,\dots,\mathbf{M}_K$ as
\begin{equation}
    \mathbf{M} = \sum_{k=1}^K \mathbf{M}_k,
\end{equation}
where each $\mathbf{M}_k$ corresponds to an outcome $k$ which is observed as a result of an interaction with parameterized process $\mathbf{\Lambda}_\theta$ with probability
\begin{equation}
    p(k|\theta,\mathbf{M}) = \mathbf{\Lambda}_\theta \ast \mathbf{M}_k = \tr(\mathbf{\Lambda}_\theta^\top \mathbf{M}_k) \;.
\end{equation}
Here ``$\ast$'' is the \textit{link product}~\cite{chiribella2009theoretical,taranto2025higherorderquantumoperations}\footnote{
    In general, for operators $X\in\L(A\otimes B)$ and $Y\in\L(B\otimes C)$, their link product is given by $X\ast Y = \tr_B((X^{\top_B}\otimes I_C) (I_A\otimes Y))$, where $\tr_B$ is partial trace over space $B$ and $(\cdot)^{\top_B}$ is the partial transpose over space $B$. Note that we can equivalently write $X\ast Y = \tr_B(XY^{\top_B})$.},
which contracts some output registers of one operation with some input register of another.
Each $\mathbf{M}_k$ can be described as
\begin{equation}
    \resizebox{0.87\linewidth}{!}{
    \begin{quantikz}[
        classical gap=0.03cm,
        column sep=0.28cm,
        font=\normalsize,
        labels={font=\footnotesize}
    ]
        \bshade{2}{8}{$\mathbf{M}_k$}
        \setwiretype{n} & \lstick[2]{$\rho$} & \wire[l][1]["C"{above,pos=0}]{q} \setwiretype{b} & \gate[2,style={fill=blue!20}]{U^{(1)}} \wire[r][1]["C"{above,pos=0.6}]{q} & 
            \invgate{\dots} & \gate[2,style={fill=blue!20}]{U^{(T-1)}} &
            \wire[l][1]["C"{above,pos=0}]{q} & \gate[2,style={fill=blue!20}]{\Pi_k} \\
        \setwiretype{n} & & \invgate{} \wire[l][1]["{A}"{above,pos=0.45}]{q} \setwiretype{q} \wshade{1}{1}{} & \wire[l][1]["{A}"{above,pos=0.45}]{q} \wire[r][1]["{A}"{above,pos=0.45}]{q} & 
            \invgate{\dots} \wshade{1}{1}{} & &
            \invgate{} \wire[l][1]["{A}"{above,pos=0.6}]{q} \wire[r][1]["{A}"{above,pos=0.45}]{q} \wshade{1}{1}{} &
    \end{quantikz} } ,
\end{equation}
corresponding to outcome $k$ from the final POVM $\Pi$.

The $i,j$-th entry of the covariance matrix for a sensing protocol with sequential probe $\mathbf{M}$ and estimator $\hat{\theta}$ is given by $\sum_{k=1}^K (\hat{\theta}_{i,k}-\theta_i)(\hat{\theta}_{j,k}-\theta_j) p(k|\theta,\mathbf{M})$.
Thus we can write the covariance matrix for this sensing protocol as
\begin{equation}
    \cov_\theta(\mathbf{M},\hat{\theta}) = \sum_{i,j=1}^m |i\>\<j| \otimes \sum_k \xi_{i,k}\xi_{j,k} \mathbf{\Lambda}_\theta \ast \mathbf{M}_k
\end{equation}
for $\xi_{i,k}=\hat{\theta}_{i,k}-\theta_i$, and therefore its MSE is given by
\begin{equation}
    \mse_\theta(\mathbf{M},\hat{\theta}) = \sum_{i=1}^m \sum_k \xi_{i,k}^2 \mathbf{\Lambda}_\theta \ast \mathbf{M}_k \;.
\end{equation}

\section{Precision bound for sequential quantum sensing}\label{sec:sequential_bound}

We will now show our main result, which provides a lower bound to the MSE for any sequential quantum sensing scenario, thus imposing a lower bound on any physically implementable strategy.
For a given $T$-time-step sequential parameter estimation task, our lower bound is efficiently computable by a semidefinite program (SDP) in the size of the probe system dimension $D$ and number of parameters $m$, that is, the SDP outputs the MSE bound in time polynomial in $D$ and $m$ (see Appendix~\ref{app:SDP_runtime_complexity}).
However the SDP runtime is exponential in the number of time steps $T$, as the size of the matrix describing such strategy is $D^{2T}$.
Throughout, we will denote the process that encodes parameter $\theta=(\theta_1,\dots,\theta_m)$ that we are estimating simply as $\mathbf{\Lambda}_\theta$.

\begin{theorem}[{Sequential Bound}]\label{thm:sequential_NH_bound}
    For any quantum process $\mathbf{\Lambda}_\theta$ parameterized by $\theta=(\theta_1,\dots,\theta_m)$ over $T$ time steps and a sensing protocol with sequential probe $\mathbf{M}'$ and locally-unbiased estimator $\hat{\theta}$, it holds that 
    \begin{equation}
        \mse_\theta(\mathbf{M}',\hat{\theta}) \geq \sdp(\mathbf{\Lambda}_\theta) \;,
    \end{equation}
    where $\sdp(\mathbf{\Lambda}_\theta)$ is the value of the following semidefinite program
    \begin{equation}\label{eqn:1var_optimization_sequential_bound}
    \begin{gathered}
        \min_\mathbf{X} \tr\Big( \overline{\mathbf{\Lambda}}_\theta \mathbf{X} \Big) \\
        \text{s.t.}\\
            \mathbf{X} = \begin{bmatrix}
                \mathbf{M} & X^\dag \\
                X & \mathbf{L}
            \end{bmatrix}\geq0 \;,\quad
            \mathbf{M} \in \seq(\Bar{A}_T) \;,\\
            X = \sum_{i=1}^m |i\>\otimes \mathbf{X}_{0,i} \;,\quad
                \mathbf{X}_{0,i} = \mathbf{X}_{0,i}^\dag \;, \\
            \mathbf{L} = \sum_{i,j=1}^m |i\>\<j|\otimes \mathbf{X}_{i,j} \;,\quad
                \mathbf{X}_{i,j} = \mathbf{X}_{i,j}^\dag \;, \\
            \tr(\mathbf{\Lambda}_\theta^\top \mathbf{X}_{0,j}) = 0\;,\quad\mathrm{and}\\
                \frac{\partial}{\partial\theta_i} \tr(\mathbf{\Lambda}_\theta^\top \mathbf{X}_{0,j}) = \delta_{i,j} \quad\forall i,j\in[1,m] \;,
    \end{gathered}
    \end{equation}
    where $\overline{\mathbf{\Lambda}}_\theta = \begin{bmatrix}
        0_{\Bar{D}\times\Bar{D}} & 0_{\Bar{D}\times m\Bar{D}} \\
        0_{\Bar{D}\times m\Bar{D}} & I_m\otimes\mathbf{\Lambda}_\theta^\top
    \end{bmatrix}$ is a block matrix with $I_m\otimes\mathbf{\Lambda}_\theta^\top$ in its lower-right block and zeros everywhere else and $\seq(\Bar{A}_T)$ is the set of all sequential probes over $T$ steps (see eqn.~\eqref{eqn:testers} in Appendix~\ref{app:multistep_processes}) and $\Bar{A}_T = A^{\otimes 2T}$ denotes the joint probe systems over $T$ steps.
    In the matrix notation above, $0_{q\times p}$ is the $p\times q$ zero matrix and $\Bar{D}=D^{2T}$ is the dimension of $\Bar{A}_T$.
\end{theorem}

The complete proof of Theorem~\ref{thm:sequential_NH_bound} is given in Appendix~\ref{app:sequential_bound_proof}.
Before giving a proof outline, we first provide some remarks regarding how $\sdp(\mathbf{\Lambda}_\theta)$ in Theorem~\ref{thm:sequential_NH_bound} is to be interpreted.
The matrix $\mathbf{X}$ is an $m+1\times m+1$ block matrix, where its $(i,j)$-th block $\mathbf{X}_{i,j}$ is an operator over the joint probe systems over $T$ steps $\Bar{A}_T = A^{\otimes 2T}$.
When the SDP gives an attainable MSE, i.e. $\sdp(\mathbf{\Lambda}_\theta) = \mathcal{C}_\mathrm{opt}$ (see eqn.~\eqref{eq:genoptimisation}), the operators in the blocks of $\mathbf{X}$ describe the entire optimal valid strategy that achieves this MSE: an initial probe state, intermediate control, final measurement, and estimator.
Its upper-left block $\mathbf{M}\in\seq(\Bar{A}_T)$ is a $T$-step tester (see Section~\ref{sec:sequential_strategies_combs}) describing the optimal initial probe state, intermediate control, and final measurement.
Each of the remaining $m$ blocks $\mathbf{X}_{0,1},\dots,\mathbf{X}_{0,m}$ in the first block-column $X$ of $\mathbf{X}$ describe the relationship between measurement outcomes from strategy $\mathbf{M}$ and estimator $\hat{\theta}=(\hat{\theta}_1,\dots,\hat{\theta}_m)$.
On the other hand, the lower-right $m\times m$ block matrix $\mathbf{L}$ describes the correlation between each pair of estimators $\hat{\theta}_1,\dots,\hat{\theta}_m$.
The positive-semidefinite constraint on $\mathbf{X}$ guarantees that the value obtained by the SDP lower-bounds any optimal strategy.
Lastly, the bottom two trace constraints ensure local-unbiasedness (see eqn.~\eqref{eqn:locally_unbiased}).

Note however, that the SDP is not necessarily a tight bound. In instances where it is not tight, the blocks of $\mathbf{X}$ may not describe a valid strategy.
{We will come back to this issue in Section~\ref{sec:single_par_tight_bound} and show that $\sdp(\mathbf{\Lambda}_\theta)$ is tight, i.e. $\sdp(\mathbf{\Lambda}_\theta)=\mathcal{C}_\mathrm{opt}$, for every regular finite-dimensional single-parameter sensing problem admitting an optimal tester, and that a valid locally optimal strategy can be constructed from an SDP optimizer.}

Now we will outline the proof of Theorem~\ref{thm:sequential_NH_bound}.
Recall from Section~\ref{sec:sequential_strategies_combs} that the probability of observing outcome $k$ is given by $p(k|\theta,\mathbf{M}) = \mathbf{\Lambda}_\theta \ast \mathbf{M}_k = \tr(\mathbf{\Lambda}_\theta^\top \mathbf{M}_k)$, so that we can express the covariance matrix for a sequential probe $\mathbf{M}=\sum_{k=1}^K\mathbf{M}_k$ as
\begin{equation}
\begin{aligned}
    \cov_\theta(\mathbf{M},\hat{\theta}) &= \sum_{i,j=1}^m |i\>\<j| \sum_k \xi_{i,k}\xi_{j,k} \tr(\mathbf{\Lambda}_\theta^\top \mathbf{M}_k)
\end{aligned}
\end{equation}
for $\xi_{i,k}=\hat{\theta}_{i,k}-\theta_i$ and its MSE as
\begin{equation}
\begin{aligned}
    \mse_\theta(\mathbf{M},\hat{\theta}) &= \sum_{i=1}^m \sum_k \xi_{i,k}^2 \tr(\mathbf{\Lambda}_\theta^\top \mathbf{M}_k) \;.
\end{aligned}
\end{equation}
On the other hand, the estimator $\hat{\theta} = (\hat{\theta}_1,\dots,\hat{\theta}_m)$ can be expressed as
\begin{equation}
   \mathbb{E}[ \hat{\theta}_j ]-\theta_j= \sum_{k} \xi_{j,k} \tr(\mathbf{\Lambda}_\theta^\top \mathbf{M}_k) \;,
\end{equation}
which gives the unbiasedness conditions from eqn.~\eqref{eqn:locally_unbiased} as
\begin{equation}\label{eqn:unbiased_estimator1}
\begin{gathered}
    \sum_k \xi_{j,k} \tr(\mathbf{\Lambda}_\theta^\top \mathbf{M}_k) = 0 
    \\\text{and}\\
    \frac{\partial}{\partial\theta_i} \sum_{k} \xi_{j,k} \tr(\mathbf{\Lambda}_\theta^\top \mathbf{M}_k) = \delta_{i,j} \;.
\end{gathered}
\end{equation}

Let us now consider an $m$-dimensional real vector space $V$ with basis $\{|i\>\}_{i=1}^m$ and matrix $\mathbf{L}_\theta(\mathbf{M},\hat{\theta})$ over $V\otimes\Bar{A}_T$, given by
\begin{equation}\label{eqn:L_matrix1}
\begin{aligned}
    \mathbf{L}_\theta(\mathbf{M},\hat{\theta}) &= \sum_{i,j=1}^m |i\>\<j|_V \otimes \mathbf{L}_\theta(\mathbf{M},\hat{\theta})_{i,j},
\end{aligned}
\end{equation}
where $\mathbf{L}_\theta(\mathbf{M},\hat{\theta})_{i,j} = \sum_k \xi_{i,k} \xi_{j,k} \mathbf{M}_k$ is Hermitian.
Thus by denoting $\Tilde{\mathbf{\Lambda}}_\theta = I_m \otimes \mathbf{\Lambda}_\theta^\top$, where $I_m$ is the $m\times m$ identity matrix, we can write the 
MSE as
\begin{equation}
\begin{aligned}
    \mse_\theta (\mathbf{M},\hat{\theta}) &= \tr\Big( \Tilde{\mathbf{\Lambda}}_\theta \mathbf{L}_\theta(\mathbf{M},\hat{\theta}) \Big) \;.
\end{aligned}
\end{equation}

Let us now consider \textit{estimator matrices} $X_\theta(\mathbf{M},\hat{\theta})_1,\dots,X_\theta(\mathbf{M},\hat{\theta})_m$, given by Hermitian matrices
\begin{equation}\label{eqn:estimator_matrix_elements1}
\begin{aligned}
    X_\theta(\mathbf{M},\hat{\theta})_i = \sum_k \xi_{i,k} \mathbf{M}_k
\end{aligned}
\end{equation}
and a block matrix of these estimator matrices
\begin{equation}\label{eqn:estimator_matrix1}
    X_\theta(\mathbf{M},\hat{\theta}) = \sum_{i=1}^m |i\>_V \otimes X_\theta(\mathbf{M},\hat{\theta})_i \;. 
\end{equation}
Block matrices $\mathbf{L}_\theta(\mathbf{M},\hat{\theta})$ and $X_\theta(\mathbf{M},\hat{\theta})$ satisfy the matrix inequality (see eqn.~\eqref{eqn:LX_inequality2} and Lemma~\ref{lem:diagonal_outer_product_block_choi_matrix_inequality})
\begin{equation}\label{eqn:LX_inequality3}
    \mathbf{L}_\theta(\mathbf{M},\hat{\theta}) \geq X_\theta(\mathbf{M},\hat{\theta}) \mathbf{M}^+ X_\theta(\mathbf{M},\hat{\theta})^\dag \;,
\end{equation}
where $\mathbf{M}^+$ is the pseudoinverse of $\mathbf{M}$.
We note that this inequality holds due to the fact that tester $\mathbf{M}$ decomposes as $\mathbf{M}=\sum_k\mathbf{M}_k$, where $\mathbf{M}_k$ is a positive-semidefinite matrix corresponding to outcome $k$ that came up in the summands of $\mathbf{L}_\theta(\mathbf{M},\hat{\theta})_{i,j}$ and $X_\theta(\mathbf{M},\hat{\theta})_i$.
Shortly, we will discuss how this inequality determines whether $\sdp(\mathbf{\Lambda}_\theta)$ is tight.

By denoting $\mathbf{M}_\opt,\hat{\theta}_\opt$ as an optimal sequential probe and optimal estimator, we can lower-bound the MSE as
\begin{equation}\label{eqn:optimization_sequential_bound_main}
\begin{aligned}
    \mse_\theta(\mathbf{M},\hat{\theta}) &\geq \mse_\theta(\mathbf{M}_\opt,\hat{\theta}_\opt) \\
    &= \tr\big(\Tilde{\mathbf{\Lambda}}_\theta \mathbf{L}_\theta(\mathbf{M}_\opt,\hat{\theta}_\opt)\big) \\
    &\geq \min_{\mathbf{L}} \tr( \tilde{\mathbf{\Lambda}}_\theta \mathbf{L} ) \,,
\end{aligned}
\end{equation}
by minimizing over all matrices $\mathbf{L} = \sum_{i,j=1}^m |i\>\<j|\otimes \mathbf{X}_{i,j}$ satisfying
\begin{equation}\label{eqn:intermediate_sdp_constraints}
\begin{gathered}
    \mathbf{X}_{i,j} = \mathbf{X}_{j,i} \;,\quad \mathbf{X}_{i,j}=\mathbf{X}_{i,j}^\dag \;, \\
    \mathbf{L} \geq X_\theta(\mathbf{M}_\opt,\hat{\theta}_\opt) \mathbf{M}_\opt^+ X_\theta(\mathbf{M}_\opt,\hat{\theta}_\opt)^\dag \;.
\end{gathered}
\end{equation}
Note that the support of each estimator matrix $X_\theta(\mathbf{M}_\opt,\hat{\theta}_\opt)_i$ is contained in the support of $\mathbf{M}_\opt$ by definition (see eqn~\eqref{eqn:estimator_matrix_elements1}).
This allows us to rewrite the constraints in eqn.~\eqref{eqn:intermediate_sdp_constraints} above as
\begin{equation}
\begin{gathered}
    \mathbf{X}_{i,j} = \mathbf{X}_{j,i} \;,\quad \mathbf{X}_{i,j}=\mathbf{X}_{i,j}^\dag \;, \\
        \mathbf{X} = \left[\begin{matrix}
        \mathbf{M}_\opt & X_\theta(\mathbf{M}_\opt,\hat{\theta}_\opt)^\dag \\
        X_\theta(\mathbf{M}_\opt,\hat{\theta}_\opt) & \mathbf{L}
    \end{matrix}\right] \geq 0 \;,
\end{gathered}
\end{equation}
since the support of each estimator matrix $X_\theta(\mathbf{M}_\opt,\hat{\theta}_\opt)_i$ is contained in the support of $\mathbf{M}_\opt$ and the matrix inequality in eqn.~\eqref{eqn:LX_inequality3} is equivalent to the block matrix
$\mathbf{X}$ being positive semidefinite (see Lemma~\ref{lem:block_martix_psd_iff_condition} and surrounding text).

Thus by simply constraining the block matrix
\begin{equation}
    \mathbf{X} = \left[\begin{matrix}
        \mathbf{M} & X^\dag \\
        X & \mathbf{L}
    \end{matrix}\right]
\end{equation}
to be positive semidefinite, we can guarantee that the support of blocks $\{\mathbf{X}_{0,i}\}_{i=1}^m$ of $X$ is in the support of $\mathbf{M}$ and that matrix inequality $\mathbf{L}\geq X\mathbf{M}^+X^\dag$ is satisfied.
Then we can further minimize over all such block matrices $\mathbf{X}$, constrain the upper-left block to be a valid sequential probe and require the other sub-blocks to be Hermitian to obtain 
\begin{equation}\label{eqn:optimization_sequential_bound3}
\begin{aligned}
    \mse_\theta(\mathbf{M},\hat{\theta}) 
    &\geq \min_{\mathbf{L}} \tr( \tilde{\mathbf{\Lambda}}_\theta \mathbf{L} )
    &\geq 
    \min_{X,\mathbf{L}} \tr\Big( \Tilde{\mathbf{\Lambda}}_\theta \mathbf{L} \Big) \;,
\end{aligned}
\end{equation}
where the minimization is over matrices $X$ and $\mathbf{L}$ satisfying the following constraints:
\begin{equation}\label{eqn:X_optimization_constraints}
\begin{gathered}
    \mathbf{X} = \begin{bmatrix}
        \mathbf{M} & X^\dag \\
        X & \mathbf{L}
    \end{bmatrix}\geq0 \;,\quad
    \mathbf{M} \in \seq(\Bar{A}_T) \;,\\
    X = \sum_{i=1}^m |i\>\otimes \mathbf{X}_{0,i} \;,\quad
        \mathbf{X}_{0,i} = \mathbf{X}_{0,i}^\dag \;, \\
    \mathbf{L} = \sum_{i,j=1}^m |i\>\<j|\otimes \mathbf{X}_{i,j} \;,\quad
        \mathbf{X}_{i,j} = \mathbf{X}_{i,j}^\dag \;, \\
    \tr(\mathbf{\Lambda}_\theta^\top \mathbf{X}_{0,j}) = 0 \;,\quad\text{and}\\
    \frac{\partial}{\partial\theta_i} \tr(\mathbf{\Lambda}_\theta^\top \mathbf{X}_{0,j}) = \delta_{i,j} \quad\forall i,j\in[1,m] \;.
\end{gathered}
\end{equation}
Further by denoting $\overline{\mathbf{\Lambda}}_\theta = \begin{bmatrix}
    0 & 0 \\
    0 & I_m\otimes\mathbf{\Lambda}_\theta^\top
\end{bmatrix}$, we can rewrite the minimization on the right-hand side to obtain
\begin{equation}
\begin{aligned}
    \mse_\theta(\mathbf{M},\hat{\theta}) 
    &\geq
    \min_\mathbf{X} \tr\Big( \overline{\mathbf{\Lambda}}_\theta \mathbf{X} \Big) \;,
\end{aligned}
\end{equation}
which is the form given in Theorem~\ref{thm:sequential_NH_bound}.

\section{MSE optimality conditions}\label{sec:MSE_optimality_conditions}

In this section, we give conditions for determining when the sequential MSE bound in Theorem~\ref{thm:sequential_NH_bound} is attainable and when Heisenberg scaling can be achieved.
First, in Section~\ref{sec:single_par_tight_bound}, we show that the bound
is locally attainable for every regular finite-dimensional single-parameter problem admitting an optimal tester and
explain how an optimal sensing strategy can be recovered from an SDP
optimizer.
We then review, in Section~\ref{subsec:HNKS}, the Hamiltonian-not-in-Kraus-span condition proposed in~\cite{zhou2021asymptotic}, which characterizes Heisenberg scaling for identical and independent uses of a known single-parameter
channel.
Finally, in Section~\ref{subsec:finite_latent_hs}, Proposition~\ref{prop:finite_latent_hs} gives a sufficient
condition for lifting a Heisenberg-limited known-noise strategy to a
temporally correlated process governed by an unknown finite latent noise
label.

\subsection{Single-parameter sequential bound tightness and optimal strategy}\label{sec:single_par_tight_bound}

The constraints of the SDP lower bound in Theorem~\ref{thm:sequential_NH_bound} imply that an optimizer $\mathbf{X}$ must be of the form
\begin{equation}
    \mathbf{X} = \begin{bmatrix}
        \mathbf{M} & X^\dag \\
        X & \mathbf{L}
    \end{bmatrix}
\end{equation}
with $\mathbf{M}\in\seq(\Bar{A}_T)$ and blocks of $X$ satisfying the local-unbiasedness estimator conditions $\tr(\mathbf{\Lambda}_\theta^\top \mathbf{X}_{0,j})=0$ and $\frac{\partial}{\partial\theta_i} \tr(\mathbf{\Lambda}_\theta^\top \mathbf{X}_{0,j}) = \delta_{i,j}$.
Whenever the optimal value $\sdp(\mathbf{\Lambda}_\theta)$ is equal to $\min_{\mathbf{M},\hat{\theta}}\mse_\theta(\mathbf{M},\hat{\theta})$, if we can express the $\mathbf{M}$ and $X$ blocks of optimizer $\mathbf{X}$ as
\begin{equation}
    \mathbf{M} = \sum_{k=1}^K \mathbf{M}_k
    \quad\text{and}\quad
    \mathbf{X}_{0,i} = \sum_{k=1}^K \xi_{i,k}\mathbf{M}_k \;,\forall i\in[m]
\end{equation}
for some positive semidefinite operators $\mathbf{M}_1,\dots,\mathbf{M}_K$, then we can obtain an optimal strategy given by the estimator $\hat{\theta}$ and sequential probe $\mathbf{M}$.

For estimating parameters of a single copy of a channel $\Lambda_\theta$, our SDP bound in Theorem~\ref{thm:sequential_NH_bound} reduces to the Nagaoka-Hayashi bound (see Appendix~\ref{app:NH_bound}).
For a regular, locally identifiable single-parameter state model, the Nagaoka--Hayashi bound reduces to the SLD Cram\'er--Rao bound and is locally attainable~\cite{Braunstein_1994,conlon2025role}.
For multi-time processes, quantum Fisher information is obtained by maximizing the output-state SLD quantum Fisher information over all admissible strategies~\cite{Yang_2019,Altherr_2021,liu2023optimal,kurdzialek2025quantum}.
The following proposition shows that the single-parameter specialization of our sequential Nagaoka--Hayashi SDP reduces exactly to the reciprocal of this established quantity.

\begin{proposition}[Single-parameter attainability]\label{prop:single_par_achievability}
    Let \(m=1\). Assume finite dimension, regularity\footnote{
        \atchg{Here regularity means that the output-state model for each fixed strategy is differentiable at $\theta$ and admits a Hermitian SLD satisfying eqn.~\eqref{eqn:single_par_SLD_definition} with finite quantum Fisher information as defined in eqn.~\eqref{eqn:single_par_SLD_QFI_definition}.}
    }, and existence of an
    optimal tester. 
    Then the SDP bound given in Theorem~\ref{thm:sequential_NH_bound} satisfies
    \begin{equation}\label{eqn:single_par_SDP_fisher_info}
    \begin{gathered}
        \sdp(\mathbf{\Lambda}_\theta)
        =
        \frac{1}{\mathcal{J}_{\seq}(\mathbf{\Lambda}_\theta)} \\
        \text{for}\quad
        \mathcal{J}_{\seq}(\mathbf{\Lambda}_\theta) := \max_{\mathbf{M}\in \mathrm{seq}(\Bar{A}_T)}
        \mathcal{J}_\text{Q}(\rho_\theta^\mathbf{M}) \;,
    \end{gathered}
    \end{equation}
    where $\mathcal{J}_\text{Q}(\rho_\theta^\mathbf{M})$ is the single parameter channel QFI of quantum state $\rho_\theta^\mathbf{M} := \mathbf{M}^{1/2}\mathbf{\Lambda}_\theta^\top \mathbf{M}^{1/2}$ immediately before the final measurement.
    Moreover, the bound is locally attainable.
\end{proposition}

A proof of Proposition~\ref{prop:single_par_achievability} is given in Appendix~\ref{app:single_par_achievability}.
In particular, a key reason for the single-parameter attainability in Proposition~\ref{prop:single_par_achievability} is that there exists a tester $\mathbf{M}_\opt$ and an estimator $\hat{\theta}_\opt$ to construct matrices $X$ and $\mathbf{L}$ satisfying the constraints in eqn.~\eqref{eqn:X_optimization_constraints},
such that the matrix inequality in eqn.~\eqref{eqn:LX_inequality3} is saturated. 
Such matrices are constructed in the proof in Appendix~\ref{app:single_par_achievability}. 
For $m>1$, on the other hand, saturation would require, for every $i,j\in[1,m]$,
\begin{equation}
    \mathbf{L}_\theta(\mathbf{M},\hat{\theta})_{i,j}
    =
    X_\theta(\mathbf{M},\hat{\theta})_i
    \mathbf{M}^+
    X_\theta(\mathbf{M},\hat{\theta})_j \;.
\end{equation}
Since $\mathbf{L}_\theta(\mathbf{M},\hat{\theta})_{i,j} = \mathbf{L}_\theta(\mathbf{M},\hat{\theta})_{j,i}$, this requires the effective estimator observables $Y_i:=\sqrt{\mathbf{M}^+}X_\theta(\mathbf{M},\hat{\theta})_i\sqrt{\mathbf{M}^+}$ to satisfy a compatibility condition $[Y_i,Y_j]=0$ on $\supp(\mathbf{M})$ for every $i,j$. 
This compatibility condition is automatic when $m=1$, but need not hold for multiparameter estimation, where the estimator observables associated with different parameters may be incompatible. 
Consequently, eqn.~\eqref{eqn:LX_inequality3} need not be saturable in the multiparameter case, and the SDP bound in Theorem~\ref{thm:sequential_NH_bound} need not be attainable.

Single-parameter metrology of quantum combs is conventionally characterized by maximizing the SLD quantum Fisher information of the output state over admissible probe combs \cite{Yang_2019,Altherr_2021,liu2023optimal,kurdzialek2025quantum}.
Proposition~\ref{prop:single_par_achievability} shows that the single-parameter specialization of our sequential Nagaoka--Hayashi SDP is exactly the reciprocal of this established sequential quantum Fisher information $\mathcal{J}_{\mathrm{seq}}$.
The examples in the next section and Appendix~\ref{app:additional_examples} explicitly realize the SLD measurement and estimator from the proof.
For a single-parameter estimation task, an explicit optimal strategy can be obtained from an SDP optimizer $\mathbf{X}$ as follows.
The spectral projectors of the SLD of $\rho_\theta^{\mathbf{M}}$ determine the tester outcome operators $\{\mathbf{M}_k\}_k$, while the rescaled SLD eigenvalues determine the locally unbiased estimator values, as shown in Appendix~\ref{app:single_par_achievability}.
The upper-left tester block $\mathbf{M}$ is then decomposed into an initial probe state and a sequence of intermediate isometries, with the SLD projectors implemented by the final measurement, using the decomposition in Appendix~\ref{app:tester_decomposition}.

From such operators $\{\mathbf{M}_k\}_k$ we can then obtain an explicit form of sequential probe $\mathbf{M}$ corresponding to an optimal strategy:
\begin{equation}
    \resizebox{0.95\linewidth}{!}{
    \begin{quantikz}[
        classical gap=0.03cm,
        column sep=0.43cm,
        font=\small,
        labels={font=\footnotesize}
    ]
        \;\; \bshade{2}{9}{$\mathbf{M}_k$} & \setwiretype{n} \lstick[2]{$|\psi\>$} & \wire[l][1]["C_1"{above,pos=0}]{q} \setwiretype{q} & \gate[2,style={fill=blue!20}]{V^{(1)}} & \wire[l][1]["C_2"{above,pos=0}]{q} \setwiretype{b} & \invgate{\dots} & \gate[2,style={fill=blue!20}]{V^{(T-1)}} & \wire[l][1]["C_T"{above,pos=0}]{q} & \gate[2,style={fill=blue!20}]{M_k'} \\
        & \setwiretype{n} & \invgate{} \wire[l][1]["{A_0}"{above,pos=0.45}]{q} \wshade{1}{1}{} & \wire[l][1]["{A_1}"{above,pos=0.45}]{q} \wire[r][1]["{A_1'}"{above,pos=0.45}]{q} & \invgate{} \wshade{1}{2}{} & \invgate{\dots} & \wire[l][1]["{A_{T-1}}"{above,pos=0.7}]{q} \wire[r][1]["{A_{T-1}'}"{above,pos=0.7}]{q} & \invgate{} \wire[r][1]["{A_T}"{above,pos=0.45}]{q} \wshade{1}{1}{} & 
    \end{quantikz} 
    }
\end{equation}
for some initial pure state $|\psi\>$ over $C_1\otimes A_0$, isometries $V^{(1)},\dots,V^{(T-1)}$, and a final POVM measurement $\{M_k'\}_k$ over $C_T\otimes A_T$.
Here we explicitly label the probe system $A$ at the input and output of the channel at time step $t$ as $A_{t-1}'$ and $A_t$, respectively, i.e. $A_t\cong A_t'\cong A$ for all $t$, where $A_0'=A_0$.
For details on how such a form is obtained for any $\mathbf{M}\in\seq(\Bar{A}_T)$, see Appendix~\ref{app:tester_decomposition}.

\subsection{HNKS}\label{subsec:HNKS}

Here we review the Hamiltonian-not-in-Kraus-span (HNKS) condition introduced in Ref.~\cite{zhou2021asymptotic}, which characterizes the achievability of Heisenberg scaling for repeated independent uses of a single-parameter quantum channel.
We will use this condition to analyze the examples in Section~\ref{sec:examples} regarding their Heisenberg limit achievability.

Consider a differentiable single-parameter quantum channel 
\begin{equation} 
    \mathcal{C}_\theta(\rho) = \sum_{a=1}^{r} K_{a,\theta}\rho K_{a,\theta}^{\dagger}, 
\end{equation} 
with Kraus operators \(\{K_{a,\theta}\}_{a=1}^{r}\). 

We define its \emph{channel Hamiltonian}, or effective parameter generator, by 
\begin{equation} 
    H_{\rm K}(\mathcal{C}_\theta) := \iota \sum_{a=1}^{r} K_{a,\theta}^{\dagger} \frac{\partial K_{a,\theta}}{\partial\theta}, 
\label{eqn:channel_hamiltonian_HNKS} 
\end{equation} 
and its \emph{Hermitian Kraus span} by 
\begin{equation} 
    \mathcal{S}_{\rm K}(\mathcal{C}_\theta) := \left\{ \sum_{a,b=1}^{r} h_{a,b} K_{a,\theta}^{\dagger}K_{b,\theta} : h=h^\dagger \right\}. 
    \label{eqn:kraus_span_HNKS} 
\end{equation} 
The operator in eqn.~\eqref{eqn:channel_hamiltonian_HNKS} is Hermitian as a consequence of \(\sum_aK_{a,\theta}^{\dagger}K_{a,\theta}=I\). 

Although the channel Hamiltonian depends on the choice of Kraus representation, changing the Kraus representation changes it only by an element of the Kraus span, so membership in \(\mathcal{S}_{\rm K}(\mathcal{C}_\theta)\) is representation independent~\cite{zhou2021asymptotic}.
The derivation is given in Appendix~\ref{app:kraus_representation_independence}.

The \emph{Hamiltonian-not-in-Kraus-span} (HNKS) condition\footnote{
HNKS is the discrete-channel analogue of HNLS.
For an infinitesimal Lindblad time step, the products of the short-time Kraus operators span $\operatorname{span}_{\C}\{I,L_k,L_k^\dagger,L_k^\dagger L_j\}$ to the relevant orders in $\sqrt{dt}$, while $H_{\rm K}=h_0\,dt+O(dt^2)$.
Hence eqn.~\eqref{eqn:HNKS_condition} reduces to HNLS as $dt\rightarrow0$.
Retaining these leading orders gives the HNLS condition for Lindblad sensing with arbitrarily fast control~\cite{zhou2018achieving,zhou2021asymptotic}.
This correspondence does not assert that the exact channel for every fixed $dt>0$ has the same Kraus span as its leading-order approximation.
Also, this limiting correspondence does not extend the HNKS theorem to an arbitrary temporally correlated multi-time process.
} 
is defined as
\begin{equation} 
    H_{\rm K}(\mathcal{C}_\theta) \notin \mathcal{S}_{\rm K}(\mathcal{C}_\theta). 
    \label{eqn:HNKS_condition} 
\end{equation} 
HNKS applies when the sensing task consists of \(N\) identical and independent uses of \(\mathcal{C}_\theta\), with arbitrary noiseless ancillas and arbitrary control operations between the uses.
It therefore does not apply directly to a general temporally correlated process \(\mathbf{\Lambda}_\theta^{(N)}\).

For \(N\) identical and independent uses of a finite-dimensional single-parameter channel, with arbitrary noiseless ancillas and arbitrary controls, Theorem~1 of Ref.~\cite{zhou2021asymptotic} directly proves that both the optimal parallel quantum Fisher information and the optimal sequential quantum Fisher information scale as \(\Theta(N^2)\) if and only if the HNKS condition in eqn.~\eqref{eqn:HNKS_condition} is satisfied.
If HNKS fails, both quantum Fisher informations are \(O(N)\); when the channel is locally identifiable and the asymptotic linear coefficient is nonzero, this scaling is \(\Theta(N)\).
Thus no additional derivation is required for the repeated independent channel considered in that theorem.

Theorem~1 of Ref.~\cite{zhou2021asymptotic}, however, does not apply directly to an arbitrary non-Markovian \(N\)-step process \(\mathbf{\Lambda}_\theta^{(N)}\).
{For such general temporally correlated process, the Kraus-gauge optimization must retain the causal constraints on the sequential strategy.
This gives the optimal comb quantum Fisher information for fixed $T$~\cite{Altherr_2021} and can also give asymptotic upper bounds for correlated processes~\cite{Kurdzialek_2025}.
In particular, a vanishing gauge term in the latter bounds can rule out Heisenberg scaling, whereas a nonzero term does not establish its achievability.
Appendix~\ref{app:correlated_kraus_bounds} discusses the construction and its relation to the single-parameter SDP bound.}

For a noise with its temporal correlation determined by a latent variable $j$, which will be considered in the following subsection, we will show how HNKS can be used to derive a sufficient criterion for the achievability of Heisenberg scaling (Proposition~\ref{prop:finite_latent_hs}): 
HNKS is first applied after conditioning on the knowledge of the latent noise parameter, and the conditional result is then lifted to the problem of estimating $\theta$ when the latent noise label is unknown.

\subsection{Finite-latent correlated-noise criterion}
\label{subsec:finite_latent_hs}

We now consider a temporally correlated sensing process governed by a finite latent noise label \(j\).
The label is sampled once at the beginning of the protocol, independently of \(\theta\), and remains fixed throughout all sensing rounds.
Conditioned on \(j\), the experimentalist faces a known-\(j\) sensing problem; without this conditioning, the process is generally not a product of independent one-step channels.

More explicitly, let $\theta\in\R$ be the parameter to be estimated and consider a process
\begin{equation}
    \mathbf{\Lambda}_{\theta}^{(T)}
    =
    \sum_{j\in\mathcal{K}}
    p_j\,
    \mathbf{\Lambda}_{\theta|j}^{(T)},
    \label{eqn:finite_latent_process}
\end{equation}
where $\mathcal{K}$ is a finite set of latent noise labels, probability $p_j$ is independent of $\theta$, and $j$ is sampled once at the beginning of the protocol and remains fixed throughout all $T$ sensing rounds.

If \(\mathbf{\Lambda}_{\theta|j}^{(T)}\) consists of \(T\) identical and independent uses of a conditional channel \(\mathcal{C}_{\theta|j}\) satisfying HNKS, then Theorem~1 of Ref.~\cite{zhou2021asymptotic} verifies the quadratic-QFI condition in the proposition.
For a repeated independent conditional channel, HNKS is necessary and sufficient for quadratic QFI.
For a more general known-$j$ process, the quadratic-QFI condition may instead be achieved using a conditional sensing strategy using initial calibration rounds, which we will discuss shortly.
Below we define a class of admissible calibration strategies with exponentially small error, allowing noise label $j$ to be learned using only $O(\log T)$ of the $T$ available sensing rounds.

\begin{definition}[{Admissible calibration strategy}]\label{def:admissible_calibration}
    For the process in eqn.~\eqref{eqn:finite_latent_process}, a family of calibration protocols and sensing continuations is
    \emph{admissible} at $\theta$ if there exist an open
    neighborhood $\mathcal{U}\subset\R$ of $\theta$,
    a compact interval $\theta\in\mathcal{I}\subset\mathcal{U}$,
    constants $A,c>0$, and integers $n_0,N_1$ such that
    the following hold for every $\theta'\in\mathcal{U}$,
    $n\geq n_0$, $N\geq N_1$, and $j$ with $p_j>0$.

    \begin{enumerate}
        \item
        \label{item:admissible_calibration_discrimination}
        \atchg{Calibration uses $n$ rounds and reports
        $\hat j$ with error probability
        $P_{\rm err}(n;\theta')$, averaged over
        $\{p_j\}_j$, satisfying}
        \begin{equation}
            P_{\rm err}(n;\theta')
            \leq A e^{-cn}
            \qquad
            \forall\theta'\in\mathcal{U}.
        \end{equation}

        \item
        \label{item:admissible_calibration_continuation}
        Conditioned on $j$, calibration admits a fresh $N$-round continuation with process $\mathbf{\Lambda}_{\theta'|j}^{(N)}$ independent of the calibration outcomes.

        At $\theta$, the known-$j$ problem satisfies the hypotheses of Proposition~\ref{prop:single_par_achievability} and admits an optimal locally unbiased estimator with all outcome values in $\mathcal{I}$.\footnote{
            \atchg{For each reported label, the selected strategy, measurement settings, and estimator values are fixed at the point $\theta$ as $\theta'$ varies.}
            \atchg{The requirement that all estimator values belong to $\mathcal{I}$ includes outcomes that have zero probability when $j$ is known but may occur when calibration reports an incorrect label.}
        }

        \item
        \label{item:admissible_calibration_encoding}
        Each sensing round of the complete $T=n+N$ protocol uses exactly one application of
        \begin{equation}
            U_{\theta'}=e^{-\iota\theta'G},
        \end{equation}
        where $G$ is a fixed, bounded, Hermitian operator that is independent of $T$.
        Conditioned on $j$, the initial joint state and all other operations are independent of $\theta'$.
    \end{enumerate}
\end{definition}

Now we state sufficient conditions for Heisenberg-limit achievability using the strategy with admissible calibration above.

\begin{proposition}[Finite-latent correlated-noise criterion]\label{prop:finite_latent_hs}
    {Suppose that the process in
    eqn.~\eqref{eqn:finite_latent_process} admits an admissible
    calibration scheme at point $\theta$ in the sense of
    Definition~\ref{def:admissible_calibration},
    with neighborhood $\mathcal{U}$.}
    For each $j$ with $p_j>0$, define the optimal known-$j$
    sequential quantum Fisher information by
    \begin{equation}
        \mathcal{J}_j(N;\theta')
        :=
        \mathcal{J}_{\seq}
        \left(
            \mathbf{\Lambda}_{\theta'|j}^{(N)}
        \right),
    \end{equation}
    for $\mathcal{J}_{\seq}$ as defined in
    eqn.~\eqref{eqn:single_par_SDP_fisher_info}.
    
    {Suppose that there exist positive constants
    $\{\kappa_j\}_{j:p_j>0}$ and an integer $N_0$,
    independent of $N$ and $\theta'$, such that for
    every $N\geq N_0$,}
    \begin{equation}
        \mathcal{J}_j(N;\theta')
        \geq \kappa_jN^2
        \qquad
        \forall\theta'\in\mathcal{U},
        \label{eqn:finite_latent_quadratic_qfi}
    \end{equation}
    where $\kappa_j>0$ whenever $p_j>0$.
    
    Then the unknown-$j$ sensing problem admits an  estimation strategy locally unbiased at $\theta$ with Heisenberg-scaling MSE,
    \begin{equation}
        \mse_\theta=O(T^{-2}).
    \end{equation}
\end{proposition}

The proof of Proposition~\ref{prop:finite_latent_hs} is given in Appendix~\ref{app:proof_finite_latent_hs}.

Proposition~\ref{prop:finite_latent_hs} transfers Heisenberg scaling from the known-$j$ sensing problems to the original problem with an unknown noise label.
An $O(\log T)$-round calibration selects which known-$j$ sensing strategy to use, leaving $T-O(\log T)$ rounds for estimating $\theta$.
The exponentially small identification error and bounded estimator values make the contribution from an incorrect label negligible compared with $T^{-2}$.
The operational encoding condition allows local unbiasedness to be restored while preserving this scaling.

\section{Noisy sensing MSE Bounds and optimal strategies}\label{sec:examples}

In this section, we demonstrate the SDP bound given in Theorem~\ref{thm:sequential_NH_bound} on a few different sequential sensing problems and derive an optimal protocol for cases which the bound is attainable.
The first sequential sensing problem (Section~\ref{sec:correlated_noise_z_rot}) involves estimating the angle of a single-qubit Pauli-$Z$ rotation channel under different temporally correlated noise models.
For its finite-latent noise models, we verify both admissible calibration and the conditional-QFI hypothesis, given in Section~\ref{subsec:finite_latent_hs}.
Second, in Section~\ref{sec:sensing_correlated_random_noise_direction}, we consider the problem of estimating the same Pauli-$Z$ rotation under temporally correlated noise with unknown noise directions.
Third, in Section~\ref{sec:example_multiparameter_rotations_depolarizing}, we consider single- and multi- parameter sequential sensing problem under independent depolarizing noise, which extends the single-step results in~\cite[Section III.A]{Conlon_2023}.
Lastly, in Section~\ref{sec:parseq2}, we consider the same sensing problem and compare a fully sequential strategy with a strategy that probes two channels in parallel at every step; 
for the single-parameter Pauli-$Z$ model, the attainable fully sequential optimal MSE is numerically smaller, whereas for the two-parameter Pauli-$YX$ model we compare only the corresponding SDP lower bounds, at the same total number of channel uses.
Additional examples are given in Appendix~\ref{app:additional_examples}.

We note that single-parameter metrology for non-Markovian processes has previously been formulated in terms of the comb quantum Fisher information, with exact small-step semidefinite programs and persistent-environment collision-model examples~\cite{Altherr_2021}.
More recently, scalable and systematically improvable upper bounds on the adaptive quantum Fisher information were derived for correlated noise, including correlated dephasing~\cite{Kurdzialek_2025}.
The models we consider here are different: a Pauli label or equatorial noise direction is sampled once and remains fixed throughout the protocol.  

\subsection{Pauli-$Z$ rotations with correlated noise}\label{sec:correlated_noise_z_rot}

Here we consider sequential Pauli-\(Z\) sensing under two temporally correlated noise models.
We write $R_z(\theta):=R_{z,\theta}=e^{-\iota\theta\sigma_z/2}$ and $\mathcal{R}_{z,\theta}(\rho):=R_z(\theta)\rho R_z(\theta)^\dagger$, and set $\sigma_0:=I$.
In both models, a finite latent Pauli label \(j\) is sampled once at the
beginning of the protocol and remains fixed throughout all \(T\) sensing
rounds.
In the strongly correlated model, conditioning on \(j\) fixes the same
unitary Pauli error at every round.
In the weakly correlated model, conditioning on \(j\) fixes the repeated
one-step Pauli noise channel, while the error-or-no-error event within
that channel remains stochastic at each round.

Neither the HNLS condition of Section~\ref{subsec:timedependent} nor the
HNKS condition of Section~\ref{subsec:HNKS} should be applied to the
one-step marginal obtained by averaging over \(j\), because this
marginalization discards the temporal memory.
Instead, we can apply HNKS to the conditional strongly correlated channel and to the conditional weakly correlated channel, where Proposition~\ref{prop:finite_latent_hs} then lifts the resulting known-\(j\) Heisenberg scaling to the original unknown-\(j\) process, once the remaining admissible calibration requirements in Section~\ref{subsec:finite_latent_hs} have been verified.

For comparison, non-Markovian extensions of HNLS have also been studied
directly.
In particular, Ref.~\cite{Mann_2025} derives the
Hamiltonian-not-in-extended-span condition and the weaker
Hamiltonian-not-in-extended-Lindblad-span condition.
These direct non-Markovian criteria are not required for the
finite-latent reduction used here.

First, let us consider sequential estimation of $T$ copies of a Pauli-$Z$ rotation under a strongly correlated noise.
We consider a correlated Pauli noise model parameterized by $p=(p_x,p_y,p_z)$ where $p_j$ is the probability of error $\sigma_j$ occurring after \textit{every} channel being probed.
We denote $p_0=1-p_x-p_y-p_z$ as the probability of no errors.
Because the noise is correlated, we must model all $T$ noisy  copies of the Pauli-$Z$ rotation as a single process $\mathbf{\Lambda}_\theta^\mathrm{st}$, which can be illustrated as
\begin{equation}
\begin{aligned}
    &\mathbf{\Lambda}_\theta^\mathrm{st} = \sum_{j\in\{0,x,y,z\}} p_j \kketbra{\sigma_jR_{z,\theta}}^{\otimes T} \\
    &\;= \sum_j p_j
    \resizebox{0.74\linewidth}{!}{
    \begin{quantikz}[
        classical gap=0.03cm,
        column sep=0.28cm,
        font=\normalsize,
        labels={font=\footnotesize}
    ]
        & \gate[1,style={fill=green!20}]{R_{z,\theta}} \wire[l][1]["{A_0}"{above,pos=0.6}]{q} \gshade{1}{2}{} & \gate[style={fill=red!20}]{\sigma_j} \wire[r][1]["{A_1}"{above,pos=0.55}]{q} & \invgate{\dots} & \gate[1,style={fill=green!20}]{R_{z,\theta}} \wire[l][1]["{A_{T-1}'}"{above,pos=0.7}]{q} \gshade{1}{2}{} & \gate[style={fill=red!20}]{\sigma_j} \wire[r][1]["{A_T}"{above,pos=0.55}]{q} \wireoverride{q} & 
    \end{quantikz} 
    } .
\end{aligned}
\end{equation}

\begin{figure}
    \centering
    \includegraphics[width=1\columnwidth]{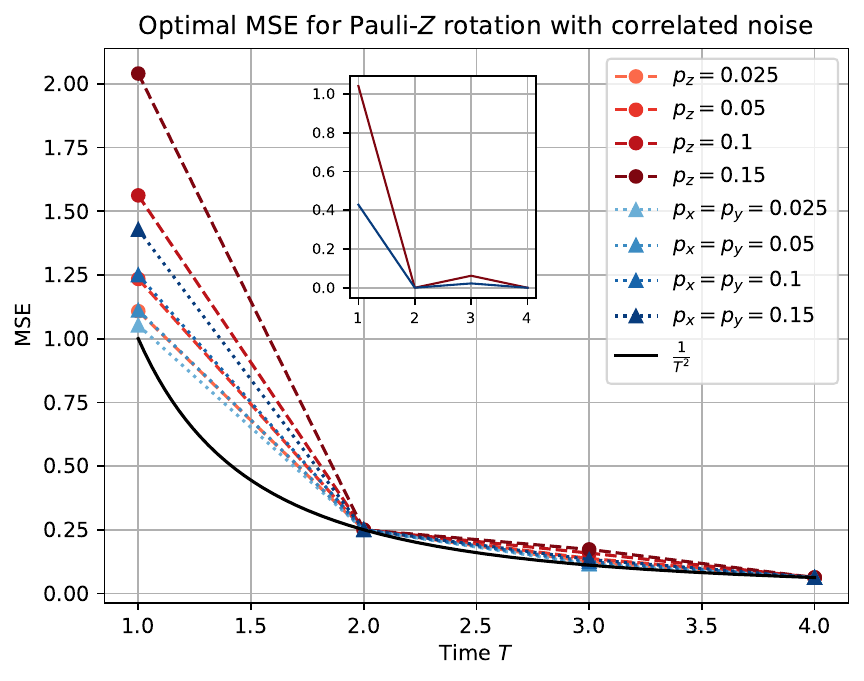}
    \caption{Optimal mean-squared error (MSE) computed from the SDP in Theorem~\ref{thm:sequential_NH_bound} for sequential Pauli-$Z$ rotation angle estimation with strongly correlated Pauli noise.
    The $x$-axis shows the number of time steps, where a single use of the channel counts as a single time step, whereas the $y$-axis shows the MSE.
    The colored lines show the optimal MSEs for different types of correlated noise on the $Z$-rotation channels at different probabilities.
    Correlated Pauli $X,Y,Z$ error probabilities are shown in the legend as $p_x,p_y,p_z$, respectively.
    Red lines are MSEs in the case of noise with non-zero $p_z$ error rate (with $p_x = p_y = 0$) and blue lines are lower bounds in the case of noise with non-zero $p_x=p_y$ error rate (with $p_z = 0$).
    Inset plot shows the difference between the computed MSEs and the Heisenberg limit $1/T^2$ for $T\in\{1,2,3,4\}$ channel uses.
    For the even values $T=2,4$ shown here, the optimal MSE equals
    the exact noiseless value $1/T^2$.
    Optimality of the computed MSE follows from Proposition~\ref{prop:single_par_achievability}
    since there is only one parameter (the Pauli-$Z$ rotation angle).
    }
    \label{fig:MSE_corr_noise_Z_rot}
\end{figure}

For both non-zero $p_z$ error rate (with $p_x = p_y = 0$) and non-zero $p_x=p_y$ error rate (with $p_z = 0$), as shown in Fig.~\ref{fig:MSE_corr_noise_Z_rot} up to $T=4$ steps, the SDP given in Theorem~\ref{thm:sequential_NH_bound} gives us a bound of $\frac{1}{T^2}$ at the tested even values $T=2,4$, matching the exact noiseless MSE.
Since this is a single-parameter problem, Proposition~\ref{prop:single_par_achievability} shows that these SDP values are attainable and optimal, rather than merely lower bounds on the MSE.
If the latent Pauli label $j$ is discarded and the process is replaced by a fictitious one-step memoryless marginal channel, the corresponding Markovian HNLS condition may be violated.
For example, when both \(X\)- and \(Y\)-type jump operators occur in the same memoryless model, the Lindblad span contains \(X Y=\iota Z\).
However, this does not diagnose the correlated process \(\mathbf{\Lambda}_\theta^{\rm st}\), because in the actual process a single value of \(j\) is sampled once.
Conditioned on \(j\), the noise is a known unitary and can be inverted without suppressing the Pauli-\(Z\) signal.

This model admits a particularly simple interpretation. 
A single Pauli label $j\in\{0,x,y,z\}$ is sampled once at the beginning of the protocol and then remains fixed throughout all $T$ sensing rounds. 
Consequently, conditioned on $j$, the noise process is deterministic rather than stochastic. 
The sensing problem therefore reduces to first identifying the latent Pauli label, followed by applying the corresponding recovery operation after every subsequent sensing round.
We describe this more precisely below.

\begin{atenv}
Conditioned on the latent label $j\in\{0,x,y,z\}$, the repeated one-step channel is
\begin{equation}
    \mathcal{C}_{\theta|j}^{\rm st}(\rho)
    =
    \sigma_jR_z(\theta)\rho R_z(\theta)^\dagger\sigma_j,
\end{equation}
with the single Kraus operator
\(K_{j,\theta}^{\rm st}=\sigma_jR_z(\theta)\).
Consequently,
\begin{equation}
\begin{aligned}
    H_{\rm K}
    \left(
        \mathcal{C}_{\theta|j}^{\rm st}
    \right)
    &=
    \iota
    (K_{j,\theta}^{\rm st})^\dagger
    \frac{\partial K_{j,\theta}^{\rm st}}{\partial\theta}
    =
    \frac{\sigma_z}{2},\\
    \mathcal{S}_{\rm K}
    \left(
        \mathcal{C}_{\theta|j}^{\rm st}
    \right)
    &=
    \operatorname{span}_{\R}\{I\}.
\end{aligned}
\label{eqn:strong_conditional_HNKS}
\end{equation}
Thus every conditional branch satisfies HNKS (c.f. eqn.~\eqref{eqn:HNKS_condition}).

Moreover, the normalized overlaps of the Kraus operators of the one-step channel $\mathcal{C}_{\theta|j}^{\rm st}$ satisfy
\begin{equation}
\begin{aligned}
    \frac{1}{2}
    \tr\left(
        (K_{j,\theta}^{\rm st})^\dagger
        K_{k,\theta}^{\rm st}
    \right)
    &=
    \frac{1}{2}\tr(\sigma_j\sigma_k)
    =
    \delta_{j,k}.
\end{aligned}
\end{equation}
\atchg{Let $C_1$ be a noiseless ancilla qubit, and consider the following states with $C_1$ as the first tensor factor and the probe register $A_0$ as the second:}
\begin{equation}\label{eqn:bell_basis}
    |\Phi^+\>
    :=
    \frac{|00\>+|11\>}{\sqrt{2}}\in C_1\otimes A_0,
    \quad
    |\Phi_j\>
    :=
    (I\otimes\sigma_j)|\Phi^+\>.
\end{equation}
To compensate the nominal rotation at the local operating point $\theta_0$, apply $R_{z,\theta_0}^\dagger$ to the probe half of $|\Phi^+\>$ immediately before the unknown channel.
For an actual parameter value $\theta'$, the resulting normalized conditional Choi state is
\begin{equation}
    |\Phi_{j,\theta'}\>
    :=
    \left(
        I\otimes
        \sigma_j
        R_z(\theta')
        R_z(\theta_0)^\dagger
    \right)
    |\Phi^+\>.
\end{equation}
At $\theta'=\theta_0$, this reduces to $|\Phi_j\>$.
The four states are orthonormal because
\begin{equation}
    \<\Phi_j|\Phi_k\>
    =
    \frac{1}{2}
    \tr(\sigma_j^\dagger\sigma_k)
    =
    \delta_{j,k}.
\end{equation}
A measurement in the Bell basis $\{|\Phi_j\>\}_j$ therefore identifies $j$ exactly in one calibration round at $\theta_0$.

For $\delta:=\theta'-\theta_0$, the Bell measurement returns the true label $j$ with probability
\begin{equation}
    r(\delta)
    =
    \cos^2\left(\frac{\delta}{2}\right)
\end{equation}
and returns only its paired label, with the pairings $0\leftrightarrow z$ and $x\leftrightarrow y$, with probability $1-r(\delta)$.
Choose
\begin{equation}
    \mathcal{U}
    :=
    \left(
        \theta_0-\frac{\pi}{3},
        \theta_0+\frac{\pi}{3}
    \right).
\end{equation}
Then $r(\delta)\geq3/4$ uniformly for $\theta'\in\mathcal{U}$.

\atchg{Use a fresh probe--ancilla pair in every calibration round; the $n$ outcomes are independent conditioned on the fixed label $j$.}
\atchg{Infer $j$ by majority vote within the observed pair $\{0,z\}$ or $\{x,y\}$, breaking ties by a fixed deterministic rule.}
Hoeffding's inequality gives
\begin{equation}
\begin{aligned}
    P_{\rm err}(n;\theta')
    &\leq
    \exp
    \left(
        -\frac{n}{2}
        \left(
            2r(\delta)-1
        \right)^2
    \right)\\
    &=
    \exp
    \left(
        -\frac{n}{2}
        \cos^2(\delta)
    \right)
    \leq
    e^{-n/8}
\end{aligned}
\end{equation}
uniformly over $\theta'\in\mathcal{U}$.
Thus the calibration-error bound in Definition~\ref{def:admissible_calibration} holds with
\begin{equation}
    A=1,
    \qquad
    c=\frac{1}{8},
    \qquad
    n_0=1.
\end{equation}

For the quadratic-QFI condition in eqn.~\eqref{eqn:finite_latent_quadratic_qfi}, when $j$ is known, applying $\sigma_j$ after every channel use gives
\begin{equation}
    \sigma_j
    \left(
        \sigma_jR_z(\theta')
    \right)
    =
    R_z(\theta').
\end{equation}
The remaining $N$ rounds are therefore exactly noiseless Pauli-$Z$ sensing and satisfy
\begin{equation}
    \mathcal{J}_j(N;\theta')
    =
    N^2
\end{equation}
for every $j$ and every $\theta'\in\mathcal{U}$.
\atchg{Optimality follows from eqn.~\eqref{eqn:finite_latent_generator_bound}: the largest and smallest eigenvalues of $G=\sigma_z/2$ differ by one, so any $N$-round strategy satisfies $\mathcal{J}_j(N;\theta')\leq N^2$.}
Hence eqn.~\eqref{eqn:finite_latent_quadratic_qfi} holds with
\begin{equation}
    \kappa_j=1
    \quad
    \text{for every $j$ with $p_j>0$},
    \qquad
    N_0=1.
\end{equation}
\begin{atenv}
The remaining requirements of Definition~\ref{def:admissible_calibration} also hold with $N_1=1$.
After the correct label has been identified, the final logical SLD measurement has locally unbiased estimator values $\theta_0\pm1/N$ and MSE $1/N^2$.
These values belong to the fixed interval $\mathcal{I}=[\theta_0-1,\theta_0+1]\subset\mathcal{U}$.
\atchg{The measurement and estimator are fixed at $\theta_0$.
The same estimator values are assigned when calibration is incorrect; any completion outcome outside the logical subspace is assigned the value $\theta_0$.
Thus all possible estimator values belong to $\mathcal{I}$, including those on incorrectly identified branches.}

The operational encoding condition holds with $U_{\theta'}=R_z(\theta')$ and $G=\sigma_z/2$: the latent distribution and Pauli errors are independent of $\theta'$, and the nominal compensation, controls, and measurements are fixed as $\theta'$ varies.

Conditioned on the fixed label, fresh preparation after calibration gives the required concatenation property.
The same operational encoding condition holds for the weakly correlated model because each $\mathcal{E}_j$ is independent of $\theta'$.
For $p_z=0$, Appendix~\ref{app:weak_pauli_calibration} verifies the remaining admissibility requirements and gives the complete logical measurement.
\end{atenv}
The assumptions of Proposition~\ref{prop:finite_latent_hs} are therefore satisfied, which gives
\begin{equation}
    \mse_\theta = O(T^{-2}).
\end{equation}

Proposition~\ref{prop:finite_latent_hs} establishes this asymptotic bound for the full sequence of sufficiently large values of $T$ using $O(\log T)$ calibration rounds, but it does not determine the exact finite-$T$ coefficient.
By contrast, the SDP calculation in Fig.~\ref{fig:MSE_corr_noise_Z_rot} gives the exact noiseless value $1/T^2$ at the even values $T=2,4$ considered numerically.
These equalities are stronger finite-$T$ statements, but the numerical values $T=2,4$ alone do not establish exact equality for every even $T$.
\end{atenv}

For the Pauli-$Z$ rotation estimation under strongly-correlated noise over $T=2$ time steps described above, we can obtain an optimal strategy from the matrices $\mathbf{X}_{0,1}$ and $\mathbf{M}$ obtained from the SDP.
In particular, the the optimal strategy $\mathbf{M}$ can be decomposed as
\begin{equation}\label{eqn:zrot_correlated_noise_strategy}
    \begin{quantikz}[classical gap=0.05cm]
        \;\; \setwiretype{n} \bshade{2}{6}{$\mathbf{M}$} & \lstick[2]{$|\psi\>$} & \setwiretype{q} & \gate[2,style={fill=blue!20}]{V} \wire[l][1]["{C_1}"{above,pos=0.6}]{q} & & \gate[2,style={fill=blue!20}]{\mathcal{M}} \wire[l][1]["{C_2}"{above,pos=0.6}]{q} \\
        \setwiretype{n} & & \invgate{} \wire[l][1]["{A_0}"{above,pos=0.4}]{q} \wire[r][1]["{A_1}"{above,pos=0.4}]{q} \wshade{1}{1}{} \setwiretype{q} & & \invgate{} \wire[l][1]["{A_1'}"{above,pos=0.4}]{q} \wire[r][1]["{A_2}"{above,pos=0.4}]{q} \wshade{1}{1}{} & 
    \end{quantikz}  \;,
\end{equation}
for initial probe state $|\psi\> = (|00\>+|11\>)/\sqrt{2}\in C_1\otimes A_0$ and isometry $V:C_1\otimes A_1\rightarrow C_2\otimes A_1'$ given by
\begin{equation}
    V =
    \begin{quantikz}
        & \gate[1,style={blue!05,fill=blue!05}]{|0\>} \setwiretype{n} \gategroup[4,steps=3,style={dashed, fill=blue!05, inner xsep=2pt},background,label style={label position=above,anchor=north,yshift=0.3cm}]{{\sc $V$}} & \setwiretype{q} & \targ{} & \rstick[3]{$C_2$} \\
        \lstick[1]{$C_1$} & & & & \\
        \lstick[1]{$A_1$} & & \ctrl{1} & \ctrl{-2} & \\
        & \gate[1,style={blue!05,fill=blue!05}]{|0\>} \setwiretype{n} & \targ{} \setwiretype{q} & & \rstick[1]{$A_1'$}
    \end{quantikz}
    \;,
\end{equation}
which is identical to the one we obtained for Pauli-$Z$ rotation with independent noise in Section~\ref{sec:indep_noise_zrot}.
The final measurement $\mathcal{M}$ is defined by POVM $\{\kketbra{M_k}\}_{k=1}^4$ for $|M_1\rr = \frac{|0^4\>-i|1^4\>}{\sqrt{2}}$, $|M_2\rr = \frac{|0^4\>+i|1^4\>}{\sqrt{2}}$, $|M_3\rr = \frac{|0101\>-i|1010\>}{\sqrt{2}}$, and $|M_4\rr = \frac{|0101\>+i|1010\>}{\sqrt{2}}$.
In Appendix~\ref{app:correlated_noise_z_rot}, we show in detail how this strategy is obtained and derive the measurement output probabilities.

\begin{figure*}[t]
    \centering
    \begin{subfigure}[t]{0.34\linewidth}
        \centering
        \includegraphics[
            width=\linewidth,
            trim={0bp 0bp 304.49bp 28bp},clip
        ]{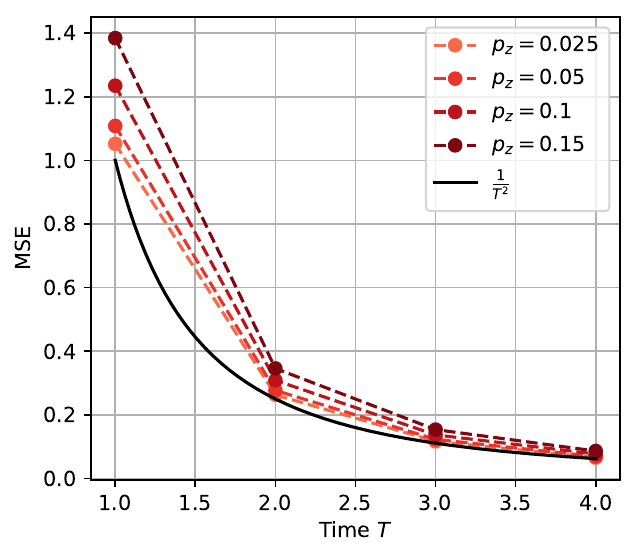}
        \caption{}
        \label{fig:MSE_weakly_correlated_z}
    \end{subfigure}\hspace{0.02\linewidth}%
    \begin{subfigure}[t]{0.34\linewidth}
        \centering
        \includegraphics[
            width=\linewidth,
            trim={304.49bp 0bp 0bp 28bp},clip
        ]{Figures/z-wkcorr_plot.pdf}
        \caption{}
        \label{fig:MSE_weakly_correlated_xy}
    \end{subfigure}
    \caption{
    Optimal mean-squared error (MSE) computed from the SDP in Theorem~\ref{thm:sequential_NH_bound} for sequential Pauli-$Z$ rotation angle estimation with weakly correlated Pauli noise.
    The $x$-axis shows the number of time steps, where a single use of the channel counts as a single time step, whereas the $y$-axis shows the MSE.
    In panel~(\subref{fig:MSE_weakly_correlated_z}), the correlated noise applies Pauli noise channel $\mathcal{E}_z$ at every step with probability $p_z$.
    In panel~(\subref{fig:MSE_weakly_correlated_xy}), the correlated noise applies Pauli noise channel $\mathcal{E}_x$ or $\mathcal{E}_y$ at every step with probability $p_x,p_y$, respectively.
    In both cases, as the noise strength \(p\) gets smaller, the MSE
    approaches the noiseless benchmark \(1/T^2\) shown by the solid black
    line. The finite-\(T\) separation from this benchmark does not by itself
    determine the asymptotic scaling exponent.
    Optimality of the computed MSE follows from Proposition~\ref{prop:single_par_achievability} since there is only one parameter (the Pauli-$Z$ rotation angle).
    }
    \label{fig:MSE_weakly_correlated}
\end{figure*}

We now consider a weakly correlated noise channel that with probability $p_j$, applies a single-qubit Pauli-$j$ channel $\mathcal{E}_j$ for $j\in\{x,y,z\}$ after all Pauli-$Z$ rotations.
Such noise channels are given by
\begin{equation}
    \mathcal{E}_j(\rho) = \frac{1}{2} (\rho + \sigma_j\rho\sigma_j) \;.
\end{equation}
We also define \(\mathcal{E}_0:=\operatorname{id}\) and \(p_0:=1-p_x-p_y-p_z\), so the latent label takes values in \(\{0,x,y,z\}\).
Thus, the entire process $\mathbf{\Lambda}_\theta^\mathrm{wk}$ can be illustrated as:
\begin{equation}
\begin{aligned}
    \sum_{j\in\{0,x,y,z\}}p_j
    \resizebox{0.74\linewidth}{!}{
    \begin{quantikz}[
        classical gap=0.03cm,
        column sep=0.28cm,
        font=\normalsize,
        labels={font=\footnotesize}
    ]
        & \gate[1,style={fill=green!20}]{R_{z,\theta}} \wire[l][1]["{A_0}"{above,pos=0.6}]{q} \gshade{1}{2}{} & \gate[style={fill=red!20}]{\mathcal{E}_j} \wire[r][1]["{A_1}"{above,pos=0.5}]{q} & \invgate{\dots} & \gate[1,style={fill=green!20}]{R_{z,\theta}} \wire[l][1]["{A_{T-1}'}"{above,pos=0.73}]{q} \gshade{1}{2}{} & \gate[style={fill=red!20}]{\mathcal{E}_j} \wire[r][1]["{A_T}"{above,pos=0.55}]{q} \wireoverride{q} & 
    \end{quantikz} 
    } .
\end{aligned}
\end{equation}
As opposed to the previous process $\mathbf{\Lambda}_\theta^\mathrm{st}$ with strongly-correlated noise, the noise in $\mathbf{\Lambda}_\theta^\mathrm{wk}$ is weaker in the sense that instead of the Pauli error $\sigma_j$ occurring at every step with probability $p_j$, the sequence of errors over the $T$ steps consists of a quantum channel after every Pauli-$Z$ rotation that applies either a Pauli error $\sigma_j$ or no error.
Taking $T=4$ as an example, the error sequence $\sigma_x,\sigma_0,\sigma_0,\sigma_x$ occurs with probability $p_x2^{-4}$, but $\sigma_z,\sigma_0,\sigma_0,\sigma_x$ cannot occur.

Panels~(\subref{fig:MSE_weakly_correlated_z}) and~(\subref{fig:MSE_weakly_correlated_xy}) of Fig.~\ref{fig:MSE_weakly_correlated} only establish that the finite-\(T\) optimal MSE lies above the noiseless value \(1/T^2\) for the parameters considered.
Whether the asymptotic \(T^{-2}\) scaling survives is instead determined below by applying HNKS (Section~\ref{subsec:HNKS}) conditionally on the persistent Pauli label.
This observation does not by itself imply the loss of Heisenberg scaling: an MSE of \(c/T^2\) with \(c>1\) depending only on error probabilities $p_x,p_y,p_z$ remains Heisenberg scaling.

\begin{atenv}
The conditional channel \(\mathcal E_j\circ\mathcal R_{z,\theta}\) has Kraus operators
\begin{equation}
    K_{0,j}(\theta)
    =
    \frac{1}{\sqrt{2}}R_z(\theta),
    \qquad
    K_{1,j}(\theta)
    =
    \frac{1}{\sqrt{2}}\sigma_jR_z(\theta).
\end{equation}
Using eqns.~\eqref{eqn:channel_hamiltonian_HNKS}
and~\eqref{eqn:kraus_span_HNKS}, the channel Hamiltonian and Hermitian
Kraus span of the conditional channel are
\begin{equation}
\begin{aligned}
    H_{\rm K}
    \left(
        \mathcal E_j\circ\mathcal R_{z,\theta}
    \right)
    &=
    \frac{\sigma_z}{2}, \quad\text{and}\\
    \mathcal S_{\rm K}
    \left(
        \mathcal E_j\circ\mathcal R_{z,\theta}
    \right)
    &=
    \operatorname{span}_{\R}
    \left\{
        I,\,
        R_z(\theta)^\dagger
        \sigma_j
        R_z(\theta)
    \right\}.
\end{aligned}
\label{eqn:weak_conditional_HNKS}
\end{equation}
These are precisely the operators appearing in the HNKS test \eqref{eqn:HNKS_condition}.
For \(j=x\) or \(j=y\),
\(R_z(\theta)^\dagger\sigma_jR_z(\theta)\) lies in the Pauli-\(X\)--\(Y\) plane, and hence HNKS is satisfied:
\begin{equation}
    \frac{\sigma_z}{2}
    \notin
    \mathcal S_{\rm K}
    \left(
        \mathcal E_j\circ\mathcal R_{z,\theta}
    \right).
\end{equation}
The conditional \(X\)- and \(Y\)-axis models therefore satisfy HNKS.
The identity branch \(j=0\) is unitary and also satisfies HNKS.

The \(Z\)-axis branch instead satisfies
\begin{equation}
    \mathcal S_{\rm K}
    \left(
        \mathcal E_z\circ\mathcal R_{z,\theta}
    \right)
    =
    \operatorname{span}_{\R}\{I,\sigma_z\},
\end{equation}
so its channel Hamiltonian belongs to its Kraus span and HNKS fails.
Indeed,
\begin{equation}
    \mathcal E_z\circ\mathcal R_{z,\theta}
    =
    \mathcal E_z,
\end{equation}
and the pure \(Z\)-dephasing branch contains no information about
\(\theta\).

When \(p_z=0\), every positive-probability conditional branch satisfies HNKS.
\atchg{The persistent label can be learned with exponentially small error using Bell-basis calibration.
At $\theta=\theta_0$, the average error probability is $(p_x+p_y)2^{-n}\leq2^{-n}$.
Appendix~\ref{app:weak_pauli_calibration} derives this expression and the neighborhood-uniform bound with $A=1$, $c=1/32$, and $n_0=1$.
It also gives the conditional error-correction protocol with $\mathcal{J}_j(N;\theta')=N^2$.}
The appendix verifies the continuation and bounded-estimator requirements of Definition~\ref{def:admissible_calibration}.
The operational encoding requirement holds with $G=\sigma_z/2$ and parameter-independent $\mathcal{E}_j$. Proposition~\ref{prop:finite_latent_hs} therefore gives $\mse_\theta=O(T^{-2})$ when $p_z=0$.

Thus the finite-\(T\) separation from \(1/T^2\) in
Fig.~\ref{fig:MSE_weakly_correlated}(\subref{fig:MSE_weakly_correlated_xy})
is a difference in the finite-\(T\) prefactor and does not
establish the loss of the Heisenberg exponent.

When \(p_z>0\), Proposition~\ref{prop:finite_latent_hs} does not apply directly because the \(Z\)-axis branch is not Heisenberg limited.
This failure of a sufficient condition is not itself a no-Heisenberg-scaling result.
In particular, if \(p_0>0\), the identity branch remains informative and can contribute an \(O(T^2)\) term after the latent branches are asymptotically resolved.
The noise model with $p_z=1$, for which $\mathcal{E}_z(\rho)=(\rho+\sigma_z\rho\sigma_z)/2$, by contrast, has zero quantum Fisher information.
A general mixed Pauli-\(Z\) noise model therefore may require a separate analysis or an asymptotic upper bound.
\end{atenv}

\subsection{Pauli-$Z$ rotations with unknown noise direction}\label{sec:sensing_correlated_random_noise_direction}

\begin{figure*}[!ht]
    \centering
    \setlength{\figHeight}{0.30\linewidth}
    
    \begin{subfigure}[t]{0.32\linewidth}
        \centering
        \includegraphics[
            height=\figHeight,
            trim={0bp 0bp 596.978bp 28bp},clip
        ]{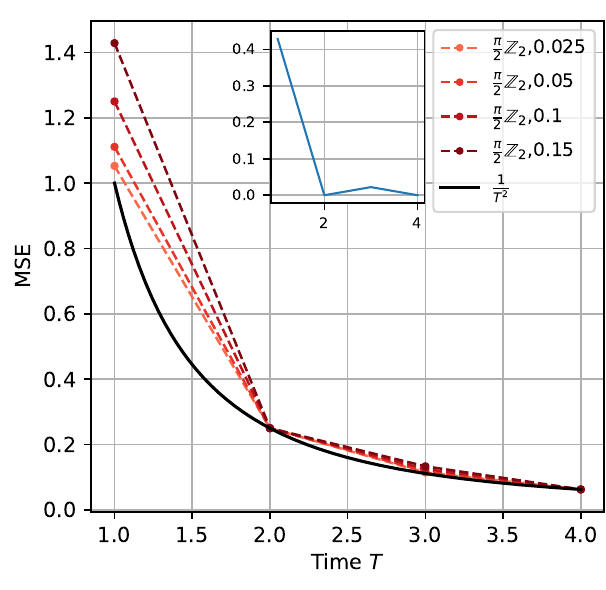}
        \caption{}
        \label{fig:MSE_z_rot_rand_noise_dir_l2}
    \end{subfigure}\hfill%
    \begin{subfigure}[t]{0.32\linewidth}
        \centering
        \includegraphics[
            height=\figHeight,
            trim={291bp 0bp 297.978bp 28bp},clip
        ]{Figures/Z_rot_unk_noise_plot.pdf}
        \caption{}
        \label{fig:MSE_z_rot_rand_noise_dir_l3}
    \end{subfigure}\hfill%
    \begin{subfigure}[t]{0.32\linewidth}
        \centering
        \includegraphics[
            height=\figHeight,
            trim={590bp 0bp 0bp 28bp},clip
        ]{Figures/Z_rot_unk_noise_plot.pdf}
        \caption{}
        \label{fig:MSE_z_rot_rand_noise_dir_l4}
    \end{subfigure}
    \caption{
    Optimal mean-squared error (MSE) computed from the SDP in Theorem~\ref{thm:sequential_NH_bound} for Pauli-$Z$ rotations with correlated random noise direction.
    Panels~(\subref{fig:MSE_z_rot_rand_noise_dir_l2}),
    (\subref{fig:MSE_z_rot_rand_noise_dir_l3}), and
    (\subref{fig:MSE_z_rot_rand_noise_dir_l4}) correspond to
    $l=2,3,4$, respectively.
    The $x$-axis shows the number of time steps, where a single use of the channel counts as a single time step, whereas the $y$-axis shows the MSE.
    The line labeled by $\frac{\pi}{l}\Z_l,p$ in the plot legend indicates the MSE corresponding to correlated noise distribution over angles $0,\pi/l, 2\pi/l,\dots,(l-1)\pi/l$ with uniform probability $p$. 
    Inset plot shows the difference between the computed MSEs and the Heisenberg limit $1/T^2$ for $T\in\{1,2,3,4\}$ channel uses.
    Note that regardless of error rate $p$, Heisenberg scaling $\frac{1}{T^2}$ is achieved at time steps that are multiples of $l$ for noise direction $\frac{\pi}{l}\Z_l$.
    Optimality of the computed MSE follows from Proposition~\ref{prop:single_par_achievability} since there is only one parameter (the Pauli-$Z$ rotation angle).
    }
\label{fig:MSE_z_rot_rand_noise_dir}
\end{figure*}

Next we consider the sequential estimation of $T$ copies of a Pauli-$Z$ rotation angle in the presence of unknown
noise. We consider two different noise models. 

\subsubsection{Unknown correlated noise}\label{sec:unknown_corr_noise}
First consider correlated noise $E_\varphi = (\cos\varphi) X + (\sin\varphi) Y$ with $\varphi$ a random angle drawn from a discrete set.
For our numerical study, we consider three uniform distributions $\{p_\varphi\}_\varphi$ over the noise angles:
$p_0^{(2)}=p_{\pi/2}^{(2)}=1/2$, $p_0^{(3)}=p_{\pi/3}^{(3)}=p_{2\pi/3}^{(3)}=1/3$, and $p_0^{(4)}=p_{\pi/4}^{(4)}=p_{2\pi/4}^{(4)}=p_{3\pi/4}^{(4)}=1/4$, where $p_\varphi$ is the probability that error $E_\varphi$ occurred after each of the $T$ copies of rotation channels being probed.
The noisy $T$-step process $\mathbf{\Lambda}_{\theta,\varphi}^{(l)}$ over angle distribution $l$ can be illustrated as
\begin{equation}
\begin{aligned}
    &\mathbf{\Lambda}_{\theta,\varphi}^{(l)} = \sum_\varphi p_\varphi^{(l)} \kketbra{E_\varphi R_{z,\theta}}^{\otimes T} \\
    &= \sum_\varphi p_\varphi^{(l)}
    \resizebox{0.75\linewidth}{!}{
    \begin{quantikz}[
        classical gap=0.03cm,
        column sep=0.28cm,
        font=\normalsize,
        labels={font=\footnotesize}
    ]
        & \gate[1,style={fill=green!20}]{R_{z,\theta}} \wire[l][1]["{A_0}"{above,pos=0.6}]{q} \gshade{1}{2}{} & \gate[style={fill=red!20}]{E_\varphi} \wire[r][1]["{A_1}"{above,pos=0.55}]{q} & \invgate{\dots} & \gate[1,style={fill=green!20}]{R_{z,\theta}} \wire[l][1]["{A_{T-1}'}"{above,pos=0.73}]{q} \gshade{1}{2}{} & \gate[style={fill=red!20}]{E_\varphi} \wire[r][1]["{A_T}"{above,pos=0.55}]{q} \wireoverride{q} & 
    \end{quantikz} 
    } .
\end{aligned}
\end{equation}

The SDP given in Theorem~\ref{thm:sequential_NH_bound} gives us a bound of $\frac{1}{T^2}$ at every step $T$ that is a multiple of $l$ for noise angle distribution $p_\varphi^{(l)}$, as shown in Fig.~\ref{fig:MSE_z_rot_rand_noise_dir} up to $T=4$ steps, which coincides with the Heisenberg scaling.
The cases $l=2,3,4$ are shown in
panels~(\subref{fig:MSE_z_rot_rand_noise_dir_l2}),
(\subref{fig:MSE_z_rot_rand_noise_dir_l3}), and
(\subref{fig:MSE_z_rot_rand_noise_dir_l4}) of
Fig.~\ref{fig:MSE_z_rot_rand_noise_dir}, respectively.
Note that distribution $p_\varphi^{(2)}$ coincides with the $p_x=p_y$ noise in Section~\ref{sec:correlated_noise_z_rot}.
For noise angle distribution $p_\varphi^{(l)}$, Heisenberg scaling is achieved for time steps $T$ that are multiples of $l$.

These numerical observations again admit a simple analytical explanation. 
\begin{atenv}
If an error is detected, apply the fixed recovery operation $Q=X$ after each noisy channel use.
Using $XY=\iota Z$, we obtain
\begin{equation}\label{eqn:discrete_unknown_direction_recovery}
\begin{aligned}
    QE_\varphi
    &=
    X\left[
        (\cos\varphi)X+(\sin\varphi)Y
    \right]\\
    &=
    (\cos\varphi)I+\iota(\sin\varphi)Z\\
    &=
    R_{z,-2\varphi}.
\end{aligned}
\end{equation}
Since this residual Pauli-$Z$ rotation commutes with the signal rotation,
one corrected noisy step is
\begin{equation}
    QE_\varphi R_{z,\theta}
    =
    R_{z,\theta-2\varphi}.
\end{equation}
After $l$ sensing rounds, the accumulated evolution is therefore
\begin{equation}
    \left(
        QE_\varphi R_{z,\theta}
    \right)^l
    =
    R_{z,l\theta-2l\varphi}.
\end{equation}
For $\varphi=j\pi/l$, the logical error relative to the ideal signal
rotation $R_{z,l\theta}$ satisfies
\begin{equation}
\begin{aligned}
    R_{z,-2l\varphi}
    &= R_{z,-2j\pi}
    = (-1)^jI.
\end{aligned}
\end{equation}
The remaining factor is only a global phase and therefore $R_{z,-2l\varphi}$ induces the identity channel.
Hence every block of $l$ corrected noisy steps is equivalent to $l$ ideal sensing steps, explaining why the SDP recovers the noiseless Heisenberg limit whenever $T$ is a multiple of $l$.
\end{atenv}

More generally, one may be interested in a continuous-time version of the above problem, modeled via eqn.~\eqref{eq:lindblad_metrology}. 
\begin{atenv}
In this case, let $j\in\{0,\dots,l-1\}$ and define $\varphi_j=\frac{j\pi}{l}$.
Conditioned on $j$, the jump operator is
\begin{equation}
    L_j
    =
    \sqrt{\gamma}E_{\varphi_j}.
\end{equation}
Since $L_j^\dagger L_j=\gamma I$, the expected number of detected jumps after a total sensing time $t$ is $\gamma t$.
\atchg{
In the discrete-time protocol for the persistent noise directions $\varphi_j=j\pi/l$ discussed in eqn.~\eqref{eqn:discrete_unknown_direction_recovery} in Section~\ref{sec:unknown_corr_noise}, the fixed recovery operation $Q=X$ is applied after every noisy channel use.
In the continuous-time protocol considered here, the jumps occur stochastically, so the same recovery operation is instead applied only after a detected jump.
Eqn.~\eqref{eqn:discrete_unknown_direction_recovery} shows that it converts $E_\varphi$ into the residual rotation $R_{z,-2\varphi}$.
In the continuous-time protocol, by contrast, the jumps are stochastic.
We use the probe--ancilla code and the code-space/error-space measurement $\{\Pi_\text{c},\Pi_\text{e}\}$ described in Appendix~\ref{sec:known_noise_dir}; whenever the error-space outcome $\Pi_\text{e}$ is observed, we apply the same fixed recovery $Q=\sigma_x=X$ to the probe.
}
\atchg{Conditioned on the error-space outcome, the recovery acts on the jump operator as}
\begin{equation}
\begin{aligned}
    QL_j
    &= \sqrt{\gamma}QE_{\varphi_j}
    = \sqrt{\gamma}R_{z,-2\varphi_j}.
\end{aligned}
\end{equation}
Consequently, after $N$ detected jumps, the accumulated logical error is
\begin{equation}
    R_{z,-2N\varphi_j}.
\end{equation}
Whenever $N=ml$, this residual error becomes
\begin{equation}
\begin{aligned}
    R_{z,-2ml\varphi_j}
    &=
    R_{z,-2\pi mj}
    = (-1)^{mj}I.
\end{aligned}
\end{equation}
Thus, independently of $j$, the residual error induces the identity
channel whenever the number of detected jumps is a multiple of $l$.
\end{atenv}

We may therefore run the sensing protocol until a desired time $t$, and then continue the evolution until the total number of detected jumps reaches the next multiple of $l$. Denote this stopping time by $\tau_t$. Since at most $l-1$ additional errors are required, the additional waiting time has bounded expectation,
\begin{equation}
    \mathbb{E}[\tau_t-t]\leq\frac{l-1}{\gamma}.
\end{equation}
For fixed $l$, this overhead is independent of $t$, and hence
\begin{equation}
\begin{aligned}
    \mathbb{E}\left[
        \left|
            \frac{\tau_t}{t}-1
        \right|
    \right]
    &=
    \frac{\mathbb{E}[\tau_t-t]}{t} 
    \leq \frac{l-1}{\gamma t} \rightarrow 0.
\end{aligned}
\end{equation}
Hence $\tau_t/t\to1$ in expectation.

At the stopping time, the unknown logical error cancels out, while the signal has coherently accumulated for time $\tau_t$.
The resulting state is therefore equivalent to ideal noiseless sensing for interrogation time $\tau_t$, giving
\begin{equation}
    \mse_{\theta}=O\left(\frac{1}{\tau_t^2}\right)
    =O\left(\frac{1}{t^2}\right).
\end{equation}
Thus, for any fixed finite set of cyclically related noise directions, the continuous-time protocol asymptotically achieves Heisenberg scaling despite the noise direction being unknown.

\begin{figure}[b]
    \centering
    \includegraphics[width=0.75\columnwidth]{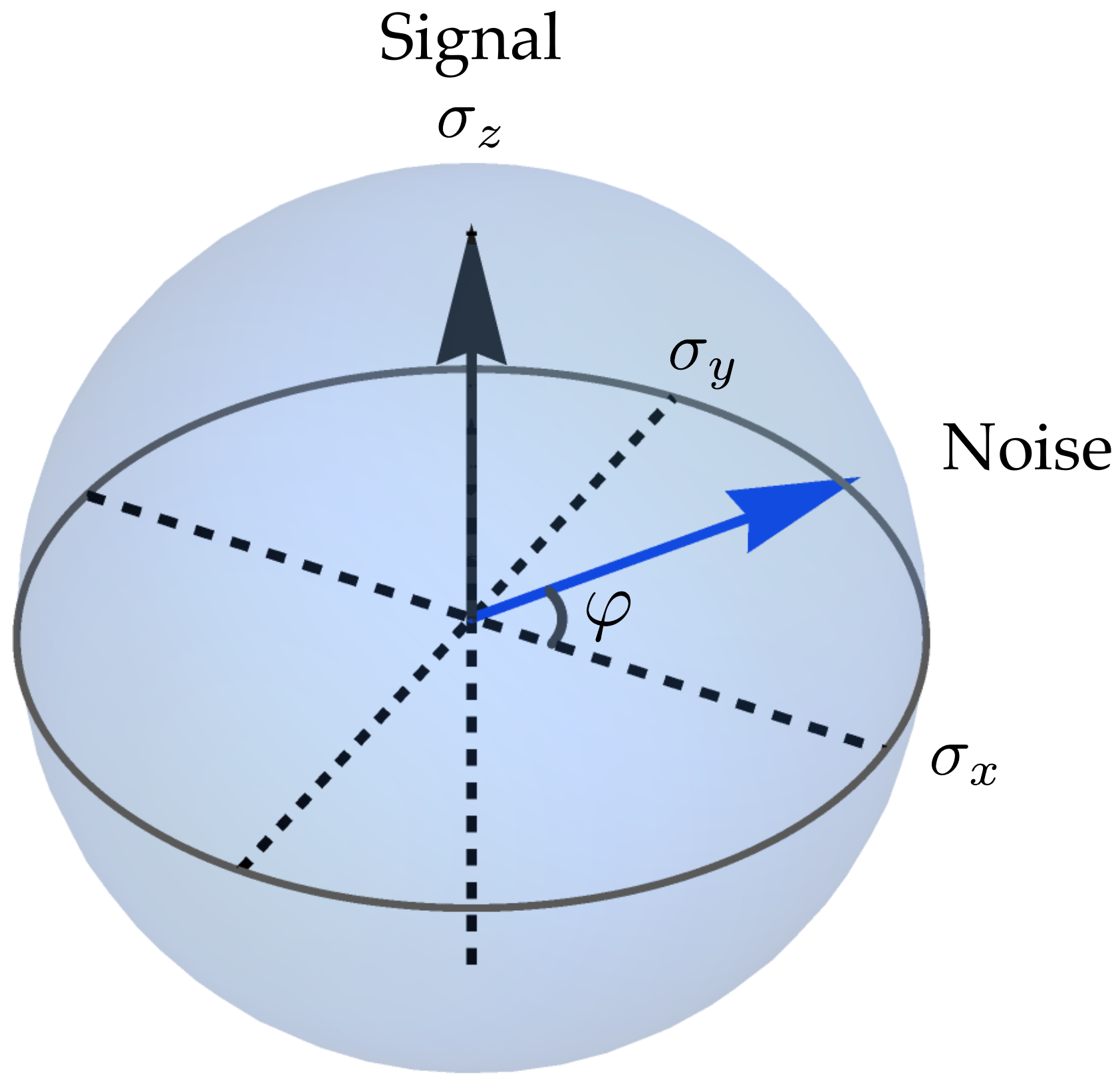}
    \caption{\textbf{Estimating $\theta$ with unknown noise characterized by $\varphi$.} The signal is a rotation about the $z$ axis of the Bloch sphere. The noise is characterized by an operator $L_\varphi=\sqrt{\gamma}((\cos\varphi)\sigma_x+(\sin\varphi)\sigma_y)$ with unknown $\varphi$.}
    \label{unknownnoise}
\end{figure}

\subsubsection{Unknown uncorrelated noise}

In Section~\ref{sec:unknown_corr_noise},
we have demonstrated that one can achieve Heisenberg scaling whenever the unknown noise angle $\varphi$ comes from a finite set of uniformly spaced periodic angles.
In the present subsection, by contrast, $\varphi$ is an arbitrary but fixed parameter in $[0,\pi)$.
Thus $\varphi$ is neither restricted to a finite set nor resampled between noise events.
We now consider evolution of a probed qubit system $A$ under the Lindblad master equation for the noise operator $L_\varphi = \sqrt{\gamma} (\cos(\varphi)\sigma_x + \sin(\varphi)\sigma_y )$ with our signal a rotation along the $z$ axis, as depicted in Fig.~\ref{unknownnoise}.
We assume that noise only happens on the probed system, between ideal control operations.
Evidently the HNLS condition is satisfied for every fixed and known value of $\varphi$, and so the Heisenberg limit is in principle attainable (see Appendix~\ref{sec:known_noise_dir}). 
However, it is not obvious what correction procedure should be applied to restore the Heisenberg limit when $\varphi$ is unknown.

If $\varphi$ is an unknown, one natural strategy is to first estimate $\varphi$, and then use this estimated noise direction for error correction. 
\atchg{Appendix~\ref{sec:unknown_noise_direction_estimation} shows that the particular method using independent Bell probes for calibration followed by plug-in correction does not achieve Heisenberg scaling for regular locally unbiased estimation.
This conclusion concerns that calibration and readout scheme, and does not exclude more general joint sensing strategies.}
We leave it as an open question whether Heisenberg scaling can be achieved under such a noise model when $\varphi$ is an unknown arbitrary continuously valued noise direction.

\subsection{Single- and multi- parameter Pauli rotations under depolarizing noise}\label{sec:example_multiparameter_rotations_depolarizing}

\begin{figure*}
    \centering
    \setlength{\figHeight}{0.25\linewidth}
    
    \begin{subfigure}[t]{0.32\linewidth}
        \centering
        \includegraphics[height=\figHeight]{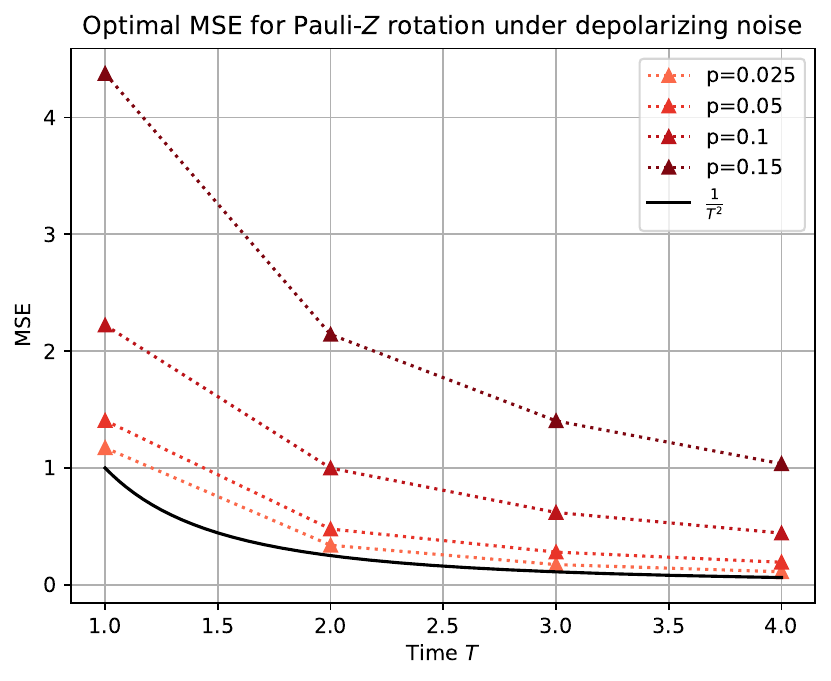}
        \caption{}
        \label{fig:MSE_depolarizing_noise_z}
    \end{subfigure}\hfill%
    \begin{subfigure}[t]{0.32\linewidth}
        \centering
        \includegraphics[height=\figHeight]{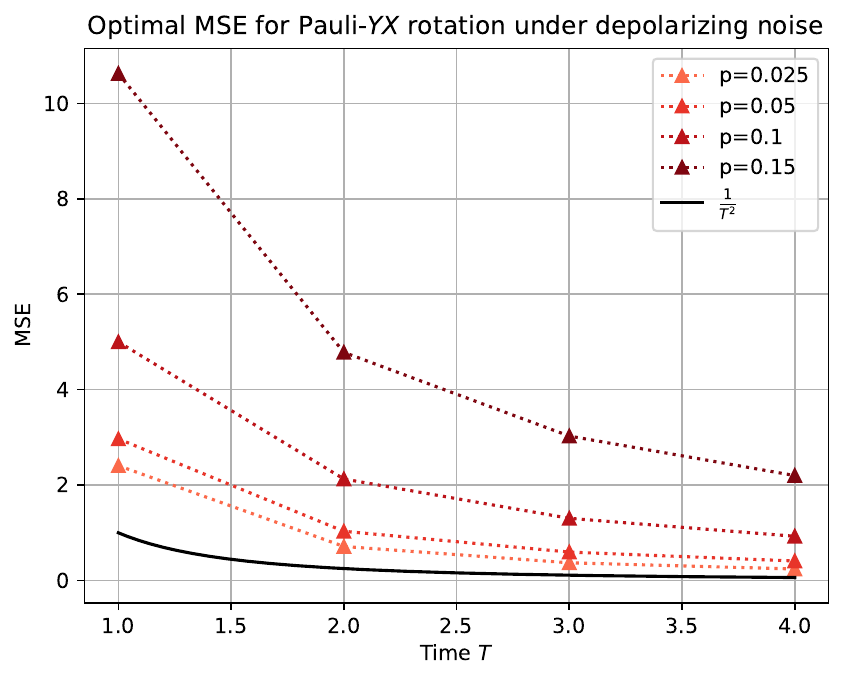}
        \caption{}
        \label{fig:MSE_depolarizing_noise_yx}
    \end{subfigure}\hfill%
    \begin{subfigure}[t]{0.32\linewidth}
        \centering
        \includegraphics[height=\figHeight]{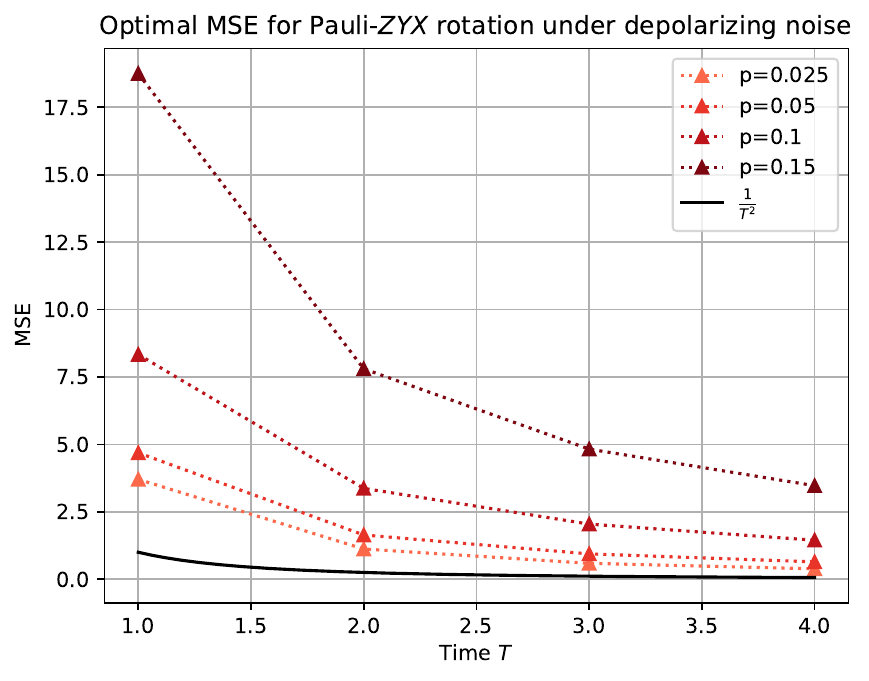}
        \caption{}
        \label{fig:MSE_depolarizing_noise_zyx}
    \end{subfigure}
    \caption{
    {Mean-squared error (MSE) SDP lower bounds for
    Pauli-$Z$ rotation in
    panel~(\subref{fig:MSE_depolarizing_noise_z}),
    Pauli-$X$ and $Y$ rotations in
    panel~(\subref{fig:MSE_depolarizing_noise_yx}), and
    Pauli-$X$, $Y$, and $Z$ rotations in
    panel~(\subref{fig:MSE_depolarizing_noise_zyx}) under
    depolarizing noise with parameter
    $p\in\{.025,.05,.1,.15\}$.}
    As noise strength $p$ gets smaller, the MSE (dotted red lines) are closer to the Heisenberg scaling (solid black line).
    The $T=1$-step bounds recover the analytic expression of the Nagaoka-Hayashi bound shown in~\cite[Section III.A]{Conlon_2023} (see eqn.~\eqref{eqn:NH_bound_depolarizing}).
    }
    \label{fig:MSE_depolarizing_noise}
\end{figure*}

Now we consider 1-, 2-, and 3- parameter Pauli rotations $R_{z,\theta_z}$, $R_{yx,\theta} = R_{y,\theta_y}R_{x,\theta_x}$, and $R_{zyx,\theta} = R_{z,\theta_z}R_{y,\theta_y}R_{x,\theta_x}$, where all three are followed by a depolarizing noise channel $\mathcal{D}_p(\rho) = (1-3p) \rho + p (\sigma_x\rho\sigma_x +\sigma_y\rho\sigma_y +\sigma_z\rho\sigma_z)$ with noise strength $p\in[0,1/3]$.
The $T$-step process $\mathbf{\Lambda}_{\theta}^{(p)}$ representing the sequence of such noisy rotations is given by
\begin{equation}
\begin{aligned}
    &\mathbf{\Lambda}_{\theta}^{(p)} =
    \resizebox{0.71\linewidth}{!}{
    \begin{quantikz}[
        classical gap=0.03cm,
        column sep=0.28cm,
        font=\normalsize,
        labels={font=\footnotesize}
    ]
        & \gate[1,style={fill=green!20}]{R_{W,\theta}} \wire[l][1]["{}"{above,pos=0.6}]{q} \gshade{1}{2}{} & \gate[style={fill=red!20}]{\mathcal{D}_p} \wire[r][1]["{}"{above,pos=0.5}]{q} & \invgate{\dots} & \gate[1,style={fill=green!20}]{R_{W,\theta}} \wire[l][1]["{}"{above,pos=0.65}]{q} \gshade{1}{2}{} & \gate[style={fill=red!20}]{\mathcal{D}_p} \wire[r][1]["{}"{above,pos=0.5}]{q} \wireoverride{q} & 
    \end{quantikz} 
    } ,
\end{aligned}
\end{equation}
for $W\in\{z,yx,zyx\}$.
Our SDP bound in Theorem~\ref{thm:sequential_NH_bound} again gives MSEs that get closer to the Heisenberg scaling as the noise strength $p\rightarrow0$, as shown in Fig.~\ref{fig:MSE_depolarizing_noise} for up to $T=4$ steps.
{The one-, two-, and three-parameter cases are shown
in panels~(\subref{fig:MSE_depolarizing_noise_z}),
(\subref{fig:MSE_depolarizing_noise_yx}), and
(\subref{fig:MSE_depolarizing_noise_zyx}) of
Fig.~\ref{fig:MSE_depolarizing_noise}, respectively.}
Note that our SDP bound recovers the analytic expressions for the Nagaoka-Hayashi bound for such 2- and 3- parameter rotations under depolarizing noise given in~\cite[Section III.A]{Conlon_2023}, which correspond to the case of $T=1$ step.
In~\cite[Section III.A]{Conlon_2023}, the 1-, 2-, and 3- parameter Nagaoka-Hayashi bound for the MSE of estimating the rotation angles under depolarizing noise with parameter $p$, denoted as $\mathrm{NH}_z(p),\mathrm{NH}_{yx}(p),\mathrm{NH}_{zyx}(p)$, respectively, are given by
\begin{equation}\label{eqn:NH_bound_depolarizing}
\begin{gathered}
    \mathrm{NH}_z(p) = \frac{1-2p}{(1-4p)^2} \;,\quad
    \mathrm{NH}_{yx}(p) = \frac{4(1-p)}{2(1-4p)^2} \;,\\
    \text{and}\quad
    \mathrm{NH}_{zyx}(p) = \frac{3}{(1-4p)^2} \;.
\end{gathered}
\end{equation}
These are precisely the points at $T=1$ in the plots in Fig.~\ref{fig:MSE_depolarizing_noise}.

\subsection{Parallel-sequential strategy for 2-parameter rotations under depolarizing noise \label{sec:parseq2}}

\begin{figure*}
    \centering
    \setlength{\figHeight}{0.37\linewidth}
    
    \begin{subfigure}[t]{0.49\linewidth}
        \centering
        \includegraphics[height=\figHeight]{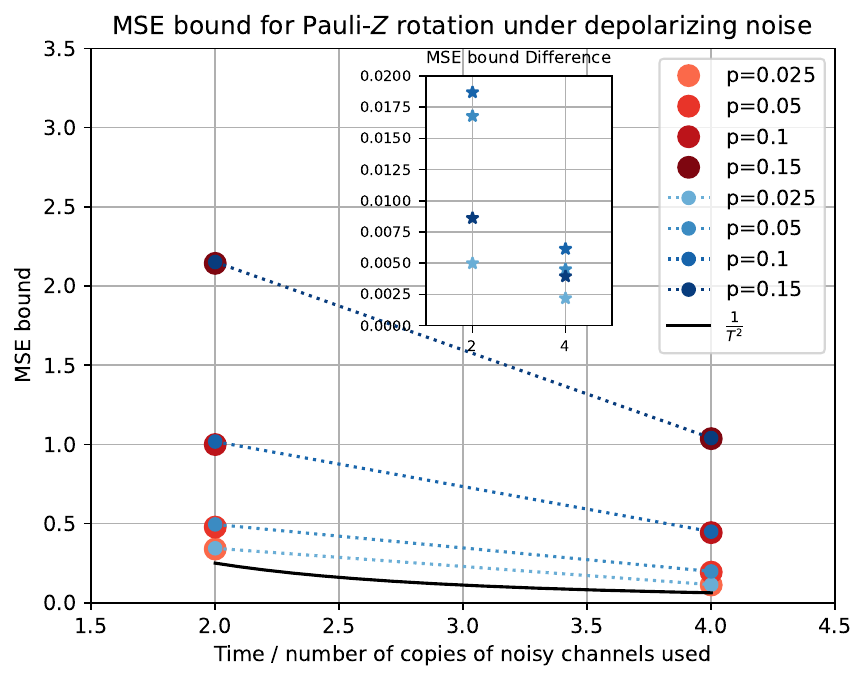}
        \caption{}
        \label{fig:MSE_par_seq_depolarizing_noise_z}
    \end{subfigure}\hfill%
    \begin{subfigure}[t]{0.49\linewidth}
        \centering
        \includegraphics[height=\figHeight]{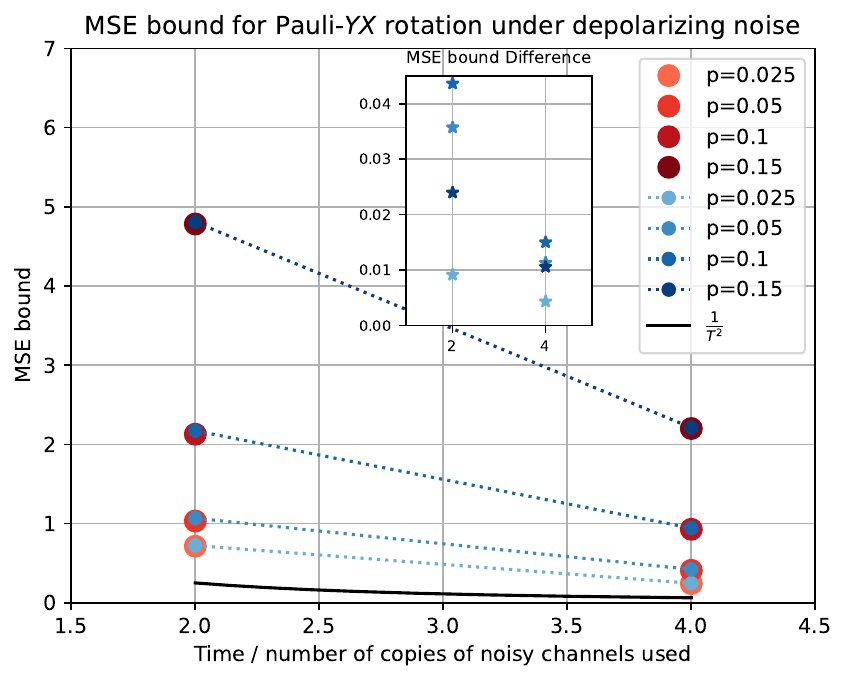}
        \caption{}
        \label{fig:MSE_par_seq_depolarizing_noise_yx}
    \end{subfigure}
    \caption{
    Mean-squared error (MSE) SDP lower bounds for Pauli-$Z$ rotation in panel~(\subref{fig:MSE_par_seq_depolarizing_noise_z}) and Pauli-$X$ and $Y$ rotations in panel~(\subref{fig:MSE_par_seq_depolarizing_noise_yx}) under depolarizing noise with parameter $p\in\{.025,.05,.1,.15\}$.
    Blue circles are SDP lower bounds for the two-copy per round strategy over $T\in\{1,2\}$ steps illustrated in eqn.~\eqref{eqn:parallel_sequential_yx_rot_depolarizing}.
    Red circles are SDP lower bounds for the one-copy per round strategy over $T\in\{2,4\}$ steps, given in panels~(\subref{fig:MSE_depolarizing_noise_z}) (Pauli-$Z$) and~(\subref{fig:MSE_depolarizing_noise_yx}) (Pauli-$X$ followed by Pauli-$Y$) of Fig.~\ref{fig:MSE_depolarizing_noise}, respectively.
    The inset in each panel shows the difference between SDP lower bounds of the one-copy and two-copy per step strategies:
    $[\text{two copy bound}]-[\text{one copy bound}]$, given $2$ copies and $4$ copies of the noisy Pauli rotation channel.
    The shades of the blue and red points indicate the noise parameter $p$ of the depolarizing noise channel, where darker shades indicate larger noise parameter.
    For the single-parameter Pauli-$Z$ model, both sets of SDP values are locally attainable by Proposition~\ref{prop:single_par_achievability}.
    For the two-parameter Pauli-$YX$ model, they are lower bounds whose attainability has not been established.
    }
    \label{fig:MSE_par_seq_depolarizing_noise}
\end{figure*}

Finally, we analyze a 1-parameter Pauli rotation $R_{z,\theta_z}$ and 2-parameter Pauli rotation $R_{yx,\theta} = R_{y,\theta_y}R_{x,\theta_x}$, both followed by a depolarizing noise channel $\mathcal{D}_p$ with noise strength $p\in\{0.025,0.05,0.1,0.15\}$ as in Section~\ref{sec:example_multiparameter_rotations_depolarizing}.
However, we now restrict the strategy to use two copies of the noisy parameterized rotation channel, in parallel, at every step.
For the single-parameter model, Proposition~\ref{prop:single_par_achievability} makes this comparison a separation between attainable precisions.
For the two-parameter model, the comparison of SDP lower bounds alone does not establish such a separation.
The $T$-step process $\mathbf{\Lambda}_{\theta}^{(p)}$ representing the sequence of such noisy rotations is given by
\begin{equation}\label{eqn:parallel_sequential_yx_rot_depolarizing}
\begin{aligned}
    &\mathbf{\Lambda}_{\theta}^{(p)} =
    \resizebox{0.71\linewidth}{!}{
    \begin{quantikz}[
        classical gap=0.03cm,
        column sep=0.28cm,
        font=\normalsize,
        labels={font=\footnotesize}
    ]
        & \gate[1,style={fill=green!20}]{R_{yx,\theta}} \wire[l][1]["{}"{above,pos=0.6}]{q} \gshade{1}{2}{} & \gate[style={fill=red!20}]{\mathcal{D}_p} \wire[r][1]["{}"{above,pos=0.5}]{q} & \invgate{\dots} & \gate[1,style={fill=green!20}]{R_{yx,\theta}} \wire[l][1]["{}"{above,pos=0.65}]{q} \gshade{1}{2}{} & \gate[style={fill=red!20}]{\mathcal{D}_p} \wire[r][1]["{}"{above,pos=0.5}]{q} \wireoverride{q} & \\
        & \gate[1,style={fill=green!20}]{R_{yx,\theta}} \wire[l][1]["{}"{above,pos=0.6}]{q} \gshade{1}{2}{} & \gate[style={fill=red!20}]{\mathcal{D}_p} \wire[r][1]["{}"{above,pos=0.5}]{q} & \invgate{\dots} & \gate[1,style={fill=green!20}]{R_{yx,\theta}} \wire[l][1]["{}"{above,pos=0.65}]{q} \gshade{1}{2}{} & \gate[style={fill=red!20}]{\mathcal{D}_p} \wire[r][1]["{}"{above,pos=0.5}]{q} \wireoverride{q} & 
    \end{quantikz} 
    } \,.
\end{aligned}
\end{equation}

In Fig.~\ref{fig:MSE_par_seq_depolarizing_noise}, we compare the two-copy-per-step strategy described in the previous paragraph with the one-copy-per-step strategy in Section~\ref{sec:example_multiparameter_rotations_depolarizing}, using the same total number of channel uses, either two or four, in each comparison.
{Panels~(\subref{fig:MSE_par_seq_depolarizing_noise_z})
and~(\subref{fig:MSE_par_seq_depolarizing_noise_yx}) show
this comparison for Pauli-$Z$ and Pauli-$YX$ rotations,
respectively.}
These plots show that for a given $T\in\{2,4\}$ copies of noisy $w\in\{z,yx\}$ rotation being probed under depolarizing noise channel with parameter $p$, the SDP lower bounds of Theorem~\ref{thm:sequential_NH_bound} for one-copy per step strategy $\mathrm{seqNH}_w^1(T,p)$ are always lower than the two-copies per step strategy $\mathrm{seqNH}_w^2(T,p)$, given the same number of copies of the noisy channel used.
This is shown by the points in the inset plots of Fig.~\ref{fig:MSE_par_seq_depolarizing_noise}, representing the difference between the MSE bound of the two: $\mathrm{seqNH}_w^2(T,p) - \mathrm{seqNH}_w^1(T,p)$.

For the single-parameter Pauli-$Z$ model, both $\mathrm{seqNH}_z^1(T,p)$ and $\mathrm{seqNH}_z^2(T,p)$ are locally attainable by Proposition~\ref{prop:single_par_achievability}.
Indeed, a two-copy parallel block is itself one step of a single-parameter process, so the proposition applies to both strategies.
Moreover, every two-copy-per-step strategy can be implemented within the fully sequential class by storing the first output until the second channel use is complete.
The positive differences in the inset of
Fig.~\ref{fig:MSE_par_seq_depolarizing_noise}%
(\subref{fig:MSE_par_seq_depolarizing_noise_z}) therefore
demonstrate, within solver precision, an advantage from allowing control after every channel use.
An attaining fully sequential strategy can be recovered from the SDP optimizer using the decomposition described in Appendix~\ref{app:tester_decomposition}.
A rigorous numerical certificate can be obtained by verifying that the MSE of the recovered fully sequential strategy is smaller than a dual-feasible lower bound for the two-copy-per-step strategy, after accounting for the corresponding numerical residuals.

For the two-parameter Pauli-$YX$ model,
Proposition~\ref{prop:single_par_achievability} does not apply, and the SDP lower bound need not be attainable.
Consequently, the inset of
Fig.~\ref{fig:MSE_par_seq_depolarizing_noise}%
(\subref{fig:MSE_par_seq_depolarizing_noise_yx}) compares only SDP lower bounds and does not by itself prove a separation between the two strategy classes.
Such a separation would follow from an attainable fully sequential strategy whose MSE is strictly below a certified lower bound for the two-copy-per-step strategy.

\section{Discussion}

In this work we study the precision of multiparameter quantum sensing, as quantified by the mean-squared error (MSE), using the most general type of protocols that sequentially probe the quantum channel or quantum process with access to coherent ancillas.
We propose a semidefinite program (SDP) optimization to compute a lower bound for the optimally achievable MSE for any sequential multiparameter quantum sensing protocol, including the most general protocol that sequentially probes the quantum channel or quantum process with access to coherent ancillas (Theorem~\ref{thm:sequential_NH_bound} in
Section~\ref{sec:sequential_bound}).
Our SDP optimization is efficiently computable in the system dimension and number of parameters for a fixed number of sensing rounds (see Appendix~\ref{app:SDP_runtime_complexity}).
Our bound applies to any quantum parameter-estimation tasks where the parameters are probed over multiple steps in sequence, including quantum channel parameter estimation scenarios under temporally correlated noise, allowing us to investigate how noise correlation affects whether Heisenberg scaling can be achieved.
Moreover, for quantum channel parameter estimation scenarios under independent noise, our bound recovers known results (Section~\ref{sec:example_multiparameter_rotations_depolarizing}).
\begin{atenv}
For regular finite-dimensional single-parameter problems admitting an optimal tester, we show that the bound is locally attainable and equals the reciprocal of the optimal sequential quantum Fisher information (Proposition~\ref{prop:single_par_achievability}).
For such attainable cases, we describe how to construct an initial state, intermediate isometries, and a locally optimal final measurement from an SDP optimizer (Section~\ref{sec:single_par_tight_bound} and Appendix~\ref{app:tester_decomposition}).
We also give a sufficient criterion for Heisenberg scaling in finite-latent correlated processes in Proposition~\ref{prop:finite_latent_hs}, using the admissible calibration schemes and operational encoding condition in Definition~\ref{def:admissible_calibration}.
Explicit protocols for correlated Pauli noise and finite sets of unknown noise directions are given in Sections~\ref{sec:correlated_noise_z_rot} and~\ref{sec:sensing_correlated_random_noise_direction}.
\end{atenv}

\begin{atenv}
A further direction is fault-tolerant quantum sensing,
where state preparation, control operations, and
measurements are also noisy.
For a restricted circuit-level noise model,
Conlon \emph{et al.}~\cite{conlon2026fault} show that an
$N$-qubit repetition code extends the time range of
Heisenberg scaling.
Our framework could provide finite-round precision bounds
for such protocols by describing the noisy encoded
evolution as a quantum process, including the syndrome
records.
For a fixed circuit architecture and noise model,
preparation, gate, and readout noise must be included
explicitly, with the optimization restricted to physically
available operations.
Relaxing these restrictions to the unrestricted testers of
Theorem~\ref{thm:sequential_NH_bound} would still give a lower
bound on the MSE, but its attainability within the
fault-tolerant architecture would require a separate proof.
This extension could quantify the trade-off between code
size, control errors, and sensing precision, including
multiparameter tasks and temporally correlated noise.
\end{atenv}

Also, it would be interesting to extend this framework to sequential function estimation~\cite{Eldredge_2018,Proctor_2018,Qian_2019,Rubio_2020,Bringewatt_2021,Bringewatt_2024}.
In particular, it is shown in~\cite{Eldredge_2018} how one can construct an optimal sequential protocol for estimating a function of a collection of Pauli-$Z$ rotation angles.
Our program can further extend this result to identify optimal MSEs for noisy function estimation scenarios and for a function of angles with different rotation axes.

\paragraph*{Note added.}
During the preparation of this manuscript, two closely related works appeared~\cite{zhou2026unified,li2026optimalstrategiesmultiparameterquantum}.
Zhou and Zhang~\cite{zhou2026unified} formulate the tight multiparameter estimation problem using quantum testers, while Li \emph{et al.}~\cite{li2026optimalstrategiesmultiparameterquantum} develop a conic framework covering several precision bounds and causal strategies. 
The present work develops a tester-based conic framework while also taking a complementary direction, focusing on general parameterized multi-time processes, single-parameter attainability, and applications to error-corrected sensing and temporally correlated noise.
We note that our SDP in Theorem~\ref{thm:sequential_NH_bound} coincides with one of the conic programs in the conic program hierarchy proposed in~\cite[Section III]{li2026optimalstrategiesmultiparameterquantum} for causally-definite strategies.
However, as shown in Section~\ref{sec:sequential_bound}, we arrive at our SDP program through a different approach.

\section*{Acknowledgements}
A.T. is supported by CQT PhD scholarship, Google PhD fellowship program, and CQT Young Researcher Career Development Grant. L.O.C.~was supported in part by the National Science Foundation under Award No. 2533041. L.O.C.~and A.V.G.~were supported in part by ONR MURI, AFOSR MURI, NSF QLCI (award No.~OMA-2120757), NSF STAQ program, DoE ASCR Quantum Testbed Pathfinder program (award No.~DE-SC0024220),    DARPA SAVaNT ADVENT, ARL (W911NF-24-2-0107), and NQVL:QSTD:Design:FTL. 
L.O.C.~and A.V.G.~also acknowledge support from the U.S.~Department of Energy, Office of Science, National Quantum Information Science Research Centers, Quantum Systems Accelerator (award No.~DE-SCL0000121) and from the U.S.~Department of Energy, Office of Science, Accelerated Research in Quantum Computing, Fundamental Algorithmic Research toward Quantum Utility (FAR-Qu). K.B.~is supported by a Hartree Fellowship from the Joint Center for Quantum Information and Computer Science (QuICS) at the University of Maryland, College Park.

\bibliographystyle{unsrt}
\bibliography{references}

\begin{thebibliography}{10}

\bibitem{giovannetti2011advances}
Vittorio Giovannetti, Seth Lloyd, and Lorenzo Maccone.
\newblock Advances in quantum metrology.
\newblock {\em Nature photonics}, 5(4):222--229, 2011.

\bibitem{aasi2013enhanced}
Junaid Aasi, Joan Abadie, BP~Abbott, Richard Abbott, TD~Abbott, MR~Abernathy, Carl Adams, Thomas Adams, Paolo Addesso, RX~Adhikari, et~al.
\newblock {Enhanced sensitivity of the LIGO gravitational wave detector by using squeezed states of light}.
\newblock {\em Nature Photonics}, 7(8):613--619, 2013.

\bibitem{tse2019quantum}
Maggie Tse, Haocun Yu, Nutsinee Kijbunchoo, A~Fernandez-Galiana, P~Dupej, L~Barsotti, CD~Blair, DD~Brown, SE~ea Dwyer, A~Effler, et~al.
\newblock {Quantum-enhanced advanced LIGO detectors in the era of gravitational-wave astronomy}.
\newblock {\em Physical Review Letters}, 123(23):231107, 2019.

\bibitem{casacio2021quantum}
Catxere~A Casacio, Lars~S Madsen, Alex Terrasson, Muhammad Waleed, Kai Barnscheidt, Boris Hage, Michael~A Taylor, and Warwick~P Bowen.
\newblock Quantum-enhanced nonlinear microscopy.
\newblock {\em Nature}, 594(7862):201--206, 2021.

\bibitem{ludlow2015optical}
Andrew~D Ludlow, Martin~M Boyd, Jun Ye, Ekkehard Peik, and Piet~O Schmidt.
\newblock Optical atomic clocks.
\newblock {\em Reviews of Modern Physics}, 87(2):637--701, 2015.

\bibitem{wu2019gravity}
Xuejian Wu, Zachary Pagel, Bola~S Malek, Timothy~H Nguyen, Fei Zi, Daniel~S Scheirer, and Holger M{\"u}ller.
\newblock Gravity surveys using a mobile atom interferometer.
\newblock {\em Science advances}, 5(9):eaax0800, 2019.

\bibitem{troullinou2023quantum}
Charikleia Troullinou, Vito~Giovanni Lucivero, and Morgan~W Mitchell.
\newblock Quantum-enhanced magnetometry at optimal number density.
\newblock {\em Physical Review Letters}, 131(13):133602, 2023.

\bibitem{marciniak2022optimal}
Christian~D Marciniak, Thomas Feldker, Ivan Pogorelov, Raphael Kaubruegger, Denis~V Vasilyev, Rick van Bijnen, Philipp Schindler, Peter Zoller, Rainer Blatt, and Thomas Monz.
\newblock Optimal metrology with programmable quantum sensors.
\newblock {\em Nature}, 603(7902):604--609, 2022.

\bibitem{conlon2023approaching}
Lorc{\'a}n~O Conlon, Tobias Vogl, Christian~D Marciniak, Ivan Pogorelov, Simon~K Yung, Falk Eilenberger, Dominic~W Berry, Fabiana~S Santana, Rainer Blatt, Thomas Monz, et~al.
\newblock Approaching optimal entangling collective measurements on quantum computing platforms.
\newblock {\em Nature Physics}, 19(3):351--357, 2023.

\bibitem{conlon2023discriminating}
Lorc{\'a}n~O Conlon, Falk Eilenberger, Ping~Koy Lam, and Syed~M Assad.
\newblock Discriminating mixed qubit states with collective measurements.
\newblock {\em Communications Physics}, 6(1):337, 2023.

\bibitem{nielsen2023deterministic}
Jens~AH Nielsen, Jonas~S Neergaard-Nielsen, Tobias Gehring, and Ulrik~L Andersen.
\newblock {Deterministic quantum phase estimation beyond N00N states}.
\newblock {\em Physical Review Letters}, 130(12):123603, 2023.

\bibitem{daryanoosh2018experimental}
Shakib Daryanoosh, Sergei Slussarenko, Dominic~W Berry, Howard~M Wiseman, and Geoff~J Pryde.
\newblock {Experimental optical phase measurement approaching the exact Heisenberg limit}.
\newblock {\em Nature communications}, 9(1):4606, 2018.

\bibitem{hou2018deterministic}
Zhibo Hou, Jun-Feng Tang, Jiangwei Shang, Huangjun Zhu, Jian Li, Yuan Yuan, Kang-Da Wu, Guo-Yong Xiang, Chuan-Feng Li, and Guang-Can Guo.
\newblock Deterministic realization of collective measurements via photonic quantum walks.
\newblock {\em Nature communications}, 9(1):1414, 2018.

\bibitem{zhou2025experimental}
Kai Zhou, Changhao Yi, Wen-Zhe Yan, Zhibo Hou, Huangjun Zhu, Guo-Yong Xiang, Chuan-Feng Li, and Guang-Can Guo.
\newblock Experimental realization of genuine three-copy collective measurements for optimal information extraction.
\newblock {\em Physical Review Letters}, 134(21):210201, 2025.

\bibitem{yi2025optimal}
Changhao Yi, Kai Zhou, Zhibo Hou, Guo-Yong Xiang, and Huangjun Zhu.
\newblock Optimal estimation of three parallel spins with genuine and restricted collective measurements.
\newblock {\em Physical Review A}, 111(6):062409, 2025.

\bibitem{mansouri2026experimental}
Arman Mansouri, Kyle~M Jordan, Raphael~A Abrahao, and Jeff~S Lundeen.
\newblock Experimental demonstration of a multi-particle collective measurement for optimal quantum state estimation.
\newblock {\em Optics Express}, 34(4):5718--5728, 2026.

\bibitem{li2023optimal}
Bacui Li, Lorc{\'a}n~O Conlon, Ping~Koy Lam, and Syed~M Assad.
\newblock Optimal single-qubit tomography: Realization of locally optimal measurements on a quantum computer.
\newblock {\em Physical Review A}, 108(3):032605, 2023.

\bibitem{yung2025saturating}
Simon~K Yung, Aritra Das, Jun Suzuki, Ping~Koy Lam, Jie Zhao, Lorc{\'a}n~O Conlon, and Syed~M Assad.
\newblock Saturating the quantum cram$\backslash$'er--rao bound in prioritised parameter estimation.
\newblock {\em arXiv preprint arXiv:2511.06704}, 2025.

\bibitem{yung2026beating}
Simon~K Yung, Wen-Zhe Yan, Lan-Tian Feng, Aritra Das, Jiayi Qin, Guang-Can Guo, Ping~Koy Lam, Jie Zhao, Zhibo Hou, Lorcan~O Conlon, et~al.
\newblock Beating three-parameter precision trade-offs with entangling collective measurements.
\newblock {\em arXiv preprint arXiv:2604.08871}, 2026.

\bibitem{bollinger1996optimal}
John~J Bollinger, Wayne~M Itano, David~J Wineland, and Daniel~J Heinzen.
\newblock Optimal frequency measurements with maximally correlated states.
\newblock {\em Physical Review A}, 54(6):R4649, 1996.

\bibitem{Huelga1997}
S.~F. Huelga, C.~Macchiavello, T.~Pellizzari, A.~K. Ekert, M.~B. Plenio, and J.~I. Cirac.
\newblock Improvement of frequency standards with quantum entanglement.
\newblock {\em Phys. Rev. Lett.}, 79:3865--3868, 1997.

\bibitem{dorner2009optimal}
Uwe Dorner, Rafal Demkowicz-Dobrzanski, Brian~J Smith, Jeff~S Lundeen, Wojciech Wasilewski, Konrad Banaszek, and Ian~A Walmsley.
\newblock Optimal quantum phase estimation.
\newblock {\em Physical review letters}, 102(4):040403, 2009.

\bibitem{Escher2011}
B.~M. Escher, R.~L. de~Matos~Filho, and L.~Davidovich.
\newblock General framework for estimating the ultimate precision limit in noisy quantum-enhanced metrology.
\newblock {\em Nature Physics}, 7:406--411, 2011.

\bibitem{demkowicz2012elusive}
Rafa{\l} Demkowicz-Dobrza{\'n}ski, Jan Ko{\l}ody{\'n}ski, and M{\u{a}}d{\u{a}}lin Gu{\c{t}}{\u{a}}.
\newblock {The elusive Heisenberg limit in quantum-enhanced metrology}.
\newblock {\em Nature communications}, 3(1):1063, 2012.

\bibitem{ozeri2013heisenberg}
Roee Ozeri.
\newblock {Heisenberg limited metrology using quantum error-correction codes}.
\newblock {\em arXiv preprint arXiv:1310.3432}, 2013.

\bibitem{kessler2014quantum}
Eric~M Kessler, Igor Lovchinsky, Alexander~O Sushkov, and Mikhail~D Lukin.
\newblock {Quantum error correction for metrology}.
\newblock {\em Physical review letters}, 112(15):150802, 2014.

\bibitem{Arrad2014}
Gilad Arrad, Yuval Vinkler, Dorit Aharonov, and Alex Retzker.
\newblock Increasing sensing resolution with error correction.
\newblock {\em Phys. Rev. Lett.}, 112:150801, 2014.

\bibitem{dur2014improved}
Wolfgang D{\"u}r, Michalis Skotiniotis, Florian Froewis, and Barbara Kraus.
\newblock {Improved quantum metrology using quantum error correction}.
\newblock {\em Physical Review Letters}, 112(8):080801, 2014.

\bibitem{sekatski2017quantum}
Pavel Sekatski, Michalis Skotiniotis, Janek Ko{\l}ody{\'n}ski, and Wolfgang D{\"u}r.
\newblock Quantum metrology with full and fast quantum control.
\newblock {\em Quantum}, 1:27, 2017.

\bibitem{demkowicz2017adaptive}
Rafa{\l} Demkowicz-Dobrza{\'n}ski, Jan Czajkowski, and Pavel Sekatski.
\newblock {Adaptive quantum metrology under general Markovian noise}.
\newblock {\em Physical Review X}, 7(4):041009, 2017.

\bibitem{zhou2018achieving}
Sisi Zhou, Mengzhen Zhang, John Preskill, and Liang Jiang.
\newblock {Achieving the Heisenberg limit in quantum metrology using quantum error correction}.
\newblock {\em Nature communications}, 9(1):78, 2018.

\bibitem{zhou2021asymptotic}
Sisi Zhou and Liang Jiang.
\newblock Asymptotic theory of quantum channel estimation.
\newblock {\em PRX Quantum}, 2(1):010343, 2021.

\bibitem{chiribella2012optimal}
Giulio Chiribella.
\newblock Optimal networks for quantum metrology: semidefinite programs and product rules.
\newblock {\em New Journal of Physics}, 14(12):125008, 2012.

\bibitem{yuan2015optimal}
Haidong Yuan and Chi-Hang~Fred Fung.
\newblock Optimal feedback scheme and universal time scaling for hamiltonian parameter estimation.
\newblock {\em Physical review letters}, 115(11):110401, 2015.

\bibitem{yuan2016sequential}
Haidong Yuan.
\newblock Sequential feedback scheme outperforms the parallel scheme for hamiltonian parameter estimation.
\newblock {\em Physical review letters}, 117(16):160801, 2016.

\bibitem{helstrom1967minimum}
Carl~W Helstrom.
\newblock Minimum mean-squared error of estimates in quantum statistics.
\newblock {\em Physics letters A}, 25(2):101--102, 1967.

\bibitem{barndorff2000fisher}
Ole~E Barndorff-Nielsen and Richard~D Gill.
\newblock Fisher information in quantum statistics.
\newblock {\em Journal of Physics A: Mathematical and General}, 33(24):4481, 2000.

\bibitem{hayashi2005statistical}
Masahito Hayashi and Keiji Matsumoto.
\newblock Statistical model with measurement degree of freedom and quantum physics.
\newblock In {\em Asymptotic theory of quantum statistical inference: selected papers}, pages 162--169. World Scientific, 2005.

\bibitem{conlon2025role}
Lorcan~O Conlon, Jun Suzuki, Ping~Koy Lam, and Syed~M Assad.
\newblock Role of the extended hilbert space in the attainability of the quantum cram{\'e}r--rao bound for multiparameter estimation.
\newblock {\em Physics Letters A}, 542:130445, 2025.

\bibitem{ragy2016compatibility}
Sammy Ragy, Marcin Jarzyna, and Rafa{\l} Demkowicz-Dobrza{\'n}ski.
\newblock Compatibility in multiparameter quantum metrology.
\newblock {\em Physical Review A}, 94(5):052108, 2016.

\bibitem{helstrom1968minimum}
Carl~W Helstrom.
\newblock The minimum variance of estimates in quantum signal detection.
\newblock {\em IEEE Trans. Inf. Theory}, 14(2):234--242, 1968.

\bibitem{young1975asymptotically}
Tzay~Y Young.
\newblock Asymptotically efficient approaches to quantum-mechanical parameter estimation.
\newblock {\em Inf. Sci.}, 9(1):25--42, 1975.

\bibitem{belavkin2004generalized}
Vyacheslav~P Belavkin.
\newblock Generalized uncertainty relations and efficient measurements in quantum systems.
\newblock {\em arXiv preprint quant-ph/0412030}, 2004.

\bibitem{holevo2011probabilistic}
Alexander~S Holevo.
\newblock {\em Probabilistic and statistical aspects of quantum theory}, volume~1.
\newblock Springer Science \& Business Media, 2011.

\bibitem{matsumoto2002new}
Keiji Matsumoto.
\newblock A new approach to the cram{\'e}r-rao-type bound of the pure-state model.
\newblock {\em J. Phys. A Math. Gen.}, 35(13):3111, 2002.

\bibitem{bradshaw2017tight}
M.~Bradshaw, S.~M. Assad, and P.~K. Lam.
\newblock A tight \text{C}ram{\'e}r--\text{R}ao bound for joint parameter estimation with a pure two-mode squeezed probe.
\newblock {\em Phys. Lett. A}, 381(32):2598--2607, 2017.

\bibitem{bradshaw2018ultimate}
M.~Bradshaw, P.~K. Lam, and S.~M. Assad.
\newblock Ultimate precision of joint quadrature parameter estimation with a {G}aussian probe.
\newblock {\em Phys. Rev. A}, 97(1):012106, 2018.

\bibitem{yung2025most}
Simon~K Yung, CM~Yung, Lorc{\'a}n~O Conlon, and Syed~M Assad.
\newblock The most informative cram$\backslash$'er--rao bound for quantum two-parameter estimation with pure state probes.
\newblock {\em arXiv preprint arXiv:2511.14950}, 2025.

\bibitem{kimizu2024adaptive}
Masataka Kimizu, Fuyuhiko Tanaka, and Akio Fujiwara.
\newblock Adaptive quantum state estimation for two optical point sources.
\newblock {\em Physical Review A}, 109(3):032434, 2024.

\bibitem{zhang2024qestoptpovm}
Jianchao Zhang and Jun Suzuki.
\newblock Qestoptpovm: an iterative algorithm to find optimal measurements for quantum parameter estimation.
\newblock {\em arXiv preprint arXiv:2403.20131}, 2024.

\bibitem{nagaoka2005new}
Hiroshi Nagaoka.
\newblock {A new approach to Cram{\'e}r-Rao bounds for quantum state estimation}.
\newblock In {\em {Asymptotic Theory of Quantum Statistical Inference: Selected Papers}}, pages 100--112. World Scientific, 2005.

\bibitem{das2025holevo}
Aritra Das, Lorc{\'a}n~O Conlon, Jun Suzuki, Simon~K Yung, Ping~K Lam, and Syed~M Assad.
\newblock Holevo cram{\'e}r-rao bound: How close can we get without entangling measurements?
\newblock {\em Quantum}, 9:1867, 2025.

\bibitem{conlon2022gap}
Lorc{\'a}n~O Conlon, Jun Suzuki, Ping~Koy Lam, and Syed~M Assad.
\newblock The gap persistence theorem for quantum multiparameter estimation.
\newblock {\em arXiv preprint arXiv:2208.07386}, 2022.

\bibitem{conlon2021efficient}
Lorc{\'a}n~O Conlon, Jun Suzuki, Ping~Koy Lam, and Syed~M Assad.
\newblock Efficient computation of the nagaoka--hayashi bound for multiparameter estimation with separable measurements.
\newblock {\em npj Quantum Information}, 7(1):110, 2021.

\bibitem{liu2023optimal}
Qiushi Liu, Zihao Hu, Haidong Yuan, and Yuxiang Yang.
\newblock Optimal strategies of quantum metrology with a strict hierarchy.
\newblock {\em Physical Review Letters}, 130(7):070803, 2023.

\bibitem{kurdzialek2025quantum}
Stanis{\l}aw Kurdzia{\l}ek, Piotr Dulian, Joanna Majsak, Sagnik Chakraborty, and Rafa{\l} Demkowicz-Dobrza{\'n}ski.
\newblock Quantum metrology using quantum combs and tensor network formalism.
\newblock {\em New Journal of Physics}, 27(1):013019, 2025.

\bibitem{liu2022optimal}
Jing Liu, Mao Zhang, Hongzhen Chen, Lingna Wang, and Haidong Yuan.
\newblock Optimal scheme for quantum metrology.
\newblock {\em Advanced Quantum Technologies}, 5(1):2100080, 2022.

\bibitem{liu2024fully}
Qiushi Liu, Zihao Hu, Haidong Yuan, and Yuxiang Yang.
\newblock Fully-optimized quantum metrology: framework, tools, and applications.
\newblock {\em Advanced Quantum Technologies}, 7(12):2400094, 2024.

\bibitem{albarelli2022probe}
Francesco Albarelli and Rafa{\l} Demkowicz-Dobrza{\'n}ski.
\newblock Probe incompatibility in multiparameter noisy quantum metrology.
\newblock {\em Physical Review X}, 12(1):011039, 2022.

\bibitem{lu2021incorporating}
Xiao-Ming Lu and Xiaoguang Wang.
\newblock Incorporating heisenberg’s uncertainty principle into quantum multiparameter estimation.
\newblock {\em Physical Review Letters}, 126(12):120503, 2021.

\bibitem{yung2024comparison}
Simon~K Yung, Lorc{\'a}n~O Conlon, Jie Zhao, Ping~Koy Lam, and Syed~M Assad.
\newblock Comparison of estimation limits for quantum two-parameter estimation.
\newblock {\em Physical Review Research}, 6(3):033315, 2024.

\bibitem{Hayashi_2024}
Masahito Hayashi and Yingkai Ouyang.
\newblock Finding the optimal probe state for multiparameter quantum metrology using conic programming.
\newblock {\em npj Quantum Information}, 10(1), November 2024.

\bibitem{fujiwara2006strong}
Akio Fujiwara.
\newblock Strong consistency and asymptotic efficiency for adaptive quantum estimation problems.
\newblock {\em Journal of Physics A: Mathematical and General}, 39(40):12489--12504, 2006.

\bibitem{gorecki2020optimal}
Wojciech G{\'o}recki, Sisi Zhou, Liang Jiang, and Rafa{\l} Demkowicz-Dobrza{\'n}ski.
\newblock Optimal probes and error-correction schemes in multi-parameter quantum metrology.
\newblock {\em Quantum}, 4:288, 2020.

\bibitem{chiribella2009theoretical}
Giulio Chiribella, Giacomo~Mauro D’Ariano, and Paolo Perinotti.
\newblock Theoretical framework for quantum networks.
\newblock {\em Physical Review A—Atomic, Molecular, and Optical Physics}, 80(2):022339, 2009.

\bibitem{taranto2025higherorderquantumoperations}
Philip Taranto, Simon Milz, Mio Murao, Marco~Túlio Quintino, and Kavan Modi.
\newblock Higher-order quantum operations, 2025.
\newblock \href{https://arxiv.org/abs/2503.09693}{arXiv:2503.09693}.

\bibitem{Braunstein_1994}
Samuel~L. Braunstein and Carlton~M. Caves.
\newblock Statistical distance and the geometry of quantum states.
\newblock {\em Phys. Rev. Lett.}, 72:3439--3443, May 1994.

\bibitem{Yang_2019}
Yuxiang Yang.
\newblock Memory effects in quantum metrology.
\newblock {\em Physical Review Letters}, 123(11), 2019.

\bibitem{Altherr_2021}
Anian Altherr and Yuxiang Yang.
\newblock Quantum metrology for non-markovian processes.
\newblock {\em Physical Review Letters}, 127(6), August 2021.

\bibitem{Kurdzialek_2025}
Stanisław Kurdziałek, Francesco Albarelli, and Rafał Demkowicz-Dobrzański.
\newblock Universal bounds for quantum metrology in the presence of correlated noise.
\newblock {\em Physical Review Letters}, 135(13), 2025.

\bibitem{Conlon_2023}
Lorcán~O. Conlon, Ping~Koy Lam, and Syed~M. Assad.
\newblock Multiparameter estimation with two-qubit probes in noisy channels.
\newblock {\em Entropy}, 25(8):1122, July 2023.

\bibitem{Mann_2025}
Zachary Mann, Ningping Cao, Raymond Laflamme, and Sisi Zhou.
\newblock Quantum error-corrected non-markovian metrology.
\newblock {\em PRX Quantum}, 6(3), August 2025.

\bibitem{conlon2026fault}
Lorcan~O. Conlon, Yu-Xin Wang, Erfan Abbasgholinejad, Victor~V. Albert, Michael~J. Gullans, and Alexey~V. Gorshkov.
\newblock Fault-tolerant heisenberg-limited quantum sensing.
\newblock {\em arXiv preprint arXiv:2608.00171}, 2026.

\bibitem{Eldredge_2018}
Zachary Eldredge, Michael Foss-Feig, Jonathan~A. Gross, S.~L. Rolston, and Alexey~V. Gorshkov.
\newblock Optimal and secure measurement protocols for quantum sensor networks.
\newblock {\em Physical Review A}, 97(4), April 2018.

\bibitem{Proctor_2018}
Timothy~J. Proctor, Paul~A. Knott, and Jacob~A. Dunningham.
\newblock Multiparameter estimation in networked quantum sensors.
\newblock {\em Physical Review Letters}, 120(8), February 2018.

\bibitem{Qian_2019}
Kevin Qian, Zachary Eldredge, Wenchao Ge, Guido Pagano, Christopher Monroe, J.~V. Porto, and Alexey~V. Gorshkov.
\newblock Heisenberg-scaling measurement protocol for analytic functions with quantum sensor networks.
\newblock {\em Physical Review A}, 100(4), October 2019.

\bibitem{Rubio_2020}
Jesús Rubio, Paul~A Knott, Timothy~J Proctor, and Jacob~A Dunningham.
\newblock Quantum sensing networks for the estimation of linear functions.
\newblock {\em Journal of Physics A: Mathematical and Theoretical}, 53(34):344001, August 2020.

\bibitem{Bringewatt_2021}
Jacob Bringewatt, Igor Boettcher, Pradeep Niroula, Przemyslaw Bienias, and Alexey~V. Gorshkov.
\newblock Protocols for estimating multiple functions with quantum sensor networks: Geometry and performance.
\newblock {\em Physical Review Research}, 3(3), 2021.

\bibitem{Bringewatt_2024}
Jacob Bringewatt, Adam Ehrenberg, Tarushii Goel, and Alexey~V. Gorshkov.
\newblock Optimal function estimation with photonic quantum sensor networks.
\newblock {\em Physical Review Research}, 6(1), March 2024.

\bibitem{zhou2026unified}
Zhao-Yi Zhou and Da-Jian Zhang.
\newblock Unified and computable approach to optimal strategies for multiparameter estimation.
\newblock {\em arXiv preprint arXiv:2603.06244}, 2026.

\bibitem{li2026optimalstrategiesmultiparameterquantum}
Linxuan Li, Zihao Hu, Longyun Chen, Yuxiang Yang, and Haidong Yuan.
\newblock Optimal strategies for multi-parameter quantum metrology, 2026.
\newblock \href{https://arxiv.org/abs/2608.01114}{arXiv:2608.01114}.

\bibitem{Zhou_2018}
Sisi Zhou, Mengzhen Zhang, John Preskill, and Liang Jiang.
\newblock Achieving the heisenberg limit in quantum metrology using quantum error correction.
\newblock {\em Nature Communications}, 9(1), January 2018.

\bibitem{nussbaum2010asymptotically}
Michael Nussbaum and Arleta Szko{\l}a.
\newblock Asymptotically optimal discrimination between pure quantum states.
\newblock In {\em Conference on Quantum Computation, Communication, and Cryptography}, pages 1--8. Springer, 2010.

\bibitem{Hayashi_2023}
Masahito Hayashi and Yingkai Ouyang.
\newblock Tight cramér-rao type bounds for multiparameter quantum metrology through conic programming.
\newblock {\em Quantum}, 7:1094, August 2023.

\bibitem{vandenberghe1996semidefinite}
Lieven Vandenberghe and Stephen Boyd.
\newblock Semidefinite programming.
\newblock {\em SIAM review}, 38(1):49--95, 1996.

\bibitem{boyd2004convex}
Stephen Boyd and Lieven Vandenberghe.
\newblock {\em Convex optimization}.
\newblock Cambridge university press, 2004.

\bibitem{gartner2012approximation}
Bernd G{\"a}rtner and Jiri Matousek.
\newblock {\em Approximation algorithms and semidefinite programming}.
\newblock Springer Science \& Business Media, 2012.

\bibitem{di2012elementarygatesternaryquantum}
Yao-Min Di and Hai-Rui Wei.
\newblock Elementary gates for ternary quantum logic circuit, 2012.
\newblock \href{https://arxiv.org/abs/1105.5485}{arXiv:1105.5485}.

\bibitem{trisciani2025decompositionmultiqutritgatesgenerated}
Daniele Trisciani, Marco Cattaneo, and Zoltán Zimborás.
\newblock Decomposition of multi-qutrit gates generated by weyl-heisenberg strings, 2025.
\newblock \href{https://arxiv.org/abs/2507.09781}{arXiv:2507.09781}.

\bibitem{Boixo_2007}
Sergio Boixo, Steven~T. Flammia, Carlton~M. Caves, and JM~Geremia.
\newblock Generalized limits for single-parameter quantum estimation.
\newblock {\em Physical Review Letters}, 98(9), February 2007.

\end{thebibliography}

\onecolumngrid

\newpage

\makeatletter
\global\let\set@footnotewidth\set@footnotewidth@one
\global\let\compose@footnotes\compose@footnotes@one
\makeatother

\appendix

\section*{Appendices}

\section{Quantum Fisher Information with an unknown noise direction}
\label{apenQFIunknownnoise}

    In this appendix, we present a more detailed analysis of the sequential sensing problem under unknown noise studied in Section~\ref{sec:sensing_correlated_random_noise_direction}. In this setting, we consider the noise operator $L_\varphi=\sqrt{\gamma}((\cos\varphi)\sigma_x+(\sin\varphi)\sigma_y)$ with our signal a rotation around the $z$ axis (see Fig.~\ref{unknownnoise}). 
    We assume that noise only happens on the probe system, where any control operations and ancillary systems are ideal. 
    In this appendix we use Hamiltonian $H_\theta=\theta\sigma_z$, consistently with the short-time Kraus operator below.
    Thus noiseless evolution for time $T$ implements $R_{z,2\theta T}$, with optimal QFI $4T^2$ and locally attainable MSE $1/(4T^2)$.
    Operators acting only on the probe are implicitly tensored with the identity on the ancillary register when applied to a probe--ancilla state.
    In Section~\ref{sec:known_noise_dir}, we first show that, when $\varphi$ is known, the Heisenberg limit is attainable as expected. Then, in Section~\ref{apensubsubestimatingnoisedir}, we compute the QFI for estimating the unknown noise direction. In Section~\ref{sec:unknown_noise_direction_estimation}, we show that a protocol implementing error correction based on imperfect knowledge of $\varphi$ does not achieve Heisenberg scaling. Finally, in Section~\ref{apenA5}, we consider estimating $\theta$ without knowing $\varphi$.
    If $\varphi$ is restricted to a fixed finite set, its value can be identified with exponentially small error, allowing $\mse_\theta=O(T^{-2})$.
    However, for an arbitrary continuously valued unknown noise direction $\varphi$, no-correction and estimate-then-correct protocols do not retain Heisenberg scaling asymptotically.

\subsection{Known noise direction}\label{sec:known_noise_dir}
First assume that $\varphi$ is known exactly. We take as our input state the maximally entangled state $\ket{\psi_0}=\frac{|00\>+|11\>}{\sqrt{2}}$ over ancillary qubit $C$ and probe qubit $A$, so that the projector onto the noiseless code space and error space can be defined as
\begin{equation}
    \begin{split}
        \Pi_\text{c}&=\ketbra{00} + \ketbra{11} \quad\text{and}\\
        \Pi_\text{e}&=\ketbra{01}+\ketbra{10} \;,
    \end{split}
\end{equation}
respectively.
Following Ref.~\cite{Zhou_2018}, we approximate the noisy channel acting on a probed qubit system $A$, where the ideal channel is parameterized by rotation angle $\theta$ and noise direction parameter $\varphi$, as follows:
\begin{equation}
    \mathcal{E}_{\theta,\varphi}(\rho)=K_0\rho K_0^\dagger+K_1\rho K_1^\dagger \;.
\end{equation}
The Kraus operators $K_0,K_1$ \footnote{Note that this is an approximation of the actual evolution since $K_0^\dag K_0+K_1^\dag K_1 = I + O(dt^2)$.} are given by
\begin{equation}
\begin{gathered}
    K_0=I+\left(-\iota\theta \sigma_z-\frac{1}{2}L_\varphi^\dagger L_\varphi\right)dt \\
    \text{and}\\
    K_1=L_\varphi\sqrt{dt}\;.
\end{gathered}
\end{equation}
Since $L_\varphi^\dagger L_\varphi = \gamma I$, 
we have $K_0=(1-(\gamma\,dt)/2) I -\iota\theta dt\sigma_z$, which acts on a subsystem of $\ket{\psi_0}=(\ket{00}+\ket{11})/\sqrt{2}$ as
\begin{equation}
    I\otimes K_0\ket{\psi_0} = \Big(1-\frac{dt\gamma}{2}\Big) \ket{\psi_0}-\iota\theta dt\sigma_z\ket{\psi_0} + O(dt^2) \;.
\end{equation}
Note that this Kraus operator keeps the evolution within the code space. 
In contrast, $K_1$ brings the state out of the code space, acting as
\begin{equation}
    I\otimes K_1\ket{\psi_0} =\sqrt{dt}L_\varphi \ket{\psi_0} + O(dt^2) \;.
\end{equation}
When the projector $\Pi_\text{e}$ is observed, one can implement the correction operator $Q_\varphi=(\cos\varphi)\sigma_x+(\sin\varphi)\sigma_y$, such that 
\begin{equation}
    I\otimes Q_\varphi K_1\ket{\psi_0} =\sqrt{dt}\sqrt{\gamma} \ket{\psi_0} + O(dt^2) \;.
\end{equation}
Thus for input $\rho=\ketbra{\psi_0}$ to the channel $\mathcal{E}_{\theta,\varphi}$, we obtain an output state $\rho_\outs=\mathcal{R}\circ(\mathcal{I}\otimes\mathcal{E}_{\theta,\varphi})(\rho)$, given by
\begin{equation}
\begin{aligned}
    \rho_\outs =& (1-dt\gamma/2)^2 \ketbra{\psi_0}-\iota\theta dt (\sigma_z\ketbra{\psi_0}-\ketbra{\psi_0}\sigma_z) + dt\gamma\ketbra{\psi_0}+O(dt^2)\\
    =&(1-dt\gamma) \ketbra{\psi_0}-\iota\theta dt (\sigma_z\ketbra{\psi_0}-\ketbra{\psi_0}\sigma_z) + dt\gamma\ketbra{\psi_0}+O(dt^2)\\
    =& \ketbra{\psi_0}-\iota\theta dt (\sigma_z\ketbra{\psi_0}-\ketbra{\psi_0}\sigma_z)+O(dt^2)\;,
\end{aligned}
\end{equation}
which is the ideal evolution to first order in $dt$. In the limit $dt\to0$, this protocol can recover Heisenberg scaling exactly.

\subsection{Estimating noise direction}
\label{apensubsubestimatingnoisedir}

We now assume that $\varphi$ is unknown and consider the task of estimating it. Starting from a perfect Bell state, we compute the quantum Fisher information for estimating $\varphi$ in the noise operator $L_\varphi=\sqrt{\gamma} (\cos(\varphi)\sigma_x+\sin(\varphi)\sigma_y)$. The quantum state after an infinitesimal time step is
\begin{equation}
\label{rhoevolve1}
    \rho=
    \begin{pmatrix}
     \frac{1}{2}-\frac{\gamma  dt}{2} & 0 & 0 & -\frac{\gamma  dt}{2}-\iota dt \theta +\frac{1}{2} \\
     0 & \frac{\gamma  dt}{2} & \frac{1}{2} \gamma  dt \mathrm{e}^{-2 \iota \varphi } & 0 \\
     0 & \frac{1}{2} \gamma dt \mathrm{e}^{2 \iota \varphi } & \frac{\gamma  dt}{2} & 0 \\
     -\frac{\gamma  dt}{2}+\iota dt \theta +\frac{1}{2} & 0 & 0 & \frac{1}{2}-\frac{\gamma  dt}{2} \\
    \end{pmatrix}\;.
\end{equation}
Differentiating with respect to $\varphi$ gives
\begin{equation}
    \hspace{-1cm}
    \dot{\rho}_\varphi=
    \begin{pmatrix}
     0 & 0 & 0 & 0 \\
     0 & 0 & -\iota \gamma  dt \mathrm{e}^{-2 \iota \varphi } & 0 \\
     0 & \iota\gamma dt \mathrm{e}^{2 \iota \varphi } & 0 & 0 \\
    0 & 0 & 0 & 0 \\
    \end{pmatrix}\;.
\end{equation}
{At the point $\theta=0$, the eigensystem to
first order in $dt$, with unnormalized eigenvectors, is}
\begin{equation}
    \begin{array}{c|cccc}
1-dt\gamma& 1 & 0 & 0 & 1 \\[4pt]
0 & -1 & 0 & 0 & 1 \\[4pt]
dt\gamma & 0 & \mathrm{e}^{-2\iota\varphi} & 1 & 0 \\[4pt]
0 & 0 & -\mathrm{e}^{-2\iota\varphi} & 1 & 0
\end{array},
\end{equation}
where the leftmost column denotes the eigenvalue of the eigenvector presented in the remainder of the corresponding row. 
From this, the SLD operator $\mathcal{L}_\varphi$ can be calculated as
\begin{equation}
    \mathcal{L}_\varphi=\begin{pmatrix}
 0 & 0 & 0 & 0 \\
 0 & 0 & -2\iota \mathrm{e}^{-2 \iota \varphi } & 0 \\
 0 & 2\iota\mathrm{e}^{2 \iota \varphi } & 0 & 0 \\
0 & 0 & 0 & 0 \\
\end{pmatrix}\;.
\end{equation}
The corresponding QFI is
\begin{equation}
    \label{eq:QFItheta}\mathcal{J}_\text{Q}=4dt\gamma\;,
\end{equation}
i.e.~the QFI for $\varphi$ grows linearly in time, meaning the variance shrinks as $1/(\gamma T)$, which gives the standard quantum limit. 
More rigorously, one can use additivity of the QFI to go from $dt$ in eqn.~\eqref{eq:QFItheta} to $T$ in general. 
Over a total sensing time $T$, additivity gives $\mathcal{J}_\text{Q}(T)=4\gamma T$, and the quantum Cram\'er--Rao bound gives $\mse_\varphi\geq1/(4\gamma T)$, which has standard-quantum-limit scaling.

\subsection{Unknown noise direction---estimating $\varphi$}\label{sec:unknown_noise_direction_estimation}

When $\varphi$ is unknown, the exact correction operator $Q$ cannot be implemented. 
As such, this is a paradigmatic example of sequential sensing with no clear optimal protocol, and can be illustrated in the multi-time-step process notation as
\begin{equation}
    \begin{quantikz}[classical gap=0.03cm]
        \gate[2,style={blue!0,fill=blue!0,fill opacity=0,opacity=0}]{|\psi_0\>} \bshade{2}{8}{} \wire[r][1]["{}"{above,pos=0.45}]{q} && \gate[2,style={fill=blue!20}]{\mathcal{Q}} \wire[l][1]["{C}"{above,pos=0.45}]{q} &
            && \gate[2,style={fill=blue!20}]{\mathcal{Q}} &
            & \gate[2,style={fill=blue!20}]{M} \\
        & \gate[1,style={fill=green!20}]{\mathcal{E}_{\theta,\varphi}} \wire[l][1]["{}"{above,pos=0.45}]{q} \setwiretype{q} \gshade{1}{1}{} & \wire[l][1]["{A}"{above,pos=0.45}]{q} & 
            \invgate{\dots} & \gate[1,style={fill=green!20}]{\mathcal{E}_{\theta,\varphi}} \wire[l][1]["{}"{above,pos=0.6}]{q} \wire[r][1]["{}"{above,pos=0.45}]{q} \gshade{1}{1}{} &&
            \gate[1,style={fill=green!20}]{\mathcal{E}_{\theta,\varphi}} \wire[l][1]["{}"{above,pos=0.6}]{q} \wire[r][1]["{}"{above,pos=0.45}]{q} \gshade{1}{1}{} &
    \end{quantikz} .
\end{equation}
A natural approach to this problem appears to be to learn $\varphi$ and then use this knowledge to implement an approximate version of $Q$, $\hat{Q}$. 
{For the independent-probe calibration protocol in Appendix~\ref{apensubsubestimatingnoisedir}, a calibration time $T_{\rm cal}$ gives $\mse_\varphi\geq1/(4\gamma T_{\rm cal})$ for locally unbiased estimation.
We define the signed calibration error as $\epsilon=\hat{\varphi}-\varphi$, so that $\mathbb{E}[\epsilon^2]=\mse_\varphi$.
The total time includes both calibration and subsequent sensing, and hence $T_{\rm cal}\leq T$.}

We now consider a protocol that uses the estimated value of $\varphi$ to predict $\theta$. 
As such, a protocol will be biased.
While unbiased estimators are optimal in the limit of many experimental repetitions\footnote{Technically, estimators with asymptotically vanishing bias are optimal.}, for finite samples, one may still wish to consider biased estimators. Given a specific estimator $\hat{\theta}_i$ with variance $v_i$ and bias $b_i$, the total MSE is given by
\begin{equation}  
\label{eq:MSEbias}
\sum_{i=1}^m\mse_{\theta_i}=\sum_i (v_i+b_i^2)\;.
\end{equation}

One particular protocol proceeds as follows: when the error projector is observed, we implement an imperfect correction operation $Q_{\hat{\varphi}}=L_{\hat{\varphi}}/\sqrt{\gamma}$ acting on the probe. 
Note that at this stage we consider a fixed $\hat{\varphi}$, such that $\epsilon$ is a fixed number, and later we will average over all possible $\epsilon$. Together with the check measurement this gives a recovery channel
\begin{equation}\label{eqn:unk_noise_recovery}
    \mathcal{Q}(\cdot) = \Pi_\text{c}\cdot\Pi_\text{c} + Q_{\hat{\varphi}}\Pi_\text{e} \cdot \Pi_\text{e} Q_{\hat{\varphi}}^\dag \;.
\end{equation}
Note that the action of correction operation $Q_{\hat{\varphi}}$ on the jump operator is proportional to a Pauli-$Z$ rotation:
\begin{equation}
\label{eq:CL}
    Q_{\hat{\varphi}} L_\varphi = \sqrt{\gamma} R_{z,2\epsilon}
\end{equation}
for $\epsilon=\hat{\varphi}-\varphi$ and $R_{z,2\epsilon} = e^{-\iota\frac{2\epsilon}{2}\sigma_z}$ denoting the single-qubit Pauli-$Z$ rotation by $2\epsilon$ angle. 
Here we used the identities $(\cos\varphi')(\cos\varphi)+(\sin\varphi')(\sin\varphi) = \cos(\varphi-\varphi')$ and $(\cos\varphi')(\sin\varphi)-(\sin\varphi')(\cos\varphi) = \sin(\varphi-\varphi')$.
Therefore, when an error is observed, the state becomes
\begin{equation}
\begin{split}
      I\otimes Q_{\hat{\varphi} }K_1\ket{\psi_0} &=\sqrt{dt}\sqrt{\gamma} R_{z,2\epsilon}\ket{\psi_0} + O(dt^2) 
      =\sqrt{dt}\sqrt{\gamma} (I-\iota\epsilon \sigma_z)\ket{\psi_0} + O(dt^2)\;,
\end{split}
\end{equation}
where in the second line we have expanded $R_{z,2\epsilon}$ assuming $\epsilon$ is small. Thus, the overall evolution is
\begin{equation}
\begin{aligned}
    \rho_\outs =& (1-dt\gamma/2)^2 \ketbra{\psi_0}-\iota\theta dt (\sigma_z\ketbra{\psi_0}-\ketbra{\psi_0}\sigma_z) + dt\gamma (I-\iota\epsilon \sigma_z)\ketbra{\psi_0}(I-\iota\epsilon \sigma_z)^\dagger+O(dt^2)\\
    =&\ketbra{\psi_0}-\iota(\theta+\gamma\epsilon) dt (\sigma_z\ketbra{\psi_0}-\ketbra{\psi_0}\sigma_z) +O(dt^2)\;.
\end{aligned}
\end{equation}
Evidently this approach introduces a bias of $\gamma\epsilon$ to first order in $dt$. 
From eqn.~\eqref{eq:MSEbias}, we can therefore see that $\mathbb{E}[(\hat{\theta}-\theta)^2]\geq\gamma^2\mathbb{E}[\epsilon^2]\geq\gamma/(4T)$, which is not Heisenberg scaling. 

These observations rule out Heisenberg scaling for the specified independent-probe calibration followed by plug-in correction. 
Appendix~\ref{apen:nuisance} separately analyzes the fixed-duration protocol with recovery $Q=X$ and an unknown continuous noise direction.
Neither observation implies a general impossibility of Heisenberg scaling for arbitrary sequential strategies.
In spite of this, there does exist a protocol achieving Heisenberg scaling, as we have shown in Section~\ref{sec:sensing_correlated_random_noise_direction} for a noise direction $\varphi$ taken from a finite set of uniformly spaced periodic angles.

There is no contradiction with the open question in the main text.
The negative conclusion above concerns an arbitrary continuously valued but fixed $\varphi$ and the particular estimate-then-correct protocol.
It does not establish a general impossibility result for all sequential strategies.
If $\varphi$ is promised to belong to a fixed finite set, one may instead identify the corresponding label with exponentially small error and then apply the known-direction recovery from Appendix~\ref{sec:known_noise_dir}.
For the uniformly spaced set considered in Section~\ref{sec:sensing_correlated_random_noise_direction}, the fixed cyclic recovery given there also cancels the unknown logical rotation.
Both constructions give Heisenberg scaling in the finite-set model, whereas the construction of a strategy for arbitrary continuous $\varphi$ remains an open problem.

\subsection{Estimation with unknown noise direction}
\label{apenA5}
In this subsection, we present further details on protocols that do not involve knowing the noise direction $\varphi$. 
{We first analyze the fixed-duration protocol with recovery $Q=X$, treating the unknown direction $\varphi$ as a nuisance parameter, giving SQL-scaling Fisher information and demonstrating that Heisenberg scaling is not attainable.}
We then analyze the alternative problem discussed in Section~\ref{sec:sensing_correlated_random_noise_direction} in the main text where $\varphi$ is drawn from a discrete set.
\atchg{The protocol in
Appendix~\ref{sec:unknown_noise_direction_estimation}
uses a separate calibration stage and applies the
estimate-dependent recovery $Q_{\hat{\varphi}}$.
Here the recovery is always $Q=X$, and the jump count and
final measurement outcome are used together to estimate
$\theta$ in the presence of the unknown parameter $\varphi$.}

\subsubsection{Nuisance parameter estimation}
\label{apen:nuisance}

\begin{atenv}
Consider the protocol that measures the code and error
syndromes and applies the fixed recovery $Q=X$ after each
detected jump.
Let $N_{\rm E}$ denote the number of detected jumps during
the fixed interrogation time $T$.
Since $L_\varphi^\dagger L_\varphi=\gamma I$,
\begin{equation}
    P(N_{\rm E}=n)
    =
    e^{-\gamma T}\frac{(\gamma T)^n}{n!},
\end{equation}
independently of $\theta$ and $\varphi$.

Each corrected jump contributes $R_{z,-2\varphi}$.
With $H_\theta=\theta\sigma_z$, the final state conditioned on $N_{\rm E}=n$ is therefore
\begin{equation}
    |\psi_{\theta,\varphi|n}\>
    =
    \frac{
        e^{-\iota(\theta T-n\varphi)}|00\>
        +
        e^{\iota(\theta T-n\varphi)}|11\>
    }{\sqrt{2}}.
\end{equation}
\atchg{After observing $N_{\rm E}=n$, measure in the following orthonormal basis of the code space:}
\begin{equation}
    |\pm_{\chi_n}\>
    :=
    \frac{|00\>\pm e^{\iota\chi_n}|11\>}{\sqrt{2}},
    \qquad
    \Pi_{\pm|n}
    :=
    \ketbra{\pm_{\chi_n}}.
\end{equation}
\atchg{Here $\chi_n$ is a controllable phase specifying the measurement basis.
This measurement has probabilities}
\begin{equation}
    p_\pm(\theta,\varphi|n)
    =
    \frac{
        1\pm\cos(2\theta T-2n\varphi-\chi_n)
    }{2}.
\end{equation}
\atchg{For example, at the point
$(\theta_0,\varphi_0)$, choose}
\begin{equation}
    \atchg{
    \chi_n
    =
    2\theta_0T-2n\varphi_0+\frac{\pi}{2}.
    }
\end{equation}
\atchg{Then $p_+=p_-=1/2$ at that point.
These measurement settings are held fixed when taking
parameter derivatives.}
The Fisher information for the phase is one.
Applying the chain rule therefore gives the outer product of the phase gradient $(2T,-2n)$ with itself.
This measurement also attains the SLD quantum Fisher information matrix of the conditional state at the operating point.
\end{atenv}

\atchg{Conditioned on the observed count $N_{\rm E}$, the CFI for simultaneously estimating both $\theta$ and $\varphi$ is}
\begin{equation}
    \mathcal{J}=\begin{pmatrix}
        \mathcal{J}_{\theta\theta}&\mathcal{J}_{\theta\varphi}\\
        \mathcal{J}_{\varphi\theta}&\mathcal{J}_{\varphi\varphi}
    \end{pmatrix}=4\begin{pmatrix}
        T^2&-N_\text{E}T\\-N_\text{E}T&N_\text{E}^2
    \end{pmatrix}\;.
\end{equation}

\begin{atenv}
The Fisher information matrix is singular because
\begin{equation}
\begin{aligned}
    \det\mathcal{J}
    &=
    16T^2N_{\rm E}^2 - 16(N_{\rm E}T)^2
    = 0.
\end{aligned}
\end{equation}
Thus its inverse $\mathcal{J}^{-1}$ does not exist.

\begin{atenv}
We treat $\varphi$ as a nuisance parameter: $\varphi$ is unknown and affects the outcome probabilities.
A locally unbiased estimator of $\theta$ must then satisfy,
\begin{equation}
    \mathbb{E}[\hat{\theta}]=\theta,
    \qquad
    \partial_\theta\mathbb{E}[\hat{\theta}]=1,
    \qquad
    \partial_\varphi\mathbb{E}[\hat{\theta}]=0.
\end{equation}
The last condition prevents a small change in the unknown noise direction from being interpreted as a change in $\theta$.
If $\varphi$ is known, this condition is unnecessary and the relevant scalar Fisher information is $\mathcal{J}_{\theta\theta}$.
If both parameters are estimated, the same nuisance penalty appears in the $\theta$ entry of the inverse joint Fisher information matrix.
\end{atenv}

For a fixed nonzero observed value of $N_{\rm E}$, the effective Fisher information for $\theta$ after treating $\varphi$ as a nuisance parameter is its Schur complement:
\begin{equation}
\begin{aligned}
    \mathcal{J}_{\theta|\varphi}
    :=
    \mathcal{J}_{\theta\theta}
    -
    \mathcal{J}_{\theta\varphi}
    \mathcal{J}_{\varphi\varphi}^{-1}
    \mathcal{J}_{\varphi\theta}
    = 4T^2 - 4(N_{\rm E}T) \frac{1}{N_{\rm E}^2} (N_{\rm E}T)
    = 0,
\end{aligned}
\end{equation}
The corresponding locally unbiased Cram\'er--Rao bound therefore diverges.
Equivalently, for fixed $N_{\rm E}\neq0$, the conditional likelihood depends only on the combination $T\theta-N_{\rm E}\varphi$, so the two parameters cannot be separated within this conditional model.
For $N_{\rm E}=0$, the conditional likelihood contains no $\varphi$ dependence and this particular nuisance-parameter conclusion does not apply.
More generally, this calculation does not constitute a no-go theorem for strategies that use the complete jump record or jointly optimize the sensing and recovery operations.
\end{atenv}

\begin{atenv}
For the complete experimental record, the observed count $N_{\rm E}$ must also be included.
Because its probability distribution is parameter independent, the full Fisher information is the average of the conditional matrices:
\begin{equation}
\begin{aligned}
    \mathcal{J}_{\rm full}
    &=
   4 \begin{pmatrix}
        T^2 & T\mathbb{E}[N_{\rm E}]\\
        T\mathbb{E}[N_{\rm E}] & \mathbb{E}[N_{\rm E}^2]
    \end{pmatrix}
    =
    4\begin{pmatrix}
        T^2 & -\gamma T^2\\
        -\gamma T^2 & \gamma T+\gamma^2T^2
    \end{pmatrix}.
\end{aligned}
\end{equation}
For $\gamma T>0$, this matrix is nonsingular.
Its effective Fisher information for $\theta$ is
\begin{equation}
    \mathcal{J}_{\theta|\varphi}^{\rm full}
    =
    \frac{4T^2}{1+\gamma T},
\end{equation}
and hence
\begin{equation}
    \mse_\theta
    \geq
    \frac{1+\gamma T}{4T^2}.
\end{equation}
Thus the full record resolves the conditional non-identifiability, but this fixed-duration, fixed-recovery protocol still has only SQL scaling for unknown $\varphi$.
This calculation does not rule out other sensing and recovery strategies.
\end{atenv}

\subsubsection{Discriminating a finite set of possible errors}
\label{apen:discretenoise}

Now we consider the more restricted noise model with $\varphi=j\pi/l$ for $j=0,...,l-1$, where $j$ is unknown but $l$ is known, as discussed in Section~\ref{sec:sensing_correlated_random_noise_direction}. 
We now ask what is the minimum time required for the experimentalist to determine $j$. For this we use the quantum Chernoff bound~\cite{nussbaum2010asymptotically}. We consider preparing the state $\ket{\psi_0}=(\ket{00}+\ket{11})/\sqrt{2}$, and using fresh copies in successive infinitesimal intervals until an error is observed. 
We denote by $\ket{\psi_{0,j}}$ the state after the jump operator corresponding to $\varphi=j\pi/l$ acts upon it. 

\atchg{To obtain independent copies of the jump states $\ket{\psi_{0,j}}$, we prepare a fresh Bell probe in each infinitesimal interval, retain the state when an error is detected, and otherwise restart the probe.
In the ideal fast-control limit, coherent evolution within one interval vanishes, while the detected jump state is $(I\otimes E_{\varphi_j})|\psi_0\>$.
This specifies how the coherent rotation is removed from the calibration argument.}
Then we have that
\begin{equation}
    |\<\psi_{0,j}|\psi_{0,k}\>|^2=\cos^2\left(\frac{j\pi}{l}-\frac{k\pi}{l}\right)\;.
\end{equation}

\begin{atenv}
For a fixed finite set of pure states $\{\ket{\psi_{0,i}}\}_i$, the multiple-state quantum Chernoff bound identifies the optimal asymptotic discrimination error exponent $\mathrm{e}^{-n\kappa}$ (c.f.~\cite{nussbaum2010asymptotically}).
Let $P_{\rm err}(n)$ denote the minimum worst-case error probability over joint measurements on $n$ copies of the unknown jump state.
Then, for $l>2$,
\begin{equation}\label{eqn:finite_jump_chernoff_rate}
    \lim_{n\to\infty}-\frac{1}{n}\log P_{\rm err}(n)=\kappa.
\end{equation}
The same exponent applies to the minimum equal-prior average error, since the minimum average and worst-case errors differ by at most a factor $l$.
For these pure states, the Chernoff exponent $\kappa$ is determined by their squared overlaps:
\end{atenv}
\begin{equation}
    \kappa=\min_{j\neq k}-\log\left(\cos^2\left(\frac{j\pi}{l}-\frac{k\pi}{l}\right)\right)
    = -2\log\left(
        \cos\left(
            \frac{\pi}{l}
        \right)
    \right)
    \quad
    \text{for }l>2 .
\end{equation}
\begin{atenv}
For $l=2$, the two possible jump states are orthogonal and can be discriminated perfectly after one observed jump, corresponding formally to $\kappa=\infty$.
Equation~\eqref{eqn:finite_jump_chernoff_rate} implies that $-n^{-1}\log P_{\rm err}(n)\geq\kappa/2$ for all sufficiently large $n$.
Consequently, there exist $A>0$ and $n_0$ such that
\begin{equation}
    P_{\rm err}(n)
    \leq
    A e^{-n\kappa/2},
    \qquad n\geq n_0.
\end{equation}

Since $L_j^\dagger L_j=\gamma I$, jumps form a Poisson process with rate $\gamma$, independently of $j$.
If $W_r$ is the waiting time between the $(r-1)$th and $r$-th jumps, then the $W_r$'s are independent and
\begin{equation}
    \mathbb{E}[W_r]
    =
    \int_0^\infty P(W_r>w)\,dw
    =
    \int_0^\infty e^{-\gamma w}\,dw
    =
    \frac{1}{\gamma}.
\end{equation}
The time $\tau_n$ needed to collect $n$ jump states satisfies
\begin{equation}
    \tau_n=\sum_{r=1}^nW_r,
    \qquad
    \mathbb{E}[\tau_n]=\frac{n}{\gamma}.
\end{equation}

For a deterministic total time $T$, define
\begin{equation}
\begin{gathered}
    c_0:=\min\left\{\frac{\kappa}{2},\frac14\right\},\qquad
    n(T):=
    \max\left\{
        n_0,\,
        \left\lceil
            \frac{3}{c_0}\log(1+\gamma T)
        \right\rceil
    \right\},\qquad
    t_{\rm cal}:=\frac{2n(T)}{\gamma}.
\end{gathered}
\end{equation}
For fixed $\gamma>0$ and fixed $l$, both $\mathbb{E}[\tau_{n(T)}]$ and $t_{\rm cal}$ are $O(\log T)$, since
\begin{equation}
    \mathbb{E}[\tau_{n(T)}]
    =
    \frac{n(T)}{\gamma}
    =
    \frac{3}{\gamma c_0}\log(1+\gamma T)+O(1)
    =
    O(\log T).
\end{equation}
Stop calibration once $n(T)$ jump states have been collected, and declare a calibration failure if this has not happened by $t_{\rm cal}$.
A Poisson lower-tail bound gives
\begin{equation}
\begin{aligned}
    \Pr(\tau_{n(T)}>t_{\rm cal})
    &=
    \Pr\!\left(
        \operatorname{Poisson}(2n(T))<n(T)
    \right)
    \leq
    e^{-n(T)/4}.
\end{aligned}
\end{equation}
Combining discrimination errors and calibration failures,
\begin{equation}
    P_{\rm fail}(T)
    \leq
    (A+1)(1+\gamma T)^{-3}
    =
    O(T^{-3}).
\end{equation}
The remaining sensing time is $T-t_{\rm cal}=T-O(\log T)$.
A bounded estimator on a failure event therefore makes its MSE contribution negligible compared with $T^{-2}$.
The known-direction recovery can then be used for the remaining sensing time.
The corresponding local SLD estimator has values $\theta_0\pm[2(T-t_{\rm cal})]^{-1}$, which remain bounded for sufficiently large $T$.
The joint Fisher information, with the latent label supplied, is at most $4T^2$.
\atchg{Here the bound follows from the continuous-time signal generator $\sigma_z$ and total interrogation time $T$; this is the continuous-time counterpart of the generator argument used in the proof of Proposition~\ref{prop:finite_latent_hs}.}
The affine normalization used in the proof of Proposition~\ref{prop:finite_latent_hs} therefore restores local unbiasedness without changing the $O(T^{-2})$ scaling.

For $l=2$, one detected jump identifies the direction perfectly and the expected waiting time is $1/\gamma$.
A deterministic calibration deadline $3\log(1+\gamma T)/\gamma$ makes the probability of observing no jump equal to $(1+\gamma T)^{-3}$.
For $l=1$, no direction calibration is needed.

The distinction between continuous and discrete time concerns the resource used for calibration.
In the continuous-time model, the informative jump events occur randomly, so collecting $n$ jump states requires a random physical time with mean $n/\gamma$.
In the discrete-time model of Section~\ref{sec:unknown_corr_noise}, the fixed unitary $E_{\varphi_j}$ occurs once per channel use, so $n$ calibration probes require exactly $n$ channel uses.
In both cases, a fixed finite set can be discriminated with exponentially small error using a logarithmic calibration overhead, under the respective ideal-control assumptions.
\atchg{The continuous-time argument uses the limit $dt\rightarrow0$ with arbitrarily frequent syndrome measurements.
It does not identify a finite-duration Lindblad channel with one application of the discrete unitary $E_{\varphi_j}R_{z,\theta}$.}
\end{atenv}

\section{Kraus representations and bounds for correlated processes}
\label{app:correlated_kraus_bounds}

In this appendix, we first give a proof the Kraus-representation independence of HNKS and then gives the details of the correlated-process Kraus-gauge construction summarized in Section~\ref{subsec:HNKS}, following the results in refs.~\cite{zhou2021asymptotic,Altherr_2021,Kurdzialek_2025}.
We distinguish the exact quantum Fisher information for a fixed number of sensing rounds from asymptotic bounds on a family of correlated processes.

\subsection{Kraus-representation independence}
\label{app:kraus_representation_independence}

\begin{atenv}
Consider the Kraus vector
\begin{equation}
    \mathbf{K}_\theta
    :=
    (K_{1,\theta},\ldots,K_{r,\theta})^\top.
\end{equation}
After padding with zero Kraus operators when necessary,
equivalent differentiable Kraus representations are
related locally by
\begin{equation}
    \mathbf{K}_\theta'
    =
    u_\theta\mathbf{K}_\theta,
    \qquad
    u_\theta^\dagger u_\theta=I.
\end{equation}
Differentiating $u_\theta^\dagger u_\theta=I$ shows that
\begin{equation}
    h_\theta
    :=
    \iota u_\theta^\dagger
    \frac{\partial u_\theta}{\partial\theta}
    =
    h_\theta^\dagger.
\end{equation}
The channel Hamiltonians therefore satisfy
\begin{equation}
\begin{aligned}
    H_{\rm K}'(\mathcal{C}_\theta)
    &=
    H_{\rm K}(\mathcal{C}_\theta) + \sum_{a,b=1}^{r} (h_\theta)_{a,b} K_{a,\theta}^\dagger K_{b,\theta}.
\end{aligned}
\end{equation}
The additional term belongs to
$\mathcal{S}_{\rm K}(\mathcal{C}_\theta)$.

The Hermitian Kraus span is also unchanged.
Indeed, a Hermitian coefficient matrix $h'$ for the primed
representation gives the Hermitian coefficient matrix
$u_\theta^\dagger h'u_\theta$ for the original
representation.
Conversely, any Hermitian coefficient matrix $h$ is obtained
by choosing $h'=u_\theta h u_\theta^\dagger$.
Consequently,
\begin{equation}
\begin{aligned}
    H_{\rm K}'(\mathcal{C}_\theta)
    &\in\mathcal{S}_{\rm K}(\mathcal{C}_\theta)
    \Longleftrightarrow
    H_{\rm K}(\mathcal{C}_\theta)
    \in\mathcal{S}_{\rm K}(\mathcal{C}_\theta).
\end{aligned}
\end{equation}
Thus HNKS is independent of the Kraus representation,
as discussed following eqn.~(11) of
Ref.~\cite{zhou2021asymptotic}.
\end{atenv}

To see the representation independence explicitly, introduce the Kraus vector
\begin{equation}
    \mathbf{K}_\theta
    :=
    (K_{1,\theta},\ldots,K_{r,\theta})^\top .
\end{equation}
Another Kraus representation can be written as
\begin{equation}
    \mathbf{K}_\theta'
    =
    u_\theta\mathbf{K}_\theta
    \quad\text{for}\quad
    u_\theta^\dagger u_\theta=I.
\end{equation}
Defining $h_\theta := \iota u_\theta^\dagger \frac{\partial u_\theta}{\partial\theta}$, the corresponding channel Hamiltonians satisfy
\begin{equation}
\begin{aligned}
    H_{\rm K}'(\mathcal{C}_\theta)
    &=
    H_{\rm K}(\mathcal{C}_\theta)
    +
    \sum_{a,b=1}^{r}
    (h_\theta)_{a,b}
    K_{a,\theta}^\dagger K_{b,\theta}.
\end{aligned}
\end{equation}
The second term belongs to \(\mathcal{S}_{\rm K}(\mathcal{C}_\theta)\), while the Hermitian Kraus span itself is unchanged by the isometric transformation \(u_\theta\).
Consequently,
\begin{equation}
    H_{\rm K}'(\mathcal{C}_\theta)
    \in
    \mathcal{S}_{\rm K}(\mathcal{C}_\theta)
    \quad\Longleftrightarrow\quad
    H_{\rm K}(\mathcal{C}_\theta)
    \in
    \mathcal{S}_{\rm K}(\mathcal{C}_\theta).
\end{equation}
This is the Kraus-gauge argument following eqn.~(11) of Ref.~\cite{zhou2021asymptotic}.

\subsection{Correlated-process bounds}

\begin{atenv}
The relevant Kraus-gauge construction starts from a decomposition of the complete parameterized process,
\begin{equation}\label{eqn:sequential_parameterized_process_kraus_decomposition}
    \mathbf{\Lambda}_\theta^{(T)}
    =
    \sum_a
    \kketbra{K_{a,\theta}^{(T)}},
\end{equation}
where $\{|K_{a,\theta}^{(T)}\rr\}_a$ are vectorized Kraus operators of the complete $T$-step process.
For a fixed admissible sequential tester $\mathbf{M}$, this decomposition induces the ensemble decomposition
\begin{equation}
\begin{gathered}
    \rho_\theta^{\mathbf{M}}
    =
    \mathbf{M}^{1/2}
    \left(
        \mathbf{\Lambda}_\theta^{(T)}
    \right)^\top
    \mathbf{M}^{1/2}
    =
    \sum_a
    |\phi_{a,\theta}^{\mathbf{M}}\>
    \<\phi_{a,\theta}^{\mathbf{M}}|
    \quad \text{for}\quad
    |\phi_{a,\theta}^{\mathbf{M}}\>
    :=
    \mathbf{M}^{1/2}
    |\overline{K_{a,\theta}^{(T)}}\rr.
\end{gathered}
\end{equation}
Equivalent differentiable decompositions are parameterized by a Hermitian matrix $h=h^\dagger$ through
\begin{equation}
    |\dot{\widetilde K}_{a,\theta}^{(T)}\rr
    =
    |\dot K_{a,\theta}^{(T)}\rr
    -
    \iota
    \sum_b
    h_{a,b}
    |K_{b,\theta}^{(T)}\rr.
\end{equation}
The purification formula for the SLD quantum Fisher information minimizes the corresponding derivative norm over $h$, while maximizing over causally admissible testers gives the comb quantum Fisher information.
This construction is given in eqns.~(6)--(8) and Theorem~1 of Ref.~\cite{Altherr_2021}, with its optimization and dual derivations in Appendices~A.1--A.3.
The equivalent fixed-$T$ expression is given in eqn.~(4) of Ref.~\cite{Kurdzialek_2025}.

For a regular, locally identifiable single-parameter problem, this quantity is $\mathcal{J}_{\seq}(\mathbf{\Lambda}_\theta^{(T)})$, whose reciprocal is exactly the SDP value in Theorem~\ref{thm:sequential_NH_bound}, as shown in Proposition~\ref{prop:single_par_achievability}.
\end{atenv}

The fixed-$T$ expression does not by itself give an asymptotic HNKS criterion.
Treating one fixed $T$-step process as a global channel and applying ordinary HNKS would characterize independent repetitions of that entire block, rather than the scaling of the correlated family $\{\mathbf{\Lambda}_\theta^{(T)}\}_{T\geq1}$ as $T$ increases.

For the latter problem, Ref.~\cite{Kurdzialek_2025} divides the process into blocks of $q$ consecutive teeth, where the corresponding block length is denoted by $m$ in that reference.
For this recursive construction, each block retains the environmental input and output through which correlations pass between blocks.
The vectors $|K_{a,\theta}^{(q)}\rr$ below refer to a Kraus decomposition of this block with its environmental boundary systems included, and $\tr_{\rm out}$ traces over its final probe and environmental outputs.
The optimization allows access to the environmental systems between blocks as a relaxation used to obtain an upper bound.
This distinction is essential when applying the construction to a correlated process.

For a Kraus gauge $h=h^\dagger$, the derivatives of the $q$-round block Kraus vectors are
\begin{equation}
\begin{aligned}
    |\dot{\widetilde K}_{a,\theta}^{(q)}\rr
    &:=
    |\dot K_{a,\theta}^{(q)}\rr
    -\iota\sum_b h_{a,b}|K_{b,\theta}^{(q)}\rr \quad
    \text{for}\quad
    |\dot K_{a,\theta}^{(q)}\rr
    :=
    \frac{\partial}{\partial\theta}
    |K_{a,\theta}^{(q)}\rr.
\end{aligned}
\end{equation}
The dot denotes differentiation with respect to $\theta$ at fixed block length $q$.
The tilde denotes a parameter-dependent change of Kraus representation, chosen to agree with the original representation at the operating point; $h$ specifies its local derivative.
This changes the representation without changing the physical block.
We then define
\begin{equation}
\begin{aligned}
    \boldsymbol{\alpha}^{(q)}(h)
    &:=
    \tr_{\rm out}
    \left(
        \sum_a
        |\dot{\widetilde K}_{a,\theta}^{(q)}\rr
        \ll\dot{\widetilde K}_{a,\theta}^{(q)}|
    \right),\quad\text{and}\quad
    \boldsymbol{\beta}^{(q)}(h)
    :=
    \tr_{\rm out}
    \left(
        \sum_a
        |\dot{\widetilde K}_{a,\theta}^{(q)}\rr
        \ll K_{a,\theta}^{(q)}|
    \right).
\end{aligned}
\end{equation}
Their causally constrained analogues of the single-channel quantities $\|\alpha\|$ and $\|\beta\|$ are
\begin{equation}
\begin{aligned}
    a^{(q)}(h)
    &:=
    \max_{\widetilde{\mathbf{C}}^{(q)}}
    \tr
    \left(
        \boldsymbol{\alpha}^{(q)}(h)
        \widetilde{\mathbf{C}}_{00}^{(q)}
    \right),
    \quad\text{and}\quad
    b^{(q)}(h)
    &:=
    \max_{\widetilde{\mathbf{C}}^{(q)}}
    \operatorname{Re}
    \tr
    \left(
        \boldsymbol{\beta}^{(q)}(h)
        \widetilde{\mathbf{C}}_{10}^{(q)}
    \right),
\end{aligned}
\end{equation}
where $\widetilde{\mathbf{C}}^{(q)}$ ranges over the relevant causal testers and the subscripts $00$ and $10$ denote their blocks with respect to the virtual two-dimensional space used in the recursive construction.
\atchg{Here $F^{(\ell)}$ denotes the recursively constructed upper bound on the optimal QFI after $\ell$ sensing rounds, with $F^{(0)}=0$, as defined in eqn.~(8) of Ref.~\cite{Kurdzialek_2025}.}
These quantities give the recursive bound
\begin{equation}
    F^{(\ell+q)}
    \leq
    F^{(\ell)}
    +
    4\min_h
    \left[
        a^{(q)}(h)
        +
        \sqrt{F^{(\ell)}}\,b^{(q)}(h)
    \right].
\end{equation}
This is eqn.~(9) of Ref.~\cite{Kurdzialek_2025}; the performance operators and the derivation of $a^{(q)}$ and $b^{(q)}$ are given in Appendix~D.2, particularly eqns.~(30) and~(34).

If there is a gauge $h$ for which $b^{(q)}(h)=0$, eqn.~(10) of Ref.~\cite{Kurdzialek_2025} gives
\begin{equation}
    \limsup_{N\rightarrow\infty}
    \frac{
        \mathcal{J}_{\seq}
        \left(
            \mathbf{\Lambda}_\theta^{(N)}
        \right)
    }{N}
    \leq
    \frac{4}{q}
    \min_{h:\,b^{(q)}(h)=0}
    a^{(q)}(h),
\end{equation}
and therefore rules out Heisenberg scaling.
If $b^{(q)}(h)\neq0$ for every gauge, eqn.~(11) gives only a quadratic upper bound.
It permits, but does not establish, Heisenberg scaling.
The asymptotic derivations are given in Appendix~D.3, particularly eqns.~(41) and~(42), and the SDP formulations are given in Appendix~D.5.
Because these correlated-process bounds are not generally tight, they do not provide a necessary and sufficient criterion for arbitrary correlated combs.

\section{Proof of Proposition~\ref{prop:finite_latent_hs}}\label{app:proof_finite_latent_hs}

\begin{proof}
    \atchg{All protocols and estimator values are held fixed when taking derivatives at the operating point.
    The latent label is used only in the analysis of the joint distribution and is not supplied to the experimentalist.}
    \atchg{Measurement settings are likewise held fixed; the scheme may be designed using the known point $\theta$.}
    First, let us choose $n = n(T) = \left\lceil \frac{3}{c}\log T \right\rceil$ so that $n(T)\geq\frac{3}{c}\log T$.
    The calibration-error bound in Definition~\ref{def:admissible_calibration} gives
    \begin{equation}
    \begin{aligned}
        P_{\rm err}(n(T);\theta')
        &\leq
        A e^{-cn(T)}
        \leq
        A e^{-3\log T}
        =
        A T^{-3}
        =
        O(T^{-3}).
    \end{aligned}
    \end{equation}
    Then, let $N(T):=T-n(T)=T-O(\log T)$ be the number of remaining sensing rounds.
    Because \(\log T=o(T)\), we have
    \begin{equation}
        \frac{N(T)}{T}
        \longrightarrow
        1
    \end{equation}
    as $T\rightarrow\infty$.
    Hence $N(T)\geq T/2$ for all sufficiently large $T$.
    
    \begin{atenv}
    Now, let us define
    \begin{equation}
        \kappa_{\min}
        :=
        \min_{j:p_j>0}\kappa_j
        >
        0,
    \end{equation}
    where positivity follows from the finiteness of $\mathcal{K}$ and the positivity of each $\kappa_j$.
    Then, take $T$ sufficiently large such that $n(T)\geq n_0$, $N(T)\geq\max\{N_0,N_1\}$, and the chosen sensing continuations satisfy Definition~\ref{def:admissible_calibration}.
    \end{atenv}

    \begin{atenv}
    Before calibration, fix for each $k$ with $p_k>0$
    one of the optimal known-$k$ strategies and locally
    unbiased estimators in
    Definition~\ref{def:admissible_calibration}.
    On the reported outcome $\hat j=k$, use that strategy
    and estimator for the remaining $N(T)$ rounds.
    Without loss of generality, the decoder reports only
    labels with nonzero prior probability: any report
    with $p_k=0$ can be replaced by a fixed
    positive-probability label without increasing
    $P_{\rm err}$.
    \end{atenv}

    \atchg{On correctly identifying $j$, the concatenation assumption makes the remaining experiment an $N(T)$-round known-$j$ sensing problem.
    Proposition~\ref{prop:single_par_achievability} and eqn.~\eqref{eqn:finite_latent_quadratic_qfi} therefore give an MSE for $\theta$ conditioned on $j$,}
    \begin{equation}
    \begin{aligned}
        \mse_{\theta|j}
        &\leq
        \frac{1}{\mathcal{J}_j(N(T);\theta)}
        \leq
        \frac{1}{\kappa_{\min}N(T)^2}
        \leq
        \frac{4}{\kappa_{\min}T^2}.
    \end{aligned}
    \end{equation}
    Let $\hat j$ denote the label returned by the calibration protocol and define the correct-calibration event as $\mathsf{C}:=\{\hat j=j\}$.
    Its complement $\mathsf{C}^{\rm c}=\{\hat j\neq j\}$ is the incorrect-calibration event.
    \atchg{Let $\hat{\theta}_{\rm cl}$ be the final estimator, with values in the interval $\mathcal{I}$ specified in Definition~\ref{def:admissible_calibration}.
    The requirement concerns every assigned outcome value, including outcomes that can occur when $\hat j\neq j$.
    The subscript ${\rm cl}$ labels this bounded estimator; no additional clipping step is needed.}
    The squared-error loss of the bounded estimator $\hat{\theta}_{\rm cl}$ is therefore $\delta:=(\hat{\theta}_{\rm cl}-\theta)^2$.
    Moreover, on $\mathsf{C}^{\rm c}$, we have
    \begin{equation}
        \delta
        \leq
        B
        :=
        \operatorname{diam}(\mathcal{I})^2.
    \end{equation}
    Here $\operatorname{diam}(\mathcal{I}) :=\sup_{\theta_1,\theta_2\in\mathcal{I}} |\theta_1-\theta_2|$ denotes the diameter of $\mathcal{I}$.
    In particular, if $\mathcal{I}= [\theta_-,\theta_+]$, then $\operatorname{diam}(\mathcal{I})=\theta_+-\theta_-$.
    Since both the {bounded} estimate and the true value $\theta$ belong to $\mathcal{I}$, their squared difference is at most $\operatorname{diam}(\mathcal{I})^2=B$.
    The law of total expectation then gives
    \begin{equation}
    \begin{aligned}
        \mathbb{E}[\delta]
        &=
        P(\mathsf{C})
        \mathbb{E}[\delta|\mathsf{C}]
        +
        P(\mathsf{C}^{\rm c})
        \mathbb{E}[\delta|\mathsf{C}^{\rm c}]
        \leq
        \frac{P(\mathsf{C})}
        {\kappa_{\min}N(T)^2}
        +
        B P(\mathsf{C}^{\rm c})
        \leq
        \frac{1}
        {\kappa_{\min}N(T)^2}
        +
        B P_{\rm err}(n(T);\theta).
    \end{aligned}
    \end{equation}
    Therefore,
    \begin{equation}
    \begin{aligned}
        \mathbb{E}[\delta]
        &\leq
        \frac{1}{\kappa_{\min}N(T)^2}
        +
        B P_{\rm err}(n(T);\theta)
        \leq
        \frac{4}{\kappa_{\min}T^2}
        +
        A B T^{-3}
        =
        O(T^{-2}).
    \end{aligned}
    \end{equation}
    It remains to impose local unbiasedness for the unknown-$j$ model.
    Write
    \begin{equation}
        f_T(\theta')
        :=
        \mathbb{E}_{\theta'}[\hat{\theta}_{\rm cl}],
        \qquad
        D:=\operatorname{diam}(\mathcal{I}).
    \end{equation}
    \atchg{For the remainder of the proof, write $P_{\rm err}:=P_{\rm err}(n(T);\theta)$.}
    On every correctly identified branch, the conditional estimator has mean $\theta$ and derivative of its mean equal to one at the operating point.
    Consequently,
    \begin{equation}
        |f_T(\theta)-\theta|
        \leq
        D P_{\rm err}.
    \end{equation}

    Let $y$ denote the complete observed record, including calibration, intermediate measurement, and final measurement outcomes.
    Let $P_{\theta'}(y\mid j)$ denote the corresponding conditional probability, or density with respect to a fixed reference measure, and write
    \begin{equation}
        \atchg{
        P_{\theta'}(j,y)
        :=p_jP_{\theta'}(y\mid j).
        }
    \end{equation}
    
    Define the joint score, on outcomes with
    $P_\theta(j,y)>0$, by
    \begin{equation}
        s_\theta(j,y)
        :=
        \left.
        \frac{\partial}{\partial\theta'}
        \log P_{\theta'}(j,y)
        \right|_{\theta'=\theta}.
    \end{equation}
    \atchg{Its second moment is the joint classical
    Fisher information
    $\mathcal{J}_{\rm joint}(T;\theta)
    :=\mathbb{E}_\theta[s_\theta^2]$.
    We derive its upper bound from the operational
    parameter-encoding condition in
    Definition~\ref{def:admissible_calibration}.}

    \atchg{The noise operations may depend on $j$ and act on persistent memory.
    Adaptive controls and measurements are allowed, with their prescribed dependence on recorded outcomes held fixed as $\theta'$ varies.}
    \atchg{We condition on a true label $j$ and purify the initial state and all intervening operations, retaining the environment and measurement records.
    The generator of the complete purified evolution is a sum of $T$ unitarily conjugated copies of $G\otimes I$ (where $G$ is the generator assumed in item 3 of Definition~\ref{def:admissible_calibration}).
    Let $g:=\lambda_{\max}(G)-\lambda_{\min}(G)$ denote the spectral range of $G$, where $\lambda_{\max}(G)$ and $\lambda_{\min}(G)$ are its largest and smallest eigenvalues, respectively.
    The total generator therefore has spectral range at most $Tg$.
    Writing $|\Psi_{j,T}(\theta')\>$ for the purified output, its SLD quantum Fisher information satisfies}
    \begin{equation}\label{eqn:finite_latent_generator_bound}
        \mathcal{J}_\text{Q}\left( \ketbra{\Psi_{j,T}(\theta)} \right)
        \leq g^2T^2.
    \end{equation}
    
    \atchg{Because $p_j$ is independent of $\theta'$,
    the joint classical Fisher information is the
    average of the conditional classical Fisher
    informations.
    Each is bounded by the corresponding purified-state quantum Fisher information, hence}
    \begin{equation}
    \begin{aligned}
        \mathcal{J}_{\rm joint}(T;\theta)
        &\leq
        \sum_{j:p_j>0}p_j\,
        \mathcal{J}_\text{Q} \left(\ketbra{\Psi_{j,T}(\theta)} \right)
        \leq g^2T^2.
    \end{aligned}
    \end{equation}
    This bound applies to the actual concatenated protocol, including incorrect calibration outcomes; no postselection on $\hat j=j$ is made.
    \begin{atenv}
    The unitary parameter encoding makes the purified output differentiable, so probabilities and expectations of bounded functions of the fixed measurement record can be differentiated at $\theta$.
    Set $s_\theta(j,y)=0$ wherever $P_\theta(j,y)=0$; these zero-density outcomes have zero first derivative at the interior point $\theta$.
    For every bounded, parameter-independent function $h(j,y)$, the score therefore satisfies
    \begin{equation}\label{eqn:finite_latent_score_identity}
        \left.\frac{d}{d\theta'}\mathbb{E}_{\theta'}[h(j,y)]\right|_{\theta'=\theta}
        =\mathbb{E}_\theta[h(j,y)s_\theta(j,y)].
    \end{equation}
    \end{atenv}
    \begin{atenv}
    Write $q_j(\theta'):=\Pr_{\theta'}(\hat j=j\mid j)$ and let $m_j(\theta')$ denote the mean of the chosen known-$j$ continuation estimator, with $m_j(\theta)=\theta$ and $m_j'(\theta)=1$.
    \atchg{The continuation condition in Definition~\ref{def:admissible_calibration}, item~\ref{item:admissible_calibration_continuation}, gives}
    \begin{equation}
        \mathbb{E}_{\theta'}[(\hat{\theta}_{\rm cl}-\theta)\mathbf{1}_{\mathsf{C}}]
        =\sum_{j:p_j>0}p_jq_j(\theta')(m_j(\theta')-\theta),
    \end{equation}
    where $\mathbf{1}_{\mathsf{C}}$ is the indicator of correct calibration and the centering value $\theta$ is fixed.
    Differentiating at $\theta$ eliminates the terms containing $q_j'(\theta)$ because $m_j(\theta)-\theta=0$, and the remaining contribution is $\sum_jp_jq_j(\theta)=1-P_{\rm err}$.
    \atchg{Apply eqn.~\eqref{eqn:finite_latent_score_identity} with $h(j,y)=(\hat{\theta}_{\rm cl}-\theta)\mathbf{1}_{\mathsf{C}^{\rm c}}$.}
    \atchg{This function is bounded and independent of $\theta'$ because the estimator values and decoder are fixed, and $\theta$ is the fixed operating point.}
    \begin{equation}
        \atchg{
        \left.\frac{d}{d\theta'}\mathbb{E}_{\theta'}[(\hat{\theta}_{\rm cl}-\theta)\mathbf{1}_{\mathsf{C}^{\rm c}}]\right|_{\theta'=\theta}
        =
        \mathbb{E}_\theta[(\hat{\theta}_{\rm cl}-\theta)s_\theta\mathbf{1}_{\mathsf{C}^{\rm c}}].
        }
    \end{equation}
    \atchg{Since $\mathbf{1}_{\mathsf{C}}+\mathbf{1}_{\mathsf{C}^{\rm c}}=1$, adding the correct-calibration contribution $1-P_{\rm err}$ gives}
    \end{atenv}
    \begin{equation}
        f_T'(\theta) = 1-P_{\rm err}
        +
        \mathbb{E}_\theta\left[
            (\hat{\theta}_{\rm cl}-\theta)
            s_\theta
            \mathbf{1}_{\mathsf{C}^{\rm c}}
        \right].
    \end{equation}
    Here $\mathbf{1}_{\mathsf{C}^{\rm c}}$ is the indicator of incorrect calibration.
    \atchg{By Definition~\ref{def:admissible_calibration}, item~\ref{item:admissible_calibration_continuation}, every assigned estimator value belongs to $\mathcal{I}$, including values on incorrectly identified branches.}
    \atchg{Since $\theta\in\mathcal{I}$, we have $|\hat{\theta}_{\rm cl}-\theta|\leq\operatorname{diam}(\mathcal{I})=D$, in particular on $\mathsf{C}^{\rm c}$.}
    Cauchy--Schwarz inequality yields
    \begin{equation}
    \begin{aligned}
        |f_T'(\theta)-1|
        &\leq
        P_{\rm err}
        +
        D\sqrt{
            P_{\rm err}\mathcal{J}_{\rm joint}(T;\theta)
        }\\
        &\leq P_{\rm err}+DgT\sqrt{P_{\rm err}} \\
        &=O(T^{-1/2}),
    \end{aligned}
    \end{equation}
    where $P_{\rm err}=P_{\rm err}(n(T);\theta)=O(T^{-3})$.
    Thus $f_T'(\theta)\rightarrow1$.
    For sufficiently large $T$, define
    \begin{equation}
        \hat{\theta}_{\rm lu}
        :=
        \theta+
        \frac{
            \hat{\theta}_{\rm cl}-f_T(\theta)
        }{f_T'(\theta)}.
    \end{equation}
    This estimator is locally unbiased at $\theta$, and
    \begin{equation}
        \mse_\theta
        =
        \frac{
            \operatorname{Var}_\theta(\hat{\theta}_{\rm cl})
        }{[f_T'(\theta)]^2}
        \leq
        \frac{\mathbb{E}[\delta]}{[f_T'(\theta)]^2}
        =
        O(T^{-2}).
    \end{equation}
    Thus the incorrect-calibration contribution is \(O(T^{-3})\), which is
    lower order than the \(O(T^{-2})\) MSE of the correctly identified
    conditional sensing problem.
\end{proof}

\section{Calibration of the weakly correlated Pauli model}
\label{app:weak_pauli_calibration}

\begin{atenv}
This appendix gives the detailed calibration and conditional
sensing arguments summarized in
Section~\ref{sec:correlated_noise_z_rot} for $p_z=0$.
A branch denotes the process conditioned on a fixed latent
label $j$.
The $X$ and $Y$ axes are called transverse because they are
perpendicular to the signal axis $Z$.
Accordingly, transverse Bell outcomes mean
$|\Phi_x\>$ and $|\Phi_y\>$.

\atchg{In each calibration round, prepare a fresh copy of $|\Phi^+\>$ from eqn.~\eqref{eqn:bell_basis}, apply $R_{z,\theta_0}^\dagger$ to the probe before the noisy channel, and then measure in the Bell basis.}
At $\theta=\theta_0$, use the decoder that returns $j=x$ or $j=y$ when the corresponding Bell outcome is observed, and returns $j=0$ if all outcomes are $|\Phi_0\>$.
Conditioned on $j=x$ or $j=y$, each round fails to reveal that
label with probability $1/2$.
The rounds are independent conditioned on $j$, so
\begin{equation}
\begin{aligned}
    P_{\rm err}(n;\theta_0)
    &=
    p_0\,0+p_x\left(\frac12\right)^n
    +p_y\left(\frac12\right)^n\\
    &=
    (p_x+p_y)2^{-n}
    \leq
    2^{-n},
\end{aligned}
\end{equation}
where the last inequality follows from
$p_x+p_y=1-p_0\leq1$.
The following argument extends this calibration to a fixed
neighborhood of $\theta_0$.
\end{atenv}

Consider the compensating rotation $R_{z,\theta_0}^\dag$ and Bell basis defined in eqn.~\eqref{eqn:bell_basis}, and let
\begin{equation}
    \delta
    :=
    \theta'-\theta_0,
    \quad
    r(\delta)
    :=
    \cos^2\left(\frac{\delta}{2}\right),
    \quad
    q(\delta)
    :=
    1-r(\delta).
\end{equation}
For the identity branch $j=0$, the Bell measurement produces only the outcomes $|\Phi_0\>$ and $|\Phi_z\>$.
For $j=x$, the probabilities of $|\Phi_x\>$ and $|\Phi_y\>$ are $r(\delta)/2$ and $q(\delta)/2$, respectively; for $j=y$, these two probabilities are interchanged.
The total probability of obtaining one of the outcomes $|\Phi_0\>$ or $|\Phi_z\>$ is $1/2$ for each of the conditional branches $j=x$ and $j=y$.

Use the following decoder: if no transverse outcome occurs, set $\hat j=0$; otherwise, choose the more frequent of the $x$ and $y$ Bell outcomes.
\atchg{Break ties by a fixed deterministic rule.}
For a true label $j\in\{x,y\}$, define 
\atchg{scores $S_t\in\{-1,0,1\}$, which are independent conditioned on $j$ and the fixed parameter value $\theta'$,} where $S_t=1$ for the Bell outcome labeled by $j$, $S_t=-1$ for the other outcome among $|\Phi_x\>$ and $|\Phi_y\>$, and $S_t=0$ to either $|\Phi_0\>$ or $|\Phi_z\>$.
\atchg{Here $t\in\{1,\ldots,n\}$ indexes calibration rounds, and $P_{j,\theta'}$ and $\mathbb{E}_{j,\theta'}$ denote probability and expectation conditioned on the latent label $j$ at parameter value $\theta'$.}
Then
\begin{equation}
    \mathbb{E}_{j,\theta'}[S_t]
    =
    \frac{r(\delta)-q(\delta)}{2}
    =
    \frac{\cos(\delta)}{2}.
\end{equation}
For
\begin{equation}
    \mathcal{U}
    :=
    \left(
        \theta_0-\frac{\pi}{3},
        \theta_0+\frac{\pi}{3}
    \right),
\end{equation}
we have $\mathbb{E}_{j,\theta'}[S_t]\geq1/4$ uniformly over $\theta'\in\mathcal{U}$.
Hoeffding's inequality therefore gives
\begin{equation}
\begin{aligned}
    P_{j,\theta'}
    \left(
        \sum_{t=1}^nS_t\leq0
    \right)
    &\leq
    \exp
    \left[
        -\frac{n}{2}
        \left(
            \mathbb{E}_{j,\theta'}[S_t]
        \right)^2
    \right]\\
    &\leq
    e^{-n/32}.
\end{aligned}
\end{equation}
The identity branch is identified without error because it never produces a transverse outcome.
Averaging over the latent distribution thus gives
\begin{equation}
    P_{\rm err}(n;\theta')
    \leq
    e^{-n/32}
\end{equation}
uniformly over $\theta'\in\mathcal{U}$.
The discrimination condition in Definition~\ref{def:admissible_calibration} therefore holds with
\begin{equation}
    A=1,
    \qquad
    c=\frac{1}{32},
    \qquad
    n_0=1.
\end{equation}

To verify eqn.~\eqref{eqn:finite_latent_quadratic_qfi}, use the probe--ancilla code
\begin{equation}
    \mathcal{C}
    =
    \operatorname{span}
    \{
        |00\>,
        |11\>
    \}.
\end{equation}
For $j=x$ or $j=y$, the no-error component of $\mathcal{E}_j$ remains in the code space, whereas the $\sigma_j$ component maps the state into the orthogonal error space.
\atchg{The first tensor factor is a noiseless ancilla and the second is the probe.}
\atchg{The code and error syndromes are the two outcomes of the projective measurement $\{\Pi_{\mathcal{C}},I-\Pi_{\mathcal{C}}\}$, where $\Pi_{\mathcal{C}}:=|00\>\<00|+|11\>\<11|$ projects onto $\mathcal{C}$.}
\atchg{After each channel use, measure these projectors and apply $I\otimes\sigma_j$ upon outcome $I-\Pi_{\mathcal{C}}$.}
The signal rotation preserves $\mathcal{C}$, and $(I\otimes\sigma_j)^2=I$, so this recovery restores the state after the signal rotation exactly, whereas the branch $j=0$ is already noiseless.

The recovered noiseless strategy with initial state $(|00\>+|11\>)/\sqrt{2}$ attains QFI $N^2$, so the optimal known-$j$ QFI satisfies $\mathcal{J}_j(N;\theta')\geq N^2$.
Conversely, eqn.~\eqref{eqn:finite_latent_generator_bound} with $T=N$ and $G=\sigma_z/2$, whose spectral range is $g=1$, bounds the QFI of every allowed $N$-round strategy by $N^2$.
Consequently,
\begin{equation}
    \mathcal{J}_j(N;\theta')=N^2
\end{equation}
for every positive-probability branch $j\in\{0,x,y\}$ and every $\theta'\in\mathcal{U}$.
Thus one may take
\begin{equation}
    \kappa_j=1,
    \qquad
    N_0=1.
\end{equation}
\begin{atenv}
The remaining admissibility requirements can be
verified explicitly.
Conditioned on $j$, different channel uses are
independent, so fresh preparation after calibration
gives the required $N$-round continuation.

Let $|0_{\rm L}\>:=|00\>$ and
$|1_{\rm L}\>:=|11\>$, and prepare their equal
superposition.
After correct recovery, the final state is
\begin{equation}
    |\psi_N(\theta')\>
    =
    \frac{
        e^{-\iota N\theta'/2}|0_{\rm L}\>
        +e^{\iota N\theta'/2}|1_{\rm L}\>
    }{\sqrt{2}}.
\end{equation}
For a fixed point $\theta_0$, define
\begin{equation}
    |\pm_{\theta_0}^{(N)}\>
    :=
    \frac{
        |0_{\rm L}\>
        \pm\iota e^{\iota N\theta_0}|1_{\rm L}\>
    }{\sqrt{2}}.
\end{equation}
Then, consider the complete POVM
\begin{equation}
\begin{gathered}
    M_\pm
    =
    |\pm_{\theta_0}^{(N)}\>
    \<\pm_{\theta_0}^{(N)}|,\qquad
    M_\perp=I-M_+-M_-,
\end{gathered}
\end{equation}
and let $\hat\theta_\pm=\theta_0\pm\frac1N$ and $\hat\theta_\perp=\theta_0$.
On a correctly identified branch, we have
\begin{equation}
    P_\pm(\theta')
    =
    \frac{1\pm\sin[N(\theta'-\theta_0)]}{2},
    \qquad
    P_\perp(\theta')=0.
\end{equation}
Consequently,
\begin{equation}
    \mathbb{E}_{\theta'}[\hat\theta]
    =
    \theta_0+
    \frac{\sin[N(\theta'-\theta_0)]}{N},
\end{equation}
so this estimator is locally unbiased at $\theta_0$
and attains MSE $1/N^2$.

The corrected state family is smooth, has constant
rank, and realizes the optimal known-$j$ QFI $N^2$.
The same outcome assignments are used for every
reported label and every syndrome record, including
incorrect calibration, so all estimator values lie
in $\mathcal{I}=[\theta_0-1,\theta_0+1]
\subset\mathcal{U}$ for $N\geq1$.

The $j=0$ strategy can use the same probe--ancilla encoding and final POVM, with no intermediate recovery.
\atchg{These continuations and estimator assignments are defined for every reported label in $\{0,x,y\}$, including a label with zero prior probability.}
For the strongly correlated model, the same measurement construction applies with $|0_{\rm L}\>=|0\>$ and $|1_{\rm L}\>=|1\>$.
Finally, the operational encoding condition holds with $G=\sigma_z/2$ and parameter-independent noise.
Thus Definition~\ref{def:admissible_calibration} is satisfied with $N_1=1$ at $\theta=\theta_0$ when $p_z=0$.
\end{atenv}
All assumptions of Proposition~\ref{prop:finite_latent_hs} are explicitly satisfied when $p_z=0$.
Proposition~\ref{prop:finite_latent_hs} therefore gives
\begin{equation}
    \mse_\theta
    =
    O(T^{-2})
\end{equation}
for the weakly correlated transverse-noise model.

\section{Sequential quantum sensing}\label{app:sequential_sensing}

In this appendix, we formalize a general setting for sequential parameter estimation, first for multiple copies of quantum channels probed in sequence, then for a more general case for quantum processes that may exhibit correlations over time.
Section~\ref{sec:sequential_strategies_combs} gives a concise overview of the quantum-comb representation used in the main text, whereas this appendix provides the detailed construction underlying that overview.
In particular, we distinguish the operational sequential probe $\mathcal{M}$ from its positive semidefinite tester representation $\mathbf{M}=\sum_k\mathbf{M}_k$, derive the outcome probabilities, covariance matrix, and MSE in this representation, and separately discuss repeated copies of a channel, the fixed-initial-state setting, and general temporally correlated processes.
Thus, the two sections describe the same framework at different levels of detail.
We show that the precision for a given estimation strategy, quantified by the mean-squared error, can be expressed as a matrix operation between some matrix representation of the strategy and the probed process over multiple time steps.

A quantum channel parameter estimation problem is defined by real-valued parameters $\theta=(\theta_1,\dots,\theta_m)\in\R^m$ of a quantum channel $\Lambda_\theta:\L(A)\rightarrow\L(A)$ over a $D$-dimensional probe register $A$.
Given the access to $T$ copies of channel $\Lambda_\theta$, the most general strategy is to probe each copy of $\Lambda_\theta$ sequentially, while keeping a quantum memory register which assists in probing the channels in sequence and measuring at the end of the protocol.

A quantum sensing strategy for estimating parameter $\theta$ from $T$ copies of $\Lambda_\theta$ is defined by a \textit{sequential probe} $\mathcal{M}$ and an \textit{estimator} $\hat{\theta}$.
Sequential probe $\mathcal{M} = (\rho, \mathcal{M}^{(1)},\mathcal{M}^{(2)},\dots,\mathcal{M}^{(T-1)}, \mathcal{M}^{(T)})$ consists of:
\begin{enumerate}
    \item An initial probe state described by density operator $\rho\in\L(C_1\otimes A)$ for some ancilla register $C_1$.
    \item A sequence of $T-1$ quantum channels $\mathcal{M}^{(1)},\mathcal{M}^{(2)},\dots,\mathcal{M}^{(T-1)}$ probing $T$ copies of channel $\Lambda_\theta$ in sequence jointly with ancilla registers $C_1,\dots,C_T$, where $C_t$ (resp. $C_{t+1}$) is the input (resp. output) ancilla registers of $\mathcal{M}_t$.
    In general the dimension of ancilla registers $C_1,\dots,C_T$ may differ, i.e. where new ancilla registers are prepared during the protocol.
    \item Final measurement $\mathcal{M}^{(T)}$ with outcomes $k\in\{1,\dots,K\}$ measuring the output after probing the $T$-th copy of $\Lambda_\theta$ jointly with ancilla $A_T$.
\end{enumerate}
On the other hand, the estimator $\hat{\theta} = (\hat{\theta}_1,\dots,\hat{\theta}_m)$ is defined by
\begin{equation}
    \hat{\theta}_j = \sum_{k} \hat{\theta}_{j,k} p(k|{\theta,\mathcal{M}}),
\end{equation}
where $p(k|\theta,\mathcal{M})$ is the probability of measurement outcome $k$ given a sequential probe $\mathcal{M}$.
In this work, we assume that the estimator is unbiased (see e.g.~\cite{conlon2021efficient,Hayashi_2023}):
\begin{equation}\label{eqn:unbiased_estimator}
\begin{gathered}
    \hat{\theta}_j = \theta_j
    \quad\text{and}\quad
    \frac{\partial}{\partial\theta_i} \hat{\theta}_j = \delta_{i,j} \;.
\end{gathered}
\end{equation}

A sequential probe $\mathcal{M}$ with measurement outcome $k$ acting on $T$ copies of $\Lambda_\theta$ can be illustrated as:
\begin{equation}
    \begin{quantikz}[classical gap=0.03cm]
        \bshade{2}{10}{$\mathcal{M}$}
        \setwiretype{n} & \lstick[2]{$\rho$} & \wire[l][1]["C_1"{above,pos=0}]{q} \setwiretype{q} & \gate[2,style={fill=blue!20}]{\mathcal{M}^{(1)}} & 
            \wire[l][1]["C_2"{above,pos=0}]{q} \setwiretype{b} & \gate[2,style={fill=blue!20}]{\mathcal{M}^{(2)}} & 
            \wire[l][1]["C_3"{above,pos=0}]{q} & \invgate{\dots} & 
            \wire[l][1]["C_T"{above,pos=0}]{q} & \gate[2,style={fill=blue!20}]{\mathcal{M}_k^{(T)}} \\
        \setwiretype{n} & & \gate[1,style={fill=green!20}]{\Lambda_\theta} \wire[l][1]["{A_0}"{above,pos=0.45}]{q} \setwiretype{q} \gshade{1}{1}{} & \wire[l][1]["{A_1}"{above,pos=0.45}]{q} \wire[r][1]["{A_1'}"{above,pos=0.45}]{q} & 
            \gate[1,style={fill=green!20}]{\Lambda_\theta} \gshade{1}{1}{} & \wire[l][1]["{A_2}"{above,pos=0.45}]{q} \wire[r][1]["{A_2'}"{above,pos=0.45}]{q} & 
            \gate[1,style={fill=green!20}]{\Lambda_\theta} \gshade{1}{1}{} \wire[r][1]["{A_3}"{above,pos=0.45}]{q} & \invgate{\dots} & 
            \gate[1,style={fill=green!20}]{\Lambda_\theta} \wire[l][1]["{A_{T-1}'}"{above,pos=0.6}]{q} \wire[r][1]["{A_T}"{above,pos=0.45}]{q} \gshade{1}{1}{} &
    \end{quantikz} \;,
\end{equation}
where the $A_t,A_t'$'s are the probe register $A$ with time-step labels and primed registers $A_t'$ for $t\geq1$ are used to explicitly indicate the output probe register of $\mathcal{M}^{(t)}$.

The sequence of $T$ copies of $\Lambda_\theta$ can described as a positive semidefinite operator $\mathbf{\Lambda}_\theta^{(T)}$, we call the \textit{channel operator}, over the joint probe systems $\Bar{A}_T = A_t\otimes(\bigotimes_{t=1}^{T-1} A_t'\otimes A_t)\otimes A_0$ over the $T$ time steps, given by
\begin{equation}\label{eqn:channels_choi}
\begin{aligned}
    \mathbf{\Lambda}_\theta^{(T)} &= \mathbf{\Lambda}_\theta^{\otimes T} 
    =
    \begin{quantikz}[classical gap=0.03cm]
        \setwiretype{n} & \gate[1,style={fill=green!20}]{\Lambda_\theta} \wire[l][1]["{A_0}"{above,pos=0.5}]{q} \wire[r][1]["{A_1}"{above,pos=0.45}]{q} \gshade{1}{1}{} & \quad & \gate[1,style={fill=green!20}]{\Lambda_\theta} \wire[l][1]["{A_1'}"{above,pos=0.45}]{q} \wire[r][1]["{A_2}"{above,pos=0.45}]{q} \gshade{1}{1}{} & \;\dots\; & \gate[1,style={fill=green!20}]{\Lambda_\theta} \wire[l][1]["{A_{T-1}'}"{above,pos=0.65}]{q} \wire[r][1]["{A_T}"{above,pos=0.55}]{q} \gshade{1}{1}{} &
    \end{quantikz} ,
\end{aligned}
\end{equation}
where $\mathbf{\Lambda}_\theta$ is the Choi matrix of channel $\Lambda_\theta$.
On the other hand, the sequential probe $\mathcal{M}$ can also be described by a positive semidefinite operator $\mathbf{M}$ over $\Bar{A}_T$, we call the \textit{sequential probe operator}, which admits a linear decomposition into collection of positive semidefinite operators $\{\mathbf{M}_k\}_{k=1}^K$, each corresponding to measurement outcome $k$, as
\begin{equation}\label{eqn:strategy_choi}
\begin{aligned}
    \mathbf{M} &= \sum_k \mathbf{M}_k \;.
\end{aligned}
\end{equation}
Throughout this paper, we will often refer to $\mathbf{M}$ as the sequential probe.
Each operator $\mathbf{M}_k$ can be thought of as a ``higher-order'' POVM element over the $T$ copies of $\Lambda_\theta$ which determines how they are probed and finally measured at time step $T$.
Each of these operators can be described as
\begin{equation}\label{eqn:strategy_choi_PPOVM}
\begin{aligned}
    &\mathbf{M}_k = \mathbf{M}_k^{(T)}\ast\mathbf{M}^{(T-1)} \ast\dots\ast \mathbf{M}^{(1)} \ast\rho \\
    &=
    \begin{quantikz}[classical gap=0.03cm]
        \bshade{2}{8}{} & \setwiretype{n} \lstick[2]{$\rho$} & \wire[l][1]["C_1"{above,pos=0}]{q} \setwiretype{q} & \gate[2,style={fill=blue!20}]{\mathcal{M}^{(1)}} & \wire[l][1]["C_2"{above,pos=0}]{q} \setwiretype{b} & \invgate{\dots} & \wire[l][1]["C_T"{above,pos=0}]{q} & \gate[2,style={fill=blue!20}]{\mathcal{M}_k^{(T)}} \\
        & \setwiretype{n} & \invgate{} \wire[l][1]["{A_0}"{above,pos=0.45}]{q} \wshade{1}{1}{} & \wire[l][1]["{A_1}"{above,pos=0.45}]{q} \wire[r][1]["{A_1'}"{above,pos=0.45}]{q} & \invgate{} \wshade{1}{3}{} & \invgate{\dots} & \invgate{} \wire[r][1]["{A_T}"{above,pos=0.45}]{q} & 
    \end{quantikz} ,
\end{aligned}
\end{equation}
where $\mathbf{M}^{(t)}$ is the Choi matrix of channel $\mathcal{M}^{(t)}$ for $t\in\{1,\dots,T-1\}$ and $\{\mathbf{M}_k^{(T)}\}_k$ is the POVM of the final measurement.
Here ``$\ast$'' is the \textit{link product}~\cite{chiribella2009theoretical,taranto2025higherorderquantumoperations}~\footnote{
    In general for operators $X\in\L(A\otimes B)$ and $Y\in\L(B\otimes C)$, their link product is given by $X\ast Y = \tr_B(X^{\top_B}Y)$ where $\tr_B$ is partial trace over space $B$ and $(\cdot)^{\top_B}$ is the partial transpose over space $B$. Note that we can equivalently write $X\ast Y = \tr_B(XY^{\top_B})$.},
which contracts some output registers of one operation with some input register of another.
In eqn.~\eqref{eqn:strategy_choi_PPOVM}, the link product $\mathbf{M}^{(t-1)} \ast \mathbf{M}^{(t)}$ contracts the ancilla registers between channels $\mathcal{M}^{(t-1)}$ and $\mathcal{M}^{(t)}$ while leaving the registers $A_{t-1}'$ and $A_t$ open to be contracted later with the $t$-th copy of the channel $\Lambda_\theta$.
On the other hand, $\rho\ast\mathbf{M}^{(1)}$ and $\mathbf{M}^{(T-1)} \ast \mathbf{M}_k^{(T)}$ contract the ancilla register $C_1$ between state $\rho$ and channel $\mathcal{M}^{(1)}$ and ancilla $C_T$ between $\mathcal{M}^{(T-1)}$ and measurement $\mathcal{M}^{(T)}$, respectively.

We note that any positive semidefinite operator $\mathbf{M}$ over $\Bar{A}_T$ that satisfy a certain set of linear constraints admits a form of eqn.~\eqref{eqn:strategy_choi_PPOVM} and thus correspond to some sequential probe $\mathcal{M}$ over $T$ time steps.
We denote the set of all such operators over $\Bar{A}_T$ as $\seq(\Bar{A}_T)$.
For a more detailed discussion and definition of $\seq(\Bar{A}_T)$, see Appendix~\ref{app:multistep_processes}, particularly eqn.~\eqref{eqn:testers}.
Also, since the sequential probe operator $\mathbf{M}$ is positive semidefinite, it admits a unique spectral decomposition
\begin{equation}\label{eqn:strategy_spectral_decomposition}
    \mathbf{M} = \sum_j \gamma_j^\mathcal{M} \kketbra{M_j},
\end{equation}
where $|M_k\rr$ is its eigenvector corresponding to eigenvalue $\gamma_k^\mathcal{M}$.

The probability of observing outcome $k$ from probing $T$ copies of $\Lambda_\theta$ with sequential probe $\mathcal{M}$ is given by
\begin{equation}\label{eqn:output_proba_channel_sequence}
\begin{aligned}
    &p(k|\theta,\mathcal{M}) = \tr\big( (\mathbf{\Lambda}_\theta^{(T)})^\top \mathbf{M}_k \big)  = 
    \begin{quantikz}[classical gap=0.03cm]
        \setwiretype{q} \lstick[2]{$\rho$} & \wire[l][1]["C_1"{above,pos=0}]{q} & \gate[2,style={fill=blue!20}]{\mathcal{M}^{(1)}} & \wire[l][1]["C_2"{above,pos=0}]{q} \setwiretype{b} & \invgate{\dots} & \wire[l][1]["C_T"{above,pos=0}]{q} & \gate[2,style={fill=blue!20}]{\mathcal{M}_k^{(T)}} \\
        \setwiretype{n} & \gate[1,style={fill=green!20}]{\Lambda_\theta} \wire[l][1]["{A_0}"{above,pos=0.5}]{q} & \wire[l][1]["{A_1}"{above,pos=0.45}]{q} \wire[r][1]["{A_1'}"{above,pos=0.45}]{q} & \gate[1,style={fill=green!20}]{\Lambda_\theta} \wire[r][1]["{A_2}"{above,pos=0.5}]{q} & \invgate{\dots} & \gate[1,style={fill=green!20}]{\Lambda_\theta} \wire[l][1]["{A_{T-1}}"{above,pos=0.6}]{q} \wire[r][1]["{A_T}"{above,pos=0.45}]{q} & 
    \end{quantikz} .
\end{aligned}
\end{equation}
Thus, for a given strategy with sequential probe $\mathcal{M}$ and estimator $\hat{\theta}$, its precision in estimating parameters $\theta$ can be quantified by its covariance matrix $\cov_\theta(\mathcal{M},\hat{\theta})$.
Covariance matrix $\cov_\theta(\mathcal{M},\hat{\theta})$ is a size-$m$ symmetric real matrix given by
\begin{equation}\label{eqn:cov_matrix_channel_sequence}
\begin{aligned}
    \cov_\theta(\mathcal{M},\hat{\theta}) &= \sum_{i,j=1}^m |i\>\<j| \bigg( \sum_k \xi_{i,k}\xi_{j,k} p(k|\theta,\mathcal{M}) \bigg) \;,
\end{aligned}
\end{equation}
where $\xi_{i,k} = \hat{\theta}_{i,k} - \theta_i$ and $p(k|\theta,\mathcal{M})$ is given by eqn.~\eqref{eqn:output_proba_channel_sequence}.
In this work, we use the \textit{mean-squared error (MSE)} to quantify how well parameters $\theta$ are estimated:
\begin{equation}\label{eqn:mse_channel_sequence}
    \mse_\theta(\mathcal{M},\hat{\theta}) = \tr\big( \cov_\theta(\mathcal{M},\hat{\theta}) \big) = \sum_{i,k} \xi_{i,k}^2 p(k|\theta,\mathcal{M}) \;.
\end{equation}

    As an example, let us consider a $T=3$ single-qubit Pauli-$Z$ rotation channel estimation.
    The Choi matrix of a single-qubit $Z$-rotation channel $\mathcal{R}_{z,\theta}(\cdot) = e^{-\iota\frac{\theta}{2}Z}\cdot e^{\iota\frac{\theta}{2}Z}$ with angle $\theta$ is given by
    \begin{equation}
        \mathbf{R}_{z,\theta} = \ketbra{00}+\ketbra{11} + e^{-\iota\theta}|00\>\<11| + e^{\iota\theta}|11\>\<00| \;.
    \end{equation}
    Three copies of channel $\mathcal{R}_{z,\theta}(\cdot)$ are therefore described by channel operator $\mathbf{R}_\theta^{(3)}$ given by
    \begin{equation}
    \begin{aligned}
        \mathbf{R}_\theta^{(3)} &= \mathbf{R}_{z,\theta}^{\otimes3} \\
        &= 
        \begin{quantikz}[classical gap=0.03cm]
            \setwiretype{n} & \gate[1,style={fill=green!20}]{\mathcal{R}_{z,\theta}} \wire[l][1]["{A_0}"{above,pos=0.5}]{q} \wire[r][1]["{A_1}"{above,pos=0.45}]{q} & \quad & \gate[1,style={fill=green!20}]{\mathcal{R}_{z,\theta}} \wire[l][1]["{A_1'}"{above,pos=0.45}]{q} \wire[r][1]["{A_2}"{above,pos=0.45}]{q} & \quad & \gate[1,style={fill=green!20}]{\mathcal{R}_{z,\theta}} \wire[l][1]["{A_2'}"{above,pos=0.45}]{q} \wire[r][1]["{A_3}"{above,pos=0.45}]{q} &
        \end{quantikz}.
    \end{aligned}
    \end{equation}
    On the other hand, the sequential probe $\mathcal{M}$ described by operator $\mathbf{M}=\sum_k\mathbf{M}_k$ with operators $\mathbf{M}_k$ given in eqn.~\eqref{eqn:strategy_choi_PPOVM} is illustrated as
    \begin{equation}
        \mathbf{M}_k =
        \begin{quantikz}[classical gap=0.03cm]
            \setwiretype{q} \lstick[2]{$\rho$} & \wire[l][1]["C_1"{above,pos=0}]{q} & \gate[2,style={fill=blue!20}]{\mathcal{M}^{(1)}} & \wire[l][1]["C_2"{above,pos=0}]{q} \setwiretype{b} & \gate[2,style={fill=blue!20}]{\mathcal{M}^{(2)}}& \wire[l][1]["C_3"{above,pos=0}]{q} & \gate[2,style={fill=blue!20}]{\mathcal{M}_k^{(3)}} \\
            \setwiretype{n} & \invgate{} \wire[l][1]["{A_0}"{above,pos=0.45}]{q} & \wire[l][1]["{A_1}"{above,pos=0.45}]{q} \wire[r][1]["{A_1'}"{above,pos=0.45}]{q} & \invgate{} & \wire[l][1]["{A_2}"{above,pos=0.45}]{q} \wire[r][1]["{A_2'}"{above,pos=0.45}]{q} & \invgate{} \wire[r][1]["{A_3}"{above,pos=0.45}]{q} & 
        \end{quantikz} .
    \end{equation}
    The action of sequential probe $\mathcal{M}$ on the three copies of $\Lambda_\theta$ with measurement outcome $k$ is therefore illustrated as
    \begin{equation}
        \begin{quantikz}[classical gap=0.03cm]
            \setwiretype{q} \lstick[2]{$\rho$} & \wire[l][1]["C_1"{above,pos=0}]{q} & \gate[2,style={fill=blue!20}]{\mathcal{M}^{(1)}} & \wire[l][1]["C_2"{above,pos=0}]{q} \setwiretype{b} & \gate[2,style={fill=blue!20}]{\mathcal{M}^{(2)}}& \wire[l][1]["C_3"{above,pos=0}]{q} & \gate[2,style={fill=blue!20}]{\mathcal{M}_k^{(3)}} \\
            \setwiretype{q} & \gate[1,style={fill=green!20}]{\mathcal{R}_{z,\theta}} \wire[l][1]["{A_0}"{above,pos=0.45}]{q} & \wire[l][1]["{A_1}"{above,pos=0.45}]{q} \wire[r][1]["{A_1'}"{above,pos=0.45}]{q} & \gate[1,style={fill=green!20}]{\mathcal{R}_{z,\theta}} & \wire[l][1]["{A_2}"{above,pos=0.45}]{q} \wire[r][1]["{A_2'}"{above,pos=0.45}]{q} & \gate[1,style={fill=green!20}]{\mathcal{R}_{z,\theta}} \wire[r][1]["{A_3}"{above,pos=0.45}]{q} & 
        \end{quantikz} .
    \end{equation}
We will come back to this example in Section~\ref{sec:z_rotation_example} to show an optimal strategy.

\subsection{Fixed initial probe state}

If we fix an initial probe state $\rho\in\D(C_1\otimes A_0)$, then we can describe the sequence of $T$ copies of $\Lambda_\theta$ along with $\rho$ as a positive semidefinite operator $\mathbf{\Lambda}_{\theta,\rho}^{(T)}$ given by
\begin{equation}\label{eqn:fixed_probe_choi}
\begin{aligned}
    &\mathbf{\Lambda}_{\theta,\rho}^{(T)} = \mathbf{\Lambda}_\theta^{\otimes T-1} \otimes (\mathcal{I}_{C_1}\otimes \Lambda_\theta)(\rho) 
    =
    \begin{quantikz}[classical gap=0.03cm]
        \gshade{2}{2}{} & \lstick[2]{$\rho$} \setwiretype{n} & \wire[l][1]["{C_1}"{above,pos=0}]{q} &\setwiretype{n}&&&& \\
        & \setwiretype{n} \gshade{1}{2}{} & \gate[1,style={fill=green!20}]{\Lambda_\theta} \wire[l][1]["{A_0}"{above,pos=0.5}]{q} \wire[r][1]["{A_1}"{above,pos=0.45}]{q} & \quad & \gate[1,style={fill=green!20}]{\Lambda_\theta} \wire[l][1]["{A_1'}"{above,pos=0.45}]{q} \wire[r][1]["{A_2}"{above,pos=0.45}]{q} \gshade{1}{1}{} & \;\dots\; & \gate[1,style={fill=green!20}]{\Lambda_\theta} \wire[l][1]["{A_{T-1}'}"{above,pos=0.65}]{q} \wire[r][1]["{A_T}"{above,pos=0.55}]{q} \gshade{1}{1}{} &
    \end{quantikz} .
\end{aligned}
\end{equation}
On the other hand, the sequential probe $\mathcal{M} = (\mathcal{M}^{(1)},\dots,\mathcal{M}^{(T-1},\mathcal{M}_k^{(T)})$ corresponding to outcomes $k$ can be described by a positive semidefinite operator
\begin{equation}\label{eqn:fixed_state_measurement_choi}
\begin{aligned}
    \mathbf{M}_k &= \mathbf{M}_k^{(T)}\ast\mathbf{M}^{(T-1)} \ast\dots\ast \mathbf{M}^{(1)}  
    =
    \begin{quantikz}[classical gap=0.03cm]
        \setwiretype{q} & \gate[2,style={fill=blue!20}]{\mathcal{M}^{(1)}} \wire[l][1]["C_1"{above,pos=0.45}]{q} \bshade{2}{5}{} & \wire[l][1]["C_2"{above,pos=0}]{q} \setwiretype{b} & \invgate{\dots} & \wire[l][1]["C_T"{above,pos=0}]{q} & \gate[2,style={fill=blue!20}]{\mathcal{M}_k^{(T)}} \\
        \setwiretype{n} & \wire[l][1]["{A_1}"{above,pos=0.45}]{q} \wire[r][1]["{A_1'}"{above,pos=0.45}]{q} & \invgate{} \wshade{1}{3}{} & \invgate{\dots} & \invgate{} \wire[r][1]["{A_T}"{above,pos=0.45}]{q} & 
    \end{quantikz} .
\end{aligned}
\end{equation}

\subsection{General parameterized multi-time processes}

Instead of $T$ copies of quantum channels, we may also consider a general multi-time quantum process which is being probed at $T$ time points, e.g. where there is an inaccessible environment $E$ that is correlated with the probe system $A$ over the $T$ time points.
In this case, we can describe this entire process parameterized by $\theta=(\theta_1,\dots,\theta_m)$ as a sequence of quantum channels $\Lambda_\theta^{(t)}$ at time step $t$ acting on probe register $A$ and some inaccessible environment $E$ as a positive semidefinite operator
\begin{equation}
\begin{aligned}
    &\mathbf{\Lambda}_\theta = \mathbf{\Lambda}_\theta^{(T)} \ast\dots\ast \mathbf{\Lambda}_\theta^{(1)}
    = 
    \begin{quantikz}[classical gap=0.03cm]
        \setwiretype{n} & \gate[2,style={fill=green!20}]{\Lambda_\theta^{(1)}} \wire[l][1]["{A_0}"{above,pos=0.45}]{q} \setwiretype{q} \gshade{1}{1}{} & \quad \wire[l][1]["{A_1}"{above,pos=0.45}]{q} \wire[r][1]["{A_1'}"{above,pos=0.45}]{q} & 
            \gate[2,style={fill=green!20}]{\Lambda_\theta^{(2)}} \gshade{1}{1}{} & 
            \invgate{\dots} \wire[l][1]["{A_2}"{above,pos=0.45}]{q} & 
            \gate[2,style={fill=green!20}]{\Lambda_\theta^{(T)}} \wire[l][1]["{A_{T-1}'}"{above,pos=0.65}]{q} \wire[r][1]["{A_T}"{above,pos=0.45}]{q} \gshade{1}{1}{} & \\
        \setwiretype{n} & \gshade{1}{5}{$\mathbf{\Lambda}_\theta$} & \wire[l][1]["{E_1}"{above,pos=0.45}]{q} \setwiretype{q} & 
            & 
            \invgate{\dots} \wire[l][1]["{E_2}"{above,pos=0.45}]{q} & 
            \wire[l][1]["{E_{T-1}}"{above,pos=0.65}]{q} & \setwiretype{n}
    \end{quantikz} ,
\end{aligned}
\end{equation}
where $\mathbf{\Lambda}_\theta^{(t)}$ is the Choi matrix of channel $\Lambda_\theta^{(t)}$.
Thus, the action of a sequential probe $\mathcal{M}$ on parameterized process $\mathbf{\Lambda}_\theta$ over $T$ time steps can be illustrated as
\begin{equation}
    \begin{quantikz}[classical gap=0.03cm]
        \bshade{2}{9}{$\mathcal{M}$}
        \setwiretype{n} & \lstick[2]{$\rho$} & \wire[l][1]["C_1"{above,pos=0}]{q} \setwiretype{q} & \gate[2,style={fill=blue!20}]{\mathcal{M}^{(1)}} & 
            \wire[l][1]["C_2"{above,pos=0}]{q} \setwiretype{b} & \gate[2,style={fill=blue!20}]{\mathcal{M}^{(2)}} & 
            \invgate{\dots} \wire[l][1]["C_3"{above,pos=0.4}]{q} & 
            \wire[l][1]["C_T"{above,pos=0}]{q} & \gate[2,style={fill=blue!20}]{\mathcal{M}_k^{(T)}} \\
        \setwiretype{n} & & \gate[2,style={fill=green!20}]{\Lambda_\theta^{(1)}} \wire[l][1]["{A_0}"{above,pos=0.45}]{q} \setwiretype{q} \gshade{1}{1}{} & \wire[l][1]["{A_1}"{above,pos=0.45}]{q} \wire[r][1]["{A_1'}"{above,pos=0.45}]{q} & 
            \gate[2,style={fill=green!20}]{\Lambda_\theta^{(2)}} \gshade{1}{1}{} & \wire[l][1]["{A_2}"{above,pos=0.45}]{q} \wire[r][1]["{A_2'}"{above,pos=0.45}]{q} & 
            \invgate{\dots} & 
            \gate[2,style={fill=green!20}]{\Lambda_\theta^{(T)}} \wire[l][1]["{A_{T-1}'}"{above,pos=0.65}]{q} \wire[r][1]["{A_T}"{above,pos=0.45}]{q} \gshade{1}{1}{} & \\
        \setwiretype{n} & & \gshade{1}{6}{$\mathbf{\Lambda}_\theta^{(T)}$} & \wire[l][1]["{E_1}"{above,pos=0.45}]{q} \setwiretype{q} & 
            & \wire[l][1]["{E_2}"{above,pos=0.45}]{q} &
            \invgate{\dots} & 
            \wire[l][1]["{E_{T-1}}"{above,pos=0.65}]{q} & \setwiretype{n}
    \end{quantikz} .
\end{equation}

We note that, in general, each channel $\Lambda_\theta^{(t)}$ may be sensitive only to some of the parameters $\theta_1,\dots,\theta_m$.
For example, for $T=2$ and $m=2$, one may have channels $\Lambda_\theta^{(1)} = \Tilde{\Lambda}_{\theta_1}^{(1)}$ and $\Lambda_\theta^{(2)} = \Tilde{\Lambda}_{\theta_1,\theta_2}^{(2)}$ for some channels $\Tilde{\Lambda}_{\theta_1}^{(1)}$ and $\Tilde{\Lambda}_{\theta_1,\theta_2}^{(2)}$.
Also, environments $E_1,\dots,E_{T-1}$ may be of different dimension.
This is the most general physical scenario of estimating parameters of quantum processes.

\section{Multi-step quantum processes}\label{app:multistep_processes}

In this appendix, we briefly review the \textit{quantum combs} or the \textit{higher-order quantum operation} formalism (see e.g.~\cite{chiribella2009theoretical,taranto2025higherorderquantumoperations}), which we use to represent both the multi-step parameterized quantum processes $\mathbf{\Lambda}_\theta$ (which include the multiple copies of parameterized quantum channels as a special case) as well as the sequential probe $\mathcal{M}$ used to estimate $\theta$.
Recall that (cf. appendix~\ref{app:sequential_sensing}) a multi-step parameterized quantum process $\mathbf{\Lambda}_\theta$ over $T$ steps can be expressed as a contraction of quantum channels represented by a positive semidefinite operator
\begin{equation}
\begin{aligned}
    &\mathbf{\Lambda}_\theta = \mathbf{\Lambda}_\theta^{(T)} \ast\dots\ast \mathbf{\Lambda}_\theta^{(1)} 
    = 
    \begin{quantikz}[classical gap=0.03cm]
        \setwiretype{n} & \gate[2,style={fill=green!20}]{\Lambda_\theta^{(1)}} \wire[l][1]["{A_0}"{above,pos=0.45}]{q} \setwiretype{q} \gshade{1}{1}{} & \quad \wire[l][1]["{A_1}"{above,pos=0.45}]{q} \wire[r][1]["{A_1'}"{above,pos=0.45}]{q} & 
            \gate[2,style={fill=green!20}]{\Lambda_\theta^{(2)}} \gshade{1}{1}{} & 
            \invgate{\dots} \wire[l][1]["{A_2}"{above,pos=0.45}]{q} & 
            \gate[2,style={fill=green!20}]{\Lambda_\theta^{(T)}} \wire[l][1]["{A_{T-1}'}"{above,pos=0.65}]{q} \wire[r][1]["{A_T}"{above,pos=0.45}]{q} \gshade{1}{1}{} & \\
        \setwiretype{n} & \gshade{1}{5}{} & \wire[l][1]["{E_1}"{above,pos=0.45}]{q} \setwiretype{q} & 
            & 
            \invgate{\dots} \wire[l][1]["{E_2}"{above,pos=0.45}]{q} & 
            \wire[l][1]["{E_{T-1}}"{above,pos=0.65}]{q} & \setwiretype{n}
    \end{quantikz} ,
\end{aligned}
\end{equation}
where $\mathbf{\Lambda}_\theta^{(t)}$ is the Choi matrix of channel $\Lambda_\theta^{(t)}$.
By picking a basis $\{|i\>\}_{i=1}^{d_{A_{t-1}'}}$ over probe register $A_{t-1}'$ and basis $\{|k\>\}_{k=1}^{d_{E_{t-1}}}$ over environment $E_{t-1}$ and denoting $|i,k\> = |i\>_{A_{t-1}} \otimes |k\>_{E_{t-1}}$, the Choi matrix $\mathbf{\Lambda}_\theta^{(t)}$ is given by
\begin{equation}
    \mathbf{\Lambda}_\theta^{(t)} = \sum_{i,j=1}^{d_{A_{t-1}'}} \sum_{k,l=1}^{d_{E_{t-1}}} \Lambda_\theta^{(t)}(|i,k\>\<j,l|) \otimes|i,k\>\<j,l| \;.
\end{equation}
The link product~\cite{chiribella2009theoretical,taranto2025higherorderquantumoperations} ``$\ast$'' between two operators $X\in\L(A\otimes B)$ and $Y\in\L(B\otimes C)$ is given by $X\ast Y = \tr_B(X^{\top_B}Y)$, where $\tr_B$ is partial trace over space $B$ and $(\cdot)^{\top_B}$ is the partial transpose over space $B$ (we can equivalently write $X\ast Y = \tr_B(XY^{\top_B})$).
So, the link product between $\mathbf{\Lambda}_\theta^{(t+1)}$ and $\mathbf{\Lambda}_\theta^{(t)}$ is given by
\begin{equation}
    \mathbf{\Lambda}_\theta^{(t+1)} \ast \mathbf{\Lambda}_\theta^{(t)} = \tr_{E_t}\Big( (\mathbf{\Lambda}_\theta^{(t+1)})^{\top_{E_t}} \mathbf{\Lambda}_\theta^{(t)} \Big) \;.
\end{equation}
Note that $\mathbf{\Lambda}_\theta^{(t+1)} \ast \mathbf{\Lambda}_\theta^{(t)}$ is an operator over $A_{t+1}\otimes A_t'\otimes A_t\otimes A_{t-1}'\otimes E_{t+1}\otimes E_{t-1}$.
By iteratively applying the link product in $\mathbf{\Lambda}_\theta^{(T)} \ast\dots\ast \mathbf{\Lambda}_\theta^{(1)}$ we can therefore obtain an explicit matrix form of operator $\mathbf{\Lambda}_\theta$ over $\Bar{A}_T$.

On the other hand, a sequential probe $\mathcal{M}$ over $T$ time steps can be described by a positive semidefinite operator
\begin{equation}\label{eqn:sequential_probe_space}
\begin{aligned}
    \mathbf{M} = \sum_{k=1}^K \mathbf{M}_k = \sum_{k=1}^K \mathbf{M}_k^{(T)}\ast\mathbf{M}^{(T-1)} \ast\dots\ast \mathbf{M}^{(1)} \ast\rho 
    = \sum_{k=1}^K
    \begin{quantikz}[classical gap=0.03cm]
        \bshade{2}{8}{} & \setwiretype{n} \lstick[2]{$\rho$} & \wire[l][1]["C_1"{above,pos=0}]{q} \setwiretype{q} & \gate[2,style={fill=blue!20}]{\mathcal{M}^{(1)}} & \wire[l][1]["C_2"{above,pos=0}]{q} \setwiretype{b} & \invgate{\dots} & \wire[l][1]["C_T"{above,pos=0}]{q} & \gate[2,style={fill=blue!20}]{\mathcal{M}_k^{(T)}} \\
        & \setwiretype{n} & \invgate{} \wire[l][1]["{A_0}"{above,pos=0.45}]{q} \wshade{1}{1}{} & \wire[l][1]["{A_1}"{above,pos=0.45}]{q} \wire[r][1]["{A_1'}"{above,pos=0.45}]{q} & \invgate{} \wshade{1}{3}{} & \invgate{\dots} & \invgate{} \wire[r][1]["{A_T}"{above,pos=0.45}]{q} & 
    \end{quantikz} 
\end{aligned},
\end{equation}
where $\mathbf{M}^{(t)}$ for $t\leq T-1$ is the Choi matrix of quantum channel $\mathcal{M}^{(t)}$, which is defined by
\begin{equation}
    \mathbf{M}^{(t)} = \sum_{i,j=1}^{d_{A_t}} \sum_{k,l=1}^{d_{C_t}} \mathcal{M}^{(t)}(|i,k\>\<j,l|) \otimes|i,k\>\<j,l| \;
\end{equation}
for a basis $\{|i\>\}_{i=1}^{d_{A_t}}$ over probe register $A_t$ and a basis $\{|k\>\}_{k=1}^{d_{C_t}}$ over ancilla register $C_t$.
The collection of operators $\{\mathbf{M}_k^{(T)}\}_{k=1}^K$ is a POVM over $A_T\otimes C_T$.

Objects described by any such $\mathbf{M}$ are known as \textit{quantum testers}~\cite[Section IV.E]{chiribella2009theoretical}\cite[Section 2.3.3]{taranto2025higherorderquantumoperations} (or simply testers).
We denote the set of all testers over $T$ steps as
\begin{equation}\label{eqn:testers}
\begin{aligned}
    \seq(\Bar{A}_T) = \Big\{ \mathbf{M}\in\L(\Bar{A}_T) &: \mathbf{M}\geq0 \;, \\
    & \mathbf{M}^{(T)} := \mathbf{M} = I_{A_T}\otimes\mathbf{M}^{(T-1)} \;,\quad
        \mathbf{M}^{(T-1)} = \tr_{A_T}(\mathbf{M}^{(T)})/d_{A_T} \geq 0 \;, \\
    & \forall t\in[1,T-1] \;,\; \mathbf{M}^{(t-1)} = \tr_{A_t'A_t}(\mathbf{M}^{(t)})/d_A \geq0 \;,\quad
        \tr_{A_t'}(\mathbf{M}^{(t)}) = I_{A_t}\otimes \mathbf{M}^{(t-1)} \;,\\
    & \mathbf{M}^{(0)} = \rho^{(0)} \geq0 \;,\quad \tr(\rho^{(0)})=1
    \Big\} \;.
\end{aligned}
\end{equation}
It is known (see \cite[Section IV.E]{chiribella2009theoretical},\cite[Section 2.4.2]{taranto2025higherorderquantumoperations}) that any $\mathbf{M}\in\seq(\Bar{A}_T)$ admits the form of
\begin{equation}
\begin{aligned}
    \mathbf{M} 
    = \sum_k \mathbf{M}_k^{(T)}\ast\mathbf{V}^{(T-1)} \ast\dots\ast \mathbf{V}^{(1)} \ast\ketbra{\psi} 
    = \sum_k
    \begin{quantikz}[classical gap=0.03cm]
        \;\; \bshade{2}{9}{} & \setwiretype{n} \lstick[2]{$|\psi\>$} & \wire[l][1]["C_1"{above,pos=0}]{q} \setwiretype{q} & \gate[2,style={fill=blue!20}]{V^{(1)}} & \wire[l][1]["C_2"{above,pos=0}]{q} \setwiretype{b} & \invgate{\dots} & \gate[2,style={fill=blue!20}]{V^{(T-1)}} & \wire[l][1]["C_T"{above,pos=0}]{q} & \gate[2,style={fill=blue!20}]{M_k'} \\
        & \setwiretype{n} & \invgate{} \wire[l][1]["{A_0}"{above,pos=0.45}]{q} \wshade{1}{1}{} & \wire[l][1]["{A_1}"{above,pos=0.45}]{q} \wire[r][1]["{A_1'}"{above,pos=0.45}]{q} & \invgate{} \wshade{1}{2}{} & \invgate{\dots} & \wire[l][1]["{A_{T-1}}"{above,pos=0.65}]{q} \wire[r][1]["{A_{T-1}'}"{above,pos=0.65}]{q} & \invgate{} \wire[r][1]["{A_T}"{above,pos=0.45}]{q} \wshade{1}{1}{} & 
    \end{quantikz} 
\end{aligned}
\end{equation}
for some ancilla systems $C_1,\dots,C_T$ such that $C_t\cong (\bigotimes_{l=1}^{t-1} A_l'\otimes A_l) \otimes A_0$, some pure state $|\psi\>$ over $A_0\otimes C_1$, isometries $V^{(1)},\dots,V^{(T-1)}$, and POVM $\{M_k'\}_k$ over $A_T\otimes C_T \cong \Bar{A}_T$.
We show in the Appendix~\ref{app:tester_decomposition} how this form can be obtained for any tester $\mathbf{M}$.

\section{Tester decomposition}\label{app:tester_decomposition}

Now we will show a decomposition of a $T$-step tester $\mathbf{M}\in\seq(\Bar{A}_T)$ into an initial state followed by a sequence of $T-1$ isometries and a final POVM measurement.
This derivation follows the proof of Theorem 12 and Theorem 3 of~\cite{chiribella2009theoretical}.
Suppose that there exist positive semidefinite operators $\{\mathbf{M}_k\}$ such that we can write $\mathcal{M}$ as
\begin{equation}
    \mathbf{M} = \sum_k \mathbf{M}_k \;,
\end{equation}
where each $\mathbf{M}_k$ corresponds to measurement outcome $k$:
\begin{equation}
    \mathbf{M}_k =
    \begin{quantikz}[classical gap=0.03cm]
        \bshade{2}{9}{} & \setwiretype{n} \lstick[2]{$\rho$} & \wire[l][1]["C_1"{above,pos=0}]{q} \setwiretype{q} & \gate[2,style={fill=blue!20}]{\mathcal{M}^{(1)}} & \wire[l][1]["C_2"{above,pos=0}]{q} \setwiretype{b} & \invgate{\dots} & \gate[2,style={fill=blue!20}]{\mathcal{M}^{(T-1)}} & \wire[l][1]["C_T"{above,pos=0}]{q} & \gate[2,style={fill=blue!20}]{\mathcal{M}_k^{(T)}} \\
        & \setwiretype{n} & \invgate{} \wire[l][1]["{A_0}"{above,pos=0.45}]{q} \wshade{1}{1}{} & \wire[l][1]["{A_1}"{above,pos=0.45}]{q} \wire[r][1]["{A_1'}"{above,pos=0.45}]{q} & \invgate{} \wshade{1}{2}{} & \invgate{\dots} & \wire[l][1]["{A_{T-1}}"{above,pos=0.65}]{q} \wire[r][1]["{A_{T-1}'}"{above,pos=0.65}]{q} & \invgate{} \wire[r][1]["{A_T}"{above,pos=0.45}]{q} \wshade{1}{1}{} & 
    \end{quantikz} \;.
\end{equation}
So, for a $T$-step channel operator $\mathbf{\Lambda}_\theta$, the probability of output $k$ from sequential strategy $\mathcal{M}$ with operator $\mathbf{M}=\sum_k\mathbf{M}_k$ is given by
\begin{equation}
    p(k|\mathcal{M},\mathbf{\Lambda}_\theta) = \tr(\mathbf{M}_k \mathbf{\Lambda}_\theta^\top) =
    \begin{quantikz}[classical gap=0.03cm]
        \bshade{2}{9}{$\mathbf{M}_k$}
        \setwiretype{n} & \lstick[2]{$\rho$} & \wire[l][1]["C_1"{above,pos=0}]{q} \setwiretype{q} & \gate[2,style={fill=blue!20}]{\mathcal{M}^{(1)}} & 
            \wire[l][1]["C_2"{above,pos=0}]{q} \setwiretype{b} & \gate[2,style={fill=blue!20}]{\mathcal{M}^{(2)}} & 
            \invgate{\dots} \wire[l][1]["C_3"{above,pos=0.4}]{q} & 
            \wire[l][1]["C_T"{above,pos=0}]{q} & \gate[2,style={fill=blue!20}]{\mathcal{M}_k^{(T)}} \\
        \setwiretype{n} & & \gate[2,style={fill=green!20}]{\Lambda_\theta^{(1)}} \wire[l][1]["{A_0}"{above,pos=0.45}]{q} \setwiretype{q} \gshade{1}{1}{} & \wire[l][1]["{A_1}"{above,pos=0.45}]{q} \wire[r][1]["{A_1'}"{above,pos=0.45}]{q} & 
            \gate[2,style={fill=green!20}]{\Lambda_\theta^{(2)}} \gshade{1}{1}{} & \wire[l][1]["{A_2}"{above,pos=0.45}]{q} \wire[r][1]["{A_2'}"{above,pos=0.45}]{q} & 
            \invgate{\dots} & 
            \gate[2,style={fill=green!20}]{\Lambda_\theta^{(T)}} \wire[l][1]["{A_{T-1}'}"{above,pos=0.65}]{q} \wire[r][1]["{A_T}"{above,pos=0.45}]{q} \gshade{1}{1}{} & \\
        \setwiretype{n} & & \gshade{1}{6}{$\mathbf{\Lambda}_\theta$} & \wire[l][1]["{E_1}"{above,pos=0.45}]{q} \setwiretype{q} & 
            & \wire[l][1]["{E_2}"{above,pos=0.45}]{q} &
            \invgate{\dots} & 
            \wire[l][1]["{E_{T-1}}"{above,pos=0.65}]{q} & \setwiretype{n}
    \end{quantikz}
    \;. 
\end{equation}
We will first show in Appendix~\ref{app:tester_isometry_sequence_measurement_decomposition} that one can decompose each $\mathbf{M}_k$ into $\mathbf{M}_k = \sqrt{\mathbf{M}}^\dag M_k' \sqrt{\mathbf{M}}$.
Here, operators $\{M_k'\}_k$ correspond to a final measurement POVM at time step $T$ over the joint ancilla--probe registers $C_T\otimes A_T\cong \Bar{A}_T$ and $\sqrt{\mathbf{M}}\cdot\sqrt{\mathbf{M}}^\dag$ is a linear map $\L(\Bar{A}_T)\rightarrow\L(C_T\otimes A_T)$ describing the intermediate operations from time step $1$ to $T-1$:
\begin{equation}
    M_k' =
    \begin{quantikz}[classical gap=0.03cm]
        & \gate[2,style={fill=blue!20}]{\mathcal{M}_k^{(T)}} \wire[l][1]["C_T"{above,pos=0.5}]{q} \\
        &  \wire[l][1]["{A_T}"{above,pos=0.5}]{q}
    \end{quantikz}
    \quad\text{and}\quad
    \sqrt{\mathbf{M}}\cdot\sqrt{\mathbf{M}}^\dag =
    \begin{quantikz}[classical gap=0.03cm]
        & \setwiretype{n} \lstick[2]{$\rho$} & \wire[l][1]["C_1"{above,pos=0}]{q} \setwiretype{q} & \gate[2,style={fill=blue!20}]{\mathcal{M}^{(1)}} & \wire[l][1]["C_2"{above,pos=0}]{q} \setwiretype{b} & \invgate{\dots} & \gate[2,style={fill=blue!20}]{\mathcal{M}^{(T-1)}} & \wire[l][1]["C_T"{above,pos=0}]{q} \\
        & \setwiretype{n} & \invgate{} \wire[l][1]["{A_0}"{above,pos=0.45}]{q} \wshade{1}{1}{} & \wire[l][1]["{A_1}"{above,pos=0.45}]{q} \wire[r][1]["{A_1'}"{above,pos=0.45}]{q} & \invgate{} \wshade{1}{2}{} & \invgate{\dots} & \wire[l][1]["{A_{T-1}}"{above,pos=0.65}]{q} \wire[r][1]["{A_{T-1}'}"{above,pos=0.65}]{q} & 
    \end{quantikz} \;.
\end{equation}
Then, in Appendix~\ref{app:tester_isometry_sequence_decomposition}, we will show how we can obtain a form of each $\mathcal{M}^{(t)}$ as an isometry channel $\mathcal{M}^{(t)}(\cdot) = V^{(t)} \cdot (V^{(t)})^\dag$, i.e. how the sequence of intermediate operations described by $\sqrt{\mathbf{M}}\cdot\sqrt{\mathbf{M}}^\dag$ can be decomposed into a sequence of isometries.

\subsection{Decomposing tester into a final measurement and intermediate operations}\label{app:tester_isometry_sequence_measurement_decomposition}

Note that by spectral decomposition of $\mathbf{M}$, we can also write $\mathbf{M} = \sum_j \alpha_j \kketbra{M_j}$, where $|M_j\rr$ is an eigenvector with eigenvalue $\alpha_j\geq0$.
Now consider the projector $\Pi_\mathbf{M} = \sum_{j:\alpha_j>0} \kketbra{M_j}$ onto the support of $\mathbf{M}$ and an operator $M_k'$ over $C_T\otimes A_T$ for $C_T\cong (\bigotimes_{t=1}^{T-1} A_t'\otimes A_t) \otimes A_0$ given by
\begin{equation}
    M_k' = \sqrt{\mathbf{M}^+} \mathbf{M}_k \sqrt{\mathbf{M}^+} + G_k,
\end{equation}
where $\{G_k\}_k$ is a set of positive semidefinite operators satisfying $\sum_k G_k = I - \Pi_\mathbf{M}$ and $K^+$ denotes the pseudo-inverse of $K$, which is defined as an operator $K^+$ satisfying
\begin{equation}
\begin{gathered}
    KK^+K = K \;,\quad
    K^+KK^+ = K^+ \;\quad
    (KK^+)^\dag = KK^+ \;,\quad
    (K^+K)^\dag = K^+K \;.
\end{gathered}
\end{equation}
Since $\mathbf{M}$ is positive semidefinite, its pseudoinverse is given by
\begin{equation}
    \mathbf{M}^+ = \sum_{j:\alpha_j>0} \alpha_j^{-1}\kketbra{M_j} \;,
\end{equation}
which implies that $\sqrt{\mathbf{M}^+}\mathbf{M}\sqrt{\mathbf{M}^+} = \Pi_\mathbf{M}$.
Thus, $\{M_k'\}_k$ is a POVM since each $M_k'$ is positive semidefinite and
\begin{equation}
    \sum_k M_k' = \sqrt{\mathbf{M}^+} \bigg( \underbrace{\sum_k \mathbf{M}_k}_{=\mathbf{M}} \bigg) \sqrt{\mathbf{M}^+} + \sum_k G_k = \Pi_\mathbf{M} + (I-\Pi_\mathbf{M}) = I \;.
\end{equation}

Now consider the operator $\sqrt{\mathbf{M}} = \sum_j \sqrt{\alpha_j}\kketbra{M_j} :\Bar{A}_T\rightarrow A_T\otimes C_T$ (recall that $A_T\otimes C_T\cong\Bar{A}_T$).
Then note that each $G_k$ satisfies $0\leq G_k\leq I-\Pi_\mathbf{M}$ so that, for any $|v\>$ in the support of $\Pi_\mathbf{M}$, we have
\begin{equation}
    0 \leq G_k|v\> \leq |v\> - \Pi_\mathbf{M}|v\> = 0,
\end{equation}
which means that $G_k\Pi_\mathbf{M} = \Pi_\mathbf{M}G_k = 0$.
Then this implies that $\sqrt{\mathbf{M}}^\dag G_k \sqrt{\mathbf{M}} = 0$.
Similarly since $0\leq \mathbf{M}_k\leq \mathbf{M}$, the support of $\mathbf{M}_k$ is contained in the support of $\mathbf{M}$, so that we have $\mathbf{M}_k\Pi_\mathbf{M} = \Pi_\mathbf{M}\mathbf{M}_k$.
Since $\sqrt{\mathbf{M}^+}\sqrt{\mathbf{M}} = \Pi_\mathbf{M}$, we have
\begin{equation}
\begin{aligned}
    \sqrt{\mathbf{M}}^\dag M_k' \sqrt{\mathbf{M}} &= \sqrt{\mathbf{M}}^\dag \sqrt{\mathbf{M}^+} \mathbf{M}_k \sqrt{\mathbf{M}^+} \sqrt{\mathbf{M}} + \sqrt{\mathbf{M}}^\dag G_k \sqrt{\mathbf{M}} 
    = \mathbf{M}_k \;.
\end{aligned}
\end{equation}
Therefore, $p(k|\mathcal{M},\Lambda_\theta) = \tr(\sqrt{\mathbf{M}}\mathbf{\Lambda}_\theta^\top \sqrt{\mathbf{M}}^\dag M_k')$, where $\sqrt{\mathbf{M}} \cdot \sqrt{\mathbf{M}}^\dag$ is a linear map of the form $\L(\Bar{A}_T)\rightarrow\L(A_T\otimes C_T)$ and $\{M_k'\}_k$ is a POVM over $A_T\otimes C_T$.

\subsection{Decomposing intermediate operations into isometries}\label{app:tester_isometry_sequence_decomposition}

Now we will describe how one can obtain an initial state $|\psi\>$ and a sequence of isometries $V^{(1)},\dots,V^{(T-1)}$ implementing the intermediate operations corresponding to the map $\sqrt{\mathbf{M}}\cdot\sqrt{\mathbf{M}}^\dag$.
We describe this through an example of sequentially estimating the angle of $T=2$ copies of single-qubit Pauli-$Z$ rotation channels as described in Section~\ref{sec:z_rotation_example} in the main text.
A general method of obtaining this decomposition for $T=2$ steps is described in Proposition~\ref{prop:2_step_tester_decomposition}.

For $T=2$, an optimal estimator matrix is $\mathbf{X}_{0,1}^*$, given by
\begin{equation}
\begin{gathered}
    \mathbf{X}_{0,1}^* = \frac{1}{4}\kketbra{M_1} -\frac{1}{4}\kketbra{M_2} 
    \quad\text{for}\quad
    |M_1\rr = \frac{|0^4\>-|1^4\>}{\sqrt{2}} 
    \;,\quad
    |M_2\rr = \frac{|0^4\>+|1^4\>}{\sqrt{2}} 
\end{gathered},
\end{equation}
and the optimal strategy can be described by positive semidefinite operator $\mathbf{M}\in\L(A_2\otimes A_1'\otimes A_1\otimes A_0)$ is given by
\begin{equation}\label{eqn:M_matrix_2_step_z_rotation}
\begin{aligned}
    \mathbf{M} &= \sum_{k=1}^{12} \mathbf{M}_k \\
    &= \frac{\kketbra{M_1}+\kketbra{M_2} + \ketbra{0111}+\ketbra{1000}}{2} + I_{A_2} \otimes \frac{I_{A_1'}\otimes(\ketbra{01}+\ketbra{10})}{4} \\
    &= I_{A_2}\otimes \mathbf{M}^{(1)} \\
    &= 
    \begin{quantikz}[classical gap=0.04cm]
        \lstick[2]{$\rho$} & \wire[r][1]["{C}"{above,pos=0}]{q} & & \gate[2,style={fill=blue!20}]{V_1} & \wire[r][1]["{C'}"{above,pos=0}]{q} & & \gate[2,style={fill=blue!20}]{\mathcal{M}'} \wireoverride{q} \\
        & \wire[l][1]["{A_0}"{above,pos=0}]{q} \setwiretype{n} & \wire[r][1]["{A_1}"{above,pos=0}]{q} & & \wire[l][1]["{A_1'}"{above,pos=0}]{q} & \wire[r][1]["{A_2}"{above,pos=0}]{q} & 
    \end{quantikz} \\
    &\qquad\text{for} \\
    \mathbf{M}^{(1)} &= \frac{\ketbra{0^3}+\ketbra{1^3}}{2} + \frac{\ketbra{001}+\ketbra{010}+\ketbra{101}+\ketbra{110}}{4} \\
    &=
    \begin{quantikz}[classical gap=0.04cm]
        \lstick[2]{$\rho$} & \wire[r][1]["{C}"{above,pos=0}]{q} & & \gate[2,style={fill=blue!20}]{V_1} & \wire[r][1]["{C'}"{above,pos=0}]{q} & \ground{} \\
        & \wire[l][1]["{A_0}"{above,pos=0}]{q} \setwiretype{n} & \wire[r][1]["{A_1}"{above,pos=0}]{q} & & \wire[l][1]["{A_1'}"{above,pos=0}]{q} & 
    \end{quantikz}
\end{aligned}
\end{equation}
and reduced initial state $\tr_C(\rho) = I_{A_0}/2$, since $\kketbra{M_1}+\kketbra{M_2} = \ketbra{0^4}+\ketbra{1^4}$ and $\tr_{A_1'}(\mathbf{M}^{(1)}) = I_{A_1}\otimes I_{A_0}/2$.
The notation $\begin{quantikz}[classical gap=0.04cm]
        & \wire[r][1]["{C'}"{above,pos=0}]{q} & \ground{} 
    \end{quantikz}$ above means that register $C'$ is being traced out.

The probability of outcome $k\in\{1,\dots,12\}$ from the optimal strategy is given by
\begin{equation}
    \tr(\mathbf{\Lambda}_\theta \mathbf{M}_k^\top) = 
    \begin{quantikz}[classical gap=0.05cm]
        \setwiretype{q} \lstick[2]{$\rho$} & & \gate[2,style={fill=blue!20}]{V_1} & \wire[l][1][""{above,pos=0}]{q} & \gate[2,style={fill=blue!20}]{\mathcal{M}_k'} \wireoverride{q}  \\
        \setwiretype{q} & \gate[1,style={fill=red!20}]{\Lambda_\theta} & & \gate[1,style={fill=red!20}]{\Lambda_\theta} & 
    \end{quantikz}, 
\end{equation}
where the $\mathbf{M}_k$'s are the 12 terms of $\mathbf{M}$.
Now consider a process map $\mathcal{T}(\cdot) = \sqrt{\mathbf{M}^\top}\cdot\sqrt{\mathbf{M}^\top}$ which action on $\mathbf{\Lambda}_\theta$ and which can be illustrated as
\begin{equation}\label{eqn:process_to_state_map}
\begin{gathered}
    \mathcal{T}(\mathbf{\Lambda}_\theta) = \sqrt{\mathbf{M}^\top} \mathbf{\Lambda}_\theta \sqrt{\mathbf{M}^\top} =
    \begin{quantikz}[classical gap=0.05cm]
        \setwiretype{q} \lstick[2]{$\rho$} & & \gate[2,style={fill=blue!20}]{V_1} & \wire[l][1]["C'"{above,pos=0}]{q} &  \\
        \setwiretype{q} & \gate[1,style={fill=red!20}]{\Lambda_\theta} & & \gate[1,style={fill=red!20}]{\Lambda_\theta} & \wire[l][1]["A_2"{above,pos=0}]{q}
    \end{quantikz} 
\end{gathered}
\end{equation}
for some density operator $\rho\in\L(C\otimes A_0)$ and isometry $V_1:C\otimes A_1\rightarrow C'\otimes A_1'$ for $C'\cong A_1'\otimes A_1\otimes A_0$.
Operator $\sqrt{\mathbf{M}^\top}:A_2\otimes A_1'\otimes A_1\otimes A_0 \rightarrow A_2\otimes C'$ with $C'\cong A_1'\otimes A_1\otimes A_0$ is explicitly given by
\begin{equation}
\begin{aligned}
    \sqrt{\mathbf{M}^\top} &= \frac{\ketbra{0^4}+\ketbra{1^4}+\ketbra{0111}+\ketbra{1000}}{\sqrt{2}} + I_{A_2} \otimes \frac{I_{A_1'}\otimes(\ketbra{01}+\ketbra{10})}{2} \\
    &= I_{A_2} \otimes \Big( \frac{\ketbra{0^3}+\ketbra{1^3}}{\sqrt{2}} + \frac{I_{A_1'}\otimes(\ketbra{01}+\ketbra{10})}{2} \Big) \;.
\end{aligned}
\end{equation}

The POVM measurement $\mathcal{M}'$ consists of POVM elements
\begin{equation}\label{eqn:2step_z_rotation_optimal_POVM}
\begin{gathered}
    M_k' = \sqrt{\mathbf{M}^+} \mathbf{M}_k \sqrt{\mathbf{M}^+} + Q_k = 
    \begin{quantikz}[classical gap=0.05cm]
        \setwiretype{q} & \gate[2,style={fill=blue!20}]{\mathcal{M}_k'} \wire[l][1]["{C'}"{above,pos=0.5}]{q} \\
        \setwiretype{q} & \wire[l][1]["{A_2}"{above,pos=0.5}]{q}
    \end{quantikz} \\
    \text{for}\\
    \sqrt{\mathbf{M}^+} = \sqrt{2}\big(\ketbra{0^4}+\ketbra{1^4}+\ketbra{0111}+\ketbra{1000}\big) + I_{A_2} \otimes 2\big(I_{A_1'}\otimes(\ketbra{01}+\ketbra{10}) \big) \\
    \text{and}\\
    Q_k =
    \begin{cases}
        \ketbra{0100} \;,\quad&\text{if }j=1 \\
        \ketbra{0011} \;,\quad&\text{if }j=2 \\
        \ketbra{1100} \;,\quad&\text{if }j=3 \\
        \ketbra{1011} \;,\quad&\text{if }j=4 \\
        0 \;,\quad&\text{otherwise} 
    \end{cases} \;,
\end{gathered}
\end{equation}
so that each POVM element $M_k'$ is a projection.
In fact, we can pick any four $j$'s and $Q_k$'s such as $\sum_k Q_k$ is a projection onto the orthogonal subspace of the support of $\mathbf{M}$ (so that $\{M_k'\}_k$ is a POVM), since
\begin{equation}
\begin{aligned}
    \tr((M_k')^\top \mathcal{T}(\mathbf{\Lambda}_\theta)) &= \tr\Big( (\sqrt{\mathbf{M}^+} \mathbf{M}_k \sqrt{\mathbf{M}^+} + Q_k)^\top (\sqrt{\mathbf{M}^\top} \mathbf{\Lambda}_\theta \sqrt{\mathbf{M}^\top}) \Big) \\
    &= \tr\Big( (\sqrt{\mathbf{M}^+} \mathbf{M}_k \sqrt{\mathbf{M}^+})^\top (\sqrt{\mathbf{M}^\top} \mathbf{\Lambda}_\theta \sqrt{\mathbf{M}^\top}) \Big) \\
    &= \tr\Big( \mathbf{M}' \mathbf{M}_k^\top \mathbf{M}' \mathbf{\Lambda}_\theta \Big) \\
    &= \tr\Big( \mathbf{M}_k^\top \mathbf{\Lambda}_\theta \Big) \\
    &\quad\text{for}\\
    \mathbf{M}' &= \ketbra{0^4}+\ketbra{1^4}+\ketbra{0111}+\ketbra{1000} + I_{A_2} \otimes I_{A_1'}\otimes(\ketbra{01}+\ketbra{10}) \;.
\end{aligned}
\end{equation}

To obtain an explicit form of $V_1$ of $\rho$, first consider the restriction $\mathcal{T}_{A_1'A_1A_0}:\L(A_1'\otimes A_1\otimes A_0)\rightarrow\L(C')$ of $\mathcal{T}$ to $\L(A_1'\otimes A_1\otimes A_0)$, given by
\begin{equation}
\begin{gathered}
    \mathcal{T}_{A_1'A_1A_0}(\cdot) = \Tilde{\mathbf{M}}\cdot\Tilde{\mathbf{M}}
    =
    \begin{quantikz}[classical gap=0.04cm]
        \lstick[2]{$\rho$} & \wire[r][1]["{C}"{above,pos=0}]{q} & & \gate[2,style={fill=blue!20}]{V_1} & \wire[r][1]["{C'}"{above,pos=0}]{q} & \\
        & \wire[l][1]["{A_0}"{above,pos=0}]{q} \setwiretype{n} & \wire[r][1]["{A_1}"{above,pos=0}]{q} & & \wire[l][1]["{A_1'}"{above,pos=0}]{q} \wire[r][1][""{above,pos=0}]{q} & 
    \end{quantikz} \\
    \text{for}\\
    \Tilde{\mathbf{M}} = \frac{\ketbra{0^3}+\ketbra{1^3}}{\sqrt{2}} + \frac{I_{A_1'}\otimes(\ketbra{01}+\ketbra{10})}{2}
\end{gathered},
\end{equation}
which satisfy $\mathcal{I}_{A_2} \otimes \mathcal{T}_{A_1'A_1A_0}= \mathcal{T}$ since $\mathcal{T}$ acts as an identity map on $\L(A_2)$ (also $\sqrt{\mathbf{M}^\top} = I_{A_2}\otimes\Tilde{\mathbf{M}}$).
Now, consider a linear map $\Tilde{\mathcal{T}}:\L(A_1)\rightarrow\L(C'\otimes A_1'\otimes A_0)$ given by
\begin{equation}\label{eqn:inputs_to_outputs_map}
\begin{gathered}
    \Tilde{\mathcal{T}}(\cdot) =
    \begin{quantikz}[classical gap=0.04cm]
        & \wire[r][1]["{A_1}"{above,pos=0}]{q} & & \gate[2,style={fill=blue!20}]{V_1} & \wire[r][1]["{C'}"{above,pos=0}]{q} & \\
        \lstick[2]{$\rho$} & \wire[r][1]["{C}"{above,pos=0}]{q} & & & \wire[l][1]["{A_1'}"{above,pos=0}]{q} \wire[r][1][""{above,pos=0}]{q} & \\
        & \setwiretype{q} & & & \wire[l][1]["{A_0}"{above,pos=0}]{q} &
    \end{quantikz}
    = \Bar{\mathbf{M}} \cdot \Bar{\mathbf{M}}^\dag \\
    \text{for}\\
    \Bar{\mathbf{M}} = \frac{|0^5\>\<0|+|1^5\>\<1|}{\sqrt{2}} + \sum_{c=0}^1 \frac{(|c01c\>_{C'A_1'} |1\>_{A_0} \<0|_{A_1}+|c10c\>_{C'A_1'} |0\>_{A_0} \<1|_{A_1})}{2} \;.
\end{gathered}
\end{equation}
By tracing out $C'$ and $A_1'$ we get
\begin{equation}\label{eqn:traced_output_equivalence}
\begin{aligned}
    \tr_{C'A_1'}\circ\Tilde{\mathcal{T}}(\cdot) &= 
    \begin{quantikz}[classical gap=0.04cm]
        & \wire[r][1]["{A_1}"{above,pos=0}]{q} & & \gate[2,style={fill=blue!20}]{V_1} & \wire[r][1]["{C'}"{above,pos=0}]{q} & \ground{} \\
        \lstick[2]{$\rho$} & \wire[r][1]["{C}"{above,pos=0}]{q} & & & \wire[l][1]["{A_1'}"{above,pos=0}]{q} \wire[r][1][""{above,pos=0}]{q} & \ground{} \\
        & \setwiretype{q} & & & \wire[l][1]["{A_0}"{above,pos=0}]{q} &
    \end{quantikz} \\
    &= \tr_{C'A_1'}(\Bar{\mathbf{M}} \cdot \Bar{\mathbf{M}}^\dag) \\
    &= \frac{1}{2} \Big( \ketbra{0}\cdot\ketbra{0} + \ketbra{1}\cdot\ketbra{1} \Big) + \frac{1}{2} \Big( |1\>\<0|\cdot|0\>\<1| + |0\>\<1|\cdot|1\>\<0| \Big) \\
    &= \<0|\cdot|0\>\frac{I_{A_0}}{2} + \<1|\cdot|1\>\frac{I_{A_0}}{2} \\
    &= \tr_{A_1}(\cdot) \frac{I_{A_0}}{2} \\
    &=
    \begin{quantikz}[classical gap=0.04cm]
        & \wire[r][1]["{A_1}"{above,pos=0}]{q} & \ground{} \\
        \lstick[1]{$\frac{I}{2}$} & \wire[l][1]["{A_0}"{above,pos=0}]{q} &
    \end{quantikz} \;.
\end{aligned}
\end{equation}
Equivalence $\tr_{C'A_1'}\circ\Tilde{\mathcal{T}}(\cdot) = \tr_{A_1}(\cdot) \frac{I_{A_0}}{2}$ between the two quantum channels implies that there exists an isometry $V:\Tilde{A}_1\otimes \Tilde{A}_0 \rightarrow C'\otimes A_1'$ with $\Tilde{A}_0\cong A_0$ and $\Tilde{A}_1\cong A_1$ that relates the Kraus operators of the two channels:
\begin{equation}\label{eqn:2step_z_rot_kraus_equivalence}
\begin{aligned}
    \<c,b',a',c'|_{C'A_1'}\Bar{\mathbf{M}} &= \sum_{b,a=0}^1 V_{c,b',a',c' ; b,a} \frac{|a\>_{A_0}\<b|_{A_1}}{\sqrt{2}} \\
    &=  \sum_{b,a=0}^1 \Big( \<c,b',a',c'|_{C'A_1'}V|b,a\>_{\Tilde{A}_1,\Tilde{A}_0} \otimes I_{A_0} \Big) \frac{|a\>_{A_0}\<b|_{A_1}}{\sqrt{2}}
    \;,\;\forall c,b',a',c'
\end{aligned},
\end{equation}
which implies that $V_1:A_1\otimes C \rightarrow C'\otimes A_1'$ can be taken as $V$ with $C\cong \Tilde{A}_0$ and $\rho$ as the maximally mixed state on $A_0$, since
\begin{equation}\label{eqn:Mbar_operator_intermediate_isometry}
\begin{aligned}
    \Bar{\mathbf{M}} = \sum_{b,a=0}^1 \Big( V|b,a\>_{\Tilde{A}_1,\Tilde{A}_0} \otimes I_{A_0} \Big) \frac{|a\>_{A_0}\<b|_{A_1}}{\sqrt{2}}
    &= (V\otimes I_{A_0}) \Big( I_{A_1} \otimes \underbrace{ \frac{|00\>_{A_0}+|11\>_{A_0}}{\sqrt{2}} }_{=|\phi\>} \Big)
    = 
    \begin{quantikz}[classical gap=0.04cm]
        & \wire[r][1]["{\Tilde{A}_1}"{above,pos=0}]{q} & & \gate[2,style={fill=blue!20}]{V} & \wire[r][1]["{C'}"{above,pos=0}]{q} & \\
        \lstick[2]{$|\phi\>$} & \wire[r][1]["{\Tilde{A}_0}"{above,pos=0}]{q} & & & \wire[l][1]["{A_1'}"{above,pos=0}]{q} \wire[r][1][""{above,pos=0}]{q} & \\
        & \setwiretype{q} & & & \wire[l][1]["{A_0}"{above,pos=0}]{q} &
    \end{quantikz} \;,
\end{aligned}
\end{equation}
where $|\phi\>$ is the maximally entangled state.

So, we can write the entries of $V$ (equivalently $V_1$) in terms of the entries of $\Bar{\mathbf{M}}$ and $\Tilde{\mathbf{M}}$ by using eqn.~\eqref{eqn:2step_z_rot_kraus_equivalence} and the relation $\Bar{\mathbf{M}} = \Tilde{\mathbf{M}}(I_{A_1}\otimes |I\rr_{A_1'} |I\rr_{A_0})$ (where in turn, $\Tilde{\mathbf{M}} = \tr_{A_2}(\sqrt{\mathbf{M}^\top})/2$, since $I_{A_2} \otimes \Tilde{\mathbf{M}} = \sqrt{\mathbf{M}^\top}$):
\begin{equation}\label{eqn:isometry_as_M_matrix}
\begin{aligned}
    V_{c,b',a',c' ; b,a} &= \sqrt{2} \<c,b',a',c',a|_{C'A_1'A_0}\Bar{\mathbf{M}}|b\>_{A_1} \\
    &= 2^{-\frac{1}{2}} \<c,b',a'|_{C'}\Tilde{\mathbf{M}}|c',b,a\>_{A_1'A_1A_0} 
    =
    \begin{cases}
        1 \quad&,\;\text{if }(c,b',a') = (c',b,a) \in\{000,111\} \\
        \frac{1}{\sqrt{2}} \quad&,\;\text{if }(c,b',a') = (c',b,a) \;,\; (b,a)\in\{01,10\} \\
        0 \quad&,\;\text{otherwise}
    \end{cases} \;.
\end{aligned}
\end{equation}
Explicitly, isometry $V$ can be realized using $|0\>$ state preparations and controlled Hadamard gates $H$:
\begin{equation}\label{eqn:app_optimal_isometry_z_rot}
\begin{aligned}
    V &=
    |0^4\>\<00| + |1^4\>\<11| + \frac{|0010\>+|1011\>}{\sqrt{2}}\<01| + \frac{|0100\>+|1101\>}{\sqrt{2}}\<10| \\
    &= 
    \begin{quantikz}
        & \gate[1,style={blue!05,fill=blue!05}]{|0\>} \setwiretype{n} \gategroup[4,steps=4,style={dashed, fill=blue!05, inner xsep=2pt},background,label style={label position=above,anchor=north,yshift=0.3cm}]{{\sc $V$}} & \gate[]{H} \setwiretype{q} & \gate[]{H} & \ctrl{3} & \rstick[3]{$C'$} \\
        \lstick[1]{$C$} & & & \ctrl{-1} & & \\
        \lstick[1]{$A_1$} & & \ctrl{-2} & & & \\
        & \gate[1,style={blue!05,fill=blue!05}]{|0\>} \setwiretype{n} & \setwiretype{q} & & \targ{} & \rstick[1]{$A_1'$}
    \end{quantikz} \;.
\end{aligned}
\end{equation}

Generally, for a strategy for $T=2$ described by positive semidefinite operator $\mathbf{M}\in\L(A_2\otimes A_1'\otimes A_1\otimes A_0)$, we have the following general form of isometry $V$ and initial pure state $|\psi\>\in C\otimes A_0$.

\begin{proposition}[{Two-step tester decomposition}]\label{prop:2_step_tester_decomposition}
    Consider a two-step tester $\mathbf{M} \in\L(A_2\otimes A_1'\otimes A_1\otimes A_0)$.
    It holds that tester $\mathbf{M}$ admits a decomposition $\mathbf{M} = \sum_k \mathbf{M}_k$ for rank-1 mutually orthogonal positive semidefinite operators $\{\mathbf{M}_k\}_k$ and is realizable as
    \begin{equation}
    \begin{aligned}
        \mathbf{M} = 
        \begin{quantikz}[classical gap=0.04cm]
            \lstick[2]{$|\psi\>$} & \wire[r][1]["{C}"{above,pos=0}]{q} & & \gate[2,style={fill=blue!20}]{V} & \wire[r][1]["{C'}"{above,pos=0}]{q} & & \gate[2,style={fill=green!20}]{\mathcal{M}'} \wireoverride{q} \\
            & \wire[l][1]["{A_0}"{above,pos=0}]{q} \setwiretype{n} & \wire[r][1]["{A_1}"{above,pos=0}]{q} & & \wire[l][1]["{A_1'}"{above,pos=0}]{q} & \wire[r][1]["{A_2}"{above,pos=0}]{q} & 
        \end{quantikz} 
    \end{aligned}
    \end{equation}
    for some POVM measurement $\mathcal{M}' = \{M_k'\}_k$ over $C'\otimes A_2$, isometry $V:C\otimes A_1\rightarrow C'\otimes A_1'$, and initial state $|\psi\>\in C\otimes A_0$ where $C'\cong A_1'\otimes A_1\otimes A_0$ and $C\cong A_0$.
    Moreover, there exists a positive semidefinite operator $\mathbf{M}^{(1)}\in\L(A_1'\otimes A_1\otimes A_0)$ such that $\mathbf{M} = I_{A_2}\otimes \mathbf{M}^{(1)}$ and a reduced initial state $\rho = \sum_k q_k\ketbra{k} = \tr_{A_2A_1'A_1}(\mathbf{M})/(d_{A_2}d_{A_1}) \in\L(A_0)$ such that POVM $\mathcal{M}'$, isometry $V:C\otimes A_1\rightarrow C'\otimes A_1'$, and initial state $|\psi\>\in C\otimes A_0$ are given by
    \begin{equation}
    \begin{gathered}
        M_k' = \sqrt{\mathbf{M}^+} \mathbf{M}_k \sqrt{\mathbf{M}^+} + Q_k \;,\quad
        V = (\sqrt{(\mathbf{M}^{(1)})^\top}\otimes I_{A_1'})(I_{A_1}\otimes |I\rr_{A_1'} \otimes\sqrt{\rho^{-1}})
        \;,\;\text{and}\quad
        |\psi\> = \sum_k \sqrt{q_k}|kk\>_{CA_0}  \;
    \end{gathered}
    \end{equation}
    for some positive semidefinite operators $\{Q_k\}_k$ such that $\sum_k Q_k$ is a projector onto the kernel of $\mathbf{M}$.
\end{proposition}
\begin{proof}
    The form of POVM $\mathcal{M}' = \{M_k'\}_k$ follows from the result of~\cite[Theorem 12]{chiribella2009theoretical}, which also asserts that there exists a linear map $\mathcal{T}(\cdot) = \sqrt{\mathbf{M}^\top} \cdot \sqrt{\mathbf{M}^\top}$ that maps the input operator to a state in $C'\otimes A_2$ (as in eqn.~\eqref{eqn:process_to_state_map}) for $C'\cong A_1'\otimes A_1\otimes A_0$ since $\sqrt{\mathbf{M}^\top}$ is a square matrix.
    The existence of positive semidefinite operator $\mathbf{M}^{(1)}\in\L(A_1'\otimes A_1\otimes A_0)$ such that $\mathbf{M} = I_{A_2}\otimes \mathbf{M}^{(1)}$ and $\rho = \tr_{A_2A_1'A_1}(\mathbf{M})/(d_{A_2}d_{A_1})$ being a density matrix are guaranteed by the fact that $\mathbf{M}$ is a valid strategy (or, a \textit{tester} in the terminology of~\cite{chiribella2009theoretical}).

    We can now define a linear map $\Tilde{\mathcal{T}}:\L(A_1)\rightarrow\L(C'\otimes A_1'\otimes A_0)$ given by $\Tilde{\mathcal{T}}(\cdot) = \Bar{\mathbf{M}}\cdot\Bar{\mathbf{M}}^\dag$ as in eqn.~\eqref{eqn:inputs_to_outputs_map} where $\Bar{\mathbf{M}} = (\sqrt{(\mathbf{M}^{(1)})^\top}\otimes I_{A_1'}\otimes I_{A_0}) (I_{A_1}\otimes |I\rr_{A_1'} |I\rr_{A_0})$.
    Tracing out $C',A_1'$ at the output of $\Tilde{\mathcal{T}}$ gives
    \begin{equation}
    \begin{aligned}
        \tr_{C'A_1'}\circ\Tilde{\mathcal{T}}(\cdot) &= \tr_{C'}\big( (\sqrt{(\mathbf{M}^{(1)})^\top}\otimes I_{A_0}) (I_{A_1'}\otimes \cdot\otimes \kketbra{I}_{A_0}) (\sqrt{(\mathbf{M}^{(1)})^\top}\otimes I_{A_0}) \big) \\
        &= \sum_{j,k} \tr\big( \sqrt{(\mathbf{M}^{(1)})^\top} (I_{A_1'}\otimes \cdot\otimes |j\>\<k|_{A_0}) \sqrt{(\mathbf{M}^{(1)})^\top} \big) |j\>\<k|_{A_0} \\
        &= \sum_{j,k} \tr_{A_1A_0}\big( \tr_{A_1'}\big((\mathbf{M}^{(1)})^\top\big) (\cdot\otimes |j\>\<k|_{A_0}) \big) |j\>\<k|_{A_0} \\
        &= \sum_{j,k} \tr_{A_1A_0}\big( (I_{A_1}\otimes\rho^\top) (\cdot\otimes |j\>\<k|_{A_0}) \big) |j\>\<k|_{A_0} \\
        &= \tr_{A_1}(\cdot) \rho \;,
    \end{aligned}
    \end{equation}
    where we use $\tr_{A_1'}(\mathbf{M}^{(1)}) = I_{A_1}\otimes\rho$ to obtain the fourth equality, since $\mathbf{M}$ is a valid strategy.
    The equality $\tr_{C'A_1'}\circ\Tilde{\mathcal{T}}(\cdot) = \tr_{A_1}(\cdot) \rho_{A_0}$ generalizes eqn.~\eqref{eqn:traced_output_equivalence}.
    Similar to eqn.~\eqref{eqn:2step_z_rot_kraus_equivalence}, there exists an isometry $V:A_1\otimes C \rightarrow C'\otimes A_1'$ which relates the Kraus operators $\{\<c,b',a',c'|_{C'A_1'}\Bar{\mathbf{M}}\}_{c,b',a',c'}$ and $\{\sqrt{q_k}|a\>_{A_0}\<b|_{A_1}\}_{a,b}$ of $\tr_{A_1'}(\mathbf{M}^{(1)})$ and $\tr_{A_1}(\cdot)\rho$, respectively, as
    \begin{equation}\label{eqn:general_isometry_as_M_matrix}
    \begin{aligned}
        \<c,b',a',c'|_{C'A_1'}\Bar{\mathbf{M}} &= \sum_{b,a} V_{c,b',a',c' ; b,a} \sqrt{q_a}|a\>_{A_0}\<b|_{A_1} 
        \;,\;\forall c,b',a',c',
    \end{aligned}
    \end{equation}
    so that we can express $\Tilde{\mathcal{T}}$ in terms of isometry $V$ and initial state $|\psi\>$:
    \begin{equation}
    \begin{aligned}
        \Bar{\mathbf{M}} 
        &= (V\otimes I_{A_0}) \Big( I_{A_1} \otimes \underbrace{ \sum_a \sqrt{q_a}|aa\>_{CA_0} }_{=|\psi\> \in C\otimes A_0} \Big) \;,
    \end{aligned}
    \end{equation}
    which generalizes eqn.~\eqref{eqn:Mbar_operator_intermediate_isometry}.
    Furthermore by using eqn.~\eqref{eqn:general_isometry_as_M_matrix} and $\Bar{\mathbf{M}} = (\sqrt{(\mathbf{M}^{(1)})^\top}\otimes I_{A_1'}\otimes I_{A_0}) (I_{A_1}\otimes |I\rr_{A_1'} |I\rr_{A_0})$, we can explicitly write entries of $V$ as
    \begin{equation}
        V_{c,b',a',c' ; b,a} = q_a^{-\frac{1}{2}} \<c,b',a'|_{C'}\<c'|_{A_1'}\<a|_{A_0}\Bar{\mathbf{M}}|b\>_{A_1} = q_a^{-\frac{1}{2}} \<c,b',a'|_{C'} \big( \sqrt{(\mathbf{M}^{(1)})^\top} \big) |c'\>_{A_1'}|b\>_{A_1}|a\>_{A_0} \;,
    \end{equation}
    generalizing eqn.~\eqref{eqn:isometry_as_M_matrix}.
    Since $\sqrt{\rho^{-1}} = \sum_a q_a^{-\frac{1}{2}}\ketbra{a}$, we obtain $V = (\sqrt{(\mathbf{M}^{(1)})^\top}\otimes I_{A_1'})(I_{A_1}\otimes |I\rr_{A_1'} \otimes\sqrt{\rho^{-1}})$.
\end{proof}

For an arbitrary $T$, the initial state $|\psi\>$, the sequence of isometries $V^{(1)},\dots,V^{(T-1)}$, and the final POVM $\{M_k'\}_k$ can be obtained iteratively using the same technique used in the proof of Proposition~\ref{prop:2_step_tester_decomposition}.
By considering operator $\mathbf{M}^{(t)} = \tr_{A_{t+1}A_{t+1}'}(\mathbf{M}^{(t+1)})/d_{A_{t+1}}$ for each $t\in\{0,\dots,T-2\}$, first we can obtain the initial state $|\psi\>$ by purifying operator $\mathbf{M}^{(0)}$ over $A_0$, which must be a density operator, since $\mathbf{M}\in\seq(\Bar{A}_T)$.
Then, we can obtain isometry $V^{(t)}$ from Kraus operator equivalence between quantum channels $\tr_{A_t'}\circ\Tilde{\mathcal{T}}^{(t)}$ and $\tr_{A_t}\otimes\Tilde{\mathcal{T}}^{(t-1)}$ where $\tr_{A_t'}\circ\Tilde{\mathcal{T}}^{(t)} = \tr_{A_t}\otimes\Tilde{\mathcal{T}}^{(t-1)}$ similar to eqn.~\eqref{eqn:general_isometry_as_M_matrix}.
Here, $\Tilde{\mathcal{T}}^{(t)}:\L(\bigotimes_{l=1}^t A_l)\rightarrow \L((\bigotimes_{l=1}^t A_l')\otimes A_0)$ is a quantum channel defined by the link product
\begin{equation}
\begin{aligned}
    \Tilde{\mathcal{T}}^{(t)}(G) &= G\ast\mathbf{M}^{(t)} \\ 
    &= \tr_{A_1\dots A_t}\Big( G (\mathbf{M}^{(t)})^{\top_{A_1\dots A_t}} \Big) \\
    &=
    \begin{quantikz}[classical gap=0.04cm]
        & \wire[r][1]["{A_t}"{above,pos=0}]{q} & & & & & \gate[2,style={fill=blue!20}]{V^{(t)}} & \wire[r][1]["{C_{t+1}}"{above,pos=0}]{q} & \ground{} \\
        & \wire[r][1]["{A_{t-1}}"{above,pos=0}]{q} & & & & \gate[2,style={fill=blue!20}]{V^{(t-1)}} & & \wire[r][1]["{A_t'}"{above,pos=0}]{q} & \\
        \setwiretype{n} & \invgate{\vdots} & & & \wire[r][1]["{C_{t-1}}"{above,pos=0}]{q} & \setwiretype{q} & & \wire[r][1]["{A_{t-1}'}"{above,pos=0}]{q} & \\
        \setwiretype{n} & \invgate{\vdots} & & & \invgate{\udots} & & & \invgate{\vdots} & \\
        & \wire[r][1]["{A_1}"{above,pos=0}]{q} & & \gate[2,style={fill=blue!20}]{V^{(1)}} & \wire[l][1]["{C_2}"{above,pos=0}]{q} & \setwiretype{n} & & \invgate{\vdots} & \\
        \lstick[2]{$\rho$} & \wire[r][1]["{C_1}"{above,pos=0}]{q} & & & & & & \wire[l][1]["{A_1'}"{above,pos=0}]{q} \wire[r][1][""{above,pos=0}]{q} & \\
        & \setwiretype{q} & & & & & & \wire[l][1]["{A_0}"{above,pos=0}]{q} &
    \end{quantikz}
\end{aligned}
\end{equation}
for all operators $G$ over $\bigotimes_{l=1}^t A_l$.
For more details, we refer the readers to~\cite[Theorem 3]{chiribella2009theoretical}.

Recall from the main text that, under the finite-dimensionality, regularity, and optimizer-existence assumptions in Proposition~\ref{prop:single_par_achievability}, the SDP bound in Theorem~\ref{thm:sequential_NH_bound} is locally attainable for single-parameter estimation.
The tester decomposition above complements this result by reconstructing an initial state, intermediate isometries, and a final measurement from the optimal tester block of the SDP optimizer.
The tester decomposition given above complements this achievability result in that one can always find such a decomposition for single-parameter estimation tasks from the optimizer given by the SDP.

\section{Proof of Theorem~\ref{thm:sequential_NH_bound}}\label{app:sequential_bound_proof}

In this appendix, we prove  Theorem~\ref{thm:sequential_NH_bound}.
First, we need to write the covariance matrix for a given parameter estimation problem as a linear function over matrices describing a strategy.
Note that the entries of the covariance matrix in eqn.~\eqref{eqn:cov_matrix_channel_sequence} can be more explicitly written in terms of the operators $\{\mathbf{M}_k\}_{k=1}^K$ and the channel operator $\mathbf{\Lambda}_\theta$ encoding parameters $\theta=(\theta_1,\dots,\theta_m)$ over $T$ time steps:
\begin{equation}
\begin{aligned}
    \cov_\theta(\mathcal{M},\hat{\theta})_{i,j} &= \sum_k \xi_{i,k} \xi_{j,k} p(k|\theta,\mathcal{M}) 
    = \tr\bigg( \mathbf{\Lambda}_\theta^\top  \Big( \underbrace{ \sum_k \xi_{i,k} \xi_{j,k} \mathbf{M}_k }_{=\mathbf{L}_\theta(\mathcal{M},\hat{\theta})_{i,j}} \Big) \bigg) 
    = \tr\bigg( \mathbf{\Lambda}_\theta^\top  \mathbf{L}_\theta(\mathcal{M},\hat{\theta})_{i,j}\bigg) \;.
\end{aligned}
\end{equation}
The matrix $\mathbf{L}_\theta(\mathcal{M},\hat{\theta})$ is given by
\begin{equation}\label{eqn:L_matrix}
\begin{aligned}
    \mathbf{L}_\theta(\mathcal{M},\hat{\theta}) &= \sum_{i,j=1}^m |i\>\<j|_V \otimes \mathbf{L}_\theta(\mathcal{M},\hat{\theta})_{i,j} 
    = \Xi 
    \left[\begin{array}{ccc}
        \mathbf{M}_1 &&  \\
        & \ddots &\\
        && \mathbf{M}_K
    \end{array}\right] \Xi^\top \;,
\end{aligned}
\end{equation}
which is a matrix over $V\otimes \Bar{A}_T$ for an $m$-dimensional real vector space $V\cong \R^m$ and real matrix
\begin{equation}\label{eqn:estimator_parameter_matrix}
    \Xi = 
    \begin{bmatrix}
        \xi_{1,1} & \dots & \xi_{1,K} \\
        \vdots & \ddots & \vdots \\
        \xi_{m,1} & \dots & \xi_{m,K}
    \end{bmatrix} \;.
\end{equation}
Hence by denoting $\Tilde{\mathbf{\Lambda}}_\theta^{(T)} = I_m \otimes (\mathbf{\Lambda}_\theta^{(T)})^\top$, where $I_m$ is the $m\times m$ identity matrix, we can write the covariance matrix as 
\begin{equation}\label{eqn:covariance_matrix_link_prod}
    \cov_\theta (\mathcal{M},\hat{\theta}) = \tr_{\Bar{A}_T}\Big( \Tilde{\mathbf{\Lambda}}_\theta^{(T)} \mathbf{L}_\theta(\mathcal{M},\hat{\theta}) \Big) \;,
\end{equation}
where $\tr_{\Bar{A}_T}$ is the partial trace over $\Bar{A}_T$ and therefore the corresponding MSE can be written as
\begin{equation}
    \mse_\theta (\mathcal{M},\hat{\theta}) = \tr\Big( \Tilde{\mathbf{\Lambda}}_\theta^{(T)} \mathbf{L}_\theta(\mathcal{M},\hat{\theta}) \Big) \;.
\end{equation}

Now let us define \textit{estimator matrices} $X_\theta(\mathcal{M},\hat{\theta})_1,\dots,X_\theta(\mathcal{M},\hat{\theta})_m$, given by
\begin{equation}\label{eqn:estimator_matrix_elements}
\begin{aligned}
    X_\theta(\mathcal{M},\hat{\theta})_i = \sum_k \xi_{i,k} \mathbf{M}_k
\end{aligned}
\end{equation}
and a block matrix of these estimator matrices
\begin{equation}\label{eqn:estimator_matrix}
    X_\theta(\mathcal{M},\hat{\theta}) = \sum_{i=1}^m |i\>_V \otimes X_\theta(\mathcal{M},\hat{\theta})_i = \Xi
    \left[\begin{array}{cc}
        \mathbf{M}_1 \\
        \vdots \\
        \mathbf{M}_K
    \end{array}\right]
\end{equation}
for $m\times K$ real matrix $\Xi$ defined in eqn.~\eqref{eqn:estimator_parameter_matrix}.
In the following, we will show that for any estimator $\hat{\theta}$ and sequential probe $\mathcal{M}$, it holds that
\begin{equation}\label{eqn:LX_inequality}
    \mathbf{L}_\theta(\mathcal{M},\hat{\theta}) \geq X_\theta(\mathcal{M},\hat{\theta}) \mathbf{M}^+ X_\theta(\mathcal{M},\hat{\theta})^\dag \;,
\end{equation}
where $\mathbf{M} = \sum_k\mathbf{M}_k$ and $\mathbf{M}^+$ is the pseudoinverse of $\mathbf{M}$.
To do this, we use the following lemma.

\begin{lemma}\label{lem:diagonal_outer_product_block_choi_matrix_inequality}
    For positive semidefinite matrices $\mathbf{M}_1,\dots,\mathbf{M}_K$ such that $\mathbf{M}_k\geq0$ and $\mathbf{M} = \sum_k\mathbf{M}_k$, it holds that
    \begin{equation}\label{eqn:diagonal_outer_product_block_choi_matrix_inequality}
    \begin{aligned}
        \left[\begin{array}{ccc}
            \mathbf{M}_1 &&  \\
            & \ddots &\\
            && \mathbf{M}_K
        \end{array}\right]
        \geq
        \left[\begin{array}{cc}
            \mathbf{M}_1 \\
            \vdots \\
            \mathbf{M}_K
        \end{array}\right]
        \mathbf{M}^+
        \left[\begin{array}{cc}
            \mathbf{M}_1 \\
            \vdots \\
            \mathbf{M}_K
        \end{array}\right]^\dag \;.
    \end{aligned}
    \end{equation}
\end{lemma}
\begin{proof}
    See Appendix~\ref{app:diagonal_outer_product_block_choi_matrix_inequality}.
\end{proof}

Lemma~\ref{lem:diagonal_outer_product_block_choi_matrix_inequality} therefore implies eqn.~\eqref{eqn:LX_inequality}:
\begin{equation}\label{eqn:LX_inequality2}
\begin{aligned}
    \mathbf{L}_\theta(\mathcal{M},\hat{\theta}) &= \Xi 
    \left[\begin{array}{ccc}
        \mathbf{M}_1 &&  \\
        & \ddots &\\
        && \mathbf{M}_K
    \end{array}\right] \Xi^\top 
    \geq \Xi
    \left[\begin{array}{cc}
            \mathbf{M}_1 \\
            \vdots \\
            \mathbf{M}_K
        \end{array}\right]
        \mathbf{M}^+
        \left[\begin{array}{cc}
            \mathbf{M}_1 \\
            \vdots \\
            \mathbf{M}_K
        \end{array}\right]^\dag \Xi^\top 
    = X_\theta(\mathcal{M},\hat{\theta}) \mathbf{M}^+ X_\theta(\mathcal{M},\hat{\theta})^\dag \;.
\end{aligned}
\end{equation}
since each matrix $\mathbf{M}_k$ is Hermitian, i.e. $\mathbf{M}_k = \mathbf{M}_k^\dag$.

Consider an optimal estimator $\hat{\theta}_\opt$ and an optimal sequential probe $\mathcal{M}_\opt$ with sequential probe operator $\mathbf{M}_\opt$. 
Then, for any sequential probe $\mathcal{M}$ and estimator $\hat{\theta}$, we have
\begin{equation}\label{eqn:optimization_sequential_bound1}
\begin{aligned}
    \mse_\theta(\mathcal{M},\hat{\theta}) &\geq \mse_\theta(\mathcal{M}_\opt,\hat{\theta}_\opt) 
    = \tr\Big( \Tilde{\mathbf{\Lambda}}_\theta \mathbf{L}_\theta(\mathcal{M}_\opt,\hat{\theta}_\opt) \Big) 
    \geq \min_\mathbf{L} \tr\Big( \Tilde{\mathbf{\Lambda}}_\theta \mathbf{L} \Big) \;,
\end{aligned}
\end{equation}
where the minimization is over matrices $\mathbf{L} = \sum_{i,j=1}^m |i\>\<j|\otimes \mathbf{X}_{i,j}$ satisfying the constraints
\begin{equation}
\begin{gathered}
    \mathbf{X}_{i,j} = \mathbf{X}_{j,i} \;,\quad \mathbf{X}_{i,j}=\mathbf{X}_{i,j}^\dag \;,\; \text{and} \\
        \mathbf{L} \geq X_\theta(\mathcal{M}_\opt,\hat{\theta}_\opt) \mathbf{M}_\opt^+ X_\theta(\mathcal{M}_\opt,\hat{\theta}_\opt)^\dag  \;.
\end{gathered}
\end{equation}
Note that the last constraint is obtained from eqn.~\eqref{eqn:LX_inequality}.

The minimization above is dependent on the optimal sequential probe $\mathcal{M}_\opt$ and optimal estimator $\hat{\theta}_\opt$.
To obtain a minimization that is independent of them, we replace $X(\mathcal{M}_\opt,\hat{\theta}_\opt)$ with a minimization
\begin{equation}\label{eqn:optimization_sequential_bound2}
\begin{aligned}
    \mse_\theta(\mathcal{M},\hat{\theta}) &\geq \min_{X,\mathbf{L}} \tr\Big( \Tilde{\mathbf{\Lambda}}_\theta \mathbf{L} \Big) \;,
\end{aligned}
\end{equation}
where the minimization is over matrices $\mathbf{L} = \sum_{i,j=1}^m |i\>\<j|\otimes \mathbf{X}_{i,j}$ and $X = \sum_{i=1}^m |i\>\otimes \mathbf{X}_{0,i}$ satisfying the constraints
\begin{equation}\label{eqn:sequential_sdp_constraints1}
\begin{gathered}
    \mathbf{X}_{i,j} = \mathbf{X}_{j,i} \;,\quad \mathbf{X}_{i,j}=\mathbf{X}_{i,j}^\dag \;, \\
        \mathbf{L} \geq X \mathbf{M}^+ X^\dag \;, \\
    \mathbf{M}\in\seq(\Bar{A}_T) \;,\quad
    X_j = X_j^\dag \;,\\
        \tr(\mathbf{\Lambda}_\theta^\top X_j) = 0
        \quad,\;\text{and}\;\quad
        \frac{\partial}{\partial\theta_i} \tr(\mathbf{\Lambda}_\theta^\top X_j) = \delta_{i,j} \;,
\end{gathered}
\end{equation}
where the constraints in the last line are due to the assumption of unbiased estimators in eqn.~\eqref{eqn:unbiased_estimator}.
Now, to obtain the SDP form in Theorem~\ref{thm:sequential_NH_bound}, we need the following lemma.

\begin{lemma}\label{lem:block_martix_psd_iff_condition}
    Consider a complex block matrix $Z$ over $\C^m\oplus\C^n$, given by
    \begin{equation}
        Z = \left[\begin{matrix}
            A & B^\dag \\
            B & C
        \end{matrix}\right]
    \end{equation}
    for $C\in\C^{n\times n}$, $B\in\C^{n\times m}$, and positive semidefinite $A\in\C^{m\times m}$.
    It holds that $Z\geq0$ if and only if both $C-BA^+B^\dag\geq0$ and $(I-AA^+)B^\dag=0$ hold, where $A^+$ is the pseudo-inverse of $A$.
\end{lemma}
\begin{proof}
    See Appendix~\ref{app:psd_block_matrix}.
\end{proof}

For a given strategy $\mathcal{M}$ with strategy operator $\mathbf{M}$ and estimator $\hat{\theta}$, consider a block matrix
\begin{equation}
    \mathbf{X} = \left[\begin{matrix}
        \mathbf{M} & X_\theta(\mathcal{M},\hat{\theta})^\dag \\
        X_\theta(\mathcal{M},\hat{\theta}) & \mathbf{L}(\mathcal{M},\hat{\theta})
    \end{matrix}\right], 
\end{equation}
and let $\Pi_\mathbf{M} = \mathbf{M}\mathbf{M}^+$ be the projector onto the eigenspace of $\mathbf{M}$ so that $\Pi_\mathbf{M}^\perp = I-\mathbf{M}\mathbf{M}^+$ is the projector onto the space orthogonal to the eigenspace of $\mathbf{M}$.
Since the support of $X_\theta(\mathcal{M},\hat{\theta})_i = \sum_{k=1}^K \xi_{i,k} \mathbf{M}_k$ is in the eigenspace of $\mathbf{M}$, it holds that
\begin{equation}
    (I-\mathbf{M}\mathbf{M}^+)X(\mathcal{M},\hat{\theta})^\dag=0 \;.
\end{equation}
By Lemma~\ref{lem:block_martix_psd_iff_condition}, we have $\mathbf{X}\geq0$, since both $\mathbf{L}(\mathcal{M},\hat{\theta}) - X(\mathcal{M},\hat{\theta}) \mathbf{M}^+ X(\mathcal{M},\hat{\theta})^\dag\geq0$ and $(I-\mathbf{M}\mathbf{M}^+)X(\mathcal{M},\hat{\theta})^\dag = 0$ hold.
Thus, the minimization on the right-hand side of eqn.~\eqref{eqn:optimization_sequential_bound2} over $X$ and $\mathbf{L}$ with constraints given in eqn.~\eqref{eqn:sequential_sdp_constraints1} can be formulated as an SDP $\sdp(\mathbf{\Lambda}_\theta)$ given by
\begin{equation}
\begin{gathered}
    \min_\mathbf{X} \tr\Big( \overline{\mathbf{\Lambda}}_\theta \mathbf{X} \Big) \\
    \text{s.t.}\\
        \mathbf{X} = \begin{bmatrix}
            \mathbf{M} & X^\dag \\
            X & \mathbf{L}
        \end{bmatrix}\geq0 \;,\quad
        \mathbf{M} \in \seq(\Bar{A}_T) \;,\\
        X = \sum_{i=1}^m |i\>\otimes \mathbf{X}_{0,i} \;,\quad
            \mathbf{X}_{0,i} = \mathbf{X}_{0,i}^\dag \;, \\
        \mathbf{L} = \sum_{i,j=1}^m |i\>\<j|\otimes \mathbf{X}_{i,j} \;,\quad
            \mathbf{X}_{i,j} = \mathbf{X}_{i,j}^\dag \;, \\
        \tr(\mathbf{\Lambda}_\theta^\top X_j) = 0
        \;,\quad\text{and}\;\quad
        \frac{\partial}{\partial\theta_i} \tr(\mathbf{\Lambda}_\theta^\top \mathbf{X}_{0,j}) = \delta_{i,j} \quad\forall i,j\in[1,m] \;,
\end{gathered}
\end{equation}
where $\overline{\mathbf{\Lambda}}_\theta = \begin{bmatrix}
    0 & 0 \\
    0 & I_m\otimes\mathbf{\Lambda}_\theta^\top
\end{bmatrix}$.

\subsection{Runtime complexity of computing $\sdp(\Lambda_\theta)$}\label{app:SDP_runtime_complexity}

The runtime of computing $\sdp(\mathbf{\Lambda}_\theta)$ depends on the dimension $D$ of the probe system $A\cong A_t\cong A_t'$, the total number of time steps $T$ involved in the protocol, and the number of parameters $m$.
Note that the size of matrix $\overline{\mathbf{\Lambda}}_\theta$ is $(1+m)D^{2T}$ where $D^{2T}$ is the size of $\mathbf{\Lambda}_\theta$. 
An SDP with objective function $\max_X\tr(AX)$ for matrices $A,X$ of size $n$ can be solved up to accuracy of $\varepsilon>0$ using various solver algorithms in $O(\sqrt{n}\log(\varepsilon^{-1}))$ steps (see e.g. \cite{vandenberghe1996semidefinite}, \cite[Section 11.5]{boyd2004convex}, \cite[Theorem 2.6.1]{gartner2012approximation}).
Thus, $\sdp(\mathbf{\Lambda}_\theta)$ can be solved up to $\varepsilon>0$ accuracy in $O(\sqrt{(1+m)}D^T\log(\varepsilon^{-1}))$ steps, which is polynomial in the probe system dimension $D$ and number of parameters $m$, but exponential in the number of time steps $T$.

\section{Proof of Lemma~\ref{lem:diagonal_outer_product_block_choi_matrix_inequality}}\label{app:diagonal_outer_product_block_choi_matrix_inequality}

In this appendix, we will prove Lemma~\ref{lem:diagonal_outer_product_block_choi_matrix_inequality}, i.e.~will show the block matrix inequality
\begin{equation}\label{eqn:diagonal_outer_product_block_choi_matrix_inequality1}
\begin{aligned}
    \left[\begin{array}{ccc}
        \mathbf{M}_1 &&  \\
        & \ddots &\\
        && \mathbf{M}_K
    \end{array}\right]
    \geq
    \left[\begin{array}{cc}
        \mathbf{M}_1 \\
        \vdots \\
        \mathbf{M}_K
    \end{array}\right]
    \mathbf{M}^+
    \left[\begin{array}{cc}
        \mathbf{M}_1 \\
        \vdots \\
        \mathbf{M}_K
    \end{array}\right]^\top \;
\end{aligned}
\end{equation}
for $\mathbf{M}_k \geq0$ and $\mathbf{M} = \sum_k\mathbf{M}_k$.
Let us denote the left-hand side and right-hand side of eqn.~\eqref{eqn:diagonal_outer_product_block_choi_matrix_inequality1} as $C$ and $D$, respectively.

First, consider block matrices 
\begin{equation}
    A=
    \left[\begin{array}{ccc}
        \sqrt{\mathbf{M}_1} &&  \\
        & \ddots &\\
        && \sqrt{\mathbf{M}_K}
    \end{array}\right]
    \quad\text{and}\quad
    B = 
    \begin{bmatrix}
        \sqrt{\mathbf{M}^+}\sqrt{\mathbf{M}_1} &\dots& \sqrt{\mathbf{M}^+}\sqrt{\mathbf{M}_K}
    \end{bmatrix},
\end{equation}
so that $C = A^\dag A$ and $D = A^\dag B^\dag BA$.
The matrix inequality in eqn.~\eqref{eqn:diagonal_outer_product_block_choi_matrix_inequality1} is equivalent to the matrix $C-D$ being positive semidefinite.
We can write this matrix as
\begin{equation}
\begin{aligned}
    C-D = A^\dag (I-B^\dag B) A \;.
\end{aligned}
\end{equation}

Now note that 
\begin{equation}
    BB^\dagger
    =\mathbf{M}^{+1/2}\Big(\sum_{k=1}^K \mathbf{M}_k\Big)\mathbf{M}^{+1/2}
    =\mathbf{M}^{+1/2}\mathbf{M}\mathbf{M}^{+1/2}
    =\Pi_{\supp(\mathbf{M})}
    \leq I \;.
\end{equation}
This implies that $||BB^\dagger|| \leq1$ for $||\cdot||$ being the operator norm.
Thus, for all normalized vectors $|v\>\in \R^K\otimes\Bar{A}_T$, it holds that
\begin{equation}
    \<v|B^\dag B|v\> \leq |\<v|v\>|^2 = 1 \;.
\end{equation}
This implies that, for all normalized vectors $|v\>$, we have
\begin{equation}
    0 \leq 1-\<v|B^\dag B|v\> = \<v|(I-B^\dag B)|v\> \;,
\end{equation}
which means that $I-B^\dag B \geq0$.
By conjugating both sides of the inequality $I-B^\dag B \geq0$ with $A^\dag$ we obtain
\begin{equation}
\begin{aligned}
    0\leq A^\dag (I-B^\dag B) A = C-D \;,
\end{aligned}
\end{equation}
which gives eqn.~\eqref{eqn:diagonal_outer_product_block_choi_matrix_inequality1}.

\section{Proof of Lemma~\ref{lem:block_martix_psd_iff_condition}}\label{app:psd_block_matrix}

In this appendix, we prove Lemma~\ref{lem:block_martix_psd_iff_condition}.

Consider a complex block matrix over $\C^m\oplus\C^n$:
\begin{equation}
    Z = \left[\begin{matrix}
        A & B^\dag \\
        B & C
    \end{matrix}\right]
\end{equation}
for $C\in\C^{n\times n}$, $B\in\C^{n\times m}$, and positive semidefinite $A\in\C^{m\times m}$.
We will show that $Z\geq0$ if and only if both $C-BA^+B^\dag\geq0$ and $(I-AA^+)B^\dag=0$ hold.

First, consider spectral decomposition $A = \sum_j a_j \ketbra{j}$ for $a_j\geq0$ so that $A^+ = \sum_j a_j^{-1}\ketbra{j}$ and $AA^+ = A^+A = \Pi_A$, where $\Pi_A = \sum_{j:a_j\neq0} \ketbra{j}$ is the projector onto the eigenspace of $A$.
Also let us denote $\Pi_A^\perp = I-\Pi_A = \sum_{j:a_j=0}\ketbra{j}$.
Now note that, for any $|u\>\in\C^m$ and $|v\>\in\C^n$, we have
\begin{equation}
\begin{aligned}
    \<u|\oplus\<v| Z |u\>\oplus|v\> &= \<u|A|u\> + \<v|B|u\> + \<u|B^\dag|v\> + \<v|C|v\> \\
    &= \<u|A|u\> + \<v|B(\Pi_A+\Pi_A^\perp)|u\> + \<u|(\Pi_A+\Pi_A^\perp)B^\dag|v\> + \<v|C|v\> \\
    &= \<u|A|u\> + (\<v|BA^+)A|u\> + \<v|B\Pi_A^\perp|u\> + \<u|A(A^+B^\dag|v\>) + \<u|\Pi_A^\perp B^\dag|v\> + \<v|C|v\> \\
    &= \<u|A|u\> + (\<v|BA^+)A|u\> + \<u|A(A^+B^\dag|v\>) \\
       &\quad  + (\<v|BA^+ + \<u|)A(A^+B^\dag|v\> + |u\>) - \underbrace{ (\<v|BA^+ + \<u|)A(A^+B^\dag|v\> + |u\>) }_{= \<v|BA^+B^\dag|v\> + (\<v|BA^+)A|u\> + \<u|A(A^+B^\dag|v\>) + \<u|A|u\> } \\
       &\quad + \<v|B\Pi_A^\perp|u\> + \<u|\Pi_A^\perp B^\dag|v\> + \<v|C|v\> \\
    &= \<v|C|v\> - \<u|BA^+B^\dag|u\> + (\<v|BA^+ + \<u|)A(A^+B^\dag|v\> + |u\>) + \<v|B\Pi_A^\perp|u\> + \<u|\Pi_A^\perp B^\dag|v\> \\
    &= \<v| (C-BA^+B^\dag) |v\> + (\<v|BA^+ + \<u|)A(A^+B^\dag|v\> + |u\>) + \<v|B\Pi_A^\perp|u\> + \<u|\Pi_A^\perp B^\dag|v\> \;.
\end{aligned}
\end{equation}

If $A\geq0$, it holds that
\begin{equation}
\begin{aligned}
    \<u|\oplus\<v| Z |u\>\oplus|v\> &= \<v| (C-BA^+B^\dag) |v\> + (\<v|BA^+ + \<u|)A(A^+B^\dag|v\> + |u\>) + \<v|B\Pi_A^\perp|u\> + \<u|\Pi_A^\perp B^\dag|v\> \\
    &\geq \<v| (C-BA^+B^\dag) |v\> + \<v|B\Pi_A^\perp|u\> + \<u|\Pi_A^\perp B^\dag|v\> \;.
\end{aligned}
\end{equation}
Therefore $Z\geq0$ if both $C-BA^+B^\dag \geq0$ and $\Pi_A^\perp B^\dag = (I-AA^+)B^\dag = 0$ hold.

Now we show that $Z\geq0$ implies that both $C-BA^+B^\dag \geq0$ and $(I-AA^+)B^\dag = 0$ must hold.
First we show that $Z\geq0$ implies $(I-AA^+)B^\dag = 0$.
Consider $|u\>\in\ker A$, $|v\>\in\C^n$, and $t\in\C$.
Since $A|u\>=0$, for all $t\in\C$ it holds that
\begin{equation}
\begin{aligned}
    0\leq (\<u|\oplus t^*\<v|) Z (|u\>\oplus t|v\>) &= t\<u|B^\dag|v\> + t^*\<v|B|u\> + |t|^2\<v|C|v\> \\
    &= 2\Re(t\<u|B^\dag|v\>) + |t|^2\<v|C|v\> \;.
\end{aligned}
\end{equation}
If $\<u|B^\dag|v\>\neq0$, we can always choose a $t\in\C$ such that $2\Re(t\<u|B^\dag|v\>) / |t|^2 < -\<v|C|v\>$, which violates the inequality above.
Therefore it holds that $\<u|B^\dag|v\>=0$ for all $|v\>\in\C^n$ and for all $|u\>\in\ker A$, which implies that $\ker A \subseteq \ker B$.
Since $\Pi_A^\perp = I-AA^+$ is the projector onto $\ker A$, we have $(I-AA^+)B^\dag=0$.

To show that $Z\geq0$ implies $C-BA^+B^\dag\geq0$, suppose that $\<v| (C-BA^+B^\dag) |v\> <0$ for some $|v\>\in\C^m$.
Then, if we pick $|u\> = -A^+B^\dag|v\>$, we have
\begin{equation}
\begin{aligned}
    0\leq \<u|\oplus\<v| Z |u\>\oplus|v\> &= \<v| (C-BA^+B^\dag) |v\> + (\<v|BA^+ + \<u|)A(A^+B^\dag|v\> + |u\>) + \<v|B\Pi_A^\perp|u\> + \<u|\Pi_A^\perp B^\dag|v\> \\
    &= \<v| (C-BA^+B^\dag) |v\> + 0 \\
    &<0 \;,
\end{aligned}
\end{equation}
since $\Pi_A^\perp A^+=0$, which gives a contradiction.
Therefore $C-BA^+B^\dag\geq0$ must hold.

\section{Sequential single-parameter local estimation achievability}\label{app:single_par_achievability}

In this appendix, we give a proof of Proposition~\ref{prop:single_par_achievability}, i.e.~show that, under the assumptions of Proposition~\ref{prop:single_par_achievability}, the SDP bound in Theorem~\ref{thm:sequential_NH_bound} is attainable for single-parameter local estimation.

\begin{proof}[Proof of Proposition~\ref{prop:single_par_achievability}]
    Consider a sequential sensing strategy described by positive semidefinite operators $\{\mathbf{M}_k\}_{k=1}^K$ satisfying $\mathbf{M}=\sum_{k=1}^K\mathbf{M}_k \in\seq(\Bar{A}_T)$ together with estimator $\{\hat{\theta}_k\}_{k=1}^K$, so that the probability of outcome $k$ is $p_k(\vartheta) = \tr(\mathbf{\Lambda}_\vartheta^\top\mathbf{M}_k)$.
    The local-unbiasedness conditions at $\theta$ are given by
    \begin{equation}
        \tr(\mathbf{\Lambda}_\theta^\top X)=0
        \quad\text{and}\quad
        \tr(\dot{\mathbf{\Lambda}}_\theta^\top X)=1\;,
    \label{eqn:single_par_centered_X_constraints}
    \end{equation}
    where $\dot{\mathbf{\Lambda}}_\theta = \frac{d}{d\vartheta} \mathbf{\Lambda}_\vartheta|_{\vartheta=\theta}$.
    Thus, the SDP in Theorem~\ref{thm:sequential_NH_bound} for $m=1$ becomes
    \begin{equation}
    \begin{gathered}
        \sdp(\mathbf{\Lambda}_\theta)
        =
        \min_{\mathbf{M},X,\mathbf{L}}
        \tr(\mathbf{\Lambda}_\theta^\top\mathbf{L})\\
        \text{s.t.}\\
        \begin{bmatrix}
            \mathbf{M}&X\\
            X&\mathbf{L}
        \end{bmatrix}\geq0\;,
        \qquad
        \mathbf{M}\in\seq(\Bar{A}_T)\;,\\
        X=X^\dag\;,
        \qquad
        \tr(\mathbf{\Lambda}_\theta^\top X)=0\;,
        \qquad
        \tr(\dot{\mathbf{\Lambda}}_\theta^\top X)=1\;,
    \end{gathered}
    \label{eqn:single_par_centered_sequential_optimization}
    \end{equation}
    where $\mathbf{L},X \in\L(\Bar{A}_T)$ (as opposed to $\mathbf{L}\in\R^{m\times m}\otimes\L(\Bar{A}_T)$ and $X\in\R^m\otimes\L(\Bar{A}_T)$ in the $m$-parameter case).

    We first solve this SDP for a fixed $\mathbf{M}\in\seq(\Bar{A}_T)$.
    By the generalized Schur-complement condition in Lemma~\ref{lem:block_martix_psd_iff_condition}, block positivity in eqn.~\eqref{eqn:single_par_centered_sequential_optimization} is equivalent to
    \begin{equation}
    \begin{gathered}
        \mathbf{L}-X\mathbf{M}^+X\geq0\quad\text{and}\quad
        (I-\mathbf{M}\mathbf{M}^+)X=0\;.
    \end{gathered}
    \label{eqn:single_par_generalized_schur}
    \end{equation}
    Since $X=X^\dag$, the second condition also gives
    $X(I-\mathbf{M}\mathbf{M}^+)=0$.
    Let $\Pi_\mathbf{M}:=\mathbf{M}\mathbf{M}^+$ denote the projector onto $\supp(\mathbf{M})$.
    We therefore have
    \begin{equation}
        X=\Pi_\mathbf{M}X\Pi_\mathbf{M}\;,
    \label{eqn:single_par_X_support}
    \end{equation}
    and every feasible lower-right block can be written as $\mathbf{L} = X\mathbf{M}^+X+\Delta$ for some $\Delta\geq0$.
    Since $\mathbf{\Lambda}_\theta^\top\geq0$, it follows that
    \begin{equation}
    \begin{aligned}
        \tr(\mathbf{\Lambda}_\theta^\top\mathbf{L})
        &=
        \tr(\mathbf{\Lambda}_\theta^\top X\mathbf{M}^+X)
        +
        \tr(\mathbf{\Lambda}_\theta^\top\Delta)
        \geq
        \tr(\mathbf{\Lambda}_\theta^\top X\mathbf{M}^+X)\;.
    \end{aligned}
    \end{equation}
    Hence, for every pair of fixed feasible operators $\mathbf{M},X$, we can saturate the inequality in eqn.~\eqref{eqn:single_par_generalized_schur} by taking
    \begin{equation}
        \mathbf{L}=X\mathbf{M}^+X \;.
    \label{eqn:single_par_optimal_L}
    \end{equation}

    For Hermitian operator $Y:= \sqrt{\mathbf{M}^+} X \sqrt{\mathbf{M}^+}$, it holds that
    \begin{equation}
        Y=\Pi_\mathbf{M}Y\Pi_\mathbf{M}
        \quad\text{and}\quad
        X=\mathbf{M}^{1/2}Y\mathbf{M}^{1/2} \;
    \label{eqn:single_par_X_Y_relation}
    \end{equation}
    by eqn.~\eqref{eqn:single_par_X_support}.
    It follows that
    \begin{equation}
    \begin{aligned}
        X\mathbf{M}^+X
        &=
        \mathbf{M}^{1/2}
        Y\Pi_\mathbf{M}Y
        \mathbf{M}^{1/2}
        =
        \mathbf{M}^{1/2}Y^2\mathbf{M}^{1/2}\;.
    \end{aligned}
    \label{eqn:single_par_L_Y_relation}
    \end{equation}
    For this fixed $\mathbf{M}$, define the output state $\rho_\vartheta^\mathbf{M} := \mathbf{M}^{1/2} \mathbf{\Lambda}_\vartheta^\top \mathbf{M}^{1/2}$.
    Positivity of $\mathbf{M}$ and
    $\mathbf{\Lambda}_\vartheta^\top$ gives
    $\rho_\vartheta^\mathbf{M}\geq0$, while tester normalization gives $\tr(\rho_\vartheta^\mathbf{M})=1$.
    Thus $\rho_\vartheta^\mathbf{M}$ is a density operator.
    Now consider
    \begin{equation}
        \dot{\rho}_\theta^\mathbf{M}
        :=
        \left.
        \frac{d}{d\vartheta}
        \rho_\vartheta^\mathbf{M}
        \right|_{\vartheta=\theta}\;,
    \end{equation}
    so that $\tr(\dot{\rho}_\theta^\mathbf{M})=0$.
    For a fixed $\mathbf{M}$, therefore $\sdp(\mathbf{\Lambda}_\theta)$ becomes
    \begin{equation}
    \begin{gathered}
        C(\mathbf{M})
        =
        \min_Y
        \tr(\rho_\theta^\mathbf{M}Y^2)\\
        \text{s.t.}\\
        Y=Y^\dag\;,
        \qquad
        Y=\Pi_\mathbf{M}Y\Pi_\mathbf{M}\;,\\
        \tr(\rho_\theta^\mathbf{M}Y)=0\;,
        \qquad
        \tr(\dot{\rho}_\theta^\mathbf{M}Y)=1\;.
    \end{gathered}
    \label{eqn:single_par_fixed_tester_optimization}
    \end{equation}

    Let $S_\mathbf{M}=\Pi_\mathbf{M}S_\mathbf{M}\Pi_\mathbf{M}$ be an SLD of $\rho_\theta^\mathbf{M}$, which is a Hermitian operator on $\supp(\mathbf{M})$, such that
    \begin{equation}
        \dot{\rho}_\theta^\mathbf{M}
        =
        \frac{1}{2}
        \left(
            \rho_\theta^\mathbf{M}S_\mathbf{M}
            +
            S_\mathbf{M}\rho_\theta^\mathbf{M}
        \right)\;,
    \label{eqn:single_par_SLD_definition}
    \end{equation}
    and its SLD quantum Fisher information is
    \begin{equation}
        \mathcal{J}_\text{Q}(\rho_\theta^\mathbf{M})
        :=
        \tr\!\left(
            \rho_\theta^\mathbf{M}S_\mathbf{M}^2
        \right)\;.
    \label{eqn:single_par_SLD_QFI_definition}
    \end{equation}

    For any $Y$ satisfying the local-unbiasedness condition in
    eqn.~\eqref{eqn:single_par_fixed_tester_optimization}, we use
    eqn.~\eqref{eqn:single_par_SLD_definition} to obtain
    \begin{equation}
    \begin{aligned}
        1
        &=
        \tr(\dot{\rho}_\theta^\mathbf{M}Y)
        =
        \frac{1}{2}
        \tr\!\left[
            (\rho_\theta^\mathbf{M}S_\mathbf{M}
            +S_\mathbf{M}\rho_\theta^\mathbf{M})Y
        \right]
        =
        \operatorname{Re}
        \tr(\rho_\theta^\mathbf{M}S_\mathbf{M}Y)
        \leq
        \sqrt{
            \tr(\rho_\theta^\mathbf{M}S_\mathbf{M}^2)
            \tr(\rho_\theta^\mathbf{M}Y^2)
        }\;,
    \end{aligned}
    \label{eqn:single_par_SLD_Cauchy_Schwarz}
    \end{equation}
    where the last line follows by applying the Cauchy--Schwarz inequality 
    \begin{equation}
        |\tr\big( \rho_\theta^\mathbf{M} S_\mathbf{M} Y \big)|^2 = \Big| \Big\< Y\sqrt{\rho_\theta^\mathbf{M}},S_\mathbf{M}\sqrt{\rho_\theta^\mathbf{M}} \Big\>_\mathrm{HS} \Big|^2
        \leq \tr(\rho_\theta^\mathbf{M} Y^2) \; \tr(\rho_\theta^\mathbf{M}S_\mathbf{M}^2)
    \end{equation}
    for the Hilbert--Schmidt inner product $\<\cdot,\cdot\>_\mathrm{HS}$.
    Therefore we obtain an inequality
    \begin{equation}
        \tr(\rho_\theta^\mathbf{M}Y^2)
        \geq
        \frac{1}{\mathcal{J}_\text{Q}(\rho_\theta^\mathbf{M})}\;.
    \label{eqn:single_par_fixed_tester_lower_bound}
    \end{equation}

    The equality in eqn.~\eqref{eqn:single_par_fixed_tester_lower_bound} is attained by $Y_\opt := \frac{S_\mathbf{M}}{\mathcal{J}_\text{Q}(\rho_\theta^\mathbf{M})}$ since
    \begin{equation}
        \tr(\rho_\theta^\mathbf{M}Y_\opt^2)
        = \frac{ \tr(\rho_\theta^\mathbf{M} S_\mathbf{M}^2) }{ \mathcal{J}_\text{Q}(\rho_\theta^\mathbf{M})^2 }
        = \frac{1}{\mathcal{J}_\text{Q}(\rho_\theta^\mathbf{M})} \;
    \end{equation}
    by eqn.~\eqref{eqn:single_par_SLD_QFI_definition}.
    We will now show that $Y_\opt$ is a feasible solution to the optimization in eqn.~\eqref{eqn:single_par_fixed_tester_optimization}.
    First note that $\tr(\rho_\theta^\mathbf{M})=1$ since it is a density operator, so that $\tr(\dot{\rho}_\theta^\mathbf{M}) = 0$.
    By taking the trace of eqn.~\eqref{eqn:single_par_SLD_definition}, we have
    \begin{equation}
    \begin{aligned}
        0=\tr(\dot{\rho}_\theta^\mathbf{M}) = \tr(\rho_\theta^\mathbf{M}S_\mathbf{M}) \;.
    \end{aligned}
    \end{equation}
    Hence we obtain the first local-unbiasedness condition
    \begin{equation}
        \tr(\rho_\theta^\mathbf{M}Y_\opt) = \frac{ \tr(\rho_\theta^\mathbf{M} S_\mathbf{M}) }{ \mathcal{J}_\text{Q}(\rho_\theta^\mathbf{M}) } = 0 \;.
    \end{equation}
    Also consider
    \begin{equation}
    \begin{aligned}
        \tr(\dot{\rho}_\theta^\mathbf{M}S_\mathbf{M})
        &=
        \frac{1}{2}
        \tr\!\left[
            (\rho_\theta^\mathbf{M}S_\mathbf{M}
            +S_\mathbf{M}\rho_\theta^\mathbf{M})S_\mathbf{M}
        \right] =
        \tr(\rho_\theta^\mathbf{M}S_\mathbf{M}^2) =
        \mathcal{J}_\text{Q}(\rho_\theta^\mathbf{M})\;.
    \end{aligned}
    \end{equation}
    Hence we obtain the second local-unbiasedness condition:
    \begin{equation}
        \tr(\dot{\rho}_\theta^\mathbf{M}Y_\opt) = \frac{ \tr(\dot{\rho}_\theta^\mathbf{M} S_\mathbf{M}) }{ \mathcal{J}_\text{Q}(\rho_\theta^\mathbf{M}) } = 1 \;,
    \end{equation}
    since $Y_\opt=\frac{S_{\mathbf{M}}}{\mathcal{J}_\text{Q}(\rho_\theta^{\mathbf{M}})}$.
    
    We have therefore shown that $Y_\opt$ is a minimizer of the $\mathbf{M}$-fixed optimization in eqn.~\eqref{eqn:single_par_fixed_tester_optimization}: 
    \begin{equation}
        C(\mathbf{M})
        =
        \frac{1}{\mathcal{J}_\text{Q}(\rho_\theta^\mathbf{M})} \;
    \label{eqn:single_par_fixed_tester_value}
    \end{equation}
    whenever $\mathcal{J}_\text{Q}(\rho_\theta^\mathbf{M})>0$.
    If $\mathcal{J}_\text{Q}(\rho_\theta^\mathbf{M})=0$, then eqn.~\eqref{eqn:single_par_SLD_Cauchy_Schwarz} shows that the derivative constraint $\tr(\dot{\rho}_\theta^\mathbf{M}Y)=1$ is infeasible.
    Optimizing over all sequential probes $\mathbf{M}\in\seq(\Bar{A}_T)$ therefore gives
    \begin{equation}
    \begin{aligned}
        \sdp(\mathbf{\Lambda}_\theta)
        &= \min_{\mathbf{M}\in\seq(\Bar{A}_T)} C(\mathbf{M}) 
        = \min_{\mathbf{M}\in\seq(\Bar{A}_T)} \frac{1}{\mathcal{J}_\text{Q}(\rho_\theta^\mathbf{M})}
        = \frac{1}{ \max_{\mathbf{M}\in\seq(\Bar{A}_T)} \mathcal{J}_\text{Q}(\rho_\theta^\mathbf{M})}
        = \frac{1}{\mathcal{J}_{\rm seq}(\mathbf{\Lambda}_\theta)}\;.
    \end{aligned}
    \label{eqn:single_par_SDP_QFI_equality}
    \end{equation}

    It remains to show that this value is attained by a physical sequential sensing strategy whose estimator obeys the local-unbiasedness conditions.
    Let $\mathbf{M}_\opt$ attain the maximum in the definition of $\mathcal{J}_{\rm seq}(\mathbf{\Lambda}_\theta)$, and let
    \begin{equation}
        S_{\mathbf{M}_\opt} = \sum_{k=1}^K\lambda_kP_k
    \label{eqn:single_par_optimal_SLD_spectral_decomposition}
    \end{equation}
    be its spectral decomposition on $\supp(\mathbf{M}_\opt)$ for $P_1,\dots,P_K$ satisfying $\sum_{k=1}^KP_k=\Pi_{\mathbf{M}_\opt}$, where $\Pi_{\mathbf{M}_\opt}$ is the projector onto the support of $\mathbf{M}_\opt$.
    Define tester outcome operator corresponding to outcome $k$ as
    \begin{equation}
        \mathbf{M}_{\opt,k}
        :=
        \mathbf{M}_\opt^{1/2}
        P_k
        \mathbf{M}_\opt^{1/2}\;,
    \label{eqn:single_par_optimal_tester_outcomes}
    \end{equation}
    which satisfies $\mathbf{M}_{\opt,k}\geq0$ and $\sum_k\mathbf{M}_{\opt,k}=\mathbf{M}_\opt$.
    Note that $\{P_k\}_k$ is the final projective measurement on the output state $\rho_\theta^{\mathbf{M}_\opt})$ (see the tester decomposition in Appendix~\ref{app:tester_decomposition}).

    Now, consider
    \begin{equation}
    \begin{aligned}
        \xi_{\opt,k}
        &:= \frac{\lambda_k} {\mathcal{J}_{\rm seq}(\mathbf{\Lambda}_\theta)} \quad\text{and}\quad
        \hat{\theta}_{\opt,k} := \theta+\xi_{\opt,k}  \;,
    \end{aligned}
    \label{eqn:single_par_optimal_estimator_values}
    \end{equation}
    so that the corresponding optimal matrices are given by
    \begin{equation}
    \begin{aligned}
        X_\opt
        &=
        \sum_k\xi_{\opt,k}\mathbf{M}_{\opt,k}
        =
        \mathbf{M}_\opt^{1/2}
        \frac{S_{\mathbf{M}_\opt}}
        {\mathcal{J}_{\rm seq}(\mathbf{\Lambda}_\theta)}
        \mathbf{M}_\opt^{1/2}\;,\\
        \mathbf{L}_\opt
        &=
        \sum_k\xi_{\opt,k}^2\mathbf{M}_{\opt,k}
        =
        \mathbf{M}_\opt^{1/2}
        \frac{S_{\mathbf{M}_\opt}^2}
        {\mathcal{J}_{\rm seq}(\mathbf{\Lambda}_\theta)^2}
        \mathbf{M}_\opt^{1/2}\;,
    \end{aligned}
    \label{eqn:single_par_optimal_X_L}
    \end{equation}
    by combining eqns.~\eqref{eqn:single_par_optimal_SLD_spectral_decomposition}, \eqref{eqn:single_par_optimal_tester_outcomes},\eqref{eqn:single_par_optimal_estimator_values}.
    In particular,
    \begin{equation}
    \begin{aligned}
        \begin{bmatrix}
            \mathbf{M}_\opt&X_\opt\\
            X_\opt&\mathbf{L}_\opt
        \end{bmatrix}
        &=
        \sum_k
        \begin{bmatrix}
            1\\
            \xi_{\opt,k}
        \end{bmatrix}
        \begin{bmatrix}
            1&\xi_{\opt,k}
        \end{bmatrix}
        \otimes\mathbf{M}_{\opt,k} \geq0\;,
    \end{aligned}
    \end{equation}
    so these operators give a feasible SDP point.

    The outcome probabilities of such a strategy are
    \begin{equation}
    \begin{aligned}
        p_k(\vartheta)
        = \tr(\mathbf{\Lambda}_\vartheta^\top \mathbf{M}_{\opt,k})
        = \tr(\rho_\vartheta^{\mathbf{M}_\opt}P_k)\;.
    \end{aligned}
    \end{equation}
    Using
    $\tr(\rho_\theta^{\mathbf{M}_\opt}S_{\mathbf{M}_\opt})=0$,
    the estimator mean at the local point is
    \begin{equation}
    \begin{aligned}
        \sum_k\hat{\theta}_{\opt,k}p_k(\theta)
        &=
        \theta
        +
        \frac{
            \tr(\rho_\theta^{\mathbf{M}_\opt}
                S_{\mathbf{M}_\opt})
        }{
            \mathcal{J}_{\rm seq}(\mathbf{\Lambda}_\theta)
        } =
        \theta\;.
    \end{aligned}
    \label{eqn:single_par_optimal_estimator_mean}
    \end{equation}
    Its derivative is
    \begin{equation}
    \begin{aligned}
        \left.
        \frac{d}{d\vartheta}
        \sum_k\hat{\theta}_{\opt,k}p_k(\vartheta)
        \right|_{\vartheta=\theta}
        &=
        \theta\tr(\dot{\rho}_\theta^{\mathbf{M}_\opt})
        +
        \frac{
            \tr(\dot{\rho}_\theta^{\mathbf{M}_\opt}
                S_{\mathbf{M}_\opt})
        }{
            \mathcal{J}_{\rm seq}(\mathbf{\Lambda}_\theta)
        } 
        =0+1 =1\;.
    \end{aligned}
    \label{eqn:single_par_optimal_estimator_derivative}
    \end{equation}
    Thus the estimator is locally unbiased at $\theta$.
    So, the MSE of such a strategy described by $\mathbf{M}_\opt$ and $\hat{\theta}_\opt$ is
    \begin{equation}
    \begin{aligned}
        \mse_\theta(\mathbf{M}_\opt,\hat{\theta}_\opt)
        = \sum_k\xi_{\opt,k}^2p_k(\theta)
        = \frac{ \tr(\rho_\theta^{\mathbf{M}_\opt} S_{\mathbf{M}_\opt}^2) }{ \mathcal{J}_{\rm seq}(\mathbf{\Lambda}_\theta)^2 }
        = \frac{1}{
            \mathcal{J}_{\rm seq}(\mathbf{\Lambda}_\theta)
        } 
        = \sdp(\mathbf{\Lambda}_\theta)\;.
    \end{aligned}
    \end{equation}
    Hence the single-parameter SDP bound is locally attainable.
\end{proof}

\section{Additional examples}\label{app:additional_examples}

In this appendix, we present additional examples of sequential sensing problems and apply to them the SDP bound given in Theorem~\ref{thm:sequential_NH_bound}.

\subsection{Sequential single-qubit Pauli-$Z$ rotation angle estimation}\label{sec:z_rotation_example}

In this section, we go through an example of estimating the angle of $T$ copies of the single-qubit Pauli-$Z$ rotation channel
\begin{equation}
    \mathcal{R}_{z,\theta}(\cdot) = e^{-\iota\frac{\theta}{2}\sigma_z} \cdot e^{\iota\frac{\theta}{2}\sigma_z} 
\end{equation}
and demonstrate how one can obtain the explicit circuit form of a sequential probe $\mathbf{M}\in\seq(\Bar{A}_T)$ obtained from the SDP in Theorem~\ref{thm:sequential_NH_bound}, using the decomposition method for a general tester in Appendix~\ref{app:tester_decomposition}.

Let us first consider the case of estimating angle $\theta$ using $T=2$ copies of $\mathcal{R}_{z,\theta}$.
The optimizer $\mathbf{X}$ of semidefinite program $\sdp(\Bar{A}_2)$ in Theorem~\ref{thm:sequential_NH_bound} gives an optimal estimator matrix $\mathbf{X}_{0,1}$ and sequential probe $\mathbf{M}$:
\begin{equation}
\begin{gathered}
    \mathbf{X}_{0,1} = \frac{\kketbra{M_1} -\kketbra{M_2}}{4}  
    \quad\text{and}\quad
    \mathbf{M} = \sum_{k=1}^{12} \mathbf{M}_k
\end{gathered}
\end{equation}
for $|M_1\rr = \frac{|0000\>-|1111\>}{\sqrt{2}}$ and $|M_2\rr = \frac{|0000\>+|1111\>}{\sqrt{2}}$ and positive semidefinite operators $\{\mathbf{M}_k\}_{k=1}^{12}$ where $\mathbf{M}_1 = \kketbra{M_1}/2$, $\mathbf{M}_2=\kketbra{M_2}/2$, and the other $\mathbf{M}_k$'s are given in eqn.~\eqref{eqn:M_matrix_2_step_z_rotation} of Appendix~\ref{app:tester_isometry_sequence_decomposition}.

We find a circuit decomposition of $\mathbf{M}=\sum_{k=1}^{12} \mathbf{M}_k$ that can be illustrated as
\begin{equation}
    \mathbf{M} = \sum_{k=1}^{12}
    \begin{quantikz}[classical gap=0.04cm]
        \lstick[2]{$|\psi\>$} & \wire[r][1]["{C_1}"{above,pos=0}]{q} & & \gate[2,style={fill=blue!20}]{V} & \wire[r][1]["{C_2}"{above,pos=0}]{q} & & \gate[2,style={fill=blue!20}]{M_k'} \wireoverride{q} \\
        & \wire[l][1]["{A_0}"{above,pos=0}]{q} \setwiretype{n} & \wire[r][1]["{A_1}"{above,pos=0}]{q} & & \wire[l][1]["{A_1'}"{above,pos=0}]{q} & \wire[r][1]["{A_2}"{above,pos=0}]{q} & 
    \end{quantikz},
\end{equation}
where $C_1$ is a single-qubit ancilla register, $C_2$ is a 3-qubit ancilla register, $|\psi\>=(|00\>+|11\>)/\sqrt{2}$ is a maximally entangled state over $C_1\otimes A_0$, $V$ is an isometry, and $\{M_k'\}_{k=1}^{12}$ is a POVM over $C_2\otimes A_2$.
The isometry $V$ can be described by the circuit
\begin{equation}\label{eqn:isometry_noiseless_z_rot}
    V = 
    \begin{quantikz}
        & \gate[1,style={blue!05,fill=blue!05}]{|0\>} \setwiretype{n} \gategroup[4,steps=4,style={dashed, fill=blue!05, inner xsep=2pt},background,label style={label position=above,anchor=north,yshift=0.3cm}]{{\sc $V$}} & \gate[]{H} \setwiretype{q} & \gate[]{H} & \ctrl{3} & \rstick[3]{$C_2$} \\
        \lstick[1]{$C_1$} & & & \ctrl{-1} & & \\
        \lstick[1]{$A_1$} & & \ctrl{-2} & & & \\
        & \gate[1,style={blue!05,fill=blue!05}]{|0\>} \setwiretype{n} & \setwiretype{q} & & \targ{} & \rstick[1]{$A_1'$}
    \end{quantikz} \;,
\end{equation}
consisting of controlled Hadamard $H$ gates and a CNOT gate.
The POVM elements $\{M_k'\}_{k=1}^{12}$ are given in eqn.~\eqref{eqn:2step_z_rotation_optimal_POVM} of Appendix~\ref{app:tester_isometry_sequence_decomposition}.
However, the only outcomes with non-zero probability are outcomes 1 and 2, with POVM elements 
\begin{equation}
\begin{gathered}
    M_1' = \frac{\kketbra{M_1}}{2} 
    \quad\text{and}\quad
    M_2' = \frac{\kketbra{M_2}}{2}
    \;.
\end{gathered}
\end{equation}
This is because, after the second copy of the rotation channel, the quantum state over $C_2\otimes A_2$ is given by
\begin{equation}
    |\psi_2\> = \frac{e^{-\iota\theta}|0000\>+e^{\iota\theta}|1111\>}{\sqrt{2}} \;.
\end{equation}
Thus we have outcome probability distribution
\begin{equation}
    p(k|\mathbf{M},\mathbf{\Lambda}_\theta) =
    \begin{cases}
        (\sin\theta)^2 \;&,\;\text{if }k=1 \\
        (\cos\theta)^2 \;&,\;\text{if }k=2 \\
        0 \;&,\;\text{otherwise}
    \end{cases}
\end{equation}
which gives an MSE of
\begin{equation}
\begin{aligned}
    \mse_\theta(\mathcal{M},\hat{\theta}) 
    &= \frac{p(1|\mathbf{M},\mathbf{\Lambda}_\theta)+p(2|\mathbf{M},\mathbf{\Lambda}_\theta)}{4}
    =\frac{1}{4} \;.
\end{aligned}
\end{equation}

The protocol described above demonstrates how one can derive an explicit protocol from the SDP optimizer.
However, an ancilla is not required for this noiseless example to obtain such MSE.
Preparing the probe in $|+\>=(|0\>+|1\>)/\sqrt{2}$ and applying $R_{z,\theta}$ successively $T$ times gives
\begin{equation}
    |\psi_\theta^{(T)}\>
    =
    \frac{
        e^{-\iota T\theta/2}|0\>
        +
        e^{\iota T\theta/2}|1\>
    }{\sqrt{2}}.
\end{equation}
Fix a point $\theta_0$ and measure in the basis
\begin{equation}
\begin{aligned}
    |\pm_{\theta_0}\>
        :=
        \frac{|0\>\pm\iota e^{\iota T\theta_0}|1\>}
        {\sqrt{2}}.
\end{aligned}
\end{equation}
The outcome probabilities are $p_\pm(\theta)=[1\pm\sin(T(\theta-\theta_0))]/2$.
The fixed estimator values $\hat{\theta}_\pm=\theta_0\pm1/T$ are locally unbiased at $\theta=\theta_0$ and give $\mse_{\theta_0}=1/T^2$.
The ancillas above arise from the chosen tester decomposition and are not necessary to attain this bound.

\subsection{Noisy Pauli-$Z$ rotations}\label{sec:indep_noise_zrot}

Here, we consider a noisy sequential estimation of $T$ copies of Pauli-$Z$ rotation channels.
We consider an independent Pauli noise model where a Pauli noise channel $\mathcal{N}_p$ is applied after each rotation channel $\mathcal{R}_{z,\theta}$, where $p_x,p_y,p_z$ are the probabilities of Pauli $\sigma_x,\sigma_y,\sigma_z$ errors, respectively and $p_0=1-(p_x+p_y+p_z)$ the probability of no errors $\sigma_0=I$.
Then the Pauli noise channel is given by 
\begin{equation}\label{eqn:pauli_noise_channel}
    \mathcal{N}_p(\cdot) = \sum_{j\in\{0,x,y,z\}} p_j \sigma_j\cdot \sigma_j \;.
\end{equation}

\begin{figure}
    \centering
    \includegraphics[width=0.55\columnwidth]{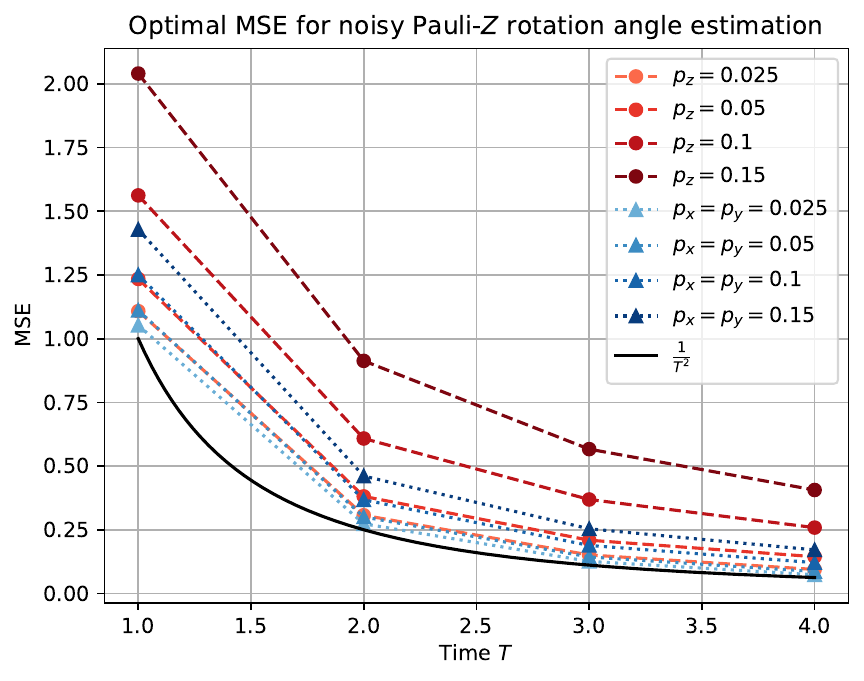}
    \caption{Mean-squared error (MSE) lower bound for sequential Pauli-$Z$ rotation angle estimation up to $T=4$ steps under independent Pauli noise.
    Red dashed lines correspond to the MSE for Pauli noise with only non-zero probability $p_z$ of $\sigma_z$ error.
    Blue dotted lines correspond to the MSE for Pauli noise with equal probability $p_x=p_y$ of either $\sigma_x$ or $\sigma_y$ errors, and zero probability of $\sigma_z$ error.
    If the Pauli noise has either only non-zero probability of $\sigma_x$ or $\sigma_y$, then the Heisenberg scaling (black line) is recovered, in agreement with the HNLS condition~\cite{zhou2018achieving}.
    }
    \label{fig:MSE_indep_noise_Z_rot}
\end{figure}

For up to $T=4$ copies of noisy $\mathcal{R}_{z,\theta}$ with non-zero $p_z$ error rates and non-zero $p_x=p_y$ error rates ranging from $.025$ to $.15$, the MSE obtained from the SDP in Theorem~\ref{thm:sequential_NH_bound} is plotted in Fig.~\ref{fig:MSE_indep_noise_Z_rot}.

For $T=2$ copies of noisy $\mathcal{R}_{z,\theta}$, the SDP program in Theorem~\ref{thm:sequential_NH_bound} gives an optimal sequential probe $\mathbf{M}$:
\begin{equation}
    \begin{quantikz}[classical gap=0.05cm]
        \;\; \setwiretype{n} \bshade{2}{8}{$\mathbf{M}$} & \lstick[2]{$|\psi\>$} & \setwiretype{q} & & \gate[2,style={fill=blue!20}]{V} & & & \gate[2,style={fill=blue!20}]{\mathcal{M}}  \\
        \setwiretype{n} & & \gate[1,style={fill=green!20}]{R_z(\theta)} \gshade{1}{2}{} \setwiretype{q} & \gate[style={fill=red!20}]{\mathcal{N}_p} & & \gate[1,style={fill=green!20}]{R_z(\theta)} \gshade{1}{2}{} & \gate[style={fill=red!20}]{\mathcal{N}_p} & 
    \end{quantikz} \;.
\end{equation}
For non-zero $p_x=p_y$ error rate and non-zero $p_z$ error rate, the isometry $V$ is identical to the one for the noiseless Pauli-$Z$ rotation estimation given in eqn.~\eqref{eqn:isometry_noiseless_z_rot}.
On the other hand, if only either $p_x$ error rate or $p_y$ error rate is non-zero, in which the Heisenberg scaling $\frac{1}{T^2}$ is achieved (due to the HNLS condition~\cite{zhou2018achieving}), the isometry $V$ is given by
\begin{equation}
    V = 
    \begin{quantikz}
        & \gate[1,style={blue!05,fill=blue!05}]{|0\>} \setwiretype{n} \gategroup[4,steps=3,style={dashed, fill=blue!05, inner xsep=2pt},background,label style={label position=above,anchor=north,yshift=0.3cm}]{{\sc $V$}} & \setwiretype{q} & \targ{} & \rstick[3]{$C_2$} \\
        \lstick[1]{$C_1$} & & & & \\
        \lstick[1]{$A_1$} & & \ctrl{1} & \ctrl{-2} & \\
        & \gate[1,style={blue!05,fill=blue!05}]{|0\>} \setwiretype{n} & \targ{} \setwiretype{q} & & \rstick[1]{$A_1'$}
    \end{quantikz} \,,
\end{equation}
where the ancilla system $C_2$ after isometry $V$ is composed of three qubits, as opposed to 1-qubit ancilla $C_1$ initially.

\subsection{Pauli $X$ and $Y$ rotations}

\begin{figure}
    \centering
    \includegraphics[width=0.55\columnwidth]{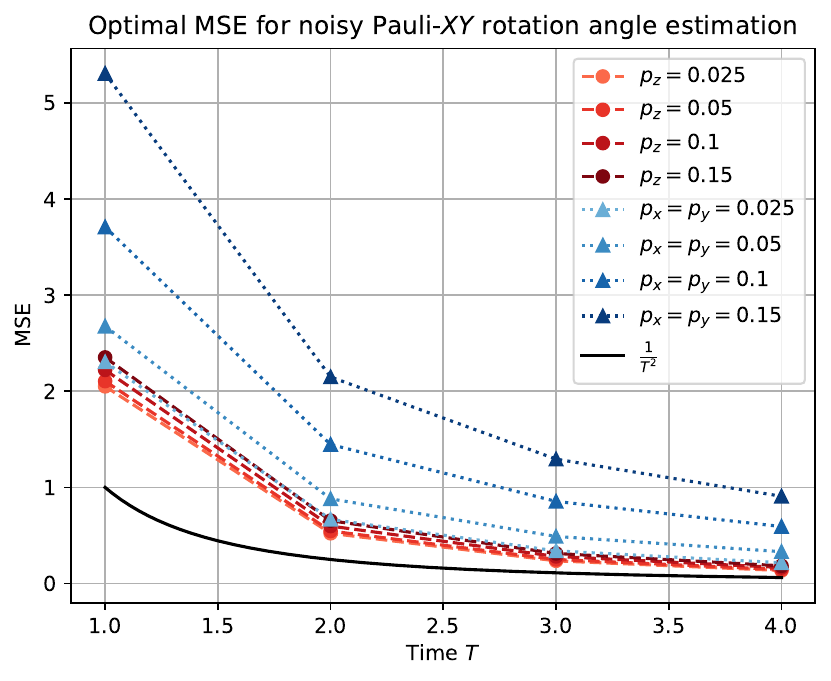}
    \caption{Mean-squared error (MSE) lower bound for noisy sequential Pauli $X$ and $Y$ rotation angle estimation up to $T=4$ steps under independent Pauli noise.
    Red dashed lines correspond to the MSE for Pauli noise with only non-zero probability $p_z$ of $\sigma_z$ error.
    Blue dotted lines correspond to the MSE for Pauli noise with equal probability $p_x=p_y$ of either $\sigma_x$ or $\sigma_y$ errors, and zero probability of $\sigma_z$ error.
    }
    \label{fig:MSE_noisy_xy_rot}
\end{figure}

Here we consider two-parameter sequential estimation of a 
quantum channel consisting of a Pauli-$Y$ rotation followed by a Pauli-$X$ rotation: $R_{xy,\theta} = R_{x,\theta_x} R_{y,\theta_y}$, then followed by either $\sigma_x$ or $\sigma_y$ or $\sigma_z$ error with probability $p_x,p_y,p_z$, respectively.
The noisy $T$-step process $\mathbf{\Lambda}_{\theta}$ consists of $T$ copies of $R_{xy,\theta}$ and can be illustrated as
\begin{equation}
\begin{aligned}
    &\mathbf{\Lambda}_{\theta}
        =
        \left(
            \sum_{j\in\{0,x,y,z\}}
            p_j\kketbra{\sigma_jR_{xy,\theta}}
        \right)^{\otimes T} \\
    &= \sum_{j_1,\ldots,j_T\in\{0,x,y,z\}} \prod_{t=1}^T p_{j_t}
    \begin{quantikz}[classical gap=0.05cm]
        & \gate[1,style={fill=green!20}]{R_{xy,\theta}} \wire[l][1]["{A_0}"{above,pos=0.6}]{q} \gshade{1}{2}{} & \gate[style={fill=red!20}]{\sigma_{j_1}} \wire[r][1]["{A_1}"{above,pos=0.5}]{q} & \invgate{\dots} & \gate[1,style={fill=green!20}]{R_{xy,\theta}} \wire[l][1]["{A_{T-1}'}"{above,pos=0.65}]{q} \gshade{1}{2}{} & \gate[style={fill=red!20}]{\sigma_{j_T}} \wire[r][1]["{A_T}"{above,pos=0.5}]{q} \wireoverride{q} & 
    \end{quantikz} \;.
\end{aligned}
\end{equation}
The SDP given in Theorem~\ref{thm:sequential_NH_bound} gives MSEs that are closer to the Heisenberg scaling $\frac{1}{T^2}$ as the error rates approach $0$, as shown in Fig.~\ref{fig:MSE_noisy_xy_rot} up to $T=4$ steps.
Note that, when the Pauli noise has non-zero probability of causing $\sigma_z$ error (blue lines in Fig.~\ref{fig:MSE_noisy_xy_rot}) or when the Pauli noise has a non-zero probabilities $p_x=p_y\neq 0$ of causing either $\sigma_x$ or $\sigma_y$ error (red lines in Fig.~\ref{fig:MSE_noisy_xy_rot}), then Heisenberg scaling cannot be recovered, in agreement with the HNLS condition~\cite{zhou2018achieving}.

\subsection{Noisy qutrit channel}

Here we will demonstrate our SDP in Theorem~\ref{thm:sequential_NH_bound} on a higher-dimensional system, namely a qutrit system.
We consider a single-parameter qutrit unitary $U_\theta = e^{i\theta G_8/2}$ where $G_8$ is one of the qutrit Gell-Mann matrices (see e.g.~\cite{di2012elementarygatesternaryquantum,trisciani2025decompositionmultiqutritgatesgenerated}), given by
\begin{equation}
    G_8 = \sqrt{\frac{1}{3}}
    \begin{bmatrix}
        1&0&0\\
        0&1&0\\
        0&0&-2
    \end{bmatrix} \;.
\end{equation}
Also, consider the qutrit Weyl-Heisenberg operators $\sigma_x,\sigma_z$ defined by
\begin{equation}
    \sigma_x|j\> = |j+1 \text{ (mod 3)}\>
    \quad\text{and}\quad
    \sigma_z|j\> = \omega^j|j\>
\end{equation}
for the third root of unity $\omega=e^{\iota2\pi/3}$ and $\{|j\>\}_{j=0}^2$ the eigenvectors of $\sigma_z$ (also the eigenvectors of $G_8$).

We consider the problem of estimating parameter $\theta$ of the unitary $U_\theta$ followed by a noise channel
\begin{equation}
    \mathcal{N}(\cdot) = \sum_{a,b=0}^2 p_{a,b} E_{a,b} \cdot E_{a,b}^\dag
\end{equation}
that applies an error $E_{a,b}=\sigma_x^a\sigma_z^b$ with probability $p_{a,b}$ (with $p_{0,0}$ being the probability of no errors).
Note that $G_3+G_8$ is not in the span of errors $E_{a,0}$, hence when the noise channel applies only errors $E_{1,0}=\sigma_x$ or $E_{2,0}=\sigma_x^2$ with non-zero probability, then the Heisenberg scaling can be achieved according to the HNLS condition~\cite{zhou2018achieving}.
On the other hand, if there is a non-zero probability of error $E_{a,b}$ for $b\neq 0$, then the MSE must be strictly larger than the Heisenberg scaling~\cite{Boixo_2007}
\begin{equation}
    \frac{1}{T^2||G_8/2||^2} = \frac{4}{3T^2}
\end{equation}
where the semi-norm $||H|| = \lambda_{\max}(H)-\lambda_{\min}(H)$ and $\lambda_{\max}(H)$ (resp. $\lambda_{\min}(H)$) is the maximum (resp. minimum) eigenvalue of $H$.

\begin{figure}
    \centering
    \includegraphics[width=0.55\columnwidth]{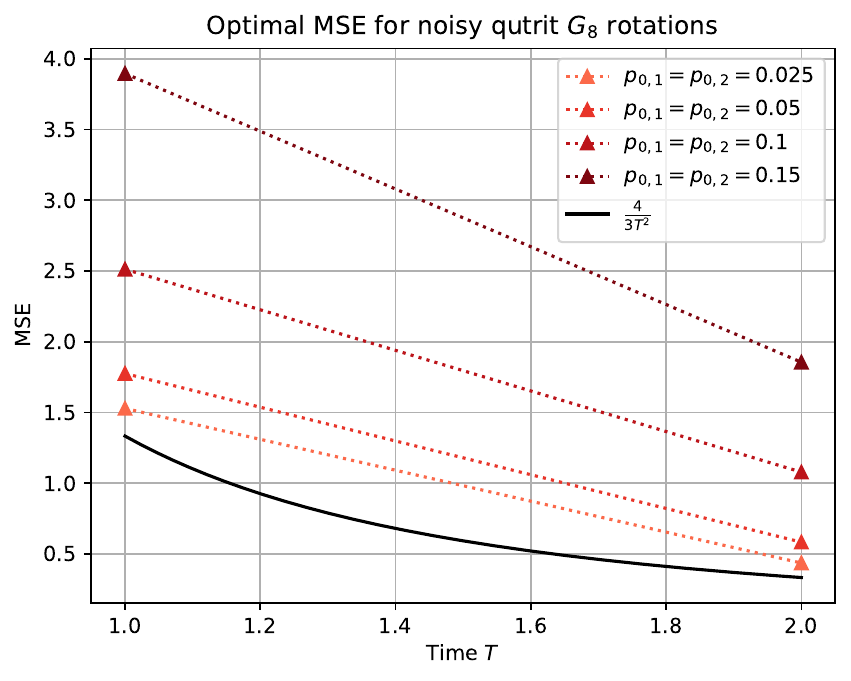}
    \caption{Mean-squared error (MSE) lower bound for noisy qutrit rotations for $T=1$ and $T=2$ steps (red lines).
    Noise channel $\mathcal{N}$ after each rotation $U_\theta$ applies either $\sigma_z$ or $\sigma_z^2$ error with (independent) probability of $p_{0,1}$ and $p_{0,2}$, respectively, where $p=p_{0,1}=p_{0,2}$. 
    The MSEs get closer to the Heisenberg scaling $4/(3T^2)$ (black line) as the error probability $p$ gets smaller.
    }
    \label{fig:MSE_qutrit}
\end{figure}

In Fig.~\ref{fig:MSE_qutrit}, we show that the SDP bound given in Theorem~\ref{thm:sequential_NH_bound} gives MSEs for estimating the parameters of $U_\theta$ followed by noise channel $\mathcal{N}$ with only non-zero probabilities $p=p_{0,1}=p_{0,2}$ of either $\sigma_z$ or $\sigma_z^2$ errors.
In this case, the resulting SDP bounds show that the MSE must be strictly larger than the Heisenberg scaling, since $G_8$ is in the linear span of $\sigma_z$ and $\sigma_z^2$.
Otherwise, if the only errors happening with non-zero probabilities are $\sigma_x$ or $\sigma_x^2$, Heisenberg scaling can be achieved.

\section{Optimal sequential strategy for $Z$-rotation with independent Pauli noise}\label{app:independent_noise_z_rot}

\begin{atenv}
In this appendix, we give the initial probe state and intermediate isometry for the two-step independent Pauli-noise examples in Appendix~\ref{sec:indep_noise_zrot}.
These examples complement the depolarizing-noise calculations in Section~\ref{sec:example_multiparameter_rotations_depolarizing} and illustrate the construction of explicit single-parameter strategies described in Section~\ref{sec:single_par_tight_bound}.
The isometry below applies to the dephasing and equal-$X/Y$ noise cases considered there.
The single-axis transverse-noise cases are discussed separately in Appendix~\ref{sec:indep_noise_zrot}.
\end{atenv}

The optimal sequential probe $\mathbf{M}$ acting on the noisy $Z$ rotations can be illustrated as
\begin{equation}\label{eqn:zrot_correlated_noise}
    \begin{quantikz}[classical gap=0.05cm]
        \;\; \setwiretype{n} \bshade{2}{8}{$\mathbf{M}$} & \lstick[2]{$|\psi\>$} & \setwiretype{q} & & \gate[2,style={fill=blue!20}]{V} & & & \gate[2,style={fill=blue!20}]{\mathcal{M}}  \\
        \setwiretype{n} & & \gate[1,style={fill=green!20}]{R_z(\theta)} \gshade{1}{2}{} \setwiretype{q} & \gate[style={fill=red!20}]{\mathcal{N}_p} & & \gate[1,style={fill=green!20}]{R_z(\theta)} \gshade{1}{2}{} & \gate[style={fill=red!20}]{\mathcal{N}_p} & 
    \end{quantikz} ,
\end{equation}
where the initial probe state is given by $|\psi\> = (|00\>+|11\>)/\sqrt{2}\in C_1\otimes A_0$ and isometry $V:C_1\otimes A_1\rightarrow C_2\otimes A_1'$ performed at time step 1 takes the form of
\begin{equation}
\begin{aligned}
    V = |0000\>\<00| + |1111\>\<11| + \frac{|0010\>+|1011\>}{\sqrt{2}}\<01| + \frac{|0100\>+|1101\>}{\sqrt{2}}\<10|
    =
    \begin{quantikz}
        & \gate[1,style={blue!05,fill=blue!05}]{|0\>} \setwiretype{n} \gategroup[4,steps=4,style={dashed, fill=blue!05, inner xsep=2pt},background,label style={label position=above,anchor=north,yshift=0.3cm}]{{\sc $V$}} & \gate[]{H} \setwiretype{q} & \gate[]{H} & \ctrl{3} & \rstick[3]{$C'$} \\
        \lstick[1]{$C$} & & & \ctrl{-1} & & \\
        \lstick[1]{$A_1$} & & \ctrl{-2} & & & \\
        & \gate[1,style={blue!05,fill=blue!05}]{|0\>} \setwiretype{n} & \setwiretype{q} & & \targ{} & \rstick[1]{$A_1'$}
    \end{quantikz} ,
\end{aligned}
\end{equation}
which is precisely the same isometry in eqn.~\eqref{eqn:app_optimal_isometry_z_rot} for noiseless $Z$-rotation.

\section{Sequential estimation with correlated noise}

In this appendix, we give explicit finite-step strategies for the correlated-noise examples in Section~\ref{sec:correlated_noise_z_rot}.
We first describe a classical mixture of noise histories, where a label is sampled once and determines the noise throughout the protocol.
We then specialize to persistent Pauli errors and give initial states, intermediate isometries, and final measurements.
These constructions illustrate the single-parameter attainability result in Proposition~\ref{prop:single_par_achievability}.

Here, we consider estimation of $T$ copies of channel $\Lambda_\theta$ over $T$ time steps with (classically) correlated noise:
\begin{equation}\label{eqn:correlated_noise}
    \sum_j p_j \left(
    \begin{quantikz}[classical gap=0.05cm]
        \setwiretype{q} \lstick[2]{$\rho$} \bshade{2}{9}{$\mathbf{M}_k$} & & & \gate[2,style={fill=blue!20}]{\mathcal{M}^{(1)}} & & & \gate[2,style={fill=blue!20}]{\mathcal{M}^{(2)}} & \invgate{\dots} & & & \gate[2,style={fill=blue!20}]{\mathcal{M}_k^{(T)}} \wireoverride{q}  \\
        \setwiretype{q} & \gate[1,style={fill=green!20}]{\Lambda_\theta} \gshade{1}{2}{} & \gate[style={fill=red!20}]{E_j^{(1)}} & & \gate[1,style={fill=green!20}]{\Lambda_\theta} \gshade{1}{2}{} & \gate[style={fill=red!20}]{E_j^{(2)}} & & \invgate{\dots} & \gate[1,style={fill=green!20}]{\Lambda_\theta} \gshade{1}{2}{} & \gate[style={fill=red!20}]{E_j^{(T)}} &
    \end{quantikz} \right) \;.
\end{equation}
The $T$ copies of noisy channels are described by the channel operator
\begin{equation}
    \mathbf{\Lambda}_\theta = 
    \sum_j p_j \left(
    \begin{quantikz}[classical gap=0.05cm]
        & \gate[1,style={fill=green!20}]{\Lambda_\theta} \wire[l][1]["{A_0}"{above,pos=0.6}]{q} \gshade{1}{2}{} & \gate[style={fill=red!20}]{E_j^{(1)}} & \wire[l][1]["{A_1}"{above,pos=0}]{q} & \wire[r][1]["{A_1'}"{above,pos=0}]{q} \setwiretype{n} & \gate[1,style={fill=green!20}]{\Lambda_\theta} \gshade{1}{2}{} & \gate[style={fill=red!20}]{E_j^{(2)}} \wireoverride{q} & \wire[l][1]["{A_2}"{above,pos=0}]{q} & \invgate{\dots} & \wire[r][1]["{A_{T-1}'}"{above,pos=0}]{q} & \gate[1,style={fill=green!20}]{\Lambda_\theta} \gshade{1}{2}{} & \gate[style={fill=red!20}]{E_j^{(T)}} \wireoverride{q} & \wire[l][1]["{A_T}"{above,pos=0}]{q}
    \end{quantikz} \right) \;.
\end{equation}

\subsection{Optimal sequential strategy for $Z$-rotation with correlated Pauli noise}\label{app:correlated_noise_z_rot}

Here we consider $Z$-rotation with angle $\theta$ over $T=2$ time steps and correlated noise process with parameters $(p_X,p_Y,p_Z) = (0.1,0.1,0)$ as discussed in Section~\ref{sec:correlated_noise_z_rot} of the main text.
The optimal estimator matrix is given by
\begin{equation}
\begin{gathered}
    \mathbf{X}_{0,1} = \frac{1}{4} \kketbra{M_1} - \frac{1}{4} \kketbra{M_2} + \frac{1}{4}\kketbra{M_3} - \frac{1}{4}\kketbra{M_4} \\
    \text{where}\\
    |M_1\rr = \frac{|0^4\>-i|1^4\>}{\sqrt{2}} \;,\quad
    |M_2\rr = \frac{|0^4\>+i|1^4\>}{\sqrt{2}} \;,\quad
    |M_3\rr = \frac{|0101\>-i|1010\>}{\sqrt{2}} \;,\quad
    |M_4\rr = \frac{|0101\>+i|1010\>}{\sqrt{2}} \;,
\end{gathered}
\end{equation}
and the optimal strategy described by
\begin{equation}
    \mathbf{M} = \frac{1}{2}\Big( \kketbra{M_1} + \dots + \kketbra{M_4} + \ketbra{0010} + \ketbra{0111} + \ketbra{1000} + \ketbra{1010} + \ketbra{1101} \Big) \;.
\end{equation}
The optimal strategy $\mathbf{M}$ acting on the noisy $Z$ rotations can be illustrated as
\begin{equation}\label{eqn:zrot_correlated_noise2}
    \sum_{j\in\{0,x,y\}} p_j \left(
    \begin{quantikz}[classical gap=0.05cm]
        \;\; \setwiretype{n} \bshade{2}{8}{$\mathbf{M}$} & \lstick[2]{$|\psi\>$} & \setwiretype{q} & & \gate[2,style={fill=blue!20}]{V} & & & \gate[2,style={fill=blue!20}]{\mathcal{M}}  \\
        \setwiretype{n} & & \gate[1,style={fill=green!20}]{R_z(\theta)} \gshade{1}{2}{} \setwiretype{q} & \gate[style={fill=red!20}]{\sigma_j} & & \gate[1,style={fill=green!20}]{R_z(\theta)} \gshade{1}{2}{} & \gate[style={fill=red!20}]{\sigma_j} & 
    \end{quantikz}  \right) \;,
\end{equation}
where the initial probe state is given by $|\psi\> = (|00\>+|11\>)/\sqrt{2}\in C\otimes A_0$ and isometry $V:C\otimes A_1\rightarrow C'\otimes A_1'$ performed at time step 1 takes the form of
\begin{equation}
    V = |0^4\>\<00| + |1011\>\<01| + |0100\>\<10| + |1^4\>\<11| 
    =
    \begin{quantikz}
        & \gate[1,style={blue!05,fill=blue!05}]{|0\>} \setwiretype{n} \gategroup[4,steps=3,style={dashed, fill=blue!05, inner xsep=2pt},background,label style={label position=above,anchor=north,yshift=0.3cm}]{{\sc $V$}} & \setwiretype{q} & \targ{} & \rstick[3]{$C'$} \\
        \lstick[1]{$C$} & & & & \\
        \lstick[1]{$A_1$} & & \ctrl{1} & \ctrl{-2} & \\
        & \gate[1,style={blue!05,fill=blue!05}]{|0\>} \setwiretype{n} & \targ{} \setwiretype{q} & & \rstick[1]{$A_1'$}
    \end{quantikz}
    \;,
\end{equation}
whereas POVM measurement $\mathcal{M}$ consists of $9$ POVM elements 
\begin{equation}
\begin{gathered}
    M_k' = \sqrt{\mathbf{M}^+}\mathbf{M}_k\sqrt{\mathbf{M}^+} = \kketbra{M_k} \;,\; k\in\{1,\dots,4\} 
    \\ \text{and}\\
    M_k' = \sqrt{\mathbf{M}^+}\frac{\ketbra{abcd}}{2}\sqrt{\mathbf{M}^+} + Q_j \;,\; k\in\{5,\dots,9\}
\end{gathered}
\end{equation}
where $|abcd\>$ are the last five eigenvectors of $\mathbf{M}$ and $Q_k$ are positive semidefinite operators chosen such that the POVM completeness condition $\sum_k M_k'=I$ holds.

The total state after the first noisy $Z$ rotation with noise unitary $E\in\{I,X,Y\}$ before isometry $V$ at time step 1 is applied is given by
\begin{equation}\label{eqn:noisy_zrot_state_t1}
    (I\otimes ER_z(\theta))|\psi\> 
    = 
    \begin{cases}
        \frac{e^{-i\theta/2}|00\> + e^{i\theta/2}|11\>}{\sqrt{2}}  \quad&,\;\text{if }E=I\\[0.3ex]
        \frac{e^{-i\theta/2}|01\> + e^{i\theta/2}|10\>}{\sqrt{2}}  \quad&,\;\text{if }E=X \\[0.3ex]
        \frac{ie^{-i\theta/2}|01\> - ie^{i\theta/2}|10\>}{\sqrt{2}}  \quad&,\;\text{if }E=Y
    \end{cases} \;,
\end{equation}
so that after $V$, we obtain the state
\begin{equation}
    V(I\otimes ER_z(\theta))|\psi\> = 
    \begin{cases}
        \frac{e^{-i\theta/2}|0^4\> + e^{i\theta/2}|1^4\>}{\sqrt{2}} \quad&,\;\text{if }E=I \\[0.3ex]
        \frac{e^{-i\theta/2}|1011\> + e^{i\theta/2}|0100\>}{\sqrt{2}} \quad&,\;\text{if }E=X \\[0.3ex]
        \frac{ie^{-i\theta/2}|1011\> - ie^{i\theta/2}|0100\>}{\sqrt{2}} \quad&,\;\text{if }E=Y
    \end{cases} \;.
\end{equation}
Since the noise at time steps 1 and 2 is correlated, the state of the system for a given noise $E$ before measurement $\mathcal{M}$ at time step 2 is
\begin{equation}
    |\psi_E\> = (I\otimes ER_z(\theta))V(I\otimes ER_z(\theta))|\psi\> = 
    \begin{cases}
        \frac{e^{-i\theta}|0^4\> + e^{i\theta}|1^4\>}{\sqrt{2}} \quad&,\;\text{if }E=I \\[0.3ex]
        \frac{e^{-i\theta}|1010\> + e^{i\theta}|0101\>}{\sqrt{2}} \quad&,\;\text{if }E=X \\[0.3ex]
        \frac{-e^{-i\theta}|1010\> - e^{i\theta}|0101\>}{\sqrt{2}} \quad&,\;\text{if }E=Y 
    \end{cases} \;.
\end{equation}
Note that only the first four measurement outcomes are nonzero:
\begin{equation}
\begin{gathered}
    \<\psi_E|M_k'|\psi_E\>|^2 =
    \begin{cases}
        \frac{1-\cos(2\theta+\frac{\pi}{2})}{2} \quad&,\;\text{if }k=1,E=I \\[0.3ex]
        \frac{1+\cos(2\theta+\frac{\pi}{2})}{2} \quad&,\;\text{if }k=2,E=I \\[0.3ex]
        \frac{1-\cos(2\theta+\frac{\pi}{2})}{2} \quad&,\;\text{if }k=3,E=X \\[0.3ex]
        \frac{1+\cos(2\theta+\frac{\pi}{2})}{2} \quad&,\;\text{if }k=4,E=X \\[0.3ex]
        \frac{1-\cos(2\theta+\frac{\pi}{2})}{2} \quad&,\;\text{if }k=3,E=Y \\[0.3ex]
        \frac{1+\cos(2\theta+\frac{\pi}{2})}{2} \quad&,\;\text{if }k=4,E=Y \\[0.3ex]
        0 \quad&,\;\text{otherwise}
    \end{cases} \;.
\end{gathered}
\end{equation}

Now we compare the optimal strategy for $T=2$ time steps above with the optimal strategy for $T=1$ time-step $Z$-rotation with $\theta=0.01$ and $(p_X,p_Y,p_Z)=(0.1,0.1,0)$.
In this case, the optimal estimator $\mathbf{X}_{0,1}$ and sequential probe $\mathcal{M}$ are given by
\begin{equation}
\begin{gathered}
    \mathbf{X}_{0,1} = 0.625(-\kketbra{M_1}+\kketbra{M_2}) \;\text{and}\\
    \mathbf{M} = \frac{1}{2}(\kketbra{M_1}+\kketbra{M_2} + \kketbra{M_3}+\kketbra{M_4}) \\
    \text{for}\\
    |M_1\rr = \frac{i|00\>+|11\>}{\sqrt{2}} \;,\quad
    |M_2\rr = \frac{-i|00\>+|11\>}{\sqrt{2}} \;,\quad
    |M_3\rr=|01\> \;,\;\text{and}\quad
    |M_4\rr=|10\> \;.
\end{gathered}
\end{equation}
Strategy $\mathbf{M}$ can be decomposed as an initial state $|\psi\>$ and a POVM measurement $\mathcal{M}$ illustrated as
\begin{equation}\label{eqn:zrot_T1_correlated_noise}
    \sum_E p_E \left(
    \begin{quantikz}[classical gap=0.05cm]
        \setwiretype{q} \lstick[2]{$|\psi\>$} & & & \gate[2,style={fill=blue!20}]{\mathcal{M}}  \\
        \setwiretype{q} & \gate[1,style={fill=green!20}]{R_z(\theta)} & \gate[style={fill=red!20}]{E} & 
    \end{quantikz} \right) \;.
\end{equation}
Since $\sqrt{\mathbf{M}^+} = \sqrt{2}(\kketbra{M_1}+\kketbra{M_2} + \kketbra{M_3}+\kketbra{M_4})$, we can take the POVM measurement $\mathcal{M} = \{M_k'\}_{k=1}^4$ to be
\begin{equation}
\begin{gathered}
    M_k = \sqrt{\mathbf{M}^+}\kketbra{M_k}\sqrt{\mathbf{M}^+} = \kketbra{M_k} \;,
\end{gathered}
\end{equation}
for initial state $|\psi\> = (|00\>+|11\>)/\sqrt{2}$.
We can use the state $|\psi_E^{(1)}\> = E(I\otimes R_z(\theta))|\psi\>$ after the noisy $Z$-rotation in eqn.~\eqref{eqn:noisy_zrot_state_t1} given error $E\in\{I,X,Y\}$ to obtain the probability of each outcome from $\mathcal{M}$:
\begin{equation}
\begin{aligned}
    \<\psi_E^{(1)}|M_k'|\psi_E^{(1)}\>|^2 =
    \begin{cases}
        \frac{|-ie^{-i\theta} + 1|^2}{4} = \frac{1-\cos(\theta+\frac{\pi}{2})}{2} \quad&,\;\text{if }k=1,E=I \\[0.3ex]
        \frac{|ie^{-i\theta} + 1|^2}{4} = \frac{1+\cos(\theta+\frac{\pi}{2})}{2} \quad&,\;\text{if }k=2,E=I \\[0.3ex]
        \frac{|e^{-i\theta}|^2}{2} = \frac{1}{2} \quad&,\;\text{if }k=3,E=X \\[0.3ex]
        \frac{|e^{i\theta}|^2}{2} = \frac{1}{2} \quad&,\;\text{if }k=4,E=X \\[0.3ex]
        \frac{|e^{-i\theta}|^2}{2} = \frac{1}{2} \quad&,\;\text{if }k=3,E=Y \\[0.3ex]
        \frac{|e^{i\theta}|^2}{2} = \frac{1}{2} \quad&,\;\text{if }k=4,E=Y \\[0.3ex]
        0 \quad&,\;\text{otherwise}
    \end{cases} \;.
\end{aligned}
\end{equation}

For noise parameters $(p_x,p_y,p_z) = (0,0,0.1)$ over $T=2$ time steps, the optimal estimator matrix is given by
\begin{equation}
\begin{gathered}
    \mathbf{X}_{0,1} = \frac{1}{4} \kketbra{M_1} - \frac{1}{4} \kketbra{M_2} \\
    \text{where}\\
    |M_1\rr = \frac{|0^4\>-i|1^4\>}{\sqrt{2}} \;,\quad
    |M_2\rr = \frac{|0^4\>+i|1^4\>}{\sqrt{2}} \;,
\end{gathered}
\end{equation}
and the optimal strategy is described by
\begin{equation}
\begin{aligned}
    \mathbf{M} &= \kketbra{M_1} + \kketbra{M_2} + \frac{1}{4}\Big( + \ketbra{0001} + \ketbra{0010} + \ketbra{0101} + \ketbra{0110} + 2\ketbra{0111} \\
        &+ 2\ketbra{1000} + \ketbra{1001} + \ketbra{1010} + \ketbra{1101} + \ketbra{1110} \Big) \\
    &= \sum_{E\in\{I,Z\}} p_E \left(
    \begin{quantikz}[classical gap=0.05cm]
        \setwiretype{q} \lstick[2]{$|\psi\>$} & & & \gate[2,style={fill=blue!20}]{V} & & & \gate[2,style={fill=blue!20}]{\mathcal{M}}  \\
        \setwiretype{q} & \gate[1,style={fill=green!20}]{R_z(\theta)} & \gate[style={fill=red!20}]{E} & & \gate[1,style={fill=green!20}]{R_z(\theta)} & \gate[style={fill=red!20}]{E} & 
    \end{quantikz} \right)
    \;.
\end{aligned}
\end{equation}

The intermediate isometry is given by
\begin{equation}
\begin{aligned}
    V &= |0^4\>\<00| + \frac{|0010\>+|1011\>}{\sqrt{2}}\<01| + \frac{|0100\>+|1101\>}{\sqrt{2}}\<10| + |1^4\>\<11| \;.
\end{aligned}
\end{equation}
So, the state after the first noisy $Z$ rotation is given by
\begin{equation}
\begin{aligned}
    (I\otimes ER_z(\theta)) = 
    \begin{cases}
        \frac{e^{-i\theta/2}|00\> + e^{i\theta/2}|11\>}{\sqrt{2}}  \quad&,\;\text{if }E=I\\[0.3ex]
        \frac{e^{-i\theta/2}|00\> - e^{i\theta/2}|11\>}{\sqrt{2}} = \frac{e^{-i\theta/2}|00\> + e^{i(\theta/2+\pi)}|11\>}{\sqrt{2}}  \quad&,\;\text{if }E=Z
    \end{cases} \;.
\end{aligned}
\end{equation}
Then the state after isometry $V$ is given by
\begin{equation}
\begin{aligned}
    V(I\otimes ER_z(\theta)) = 
    \begin{cases}
        \frac{e^{-i\theta/2}|0000\> + e^{i\theta/2}|1111\>}{2}  \quad&,\;\text{if }E=I\\[0.3ex]
        \frac{e^{-i\theta/2}|0000\> + e^{i(\theta/2+\pi)}|1111\>}{2}  \quad&,\;\text{if }E=Z
    \end{cases} \;.
\end{aligned}
\end{equation}
Then finally we obtain the state before the measurement:
\begin{equation}
\begin{aligned}
    |\psi_E\> = (I\otimes ER_z(\theta))V(I\otimes ER_z(\theta)) = 
    \begin{cases}
        \frac{e^{-i\theta}|0000\> + e^{i\theta}|1111\>}{2}  \quad&,\;\text{if }E=I\\[0.3ex]
        \frac{e^{-i\theta}|0000\>  + e^{i\theta}|1111\>}{2}  \quad&,\;\text{if }E=Z
    \end{cases} \;,
\end{aligned}
\end{equation}
which gives outcome probability distribution
\begin{equation}
\begin{aligned}
    |\ll M_k|\psi_E\>|^2 =
    \begin{cases}
        \frac{1-\cos(2\theta)}{2} \quad&,\;\text{if }k=1,E=I \\[0.3ex]
        \frac{1+\cos(2\theta)}{2} \quad&,\;\text{if }k=2,E=I \\[0.3ex]
        \frac{1-\cos(2\theta)}{2} \quad&,\;\text{if }k=1,E=Z \\[0.3ex]
        \frac{1+\cos(2\theta)}{2} \quad&,\;\text{if }k=2,E=Z \\[0.3ex]
        0 \quad&,\;\text{otherwise}
    \end{cases} \;.
\end{aligned}
\end{equation}

\section{Nagaoka-Hayashi bound}\label{app:NH_bound}

\begin{atenv}
Section~\ref{sec:existing_bounds} reviews the Nagaoka--Hayashi bound for a fixed output-state model $\rho_\theta$, whereas Theorem~\ref{thm:sequential_NH_bound} in Section~\ref{sec:sequential_bound} extends this bound to arbitrary sequential strategies of a parameterized multi-time process.
Here we verify explicitly that the latter reduces to the former when $T=1$ and the initial probe state $\rho$ is fixed.
In this specialization, $\mathbf{\Lambda}_{\theta,\rho}=\Lambda_\theta(\rho)=\rho_\theta$ and there are no intermediate control operations.
\end{atenv}

For a quantum channel $\Lambda_\theta:\L(A)\rightarrow\L(A)$ parameterized by $\theta=(\theta_1,\dots,\theta_m)$, 
a single measurement $\mathcal{M}$ and an initial probe state $\rho$ to estimate $\theta$ are illustrated as
\begin{equation}\label{eqn:NH_circuit}
\begin{quantikz}[classical gap=0.05cm]
    \setwiretype{q} \lstick{$\rho$} & \gate[1,style={fill=green!20}]{\Lambda_\theta} & \gate[1,style={fill=blue!20}]{\mathcal{M}_k} 
\end{quantikz} .
\end{equation}
By denoting $\rho_\theta = \Lambda_\theta(\rho)$ and for measurement $\mathcal{M}$ with POVM element $M_k^\dag M_k$ corresponding to outcome $k$, we have 
\begin{equation}
\begin{aligned}
    \mathbf{\Lambda}_{\theta,\rho} &=
    \begin{quantikz}[classical gap=0.05cm]
        \setwiretype{n} & \lstick{$\rho$} \wire[r][1][]{q} & \gate[1,style={fill=green!20}]{\Lambda_\theta} & \wire[l][1]["{A_1}"{above,pos=0}]{q} 
    \end{quantikz} = \rho_\theta \;.
\end{aligned}
\end{equation}
The measurement map $\mathcal{M}_k = M_k\cdot M_k^\dag$ corresponding to outcome $k$ can be described by a positive semidefinite operator
\begin{equation}
    \mathbf{M}_k =
    \begin{quantikz}[classical gap=0.05cm]
        \setwiretype{n} & \wire[r][1]["A_1"{above,pos=0}]{q} & \gate[1,style={fill=blue!20}]{\mathcal{M}_k}  
    \end{quantikz} \;.
\end{equation}
So, we can write the $\mathbf{L}$ matrix and estimator matrix $X$ corresponding to measurement $\mathcal{M}$ and estimator $\hat{\theta}_j=\sum_k\hat{\theta}_{j,k}p(k|\theta,\mathcal{M})$ as
\begin{equation}
\begin{aligned}
    \mathbf{L}_\theta(\mathcal{M},\hat{\theta})
    &= \sum_{i,j=1}^m |i\>\<j| \otimes \sum_k \xi_{i,k}\xi_{j,k} \mathbf{M}_k 
    \quad\text{and}\quad
    X_\theta(\mathcal{M},\hat{\theta}) 
    = \sum_{i=1}^m |i\> \otimes \sum_k \xi_{i,k} \mathbf{M}_k \;,
\end{aligned}
\end{equation}
where $\xi_{i,k} = (\hat{\theta}_{i,k}-\theta_i)$.
For this fixed-state strategy, define
\begin{equation}
    \overline{\mathbf{\Lambda}}_{\theta,\rho}
    :=
    \begin{bmatrix}
        0 & 0\\
        0 & I_m\otimes\mathbf{\Lambda}_{\theta,\rho}
    \end{bmatrix}.
\end{equation}
Thus the SDP bound in Theorem~\ref{thm:sequential_NH_bound} becomes
\begin{equation}
\begin{aligned}
    \mse_\theta(\mathcal{M},\hat{\theta}) \geq \min_\mathbf{X} \bigg\{ \tr\Big( \overline{\mathbf{\Lambda}}_{\theta,\rho} \mathbf{X} \Big) :  
        &\; \mathbf{X} = 
        \begin{bmatrix}
            I & X^\dag \\
            X & \mathbf{L}
        \end{bmatrix}\geq0 
        \;,\\
        &\qquad X = \sum_{i=1}^m |i\>\otimes \mathbf{X}_{0,i} \;,\quad
                \mathbf{X}_{0,i} = \mathbf{X}_{0,i}^\dag \;, \\
        &\qquad \mathbf{L} = \sum_{i,j=1}^m |i\>\<j|\otimes \mathbf{X}_{i,j} \;,\quad
            \mathbf{X}_{i,j} = \mathbf{X}_{i,j}^\dag \;, \\
        &\qquad \tr(\mathbf{\Lambda}_{\theta,\rho} \mathbf{X}_{0,j})=0
        \;,\quad \frac{\partial}{\partial\theta_i} \tr\Big( \mathbf{\Lambda}_{\theta,\rho} \mathbf{X}_{0,j} \Big) = \delta_{i,j} \bigg\} \;,
\end{aligned}
\end{equation}
which is precisely the Nagaoka-Hayashi bound as proposed in~\cite{conlon2021efficient}.

\end{document}